\documentclass[aps,prb,twocolumn,showpacs,preprintnumbers,amsmath,amssymb,superscriptaddress,longbibliography,nofootinbib]{revtex4-2}
\usepackage{hyperref}
\hypersetup{colorlinks=true, citecolor=blue, urlcolor=cyan, linkcolor=blue}
\usepackage{xcolor}
\usepackage{graphicx}
\usepackage{dcolumn}
\usepackage{bm}
\usepackage{braket}
\usepackage{amsthm}
\usepackage[caption=false]{subfig}
 \newcommand{\be}{\begin{equation}}
\newcommand{\ee}{\end{equation}} 
\colorlet{darkred}{red!70!black}

\newcommand{\supp}{\operatorname{supp}}
\newcommand{\CharG}{\operatorname{Char}(G)}

\newcommand{\ex}[1]{\mathrm{e}^{#1}}
\newcommand{\fr}{\frac}

\newcommand{\ca}[1]{\mathcal{#1}}
\newcommand{\bb}[1]{\mathbb{#1}}
\newcommand{\abs}[1]{\left|#1\right|}

\def\tr{{\mathrm{Tr}}}

\usepackage[normalem]{ulem}

\newcommand{\taust}{\tau_{\mathrm{strong}}}

\newcommand{\VS}{V^{\mathrm{sym}}}
\newcommand{\VNS}{V^{\mathrm{n-sym}}}
\newcommand\Ssa{A_{G,\mathrm{strong}}}
\newcommand\Sa{A_{G}}

\newcommand\Tr{\mathrm{Tr}}

\newcommand\calB{{\cal B}}

\newcommand\calD{{\cal D}}
\newcommand\calE{{\cal E}}
\newcommand\calF{{\cal F}}
\newcommand\calG{{\cal G}}
\newcommand\calH{{\cal H}}
\newcommand\calI{{\cal I}}

\newcommand\calK{{\cal K}}
\newcommand\calL{{\cal L}}
\newcommand\calM{{\cal M}}

\newcommand\calO{{\cal O}}

\newcommand\calQ{{\cal Q}}
\newcommand\calR{{\cal R}}
\newcommand\calS{{\cal S}}
\newcommand\calT{{\cal T}}
\newcommand\calU{{\cal U}}
\newcommand\calV{{\cal V}}

\newtheorem{definition}{Definition}
\newtheorem{theorem}{Theorem}
\newtheorem{lemma}{Lemma}
\newcommand{\eq}[1]{\begin{align} #1 \end{align}}
\newcommand{\ii}{\mathrm{i}}

\newenvironment{proofof}[1]{%
  \begin{proof}[\textbf{Proof of #1}]%
}{%
  \end{proof}%
}

\makeatletter
\let\cat@comma@active\@empty
\makeatother

\begin{document}
\preprint{KYUSHU-HET-344}
\preprint{RIKEN-iTHEMS-Report-26}

\title{Resource-Theoretic Quantifiers of Weak and Strong Symmetry Breaking:\\ Strong Entanglement Asymmetry and Beyond}

\author{Yuya Kusuki}
\email[]{kusuki.yuya@phys.kyushu-u.ac.jp}
\affiliation{Institute for Advanced Study, 
Kyushu University, Fukuoka 819-0395, Japan.}
\affiliation{Department of Physics, 
Kyushu University, Fukuoka 819-0395, Japan.}
\affiliation{RIKEN Interdisciplinary Theoretical and Mathematical Sciences (iTHEMS),
Wako, Saitama 351-0198, Japan.}

\author{Sridip Pal}
\email[]{sridip@ihes.fr}
\affiliation{Institut des Hautes Études Scientifiques (IHES), 91440 Bures-sur-Yvette, France}

\author{Hiroyasu Tajima}
\email[]{hiroyasu.tajima@inf.kyushu-u.ac.jp}
\thanks{\\Three authors contributed equally to this work.}
\affiliation{Department of Informatics, Faculty of Information Science and Electrical Engineering, 
Kyushu University, Fukuoka 819-0395, Japan.}
\affiliation{JST, FOREST, 4-1-8 Honcho, Kawaguchi, Saitama, 332-0012, Japan.}

\begin{abstract}
Quantifying \emph{how much} a quantum state breaks a symmetry is essential for characterizing phases, nonequilibrium dynamics, and open-system behavior.
For mixed states, however, conventional diagnostics of weak symmetry breaking can miss a stronger form of symmetry breaking associated with the possibility of exchanging conserved charges with an environment.
This motivates an operational theory of strong symmetry breaking.
We develop a resource theory of strong symmetry breaking by identifying the appropriate free states and free operations, and, for channels on a fixed system, by showing that the latter are precisely the operations realizable without exchanging the conserved charge with the environment.
Within this framework, we systematically construct measures of strong symmetry breaking, including strong entanglement asymmetry and covariance matrices of symmetry generators, for a broad class of symmetry groups.
We further completely characterize i.i.d.\ convertibility between arbitrary states under any compact Lie group strong symmetry.This identifies the combinations of resource-theoretically valid measures that quantify strong symmetry breaking for arbitrary states, including all mixed states. In particular, for $U(1)$ symmetry, if the states have positive variances of the conserved quantity and equal strong-symmetry periods, the conversion rate is completely determined by the variance ratio whenever the input is pure or the output is weak symmetric.
Thus, for these state conversions, the variance of the conserved quantity plays an operational role analogous to that of entanglement entropy in entanglement theory or quantum Fisher information in the resource theory of weak asymmetry.
We further show how weak symmetry breaking is irreversibly converted into strong symmetry breaking in open-system dynamics.
We illustrate the framework with examples from quantum field theory, strong-to-weak spontaneous symmetry breaking, and a strong-symmetry analogue of Mpemba-type dynamics.
\end{abstract}

\maketitle

\section{Introduction \& Summary}

Symmetry breaking is a paradigmatic concept in modern physics, underpinning our understanding of phases of matter, nonequilibrium dynamics, and high-energy physics.
In quantum many-body systems in and out of equilibrium, however, it is not sufficient to ask only whether a symmetry is present or broken.
One needs a quantitative and operational answer to a sharper question: 
\emph{how much a given state breaks a symmetry, and what physical tasks can it enable?}

This question becomes especially subtle in open quantum systems. 
For pure states, the distinction between different notions of symmetry is often invisible. 
For mixed states, however, symmetry naturally splits into two inequivalent notions, as has become increasingly important in recent studies~\cite{Buca2012,Groot2021,Ma2022,Lee2022,Zhang2022,Lee2023,Ogunnaike2023,Ma2023,Lessa2024,Sala2024}. 
The first, often called \emph{weak symmetry}, requires invariance under the adjoint action of the symmetry group,
\begin{equation}
  U_g \rho U_g^\dagger = \rho, \qquad g\in G .
\end{equation}
The second, known as \emph{strong symmetry}, requires the state to remain in a single charge sector,
\begin{equation}
  U_g\rho = e^{i\theta_g}\rho, \qquad g\in G .
\end{equation}
Every strong-symmetric state is weak symmetric, but the converse is false. 
This is not a merely formal distinction. 
Weak symmetry is compatible with exchanging conserved charges with an environment, whereas strong symmetry forbids such exchange. 
Thus, for mixed states, weak and strong symmetry encode genuinely different physical constraints.

This distinction exposes a limitation of existing quantitative diagnostics. 
The entanglement asymmetry, introduced as a subsystem measure of symmetry breaking~\cite{Ares2022}, has been widely used in many-body physics and quantum field theory (QFT), both in nonequilibrium settings~\cite{Murciano2023,Ares2023,Lastres2024,Fossati2024a,Kusuki2024,Chen2024,Ares2023a,Russotto2024} and near or at equilibrium~\cite{Capizzi2023,Capizzi2023a,Fossati2024,Chen2023,Lastres2024,Fossati2024a,Kusuki2024}. 
It has also been applied to anomalous symmetry breaking~\cite{Florio2025}, holographic settings~\cite{Benini:2024xjv}, and, most prominently, dynamical symmetry restoration and the quantum Mpemba effect~\cite{Ares2022,Murciano2023,Ares2024,Rylands2023,Yamashika2024,Turkeshi2024,Liu2024,Liu2024a,Chalas2024,Caceffo2024,Rylands2024,Yamashika2024a,Joshi2024,DiGiulio:2025ems}; see also the review~\cite{Ares:2025onj}. 
From the resource-theoretic viewpoint, the entanglement asymmetry is a measure of weak asymmetry: it coincides with the relative entropy of asymmetry in the standard resource theory of asymmetry~\cite{Vaccaro2008,Gour2009}. 
It therefore detects whether a density matrix is invariant under the adjoint action of the symmetry group, but it is insensitive to whether the state is confined to a single charge sector. 
Consequently, two mixed states can be indistinguishable by conventional weak-asymmetry diagnostics while being physically different from the viewpoint of strong symmetry.

This identification between entanglement asymmetry and a distinguished resource-theoretic monotone has also motivated recent attempts to make resource-theoretic structures explicit in studies of dynamical symmetry restoration and Mpemba-type phenomena. 
For example, Ref.~\cite{Summer2025} investigated Mpemba effects from a resource-theoretic perspective and clarified the relation between thermal and weak-asymmetry Mpemba effects. 
Other recent works in this direction include Refs.~\cite{Aditya2025,Mazzoni2025,Yamashika2025}. 
Our focus is complementary: we formulate a resource theory tailored to strong symmetry and use it to identify the corresponding monotones and a strong-symmetry analog of Mpemba-type crossover.

This motivates the central question of this work:
\begin{quote}
\emph{Can one construct an operational theory that quantifies strong symmetry breaking?}
\end{quote}

We answer this question by developing a resource theory of strong symmetry breaking. 
The guiding principle is that a physically meaningful measure of symmetry breaking should not merely be a convenient diagnostic. 
It should be monotone under operations that do not supply the relevant resource, and, in the i.i.d. asymptotic setting, it should determine the conversion rate at which one resource state can be converted into another. 
This perspective has been central in other resource theories: entanglement entropy determines pure-state i.i.d. asymptotic interconversion~\cite{Bennett1996}, while quantum Fisher information plays an analogous role for $U(1)$ asymmetry and energetic coherence~\cite{Gour2008,Marvian2020,Marvian2022}. 
For finite groups and compact Lie groups, related i.i.d. convertibility structures have also been identified through logarithmic characteristic functions and quantum geometric tensors~\cite{Shitara2023,Yamaguchi2024}. 
Our work establishes the corresponding operational structure for \emph{strong} symmetry breaking.

The first ingredient is the identification of the appropriate free states. 
For ordinary group symmetries, the natural free states are strong-symmetric states. 
For non-Abelian symmetries, strictly strong-symmetric states may be absent even in simple representations. 
We therefore also introduce \emph{single-sector states}, which provide an intermediate notion between weak and strong symmetry: they are weak symmetric but supported on a single irreducible sector. 
For Abelian groups with non-projective unitary representations, the single-sector and strong-symmetric notions coincide.
This gives a unified way to treat strong symmetry breaking across Abelian and non-Abelian settings.

The second ingredient is the identification of the appropriate free operations. 
We define strong-covariant operations by
\begin{equation}
  \Lambda(U_g X)=U'_g\Lambda(X), \qquad g\in G ,
\end{equation}
for arbitrary operators $X$. 
These operations are not introduced merely for formal convenience.
For channels acting on a single system, namely $(U,U)$-strong-covariant maps $\Lambda:\mathcal B(\mathcal H)\to\mathcal B(\mathcal H)$, we prove that strong covariance is equivalent to a Stinespring realization in which the system exchanges no conserved charge with the environment. 
For general channels between possibly different input and output representations, strong covariance is equivalently characterized by Kraus operators satisfying
\begin{equation}
  K_\alpha U_g = U'_g K_\alpha .
\end{equation}
For Markovian open-system dynamics whose Hamiltonian part respects the symmetry, this implies that strong-covariant GKSL evolutions have jump operators commuting with the symmetry. 
This characterization is consistent with the notion of strong-symmetry-preserving dynamics discussed in the condensed-matter and open-system literature~\cite{Lessa2024,Sala2024,Zapusek2025}. 
Thus, the free operations in our resource theory are directly tied to a concrete physical constraint in open quantum dynamics: no exchange of the conserved charge with the environment.

The third ingredient is the construction of resource monotones. 
We introduce several complementary monotones of strong symmetry breaking. 
For general groups, we define a strong analogue of entanglement asymmetry, which combines the weak relative entropy of asymmetry with the entropy of the distribution over charge sectors. 
It is faithful when single-sector states are taken as free states and directly quantifies strong symmetry breaking. It is also faithful with respect to strong-symmetric states whenever the two classes coincide.
We also define an averaged logarithmic characteristic function \(L(\rho)\), inspired by logarithmic characteristic functions that appear in finite-group asymmetry~\cite{Shitara2023}. 
This quantity is faithful to strong symmetry and additive under tensor products.
For compact Lie groups, we show that the non-symmetrized and symmetrized covariance matrices of the symmetry generators are matrix-valued monotones under strong-covariant operations. 
These covariance-based monotones are naturally connected to quantum geometric and Fisher-information structures~\cite{Provost1980,Berry1989,Kudo2023,Yamaguchi2024}; for connected compact Lie groups, they are faithful to strong symmetry. 
These quantities distinguish mixed states that are invisible to weak-asymmetry measures and provide an additional criterion for classifying symmetry breaking in open quantum systems.

We further completely characterize i.i.d.\ asymptotic convertibility between arbitrary finite-dimensional states for non-projective continuous unitary representations of any compact Lie group.
Here, we allow a strong-symmetric ancillary state to remain in the output of a strong-covariant operation, with the error relative to the product of the target state and the ancillary state vanishing asymptotically.
Our approach builds on the analysis of pure-state asymptotic conversion in terms of the quantum geometric tensor (QGT)~\cite{Yamaguchi2024}, as well as the asymptotic bounds for metric adjusted QGTs and the methods for comparing Gaussian models~\cite{YT2026_1}.
Using these results, we extend the analysis of asymptotic conversion to incorporate the exact charge-conservation constraints imposed by strong covariance.
We thereby express the necessary and sufficient conditions for i.i.d.\ asymptotic conversion under strong-covariant operations in terms of single-copy data: proportionality of the non-symmetrized covariance matrices of the symmetry generators and residual mean charges, matrix inequalities for the family of metric adjusted QGTs, and inclusion of weak-symmetry subgroups and equality of strong-symmetry subgroups.
This characterization requires neither connectedness nor commutativity of the group and shows how conservation of charge fluctuations and constraints arising from weak symmetry breaking jointly govern strong-covariant conversion.

As one application of this general theorem, we consider cases in which the conversion rate under $U(1)$ symmetry reduces to the ratio of the variances of the conserved quantity.
Let \(H\) be the conserved charge generating the $U(1)$ action, and let
\begin{equation}
  V_H(\rho):=\tr[\rho H^2]-\tr[\rho H]^2
\end{equation}
be its variance. 
In the following, we assume that the input and output states have positive variances and equal strong-symmetry periods.
Under these assumptions, if the input is pure and the output is arbitrary, or if the input is arbitrary and the output is weak symmetric, the optimal conversion rate is
\begin{equation}
  R(\rho\to\sigma)=\frac{V_H(\rho)}{V_H(\sigma)}.
\end{equation}
Thus, for these state conversions, the variance of the conserved quantity completely characterizes i.i.d.\ convertibility for strong $U(1)$ symmetry breaking.
In this setting, the variance of the conserved quantity plays an operational role analogous to that of entanglement entropy in entanglement theory~\cite{Bennett1996} or quantum Fisher information in the standard resource theory of $U(1)$ asymmetry~\cite{Gour2008,Marvian2020,Marvian2022}.

This complete monotone also clarifies the relation between weak and strong symmetry breaking in open-system dynamics. 
For $U(1)$ symmetry, the variance measures the total strong-symmetry resource, while the SLD quantum Fisher information measures the coherent, weak-asymmetry component. 
The SLD quantum Fisher information and its matrix generalizations are standard resource monotones for weak asymmetry~\cite{Yadin2016,Kudo2023}, and the skew-information interpretation of its quantum component is closely related to the convex-roof characterization of variance~\cite{Yu2013}. 
Using the decomposition
\begin{equation}
  V_H(\rho)=\widetilde I_H(\rho)+C_H(\rho),
\end{equation}
where $\widetilde I_H$ is the SLD skew information and $C_H$ is the classical sector-fluctuation contribution~\cite{Gibilisco2007,Gibilisco2011}, we obtain a quantitative picture of how weak symmetry breaking is converted into strong symmetry breaking. 
Under strong-covariant dynamics, the total variance is conserved, whereas the quantum Fisher information decreases. 
Therefore, such dynamics preserve the total amount of strong symmetry breaking while irreversibly converting its coherent component into classical sector fluctuations. 
The same geometric picture extends to compact Lie groups through covariance matrices and the SLD quantum Fisher information matrix~\cite{Kudo2023}.

The discussion above is based on choosing strong covariant operations as the free operations.
However, part of the literature uses the term strong symmetry in a slightly different sense, namely for channels with the same input and output representation whose Kraus operators transform with a common, possibly nontrivial, phase factor~\cite{Groot2021,Lee2022,Zapusek2025}.
To incorporate such operations into our framework, we introduce quasi-strong covariant operations: CPTP maps satisfying $\Lambda(U_g X)=\xi(g)U'_g\Lambda(X)$ for some continuous one-dimensional unitary representation $\xi$.
Any individual quasi-strong covariant operation reduces to a strong covariant one upon redefining the output representation as $U'_g\mapsto\xi(g)U'_g$.
For fixed input and output representations, the two classes are therefore related by a simple identity: the quasi-strong covariant class is the union, over all characters $\xi$, of the strong covariant classes with output representation $\xi U'$.
This union can be strictly larger than the strong covariant class, because different characters cannot be absorbed simultaneously by a single common rephasing.
On a fixed finite-dimensional system with the same input and output representation, the two classes coincide for connected topological groups, whereas they can differ for disconnected groups.
More importantly, because our free states and resource measures are invariant under such character rephasings, this identity reduces the resource theory with quasi-strong covariant free operations to our results for strong covariant operations even when the two classes differ: the same choices of free states remain valid, and the monotonicity and conservation laws of our resource measures, as well as the optimal i.i.d.\ conversion rate for $U(1)$ symmetry, all carry over unchanged.

The resource-theoretic viewpoint also has immediate consequences for current diagnostics of dynamical symmetry restoration. 
A widely used computational proxy is the second-R\'enyi analogue of entanglement asymmetry, or closely related R\'enyi-type quantities, which have appeared in many studies of entanglement asymmetry and quantum Mpemba phenomena.
Although such proxies are often easier to evaluate than the von Neumann relative-entropy version, we show that the second-R\'enyi proxy is not a resource monotone: even in a single-qubit example with a $Z_2$ symmetry, it can increase under a symmetric operation. 
This result is not a consequence of the strong-symmetry resource theory developed in this work. 
Rather, it is an application of the standard resource theory of weak asymmetry, and it illustrates why the resource-theoretic formulation is needed in the first place. 
Without monotonicity under the relevant free operations, a diagnostic cannot by itself quantify the amount of symmetry breaking or certify Mpemba-type behavior. 
Thus, this example clarifies how resource theory can refine and correct existing diagnostics.

Finally, we illustrate the physical consequences of strong symmetry breaking in analytically tractable examples. 
These include QFT examples such as reduced density matrices in CFT, global quenches in boundary CFT, and thermal states with rotational or global symmetries, as well as strong-to-weak spontaneous symmetry breaking. 
We also discuss a strong-symmetry analogue of Mpemba-type crossover. 
In addition, we outline natural extensions of our monotones to generalized and non-invertible symmetries, building on recent developments in entanglement asymmetry for generalized symmetries~\cite{Ahmad2025,Benini2025}, while leaving a full resource-theoretic formulation of generalized-symmetry asymmetry to future work. 
These examples and extensions show that strong symmetry breaking is not a formal refinement of weak asymmetry, but a physically distinct phenomenon that appears naturally in QFTs.

The rest of the paper is organized as follows.

\begin{itemize}
  \item \textbf{Weak asymmetry and entanglement asymmetry [\S\ref{sec:basic}].}
  We review the resource theory of asymmetry and explain why the entanglement asymmetry is a principled weak-asymmetry measure rather than merely a convenient diagnostic.

  \item \textbf{A limitation of second-R\'enyi proxies [\S\ref{Sec:RRWS}].}
  We show that the second-R\'enyi asymmetry can increase under symmetric operations and hence is not a resource monotone. 
  This motivates the need for operationally well-defined monotones when discussing symmetry restoration and Mpemba-type phenomena.

  \item \textbf{Resource theory of strong symmetry [\S\ref{sec:RStrong}].}
We define strong-symmetric and single-sector free states, introduce strong-covariant operations, and characterize the latter via a Stinespring-type realization, a Kraus representation, and the corresponding GKSL condition for Markovian dynamics.
In particular, the Stinespring characterization shows that, for channels on a single system, strong covariance is equivalent to a realization in which the system exchanges no conserved quantity with the environment.

  \item \textbf{Strong-symmetry monotones and i.i.d. asymptotic convertibility [\S\ref{sec:RM}].}
  We construct a strong analogue of entanglement asymmetry, the averaged logarithmic characteristic function \(L(\rho)\), and the non-symmetrized and symmetrized covariance matrices of symmetry generators. 
  We completely characterize i.i.d.\ convertibility between arbitrary finite-dimensional states for non-projective continuous unitary representations of any compact Lie group, allowing a strong-symmetric ancillary state to remain in the output.
  In particular, for $U(1)$ symmetry, when the input and output states have positive variances and equal strong-symmetry periods, the optimal conversion rate is given by the variance ratio if the input is pure or the output is weak symmetric.

  \item \textbf{Conversion from weak to strong symmetry breaking [\S\ref{sec:RM}].}
  We use the relation between variance, quantum Fisher information, and information geometry to quantify how weak symmetry breaking is irreversibly converted into strong symmetry breaking under strong-covariant dynamics.

  \item \textbf{Generalized symmetries and examples [\S\ref{sec:general}--\S\ref{sec:so}].}
  We outline extensions to generalized symmetries and apply the framework to QFT examples, strong-to-weak spontaneous symmetry breaking, and a strong-symmetry Mpemba-type crossover.
\end{itemize}

\section{Preliminary: Basics of the Resource Theory of Asymmetry}\label{sec:basic}

When one is asked to ``quantify how strong a quantum state breaks a given symmetry'',
a natural first step in many-body and high-energy physics is to look at an order parameter.  
For a symmetry group $G$,
one might consider the expectation value of a symmetry-breaking operator $O$,
\begin{equation}
  m_O(\rho) := \langle O \rangle_\rho = \tr(\rho\, O),
\end{equation}
or closely related correlation functions.
While such quantities are extremely useful for diagnosing phases and symmetry breaking,
they are not satisfactory as general measures of ``how much'' symmetry is broken.
They depend on the particular choice of operator $O$:
two different symmetry-breaking operators may induce different orderings of states, and there is no canonical way to decide which one should define the degree of symmetry breaking.
Moreover, they need not be monotone under symmetric operations:
for instance, coupling to a symmetric environment or coarse-graining can increase $|m_O|$, 
even though no additional  ``symmetry-breaking'' has been supplied.

These limitations motivate a structural approach in which one quantifies how much
``asymmetry'' a state possesses with respect to a symmetry group $G$
by treating asymmetry itself as a resource.
This is the viewpoint of the \emph{resource theory of asymmetry}.

In this section, we therefore briefly review the basic notions of resource theories,
specialized to symmetry breaking.
Many of these concepts are standard in the quantum information literature, but they are perhaps less familiar in the high-energy and condensed-matter communities.
For this reason, we include here a self-contained overview.
In the present work, the ``resource'' will be symmetry breaking.
Our goal in this section is to explain how resource theory provides a principled framework for defining
quantifiers of symmetry breaking that are selected by clear operational and axiomatic requirements, and therefore, all of this section serves as a review of the previous results, unlike other sections.

\subsection{Resource theories}\label{sec:RT_basic}
Let $\ca{H}$ be a Hilbert space and let $\ca{B}(\ca{H})$ denote the set of bounded operators acting on $\ca{H}$.
We write
\begin{equation}
  \ca{S}(\ca{H}) := \left\{ \rho \in \ca{B}(\ca{H}) \,\middle|\, \rho \ge 0,\ \tr\rho = 1 \right\}
\end{equation}
for the set of density operators on $\ca{H}$.
The allowed operations on $\ca{H}$ are described by quantum channels,
i.e.\ completely positive trace-preserving (CPTP) maps
\begin{equation}
  \Phi : \ca{B}(\ca{H}) \to \ca{B}(\ca{H}') ,
\end{equation}
which {\color{black}satisfy the following two properties
\begin{itemize}
  \item[(i)] \emph{Completely positivity:} For any additional Hilbert space $\calH_R$ and any  operator $A\in\calB(\calH\otimes\calH_R)$, 
  \begin{align}
A\ge0\Rightarrow\Phi\otimes\mathrm{id}_R(A)\ge0,
  \end{align}
  where $\mathrm{id}_R$ is the identity map on $\calB(\calH_R)$.
  \item[(ii)] \emph{Trace preserving:} For any $A\in\calB(\calH)$,
  \begin{equation}
    \Tr\Phi(A)=\Tr A .
  \end{equation}
\end{itemize}

Because of these properties, any CPTP map satisfies
$\Phi\otimes\mathrm{id}_R\bigl(\ca{S}(\ca{H}\otimes \calH_R)\bigr) \subseteq \ca{S}(\ca{H}'\otimes \calH_R)$.
}
By the Kraus representation theorem, any CPTP map
$\Phi : \ca{B}(\ca{H}) \to \ca{B}(\ca{H}')$ can be written as
\begin{equation}
  \Phi(X) = \sum_{\alpha} K_\alpha X K_\alpha^\dagger,
  \qquad X \in \ca{B}(\ca{H}),
\end{equation}
where the {\it Kraus operator} $K_\alpha : \ca{H} \to \ca{H}'$ obey
\begin{equation}
  \sum_{\alpha} K_\alpha^\dagger K_\alpha = \mathbf{1}.
\end{equation}

\medskip
\noindent\textbf{Free states and free operations.}
A resource theory specifies which states are considered ``free'' (containing no resource)
and which physical operations are allowed to be performed at no cost.

\medskip
\noindent\emph{Free states.}
A set of free states is a subset
\begin{equation}
  \ca{F}(\ca{H}) \subseteq \ca{S}(\ca{H})
\end{equation}
representing those states that are accessible without consuming the resource.
For instance, in the resource theory of entanglement, $\ca{F}$ consists of separable states;
in the resource theory of asymmetry, $\ca{F}$ will be the set of symmetric states.

\medskip
\noindent\emph{Free operations.}
A set of free operations is a subset $\ca{O}$ of CPTP maps that satisfy
\begin{itemize}
  \item[(i)] \emph{Closure:} $\ca{O}$ is closed under finite composition and contains the identity channel.
  \item[(ii)] \emph{Free preservation:} For any $\Phi \in \ca{O}$ and any $\sigma \in \ca{F}$,
  \begin{equation}
    \Phi(\sigma) \in \ca{F}.
  \end{equation}
\end{itemize}
In words, free operations cannot generate resourceful states from free ones.

\medskip
\noindent\textbf{Resource monotones.}
A central object in a resource theory is a numerical quantity that measures how much resource
a state contains.

For the resource theory $(\ca{F},\ca{O})$,
let $M$ be a function
\begin{equation}
  M : \ca{S}(\ca{H}) \to \mathcal{X},
\end{equation}
where $\mathcal{X}$ is a partially ordered set, typically $\mathcal{X}=[0,\infty)$,
while matrix-valued monotones arise naturally when $\mathcal{X}$ is taken to be a space of positive semidefinite matrices.
The function $M$ is called a \emph{resource monotone} if it satisfies the following conditions:
\begin{itemize}
  \item[(i-a)] \emph{Positivity and vanishing on free states:} 
  $M(\rho) \ge 0$ for all $\rho \in \ca{S}(\ca{H})$, and
  \begin{equation}
    M(\sigma) = 0 \qquad \text{for all } \sigma \in \ca{F}.
  \end{equation}
  \item[(i-b)] \emph{Monotonicity under free operations:} For all $\rho \in \ca{S}(\ca{H})$
  and all $\Phi \in \ca{O}$,
  \begin{equation}
    M\bigl( \Phi(\rho) \bigr) \le M(\rho).
  \end{equation}
\end{itemize}
Here, when $\mathcal{X} = [0,\infty)$, the inequalities appearing in these conditions
are understood as inequalities between real numbers.
When $\mathcal{X}$ is a space of positive semidefinite matrices, the inequalities
are understood as matrix inequalities, where $A \ge B$ means that $A-B$ is positive semidefinite.

Beyond these minimal axioms, additional structural requirements are often imposed.

\begin{itemize}
\item[(ii)]\emph{Faithfulness}:  A resource monotone $M$ is \emph{faithful} if, in addition to (i-a),
\begin{equation}
  M(\rho) = 0 \quad \Longrightarrow \quad \rho \in \ca{F}.
\end{equation}
\item[(iii)]\emph{Convexity:}
\begin{equation}
  M\Bigl(\sum_i \lambda_i \rho_i\Bigr) \le \sum_i \lambda_i M(\rho_i),
  \qquad \lambda_i \ge 0,\ \sum_i \lambda_i = 1,
\end{equation}
expressing that classical mixing of states cannot increase the resource on average.

\item[(iv)]\emph{Strong monotonicity:}

Let $\{\Phi_k\}_{k\in K}$ be completely positive maps such that the associated
quantum--classical (QC) measurement channel $\sum_k\Phi_k(\cdot)\otimes\ket{k}\bra{k}_C$ belongs to the set of free operations $\mathcal{O}$, where $C$ is a ``classical"
register with an orthonormal basis $\{\ket{k}\}$.
Such QC measurement channels may not be allowed as free operations in every resource theory. Indeed, in some resource theories, e.g. quantum thermodynamics, no such free channels are allowed at all. When they are allowed, however, we call the measurement $\{\Phi_k\}_{k\in K}$ \emph{free}.
With respect to free measurements, strong monotonicity is defined as follows.
Starting from a state $\rho$, outcome $k$ occurs with probability
\begin{equation}
  p_k := \tr\bigl[\Phi_k(\rho)\bigr],
\end{equation}
and post-measurement state
\begin{equation}
  \rho_k := \frac{\Phi_k(\rho)}{p_k}.
\end{equation}
Strong monotonicity requires
\begin{equation}
   \sum_k p_k\, M(\rho_k) \le M(\rho).
\end{equation}
We remark that a measure $M$ satisfies strong monotonicity if it satisfies the following two properties for an arbitrary QC state $\sum_kq_k\eta_k\otimes\ket{k}\bra{k}_C$: (a) \textit{Linear convexity}
$M(\sum_kq_k\eta_k\otimes\ket{k}\bra{k}_C)=\sum_kq_kM(\eta_k\otimes\ket{k}\bra{k}_C)$, and (b) \textit{Invariance under classical flags}: $M(\eta_k\otimes\ket{k}\bra{k}_C)=M(\eta_k)$.

\item[(v)]\emph{Additivity.}
\begin{equation}
  M(\rho \otimes \sigma) = M(\rho) + M(\sigma),
\end{equation}
which in particular implies $M(\rho^{\otimes n}) = n\, M(\rho)$.

\item[(vi)]\emph{Completeness for i.i.d.\ convertibility.}
Assume the resource theory $(\ca{F}, \ca{O})$
is consistently defined on tensor-product systems so that free operations can act on $N$ copies of the
system and produce $\lfloor rN \rfloor$ copies of another system.
For two resource states $\rho,\sigma \not\in \ca{F}$, define the optimal asymptotic
conversion rate $R(\rho \to \sigma)$ as the supremum of real number $r \ge 0$ such that there exists
a sequence of free operations
\begin{equation}
  \Phi_N : \ca{S}(\ca{H}^{\otimes N}) \longrightarrow \ca{S}(\ca{H}^{\otimes \lfloor rN \rfloor}),
  \qquad N \in \bb{N},\label{CPTPs}
\end{equation}
satisfying
\begin{equation}
  \lim_{N\to\infty}
  \Bigl\|
    \Phi_N\bigl(\rho^{\otimes N}\bigr)
    - \sigma^{\otimes \lfloor rN \rfloor}
  \Bigr\|_1
  = 0,
\end{equation}
where $\|\cdot\|_1$ denotes the trace norm.
We say that a set of resource monotone $\{M_i\}$ is \emph{i.i.d.\ complete}
if for any resource states $\rho,\sigma$ in some subclass of states, e.g. pure states,
\begin{equation}
\enskip\enskip\enskip\enskip\enskip\enskip R(\rho \to \sigma) =\sup\{r| \forall i,\ M_i(\rho)\ge rM_i(\sigma)\}.
\end{equation}
In particular, when each $M_i$ is a real number, i.e. $\mathcal{X}=[0,\infty)$, the above condition becomes
\begin{align}
R(\rho\rightarrow\sigma)=\min_i\frac{M_i(\rho)}{M_i(\sigma)},
\end{align}
and furthermore, when $\{M_i\}$ contains only single real-number measure $M$, 
\begin{equation}
R(\rho \to \sigma) =\frac{M(\rho)}{M(\sigma)}.
\end{equation}
Then $M$ alone determines the optimal rate at which many copies of $\rho$
can be reversibly converted into copies of $\sigma$ by free operations.
A familiar example is the entanglement entropy: the entanglement entropy determines the optimal rate for the i.i.d.\ convertibility between any pure states in the entanglement theory~\cite{Bennett1996}.
Another example is the quantum Fisher information in the resource theory of asymmetry for the $U(1)$-symmetry, which plays an analogous role to the entanglement entropy in the entanglement theory~\cite{Gour2008,Marvian2020,Marvian2022}.

We also define the optimal exact conversion rate $R_{\mathrm{ex}}(\rho\rightarrow\sigma)$ as the supremum of real number $r\ge0$ such that there exists a sequence of free operations \eqref{CPTPs} satisfying the following relation for any sufficiently large $N$:
\begin{equation}
\Phi_N(\rho^{\otimes N})=\sigma^{\otimes \lfloor rN\rfloor}\,.
\end{equation}

There are two main advantages to introducing i.i.d.\ complete resource measure(s).
The first advantage is purely operational.
For example, in the entanglement theory, one often considers distillation protocols, in which many copies of a weakly entangled state are processed to produce copies of a maximally entangled state, namely a Bell pair.
If an i.i.d.\ complete measure exists, it allows one to determine the ultimate efficiency limit of such distillation procedures.
Indeed, in entanglement theory, the entanglement entropy is well known to quantify precisely this optimal distillation rate~\cite{Bennett1996}.

The second advantage is that resource measures enable a characterization of macroscopic states.
Two states $\psi$ and $\phi$ satisfying $R(\psi\rightarrow\phi)=R(\phi\rightarrow\psi)=1$ are equivalent in the sense that their i.i.d. copies can be approximately interconverted by free operations with vanishing error, allowing suitable sublinear losses in copy number.
An i.i.d.\ complete resource measure implies that the equivalence classes of macroscopic states defined by such asymptotic free interconvertibility are fully characterized by a single numerical quantity.

\item[(vi$^\prime$)]
\emph{Completeness for exact i.i.d.\ convertibility.}
For the exact conversion rate $R_{\mathrm{ex}}$ defined above, we say that a family of resource monotones $\{M_i\}$ is \emph{exact-i.i.d.\ complete} on a specified class of states if, for any two resource states $\rho,\sigma$ in that class,
\begin{equation}
\begin{aligned}
R_{\mathrm{ex}}(\rho\to\sigma)
&=\sup\bigl\{r\ge0\mid{}\\
&\qquad M_i(\rho)\ge rM_i(\sigma)\ \text{for all }i\bigr\}.
\end{aligned}
\end{equation}
This property is distinct from (vi), which concerns approximate conversions with vanishing error.
\end{itemize}

In this language, a quantity is a ``good'' measure for a given resource
precisely when it obeys such axioms for an appropriate choice of free states and free operations.
For instance, the entanglement entropy is a good entanglement monotone under LOCC (local operations and classical communication).
In the following, we specialize this general structure to the case where the resource is symmetry-breaking.

\subsection{Resource theory of asymmetry for weak symmetry}

We now review the standard formulation of the resource theory of asymmetry~\cite{Marvian_thesis2012,Marvian2013}, a.k.a. the resource theory of quantum reference frames~\cite{Gour2008}.
Let $G$ be a (compact) group with unitary representation $U : G \to \mathrm{U}(\ca{H})$,
$g \mapsto U_g$.

\medskip
\noindent\textbf{Free states: symmetric states:}
A state $\rho \in \ca{S}(\ca{H})$ is \emph{$G$-symmetric} (or weak symmetric)
if it is invariant under the adjoint action of $U_g$:
\begin{equation}
  U_g \rho U_g^\dagger = \rho \qquad \forall\, g \in G.
\end{equation}
The set of free states in the resource theory of asymmetry is
\begin{equation}
  \ca{F}_G := \bigl\{ \rho \in \ca{S}(\ca{H}) \,\big|\, U_g \rho U_g^\dagger = \rho\ \forall g \in G \bigr\}.
\end{equation}

A useful object is the \emph{group twirling} (or $G$-twirling) channel $\ca{G}$ defined by
\begin{equation}\label{eq:twir}
  \ca{G}(X)
  := \int_G \mathrm{d}g\, U_g X U_g^\dagger,
  \qquad X \in \ca{B}(\ca{H}),
\end{equation}
where $\mathrm{d}g$ denotes the normalized Haar measure on $G$.
This defines a CPTP map $\ca{G} : \ca{B}(\ca{H}) \to \ca{B}(\ca{H})$.
For states $\rho \in \ca{S}(\ca{H})$ we have $ \ca{G}(\rho) \in \ca{F}_G$, and
\begin{equation}
  \ca{G}(\sigma) = \sigma \ \ \text{for all }\sigma \in \ca{F}_G.
\end{equation}

\medskip
\noindent\textbf{Free operations: covariant channels:}
The standard free operations in the resource theory of asymmetry are
$G$-covariant quantum channels.

A CPTP map
\begin{equation}
  \Phi : \ca{B}(\ca{H}) \to \ca{B}(\ca{H}')
\end{equation}
is \emph{$G$-covariant} {\color{black}(or  \emph{$(U,U')$-covariant})
if $\Phi$ satisfies the following condition for a unitary representation $U' : G \to \mathrm{U}(\ca{H}')$}
\begin{equation}
  \Phi\bigl(U_g X U_g^\dagger\bigr)
  = U'_g\, \Phi(X)\, {U'_g}^\dagger
  \qquad \forall\, X \in \ca{B}(\ca{H}),\ \forall\, g \in G.
\end{equation}
We denote the set of $G$-covariant channels by $\ca{O}_G$.
By construction, such channels map symmetric states to symmetric states:
if $\sigma \in \ca{F}_G$ then $\Phi(\sigma) \in \ca{F}_G$.
Thus $(\ca{F}_G,\ca{O}_G)$ forms a resource theory  
in which any deviation from $G$-symmetry is interpreted as a resource.

{\color{black}Intuitively, a covariant operation can be understood as an operation that is realized by interacting the system with an environment prepared in a symmetric state, through an interaction that respects the symmetry. 
Indeed, the following theorem holds:
\begin{theorem}\label{thm:SD_weak_StoS'}
Let $G$ be a group, and let $U:G\rightarrow U(\calH)$ and $U':G\rightarrow U'(\calH')$ be unitary representations acting on $\calH$ and $\calH'$.
Let a CPTP map $\Lambda:\calB(\calH)\rightarrow \calB(\calH')$ be a $(U,U')$-covariant operation.
Then, there exist two auxiliary Hilbert spaces $\mathcal{H}_E$ and $\mathcal{H}_{E'}$ satisfying $\calH\otimes\calH_{E} \cong \calH'\otimes\calH_{E'}$, a unitary operator $V: \mathcal{H}_S \otimes \mathcal{H}_E \longrightarrow \mathcal{H}_{S'} \otimes \mathcal{H}_{E'}$, two (projective) unitary representations $\{U_g^{(E)}\}_{g\in G}$ and $\{U_g^{(E')}\}_{g\in G}$ acting on $\mathcal{H}_E$ and $\mathcal{H}_{E'}$, and a symmetric state $\sigma_E$ on the ancillary system $E$ such that
\eq{
    (U'_g \otimes U_g^{(E')})\, V
    &= V\, (U_g \otimes U_g^{(E)})
    \quad \forall g \in G
\label{eq:covariant_unitary_general},\\
    \sigma_E&=U_g^{(E)} \sigma_E \big(U_g^{(E)}\big)^\dagger
    \quad \forall g \in G,\\
    \Lambda(\rho)
    &= \Tr_{E'}\!\left[\,V (\rho \otimes \sigma_E) V^\dagger \,\right].\label{eq:covariant_unitary_general3}
}
Conversely, when a CPTP map $\Lambda:\calB(\calH)\rightarrow \calB(\calH')$ can be realized by $(V,\sigma_E)$ satisfying \eqref{eq:covariant_unitary_general}--\eqref{eq:covariant_unitary_general3}, the map is $(U,U')$-covariant.
\end{theorem}
}

{\color{black}
\textbf{Covariant measurements:} To discuss strong monotonicity, we also introduce QC states and free measurements.
We define a classical register as a system on which the unitary representation of the group $G$
is always given by the identity operator $I$.
A QC state is then taken to be of the form
\[
\sum_k p_k\, \rho_k \otimes \ket{k}\bra{k}_C .
\]

We define a $(U,U')$-covariant measurement
$\{\Psi_k : \mathcal{B}(\mathcal{H}) \to \mathcal{B}(\mathcal{H}')\}$
as a collection of completely positive maps such that the associated QC channel
\[
\sum_k \Psi_k(\cdot)\otimes \ket{k}\bra{k}_C
\]
is a $(U,\, U'\otimes I_C)$-covariant operation.
Note that $\{\Psi_k\}$ is a $(U,U')$-covariant measurement iff each CP map $\Psi_k$ is individually $(U,U')$-covariant and $\sum_k\Psi_k$ is trace preserving.
}

\subsection{Asymmetry monotones for weak symmetry}

Within this framework, a good quantifier of symmetry breaking is a resource monotone
for the resource theory $(\ca{F}_G,\ca{O}_G)$.

\subsubsection{Relative entropy of asymmetry}
We now recall the relative entropy of asymmetry,
which plays a central role in our work.

\medskip
\noindent\textbf{Relative entropy of asymmetry (a.k.a. Entanglement asymmetry):}
Given a symmetry group $G$ and its twirling channel $\ca{G}$,  
the \emph{relative entropy of asymmetry \cite{Vaccaro2008,Gour2009}} is defined by
\eq{
  A_G(\rho)
:=\min_{\sigma\in\calF_{G}}S(\rho\|\sigma),
}
where $S(\rho \Vert \sigma):=  \tr\bigl[ \rho (\log \rho - \log \sigma) \bigr]$ is the relative entropy.
Equivalently, it is written as~\cite{Gour2009}
\begin{equation}
  A_G(\rho)=S\bigl( \rho \,\Vert\, \ca{G}(\rho) \bigr)
  = S\bigl(\ca{G}(\rho)\bigr) - S(\rho),\label{alt_def_EA}
\end{equation}
with $S(\rho) := - \tr[\rho \log \rho]$ the von Neumann entropy.
The relative entropy of asymmetry $A_G$ satisfies:
(i) \emph{the minimal requirements (i-a) and (i-b) as a resource monotone}~\cite{Vaccaro2008}, 
(ii) \emph{Faithfulness}~\cite{Vaccaro2008}, (iii) \emph{Convexity}~\cite{Marvian2016}, and (iv) \emph{Strong monotonicity}~\cite{Vaccaro2008}. 

\medskip
The quantity $A_G$ thus provides a mathematically well-defined and operationally meaningful measure
of symmetry breaking with respect to $G$:
it vanishes exactly on symmetric states,
does not increase under $G$-covariant operations.
In this way, resource theory singles out canonical measures such as the relative entropy
of asymmetry, and also provides a systematic language for comparing and selecting among
multiple candidate quantifiers of symmetry breaking.

\subsubsection{Complete monotones for the i.i.d.\ convertibility}

The virtues of the relative entropy of asymmetry are that it is faithful and can be defined for any group $G$ for which the twirling operation is well defined. 
However, as we will discuss later, the rate of this quantity $A_G(\rho^{\otimes N})/N$ vanishes in the independent and identically distributed (i.i.d.)\ asymptotic limit and therefore cannot characterize the optimal ratio of the conversions between the i.i.d.\ states. 
In other words, the measures that play the role analogous to the entanglement entropy in entanglement theory--an i.i.d.\ complete monotone--must be different quantities.
In the case of weak symmetry, it has been clarified for various symmetries which quantities serve as such i.i.d.\ complete monotones. 
In the following, we briefly review these results.

\medskip
\noindent\textbf{U(1): Quantum Fisher information:}
Given a unitary representation $\{e^{-iHt}\}_{t\in[0,2\pi)}$ of $U(1)$ and its generator $H$,
the \emph{Symmetric Logarithmic Derivative (SLD)-quantum Fisher information}~\cite{Helstrom1969} for the state family $\{e^{-iHt}\rho e^{iHt}\}$ is defined by
\begin{equation}
\calI_H(\rho):=2\sum_{k,l}\frac{(p_k-p_l)^2}{p_k+p_l}|\bra{k}H\ket{l}|^2,
\end{equation}
where $\{p_k,\ket{k}\}$ are the eigenvalues and eigenstates of $\rho$.
This quantity is also called \textit{SLD-skew information}~\cite{Hansen2008}:
\eq{
\tilde{\calI}_H(\rho):=\frac{1}{4}\calI_H(\rho).
}
When $\rho$ is pure, it is equal to 4 times of the variance of $L$:
\begin{equation}
  \calI_H(\ket{\psi}\bra{\psi})
  = 4V_H(\ket{\psi}\bra{\psi}),
\end{equation}
where $V_H(\rho):=\Tr[H^2\rho]-\Tr[H\rho]^2$.
The Fisher information $\calI_H(\rho)$  satisfies:
(i) \emph{the minimal requirements (i-a) and (i-b) as a resource monotone}~\cite{Yadin2016}, 
(ii) \emph{Faithfulness}~\cite{Yadin2016}, (iii) \emph{Convexity}~\cite{Yadin2016}, and (iv) \emph{Strong monotonicity}~\cite{Yadin2016}, and (v) \emph{Additivity}~\cite{Hansen2008}.
Furthermore, it also satisfies (vi) \emph{Completeness for i.i.d.\ convertibility between pure states}~\cite{Gour2008} \emph{and from a pure state to a mixed state}~\cite{Marvian2022}:
\begin{equation}
R(\ket{\psi}\rightarrow\ket{\phi})=\frac{\calI_H(\ket{\psi}\bra{\psi})}{\calI_H(\ket{\phi}\bra{\phi})},
\end{equation}
when $\mathrm{Sym}_{U(1)}(\psi)\subset\mathrm{Sym}_{U(1)}(\phi)$, where $\mathrm{Sym}_G(\rho):=\{g\in G|U_g\rho U^\dagger_g=\rho\}$.
(If $\mathrm{Sym}_{U(1)}(\psi)\not\subset\mathrm{Sym}_{U(1)}(\phi)$, $R(\psi\rightarrow\phi)=0$.)

This quantity is also computationally accessible in many-body physics and can serve as a useful indicator of symmetry breaking.
For instance, in the BCS model, it is proportional to the number of Cooper pairs, and allows us to analyze the quantum Mpemba effect in the thermodynamic limit~\cite{Yamashika2025}.
Another related topic is the conversion theory. In the case of $U(1)$ symmetry, the asymptotic convertibility of pure states beyond the i.i.d.\ regime
(i.e.\ allowing arbitrary correlations)
admits both a necessary condition and a sufficient condition,
given by an extension of the quantum Fisher information~\cite{Yamaguchi2023,Yamaguchi2023a}.

The SLD quantum Fisher information determines the i.i.d.\ conversion rate between pure states and from a pure state to a mixed state, but does not determine the rate between arbitrary states.
This limitation is resolved by the metric adjusted quantum geometric tensor introduced below~\cite{YT2026_1}.

\medskip
\noindent\textbf{Compact Lie group: Quantum geometric tensor and metric adjusted quantum geometric tensor}
Given a compact Lie group $G$ and its (projective) unitary representation $U_{g\in G}$,  
the \emph{quantum geometric tensor}~\cite{Provost1980,Berry1989} for a pure state $\ket{\psi}$ is defined by
\begin{equation}\calQ^{\psi}_{i,j}:=\bra{\psi}X_i X_j\ket{\psi}-\bra{\psi}X_i \ket{\psi}\bra{\psi}X_j\ket{\psi},
\end{equation}
where $\{X_i\}$ are the Hermitian operators defined as 
\begin{equation}
 X_i:=-i\frac{\partial}{\partial \lambda^i} U_{g(\bm{\lambda})}|_{\bm{\lambda}=\bm{0}}\label{X_j_def}
\end{equation}
and $g(\bm{\lambda})=\ex{i\sum_{i}\lambda^i A_i}$ is a parametrization of elements in the neighborhood of the identity $e\in G$ with a basis $\{A_i\}$ of the Lie algebra $\mathfrak{g}$. 
The quantum geometric tensor $\calQ^\psi$ satisfies:
(i) \emph{the minimal requirements (i-a) and (i-b) as a resource monotone}~\cite{Yamaguchi2024}, 
(ii) \emph{Faithfulness} (for connected $G$)~\cite{Yamaguchi2024}, (iv) \emph{Strong monotonicity}~\cite{Yamaguchi2024}, and (v) \emph{Additivity}~\cite{Yamaguchi2024}.
Furthermore, it also satisfies (vi) \emph{Completeness for i.i.d.\ convertibility between pure states}~\cite{Yamaguchi2024}: when $\mathrm{Sym}_G(\psi)\subset\mathrm{Sym}_G(\phi)$,
\begin{equation}
R(\ket{\psi}\rightarrow\ket{\phi})=\sup\{r\ge0|\calQ^{\psi}\ge r\calQ^{\phi}\}.
\end{equation}
(If $\mathrm{Sym}_G(\psi)\not\subset\mathrm{Sym}_G(\phi)$, $R(\psi\rightarrow\phi)=0$.)

The quantum geometric tensor defined above applies to pure states, whereas quantum Fisher information also quantifies the asymmetry of mixed states.

For mixed states, we can employ the SLD-quantum Fisher information matrix, which satisfies the properties (i)-(v) of resource measure for an arbitrary connected Lie group (and (i) and (iii)-(v) for non-connected Lie group)~\cite{Kudo2023}:
\begin{equation}
\calI^{\rho}_{i,j}:=2\sum_{k,l}\frac{(p_k-p_l)^2}{p_k+p_l}\bra{k}X_i\ket{l}\bra{l}X_j\ket{k},
\end{equation}
where $\{p_k,\ket{k}\}$ are the eigenvalues and eigenstates of $\rho$.
For pure states, the quantum Fisher information matrix coincides with the four times of the real part of the quantum geometric tensor.

The metric adjusted quantum geometric tensor, which unifies the SLD quantum Fisher information matrix and the pure-state quantum geometric tensor, was introduced in Ref.~\cite{YT2026_1}.
It is a complex matrix-valued quantity also defined for mixed states, and its one-parameter family provides an i.i.d.\ complete measure for arbitrary quantum states, as described below.

Again, consider a compact Lie group $G$ and a continuous (projective) unitary representation $U_{g\in G}$ on a finite-dimensional system. For a state $\rho$ and $0<q<1$, the \emph{metric adjusted quantum geometric tensor} is defined by
\eq{
\calQ^{\rho,q}_{i,j}
&:=q(1-q)\nonumber\\
&\times\sum_{\substack{k,l\\qp_k+(1-q)p_l>0}}
\frac{(p_k-p_l)^2}{qp_k+(1-q)p_l}
\bra{l}X_i\ket{k}\bra{k}X_j\ket{l}\label{eq:metric_adjusted_QGT}
}
as in Ref.~\cite{YT2026_1}.
Here, $\{p_k,\ket{k}\}$ are the eigenvalues and eigenstates of $\rho$, and $X_j$ are the generators introduced in \eqref{X_j_def}.
At the endpoints $q=0,1$, it is defined by the continuous limits $\calQ^{\rho,0}:=\lim_{q\to0^+}\calQ^{\rho,q}$ and $\calQ^{\rho,1}:=\lim_{q\to1^-}\calQ^{\rho,q}$ for fixed $\rho$.
For $0<q<1$, the metric adjusted quantum geometric tensor $\calQ^{\rho,q}$ satisfies:
(i) the minimal requirements (i-a) and (i-b) as a resource monotone~\cite{Yamaguchi2024,YT2026_1},
(ii) faithfulness (for connected $G$)~\cite{Yamaguchi2024,YT2026_1},
(iii) convexity~\cite{Yamaguchi2024,YT2026_1}, (iv) strong monotonicity~\cite{Yamaguchi2024,YT2026_1}, and (v) additivity~\cite{Yamaguchi2024,YT2026_1}.
The endpoint quantities at $q=0,1$ also satisfy (i) and (iii)--(v), but generally do not satisfy (ii).

Furthermore, the family $\{\calQ^{\rho,q}\}_{q\in[1/2,1]}$ satisfies (vi) completeness for i.i.d.\ convertibility between arbitrary states, including mixed states~\cite{YT2026_1}:
when $\mathrm{Sym}_G(\rho)\subseteq\mathrm{Sym}_G(\sigma)$,
\begin{equation}\label{eq:metric_adjusted_QGT_rate}
\begin{aligned}
R(\rho\to\sigma)
=\sup\bigl\{r\ge0\mid{}&\calQ^{\rho,q}\ge r\calQ^{\sigma,q}\\
&\text{for all }q\in[1/2,1]\bigr\}.
\end{aligned}
\end{equation}
(If $\mathrm{Sym}_G(\rho)\not\subseteq\mathrm{Sym}_G(\sigma)$, then $R(\rho\to\sigma)=0$.)
The rate formula \eqref{eq:metric_adjusted_QGT_rate} also holds when $G$ is disconnected.
To see how this is possible, let $G_0$ be the subgroup consisting of all group elements that can be joined to the identity by a continuous path.
The generators $X_i$ describe infinitesimal transformations in $G_0$, and the metric adjusted quantum geometric tensor is constructed from the response of the state to these transformations.
Thus, for $0<q<1$, the condition $\calQ^{\rho,q}=0$ is equivalent to $U_g\rho U_g^\dagger=\rho$ for every $g\in G_0$.
If $G$ is disconnected, it also contains elements outside $G_0$, so this condition alone does not determine symmetry under the whole group.
This is why connectedness is assumed in (ii) above.

The conversion-rate formula, however, requires $\mathrm{Sym}_G(\rho)\subseteq\mathrm{Sym}_G(\sigma)$ in addition to the tensor inequalities.
This inclusion means that every group element that leaves the input $\rho$ invariant also leaves the output $\sigma$ invariant.
Indeed, covariant operations preserve the invariance of the input: if $U_g\rho U_g^\dagger=\rho$, the output is invariant under the same $g$.
Even for i.i.d.\ conversion with vanishing error, the one-copy marginal of the actual output retains this invariance and converges to $\sigma$ as the error vanishes, so $\sigma$ must also be invariant under the same $g$.
This condition also applies to elements outside $G_0$.
The content of \eqref{eq:metric_adjusted_QGT_rate} is that, once this inclusion condition for the whole group is satisfied, the remaining constraints on the conversion rate are determined by the matrix inequalities for the tensor family.
For example, a finite group has no infinitesimal generators, so the tensors have no matrix entries to compare.
The tensor inequalities alone would then allow any rate between any pair of states.
The inclusion condition on the symmetry subgroups remains, however, and excludes, for example, conversion from a symmetric state to an asymmetric state.
Thus, (ii) concerns the faithfulness of the tensor alone, whereas the conversion-rate formula in (vi) combines the tensors with the inclusion condition on symmetry subgroups.

In \eqref{eq:metric_adjusted_QGT_rate}, even when $\rho$ or $\sigma$ is free, $R(\rho\to\sigma)$ is defined as the supremum of the rates $r$ for which the i.i.d.\ conversion $\rho^{\otimes N}\to\sigma^{\otimes\lfloor rN\rfloor}$ is possible with vanishing error.
For example, if the output $\sigma$ is free, discarding the input and preparing the required number of copies of $\sigma$ is weak covariant, giving $R(\rho\to\sigma)=\infty$.
If the input $\rho$ is free and the output $\sigma$ is a resource state, then $R(\rho\to\sigma)=0$.
Since $U(1)$ is a compact Lie group, this result also resolves the previously open problem of i.i.d.\ completeness for arbitrary states under $U(1)$ symmetry~\cite{YT2026_1}.

The relations to the quantities introduced above are
\begin{align}
\calQ^{\rho,1/2}&=\frac14\calI^\rho,\label{eq:metric_adjusted_QGT_SLD}\\
\calQ^{\psi,q}&=q\calQ^\psi+(1-q)(\calQ^\psi)^{\mathsf T}\label{eq:metric_adjusted_QGT_pure}
\end{align}
as shown in Ref.~\cite{YT2026_1}.
Here, $\psi=\ket{\psi}\bra{\psi}$ is a pure state and $\mathsf T$ denotes the transpose.
Thus, at $q=1/2$ the tensor equals $1/4$ of the SLD quantum Fisher information matrix, and at $q=1$ it reduces to the usual quantum geometric tensor for pure states.

\medskip
\noindent\textbf{Finite group: logarithmic fidelity-based characteristic function:}
Given a finite group $G$ and its (projective) unitary representation $U$, we define
the \emph{logarithmic fidelity-based characteristic function (LFCF)} as
\begin{equation}
 \label{eq:logCh}
 LF_{G,g}(\rho):=-\log F(\rho,U_g\rho U^\dagger_g),
\end{equation}
where $F(\rho,\sigma):=\Tr[\sqrt{\sqrt{\sigma}\rho\sqrt{\sigma}}]$ is the Uhlmann fidelity.
When $\rho$ is pure, the LFCF coincides with the logarithmic characteristic function (LCF) $L_{G,g}(\rho):=-\log|\Tr[U_g\rho]|$, which is defined only for pure states in Ref.~\cite{Shitara2023}.
For any $g\in G$, $LF_{G,g}(\rho)$ satisfies (i) \emph{the minimal requirements (i-a) and (i-b) as a resource monotone}, (iii) \emph{Convexity}, and (v) \emph{Additivity}, but it does not satisfy (iv) \emph{Strong monotonicity} in general.
The set $\{LF_{G,g}(\rho)\}_{g\in G}$ satisfies (ii) \emph{Faithfulness}: $\rho\in\calF_G$ iff $LF_{G,g}(\rho)=0$ for every $g\in G$.
For pure states, (i), (ii), and (v) were established in Ref.~\cite{Shitara2023} using the techniques of Ref.~\cite{Marvian2013}.
Appendix~\ref{sec:prop_LCF} proves (i), (ii), (iii), and (v) for arbitrary states and gives a counterexample to (iv).

Furthermore, the set $\{LF_{G,g}(\rho)\}_{g\in G}$ satisfies (vi$^\prime$) \emph{Completeness for exact i.i.d.\ convertibility} between pure states~\cite{Shitara2023}: when $\mathrm{Sym}_G(\psi)\subset\mathrm{Sym}_G(\phi)$ holds,
\begin{equation}
R_{\mathrm{ex}}(\ket{\psi}\rightarrow\ket{\phi})=\min_{g\in G}\frac{LF_{G,g}(\psi)}{LF_{G,g}(\phi)}.
\end{equation}
If $\mathrm{Sym}_G(\psi)\not\subset\mathrm{Sym}_G(\phi)$, $R_{\mathrm{ex}}(\psi\rightarrow\phi)=0$.

By contrast, for finite groups, the approximate asymptotic conversion rate satisfies the following property~\cite{Shitara2023}:
\begin{equation}
R(\ket{\psi}\rightarrow\ket{\phi})=\left\{
\begin{array}{ll}
0 &  (\mathrm{Sym}_G(\psi)\not\subset\mathrm{Sym}_G(\phi))\\
\infty & \mathrm{(Otherwise)}
\end{array}
\right.
\end{equation}

\section{Resource-theoretic Results for Weak Symmetry}\label{Sec:RRWS}

Before turning to strong symmetry, it is useful to highlight what one can learn
by applying the resource theory to weak symmetry.
In this section, we briefly summarize two simple but important results.
Both rely on taking the basic axioms of a resource monotone seriously
(faithfulness, monotonicity under free operations, etc.),
and both illustrate how resource theory helps to discriminate between
``good'' and ``merely convenient'' measures of symmetry breaking.

\subsection{Rényi-2 proxies are not resource monotones}

Motivated by the simplicity of R\'enyi entropies, 
a ``second R\'enyi'' analogue  of the relative entropy of asymmetry,
\begin{equation}
  A_G^{(2)}(\rho)
  := S^{(2)}\!\bigl(\ca{G}(\rho)\bigr) - S^{(2)}(\rho),
  \ \ \ 
  S^{(2)}(\rho) := - \log \tr(\rho^2),
\end{equation}
or closely related quantities have been considered \cite{Ares2022,Ares2023,Ferro2023,Capizzi2023a,Chen2023,Lastres2024,Fossati2024a,Chen2024,Ares2023a,Russotto2024,Murciano2023,Yamashika2024,Caceffo2024,Fujimura2025}.
Such quantities are easy to compute and often behave qualitatively like $A_G$
in explicit examples, so they have been widely used in the context of symmetry breaking.

The quantity, $A_G^{(2)}(\rho)=0$ if and only if $\rho$ is weak-symmetric, hence a nice proxy as an indicator of symmetry breaking. However, as a ``quantifier", it fails the crucial test of monotonicity under
free operations (symmetric operation). Hence, one should not use $A_G^{(2)}(\rho)$ \emph{solely} to infer anything about ``amount" of symmetry breaking, e.g.\ to deduce quantum Mpemba effect.

Consider a single qubit with a $\bb{Z}_2$ symmetry generated by the Pauli operator $X$,
so that the twirling map is $\ca{G}_X(\rho)=\tfrac12(\rho+X\rho X)$.
Then the $Z$-basis dephasing channel
\begin{equation}
\Delta_Z(\rho)=\sum_{s=\pm 1} P_s \rho P_s,   \ \ \ \ \ 
P_{\pm}=\tfrac12(\mathbf{1}\pm Z)
\end{equation}
is $\bb{Z}_2$-covariant (hence free), yet it can strictly \emph{increase} the R\'enyi-2 proxy.
For example, for
\begin{equation}
  \rho=\frac12\!\left(\mathbf{1}+\frac12 X+\frac12 Z\right),
\end{equation}
one finds
\begin{equation}
  A_G^{(2)}\bigl(\Delta_Z(\rho)\bigr) > A_G^{(2)}(\rho).
\end{equation}
We present the detailed calculation in Appendix~\ref{app:Renyi2-not-monotone}.

Thus, $A_G^{(2)}$ can be strictly \emph{increased} by free operations, and therefore fails to be a resource monotone, even though it looks natural
and is convenient in many-body calculations.
Accordingly, $A_G^{(2)}$ should be regarded as a
model-dependent proxy whose behavior must be interpreted with care,
rather than as a principled, universally valid quantifier of weak symmetry breaking.
For example, from the results of $A_G^{(2)}$ alone, one cannot determine precisely whether the quantum Mpemba effect is occurring.
This simple but important application to physics clearly illustrates the usefulness of the resource-theoretic framework,
since it is impossible to discuss monotonicity rigorously without resource theory.

\subsection{Relative entropy of asymmetry and the thermodynamic limit}

For weak symmetry with respect to a compact group $G$, the standard resource-theoretic
measure is the relative entropy of asymmetry $A_G$.
For many-body systems, however, resource theory also makes the limitation of $A_G$ transparent.

It is known that the quantity $A_G$ exhibits at most logarithmic growth in the system size $L$ in two cases: 
i.i.d.\ states $\rho^{\otimes L}$ of any state $\rho$~\cite{Gour2009} and translationally invariant matrix product states (MPS) with finite bond dimension~\cite{Capizzi2023a}.
Therefore, in at least these two cases, the asymmetry density vanishes in the thermodynamic limit $L\to\infty$,
\begin{equation}
  \frac{A_G\bigl(\rho_L\bigr) }{L}  \to 0.
\end{equation}
In resource-theoretic language, $A_G$ is not additive and not extensive, and consequently
cannot be a complete measure determining optimal asymptotic conversion rates between many copies
of resource states.
Thus, while $A_G$ is conceptually clean and operationally meaningful at finite size,
it does not define a nontrivial extensive resource measure in the thermodynamic limit
and therefore is not an appropriate thermodynamic quantifier of weak symmetry breaking.
This conclusion naturally emerges once one asks for additivity and completeness within
the resource theory, and it motivates the search for alternative measures
(such as the quantum Fisher information) when dealing with macroscopic systems \cite{Yamashika2025}.

These examples show the practical value of the resource theory.
They do not only generate new candidates; they also provide a systematic way to test
and sometimes rule out quantities that are introduced ad hoc.
In the remainder of the paper, we will adopt the same philosophy for strong symmetry:
we will first specify the appropriate notion of free states and free operations,
and then derive a measure of strong symmetry breaking that satisfies the basic resource-theoretic requirements by construction.

\section{Resource Theory of Asymmetry for Strong Symmetry}\label{sec:RStrong}
Now, let us formulate the resource theory of asymmetry for strong symmetry.
{\color{black}Throughout this section, we restrict attention to finite-dimensional Hilbert spaces and provide proofs under this assumption, as is common in many resource theories.
We {\color{black}expect} that the framework and most of the results themselves do not rely on finite dimensionality, and many of the quantities we introduce remain well defined (and often computable) in physically relevant infinite-dimensional settings.
}See section~\!\ref{sec:ex} for examples, where we compute asymmetry in the settings involving CFT.

\subsection{Free states and free operations}
We first define the free states.
A natural choice is to take the following ``strong symmetric states" as free states.

\medskip
\noindent\textbf{Free states: Strong symmetric states:}
Let $S$ be a quantum system, and let $G$ be a group with a (projective) unitary representation $U$ acting on $S$. 
A state $\rho$ on $S$ is said to be strong symmetric with respect to $U$ if it satisfies
\eq{
U_g\rho=e^{i\theta_{g,\rho}}\rho,\enskip\forall g\in G,
}
where $\theta_{g,\rho}$ is a real-valued function of $g$ and $\rho$. We refer to the whole set of strong symmetric states as ${\cal F}_{G,\mathrm{strong}}$.

The intuition behind this class becomes clear when we consider the case of $U(1)$ symmetry. For example, if $U$ can be expressed as $\{e^{i Q \theta}\}$, then a strong symmetric state is simply an (not necessarily pure) eigenstate of $Q$.
This class of states is a natural extension of symmetric states and is widely used in the field of condensed matter physics. 
Hereafter, we refer to symmetric states for weak symmetry as weak symmetric states.

An important remark on this class is that there are many cases where no strong symmetric state exists. 
As an example of such cases, let us consider the situation where the system of interest is a qubit, $G$ is $SU(2)$, 
and $U$ is the natural irreducible unitary representation of $SU(2)$ acting on the qubit:  
$U_g := e^{-i\frac{\theta_g}{2}\vec{n}_g\cdot\vec{\sigma}}$, 
where $\theta_g$ is a real parameter describing the rotation angle, $\vec{n}_g$ is a real vector parameter defined as 
$\vec{n}_g := (n_x(g), n_y(g), n_z(g))$ describing the rotation axis, and $\vec{\sigma} := (\sigma_x, \sigma_y, \sigma_z)$ 
denotes the Pauli operators. 
Then, for any strong symmetric state $\rho$, the condition $U_g\rho U_g^\dagger = \rho$ must hold 
for all $g\in G$. Since the representation $U$ is irreducible, it follows that 
$\rho = I/2$. However, since $U$ includes $\sigma_x$, we have 
$\sigma_x \rho = \sigma_x/2 \not\propto \rho$. 
This leads to a contradiction, and hence there is no strong symmetric state in this case. We further note that when $G$ is an Abelian group and $U$ is a non-projective unitary representation, the strong symmetric states always exist, only in the case of non-Abelian symmetry or the case of projective unitary representation, the subtlety as mentioned in this paragraph may appear.

To compensate for the limitation, we introduce another class of states as follows:

\medskip
\noindent\textbf{Additional class of states: Single-sector states:}
Let $S$ be a quantum system, and let $G$ be a group with a (projective) unitary representation $U$ acting on $S$. 
We assume that $U_{g\in G}$ admits the following irreducible decomposition:
\eq{
U_g = \bigoplus_{\nu} U_g^{(\nu)} \otimes I_{m_\nu},
}
where $m_{\nu}$ denotes the multiplicity of the irrep labeled by $\nu$.  
Correspondingly, the Hilbert space $\calH$ can be decomposed as
\eq{\label{sum_a}
\calH = \bigoplus_{\nu} \calH_{\nu} \otimes \mathbb{C}^{m_\nu}.
}
A state $\rho$ on $S$ is said to be \textit{single-sector} with respect to $U$ if it is weak symmetric with respect to $U$ and satisfies
\eq{
\exists\, \nu \text{ such that } P_\nu \rho P_\nu = \rho,
}
where $P_\nu$ is the projection operator onto $\calH_{\nu} \otimes \mathbb{C}^{m_\nu}$.
We refer to the whole set of single-sector states as ${\cal F}_{G,\mathrm{single}}$.

Clearly, any single-sector state $\rho$ must take the following form 
(note that a single-sector state is also weak symmetric):
\eq{\label{51}
\rho = \frac{I_{\nu}}{d_\nu} \otimes \rho_{m_{\nu}},
}
where $I_{\nu}$ is the identity operator on $\calH_{\nu}$, $d_\nu:=\mathrm{Tr}\ I_\nu$ and $\rho_{m_{\nu}}$ is a quantum state on $\mathbb{C}^{m_{\nu}}$.

When $G = U(1)$, the single-sector states coincide with the strong symmetric states. 
For a general group $G$, however, the single-sector states constitute a distinct class from the strong symmetric states. 
They always include all strong symmetric states and exist even in situations where no strong symmetric state exists. 
The relationship among weak symmetric, single-sector, and strong symmetric states is illustrated in Figure \ref{fig_venn}.

\begin{figure}[t]
\begin{center}
\includegraphics[width=.45\textwidth]{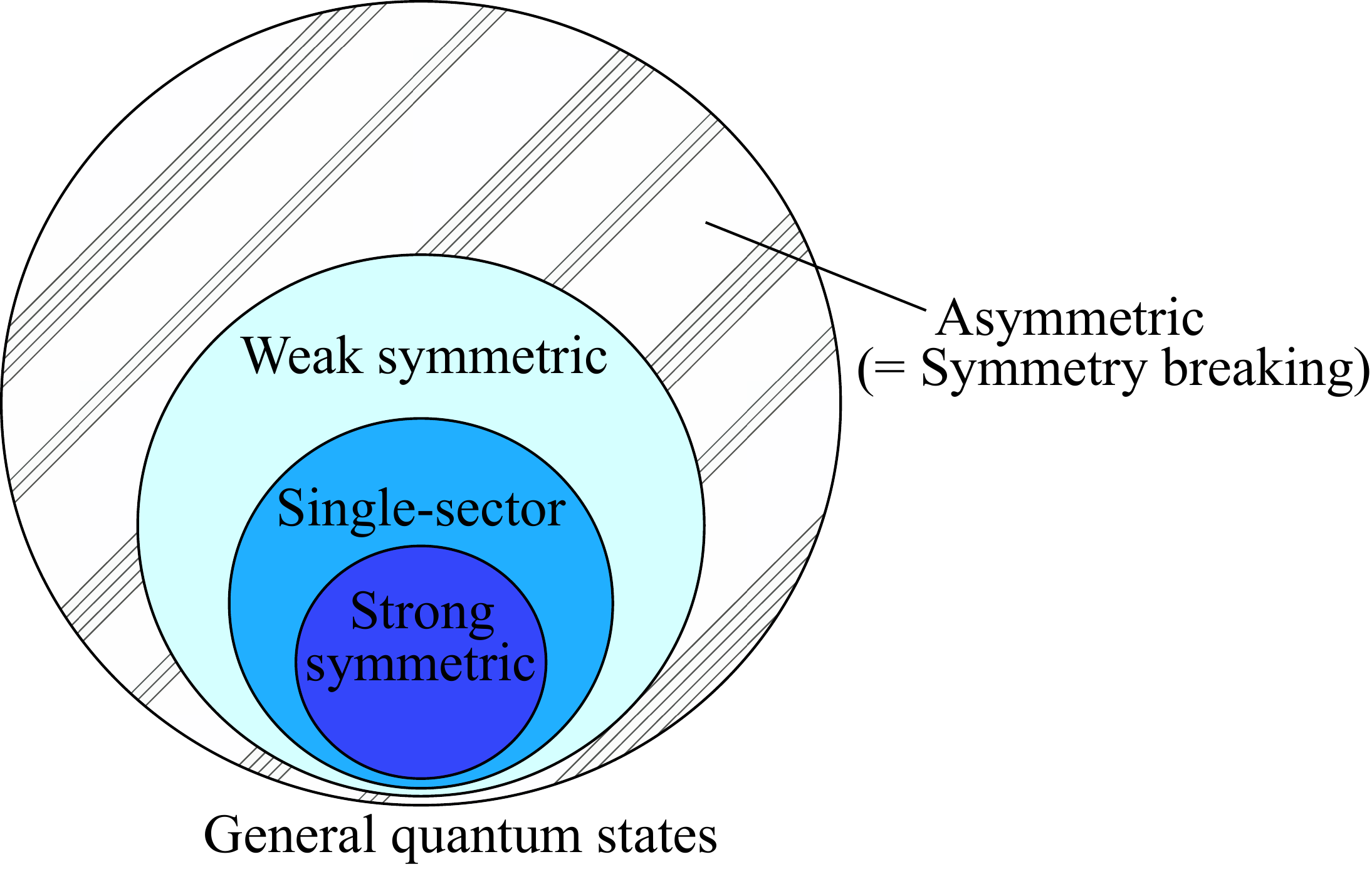}
\caption{Classifications of states. 
In the figure, asymmetric states are indicated by the hatched region, whereas symmetric states are shown as the three shaded regions.
The set of symmetric states exhibits a three-level nested hierarchical structure, i.e. (strong symmetric)$\subset$(single-sector)$\subset$(weak symmetric). The important features of this structure is the following four: (a) There exist cases where (strong symmetric)$\subsetneq$(single-sector) and (single-sector)$\subsetneq$(weak symmetric) hold. (b) When $G=U(1)$, (strong symmetric)$=$(single-sector) holds. (c) When $G$ is some non-Abelian group, sometimes there are no strong symmetric states. (d) For any $G$, there is at least a single-sector state.}
\label{fig_venn}
\end{center}
\end{figure}
\medskip
\noindent
\textbf{Examples:}
{\color{black}
Let us unpack the definitions above for the purpose of illustration. We consider a single qubit system. Using the Bloch-sphere representation, we have
\begin{equation}
    \rho=\frac{1}{2}\left(\mathbf{1}+r_x X+r_yY+r_zZ\right)\,,
\end{equation}
where $\vec{r}=(r_x,r_y,r_z)$ lies on/inside the Bloch-sphere i.e.\ $|\vec{r}|\leq 1$.
There is a natural action of $SU(2)$ on the space of these density matrices: it rotates the vector $\vec{r}$. We further consider the $U(1)$ subgroup, generated by $Z$. 

The weak-symmetric states under $U(1)$ are given by the following family of density matrices, parametrized by $0\leq p\leq 1$.
\begin{equation}
    \rho_{[p]}:=\begin{pmatrix}
        p & 0\\
    0 & 1-p\,.
    \end{pmatrix}
\end{equation}

Among these, there exist only two strong symmetric states under $U(1)$, given by $\rho_{[0]}$ and $\rho_{[1]}$. In terms of Bloch-vectors, $\vec{r}=(0,0,r)$ denotes the weak symmetric ones (with $2p=1+r$), and $(0,0,\pm 1)$ denotes the strong symmetric ones. As previously mentioned, for Abelian groups with non-projective unitary representations, strong symmetric states coincide with single-sector states.

Now let us turn our attention to $SU(2)$, the only vector, invariant under $SU(2)$ is the zero vector corresponding to $\rho=\frac{1}{2}\mathbf{1}=\rho_{[1/2]}$. This is the \textit{only} weak-symmetric state under $SU(2)$. Note that $\rho_{[p]}$, in spite of being diagonal, is not weak-symmetric under $SU(2)$ unless $p=1/2$. Moreover, $\rho_{[1/2]}$ is also a single sector state because it obeys \eqref{51}, with $\nu$ being the spin-$1/2$ irrep of $SU(2)$ and $m_\nu=1$. Furthermore, as mentioned before, there is no strong symmetric state in this case.

To explore further, we consider a two-qubit system. Again, $SU(2)$ has a natural action. Now we have spin-$1$ irrep and spin-$0$ irrep. In the basis adapted to spin-$1$ irrep and spin-$0$ irrep, the weak-symmetric density matrices under $SU(2)$ are given by the following family, parametrized by $0\leq p\leq 1/3$:
\begin{equation}
    \rho_{\text{2-qubit}[p]}:=
    \begin{pmatrix}
        p & 0 & 0 & 0\\
        0 & p & 0 & 0\\
        0 & 0 & p & 0\\
        0 & 0 & 0 & 1-3p
    \end{pmatrix}
\end{equation}
Once again, any diagonal density matrix may not be weak-symmetric. It has to be of the form given above.
The \textit{only} two single-sector states are $\rho_{\text{2-qubit}[0]}$ and $\rho_{\text{2-qubit}[1/3]}$. The \textit{only} strong-symmetric state is $\rho_{\text{2-qubit}[0]}$. \emph{Note that for non-abelian group, the strong-symmetric state exist only if there exists a single-sector state, transforming under dimension $1$ irrep of the relevant group.}

These examples illustrate the Venn-diagram, depicted in the figure.~\!\ref{fig_venn}
}.

As we shall see later, both the strong symmetric states and the single-sector states can serve as free states under strong symmetry. 
Moreover, faithful resource measures can be defined for each of them. 
We remark that, while single-sector states include all strong symmetric states in general, the two classes coincide for Abelian groups with non-projective unitary representations.
In this paper, we treat the strong symmetric states as the free states, 
while also discussing the results for the single-sector states.

Next, let us define the free operations:

\medskip
\noindent
\textbf{Free operations: strong covariant operations:}
Let $G$ be a group, and let $U$ and $U'$ be unitary representations of $G$ acting on two Hilbert spaces $\calH$ and $\calH'$, respectively.
A CPTP map $\Lambda: \calB(\calH)\to \calB(\calH')$ is said to be $(U,U')$-strong covariant if it satisfies
\begin{equation}
    \Lambda\!\left(U_g...\right)
    = U'_g \Lambda(...) 
    \quad \forall g \in G.
    \label{eq:s-covariance_def}
\end{equation}
We refer to the whole set of strong covariant operations as ${\cal O}_{G,\mathrm{strong}}$.

The following theorem ensures that the above definitions of free states and free operations give rise to a well-defined resource theory.
\begin{theorem}\label{thm:GR_strong}
Let $G$ be a (compact) group, and let $U$, $U'$ and $U''$ be unitary representations acting on Hilbert spaces $\calH$, $\calH'$ and $\calH''$, respectively.
Then, the following three are valid:$\\$
(i) The identity operation on $\calB(\calH)$ is $(U,U)$-strong covariant. $\\$
(ii) If $\Lambda:\calB(\calH)\rightarrow \calB(\calH')$ and $\Lambda':\calB(\calH')\rightarrow \calB(\calH'')$ are $(U,U')$- and $(U',U'')$-strong covariant, respectively, $\Lambda'\circ\Lambda$ is also $(U,U'')$-strong covariant. $\\$
(iii) If $\Lambda:\calB(\calH)\rightarrow \calB(\calH')$ is a $(U,U')$-strong covariant CPTP map, the following relations hold:
\eq{
\rho\in{\cal F}_{G,\mathrm{strong}}&\Rightarrow\Lambda(\rho)\in{\cal F}_{G,\mathrm{strong}},\label{GR_strong}\\
\rho\in{\cal F}_{G,\mathrm{single}}&\Rightarrow\Lambda(\rho)\in{\cal F}_{G,\mathrm{single}}.\label{GR_single}
}
\end{theorem}
Namely, either the combination of $(\calF_{G,\mathrm{strong}},\calO_{G,\mathrm{strong}})$ or $(\calF_{G,\mathrm{single}},\calO_{G,\mathrm{strong}})$ satisfies the minimal requirements of the resource theory.

\subsection{Physical realization and Kraus representation of strong covariant operations}
Intuitively, a strong covariant operation can be understood as an operation realized by dynamics that do not exchange any conserved quantity with the environment. 
Indeed, the following theorem is valid:
\begin{theorem}\label{thm:SD_strong_StoS}
Let $G$ be a (compact) group, and let $U$ be a unitary representation acting on a Hilbert space $\calH$.
Let $\Lambda:\calB(\calH)\rightarrow \calB(\calH)$ be a $(U,U)$-strong covariant CPTP map.
Then, there exists a auxiliary Hilbert space $\mathcal{H}_E$, a unitary operator $V$ on $\calH\otimes\calH_{E}$, and a state $\sigma_E\in\calS(\calH_E)$ such that
\eq{
    (U_g \otimes I^{(E)})\, V
    &= V\, (U_g \otimes I^{(E)})
    \quad \forall g \in G
    \label{eq:covariant_unitary_general_strong},\\
    \Lambda(\rho)
    &= \Tr_E\!\left[\,V (\rho \otimes \sigma_E) V^\dagger \,\right].\label{eq:covariant_unitary_general_strong2}
}
Conversely, when a CPTP map $\Lambda:\calB(\calH)\rightarrow \calB(\calH)$ can be realized by $(V,\sigma_E)$ satisfying \eqref{eq:covariant_unitary_general_strong}  and \eqref{eq:covariant_unitary_general_strong2}, the map is $(U,U)$-strong covariant.
\end{theorem}

By comparing this theorem with Theorem \ref{thm:SD_weak_StoS'}, one can gain an intuitive understanding of the difference between weak covariant and strong covariant operations. 
As an illustrative example, consider a thermodynamic process involving multiple conserved quantities. 
Suppose that the system of interest has both particle number $N$ and energy $E$ as conserved quantities, but that it exchanges only energy $E$ with a heat bath, while no exchange of particle number occurs. 
Assume that the heat bath is prepared in a Gibbs state with respect to its Hamiltonian $H_B$, which is a weak symmetric state. 

In such a setting, Theorem \ref{thm:SD_strong_StoS} guarantees that the CPTP map induced on the system $S$ is strong covariant with respect to the unitary representation $e^{iN_S\theta}$ generated by the particle-number operator $N_S$. 
At the same time, it is only weak covariant with respect to the unitary representation $e^{iH_St}$ generated by the Hamiltonian, reflecting the fact that energy is exchanged with the heat bath (Figure \ref{fig_st_imp}).

\begin{figure}[t]
\begin{center}
\includegraphics[width=.45\textwidth]{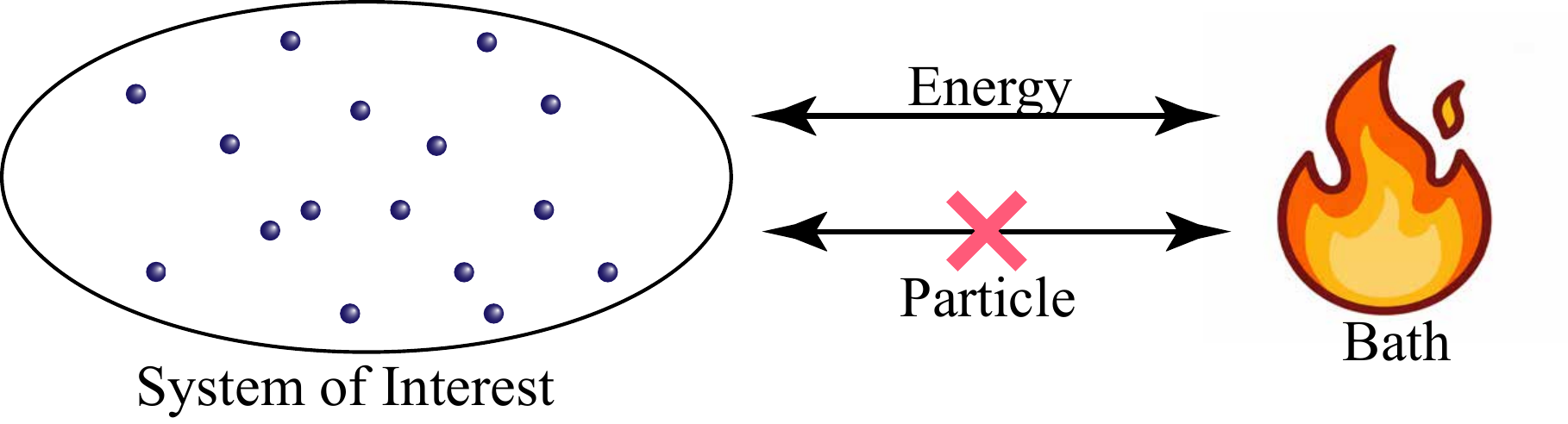}
\caption{Schematic diagram of physical realization of operations that are strong covariant with respect to particle number
and weak covariant with respect to energy.
The system exchanges energy with an external heat bath, while no particles are exchanged.
}
\label{fig_st_imp}
\end{center}
\end{figure}

Furthermore, the Kraus operators of a strong covariant operation always commute with the unitary representation $U_g$ of the symmetry. 
This statement can be established by the following theorem:
\begin{theorem}\label{thm:Kraus_strong}
Let $G$ be a (compact) group, and let $U$ and $U'$ be unitary representations acting on Hilbert spaces $\calH$ and $\calH'$, respectively.
Let $\Lambda:\calB(\calH)\rightarrow \calB(\calH')$ be a $(U,U')$-strong covariant CPTP map.
Then, any Kraus representation $\{K_m\}$ of $\Lambda$ satisfies
\eq{
    K_mU_g=U'_gK_m,\enskip \forall g\in G.\label{eq:K_comm}
}
Conversely, when a CPTP map $\Lambda:\calB(\calH)\rightarrow \calB(\calH')$ has a Kraus representation satisfying \eqref{eq:K_comm}, it is $(U,U')$-strong covariant.
\end{theorem}

This is consistent with the definition of strong symmetry in cond-mat literature, see e.g.\ \cite{Lessa2024,Sala2024,Zapusek2025}, in particular, sec. $II.B$ of \cite{Sala2024}.
As a consequence of Theorem \ref{thm:Kraus_strong}, it follows that if the time evolution generated by the Gorini–Kossakowski–Sudarshan–Lindblad (GKSL) equation is strong covariant, then all jump operators necessarily commute with the symmetry:
\begin{theorem}
Let $S$ be a quantum system whose dynamics obeys the following GKSL equation:
\eq{
\partial_t \rho=-i[H,\rho]+\sum_k\left(L_k\rho L^\dagger_k-\frac{1}{2}\{L^\dagger_kL_k,\rho\}\right).
}
Let $G$ and $U$ be a group and its (projective) unitary representation on $\calH$, the Hilbert space of $S$.
We also assume that the CPTP map $\Lambda_t:\calB(\calH)\rightarrow\calB(\calH)$ realized by the master equation is $(U,U)$-strong covariant for any $t$, and that $[H,U_g]=0$.
Then, any jump operators $\{L_k\}$ of the GKSL equation satisfy
\eq{
[L_k,U_g]=0, \enskip\forall g\in G.\label{com_GKSL}
}
\end{theorem}
Although there may exist multiple sets of jump operators $\{L_k\}$ that generate the same time evolution, \eqref{com_GKSL} holds for any such choice.

\subsection{Relation between strong covariant and weak covariant}\label{sec:strong_weak_relation}
We make several remarks on the relation between the strong covariant operations and the weak covariant operations. First, by definition, any strong covariant operation is also weak covariant.
Second, the converse is not true. Two illustrative examples of weak covariant operations which are not always strong covariant are as follows:
\begin{description}
\item[The partial trace]{The partial trace is not always a strong covariant operation. 
As an example, consider two qubits $S$ and $S'$, and equip the composite system with a $U(1)$ symmetry represented by
\begin{equation}
U_t := e^{-i (H_S + H_{S'}) t},
\end{equation}
where $H_S = \ket{1}\bra{1}_S$ and $H_{S'} = \ket{1}\bra{1}_{S'}$.
In this setting, the state
\begin{equation}
\ket{\psi}_{SS'} := \frac{1}{\sqrt{2}}(\ket{01} + \ket{10})
\end{equation}
is a strong symmetric state.
However, tracing out $S'$ yields the reduced state
\begin{equation}
    \rho_S = \frac{1}{2}\bigl(\ket{0}\bra{0} + \ket{1}\bra{1}\bigr),
\end{equation}
which is not strong symmetric.
Therefore, in this example, taking the partial trace over $S'$ is not a strong covariant operation.

We stress that, sometimes the partial trace is strong covariant. Indeed, for two Hilbert spaces $\calH$ and $\calH_E$, let us assume that $U:G\rightarrow U(\calH\otimes\calH_E)$ and $U':G\rightarrow U(\calH)$ satisfy
\eq{
U_g&=V_g\otimes I_E\\
U'_g&=V_g,
}
where $V:G\rightarrow U(\calH)$ is a unitary representation acting on $\calH$.
Then, we can show that the partial trace $\Tr_E$ is $(U,U')$-strong covariant as follows:
\eq{
Tr_E[U_g(\cdot)]&=V_g\sum_j\bra{j}_E(\cdot)\ket{j}_E\nonumber\\
&=U'_g\Tr_E[\cdot].
}
Therefore, the partial trace taken in Theorem \ref{thm:SD_strong_StoS} is a free (=strong covariant) operation.
} 
\item[Appending a strong symmetric state]{Appending a strong symmetric state may also fail to be a strong covariant operation, depending on the situation. 
This can be seen by revisiting the two-qubit example above with systems $S$ and $S'$. 
Define a map $\Lambda:S\to SS'$ by
\begin{equation}
\Lambda(\cdot)= (\cdot)\otimes \ket{1}\bra{1}_{S'} .
\end{equation}
Then,
\begin{align}
e^{-i(H_S+H_{S'})t}\,\Lambda(\ket{0}\bra{0})&={\color{black}e^{-i(H_S+H_{S'})t}\ket{0}\bra{0}_S\otimes\ket{1}\bra{1}_{S'}}\nonumber\\
&= e^{-it}\,\Lambda(\ket{0}\bra{0})\nonumber\\
&\ne \Lambda\!\left(e^{-iH_S t}\ket{0}\bra{0}\right),
\end{align}
and hence $\Lambda$ is not strong covariant.

We stress that, as the partial trace, sometimes the map appending a strong symmetric state is strong covariant. Indeed, for two Hilbert spaces $\calH$ and $\calH_E$, let us assume that $U:G\rightarrow U(\calH)$ and $U':G\rightarrow U(\calH\otimes\calH_E)$ satisfy
\eq{
U_g&=V_g\\
U'_g&=V_g\otimes I_E,
}
where $V:G\rightarrow U(\calH)$ is a unitary representation acting on $\calH$.
Now, let us define a CPTP map $\Lambda_{\mathrm{add}}:\calB(\calH)\rightarrow\calB(\calH\otimes\calH_E)$ as $\Lambda_{\mathrm{add}}(\cdot):=\cdot\otimes\sigma_E$, where $\sigma_E$ is an arbitrary state $\sigma_E\in\calS(\calH_E)$, which is a strong symmetric state with respect to the trivial representation on $\calH_E$. 
Then, we can easily see that 
\eq{
\Lambda_{\mathrm{add}}(U_g\cdot)&=(V_g\cdot)\otimes\sigma_E\nonumber\\
&=U'_g\Lambda_{\mathrm{add}}(\cdot).
}
Therefore, the operation appending $\sigma_E$ which is taken in Theorem \ref{thm:SD_strong_StoS} is a free (=strong covariant) operation.
}
\end{description}

These operations are usually regarded as free operations in standard resource theories,
even though they are not included in the minimal requirements for free states and free operations.
As a result, readers familiar with resource-theoretic frameworks may find this somewhat unnatural.
However, this behavior is a direct consequence of the required property of strong covariant operations: 
namely, that the channel does not exchange any conserved quantity with the external environment. 
This is precisely the property of ``strong symmetric-preserving'' dynamics often emphasized in previous studies of strong symmetry.

Indeed, the two operations discussed above--the partial trace and the appending of a strong symmetric state--both
create an opportunity either to leak conserved quantities to an external system or to inject conserved quantities
from the outside into the system of interest.
For this reason, neither operation qualifies as strong covariant under our definition.

The restriction that forbids any opportunity for leakage or injection of conserved quantities is very strong, 
and it imposes several additional constraints on strong covariant operations beyond those discussed above.
For example, when a Hilbert space $\calH$ has a larger dimension than another Hilbert space $\calH'$, 
a strong covariant operation from $\calB(\calH)$ to $\calB(\calH')$ does not necessarily exist. 
An illustrative example is in the case of $U(1)$ symmetry, as shown in the following lemma:
\begin{lemma}\label{lemm:no_exist}
Let $\calH$ and $\calH'$ be Hilbert spaces whose dimensions are $d$ and $d'$, respectively, 
and assume $d > d'$. 
Let $U:U(1)\rightarrow U(\calH)$ and $U':U(1)\rightarrow U(\calH')$ be unitary representations which are defined as $U_t:=e^{-iHt}$ and $U'_t:=e^{-iH't}$, respectively. 
If $H$ has strictly more distinct eigenvalues than $H'$, then there exists no $(U,U')$-strong covariant operation 
$\Lambda: \calB(\calH) \rightarrow \calB(\calH')$.
\end{lemma}

\subsection{Resource measures}\label{sec:RM}
In this subsection, we introduce resource measures of strong symmetry breaking, which is one of the main goals of this paper.
Unless otherwise stated, in this section we treat the strong symmetric states as the free states. 
When the single-sector states are taken as the free states, it will be explicitly mentioned.

Since the symmetry depends on the choice of the (projective) unitary representation, 
all quantities introduced in this section are functions of the representation. 
For simplicity, we omit the explicit dependence on the unitary representation 
whenever it is clear from the context, and specify it only when necessary. 
In particular, we write a measure $M$ as $M(\rho)$ when the representation is omitted, 
and write it as $M(\rho\,\|U)$ when the representation is explicitly indicated, e.g. $\Ssa(\rho\|U)$, $L(\rho\|U)$, etc.

\subsubsection{Measures for general-group symmetry}
In this subsection, we introduce two measures which are applicable to a general group $G$.

\begin{definition}[Entanglement asymmetry of strong symmetry breaking]
Let $G$ be a group, and let $U$ be its (projective) unitary representation acting on a system $S$.
We assume that $G$ has Haar measure $\int_G dg=1$ and that $U$ has the irreducible decomposition 
\eq{
U_g = \bigoplus_{\nu} U_g^{(\nu)} \otimes I_{m_\nu}.\label{eq:decom_U}
}
Then, we define $\Ssa(\rho)$ as
\eq{
\Ssa(\rho)&:=H\{p_\nu(\rho)\}+S(\ca{G}(\rho))-S(\rho),\nonumber\\
&=H\{p_\nu(\rho)\}+\Sa(\rho)
}
where
\eq{
S(\rho)&:=-\Tr[\rho\log\rho],\\
\ca{G}(\rho)&:=\int_Gdg U_g\rho U^\dagger_g,\\
p_{\nu}(\rho)&:=\Tr[P_\nu\ca{G}(\rho)]=\Tr[P_\nu\rho],\\
H\{p_\nu\}&:=-\sum_{\nu}p_\nu\log p_{\nu}\,,
}
and $P_\nu$ is the projection to $\calH_{\nu}\otimes \mathbb{C}^{m_\nu}$ in the decomposition of $\calH$ under $U$:
\eq{
\calH = \bigoplus_{\nu} \calH_{\nu} \otimes \mathbb{C}^{m_\nu}.
}
\end{definition}

\medskip
\noindent
\textbf{Examples:}
{\color{black}
Let us consider the one qubit example, mentioned earlier:
$$\rho=\frac{1}{2}\left(\mathbf{1}+r_xX+r_yY+r_zZ\right)\,.$$
For $\mathbb{Z}_2$ or $U(1)$, generated by $Z$, we have 
$$\mathcal{G}(\rho)=\frac{1}{2}\left(\mathbf{1}+r_zZ\right)\,,$$
so that we have 
\begin{equation}
   \begin{aligned}p_\pm &=(1\pm r_z)/2\,,\\
   \Ssa^{\mathbb{Z}_2/U(1)}(\rho)&=\log 2-2f(r_z)+f(|\vec{r}|) \,,  \end{aligned}
\end{equation}
with 
$f(x)=\frac{1+x}{2}\log (1+x)+\frac{1-x}{2}\log (1-x)$. Note that 
\begin{equation}
   0\leq x\leq y \Leftrightarrow \frac{f(x)}{x^2}\leq \frac{f(y)}{y^2}\,,
\end{equation}
(the proof of the inequality can be found in the appendix, in particular, see \eqref{eq0},\eqref{eq1},\eqref{ineq1}) from which it follows that
\begin{equation}
\begin{aligned}
  \Ssa^{\mathbb{Z}_2/U(1)}(\rho)= & \log 2-2 f(r_z)+f(|\vec{r}|)\\
   &\geq \log2- f(r_z)+f(|\vec{r}|)\left(1-\frac{r_z^2}{|\vec{r}|^2}\right)\geq 0\,,
\end{aligned}
\end{equation}
The first inequality above is saturated if $r_x=r_y=0$ and the second inequality is saturated if $r_z=1$ (and by evenness if $r_z=-1$). The Bloch vector $\vec{r}=(0,0,\pm 1)$ precisely corresponds to strong symmetric states under $\mathbb{Z}_2$ and $U(1)$.

Hence, $\Ssa^{\mathbb{Z}_2/U(1)}(\rho)$ , in this example, is faithful i.e.\ it is zero iff the state is strong-symmetric. 
More generally, this measure is faithful with respect to strong symmetric states if and only if all irreducible representations occurring in \(U\) are one-dimensional. The following \(SU(2)\) example illustrates the failure of faithfulness for a two-dimensional irreducible representation.

For $SU(2)$, let us consider the same density matrix $\rho$, now the relevant irrep is spin-$1/2$ representation and $\mathcal{G}(\rho)=\frac{1}{2}\mathbf{1}$, resulting in $p_{\nu}=1$, $H\{p_\nu(\rho)\}=0$ and 
\begin{equation}
    \Ssa^{SU(2)}=f(|\vec{r}|)
\end{equation}
As mentioned previously, there is no strong-symmetric state in this case, however, $f(0)=0$, hence the resource monotone is not faithful.
}
\begin{theorem}\label{thm:strong_EA}
The quantity $\Ssa(\rho)$ satisfies the following features
\begin{description}
\item[(A)] $\Ssa(\rho)$ is a resource measure. Namely, (i-a) it is non-negative and when $\rho$ is strong symmetric, it is zero, and (i-b) when a CPTP map $\Lambda:\calB(\calH)\rightarrow\calB(\calH')$ is $(U,U')$-strong covariant, \eq{\Ssa(\rho\|U)\ge\Ssa(\Lambda(\rho)\|U')
}
\item[(B)] It is faithful when we employ $\calF_{G,\mathrm{single}}$ as free states. In other words, $\Ssa(\rho)=0$ if and only if $\rho$ is a single-sector state.
\item[(C)] It is \textit{not} necessarily faithful (=property (ii)) when we employ $\calF_{G,\mathrm{strong}}$ as free states. In other words, there exists a resource state $\rho$ (=state which is not strong symmetric) satisfying $\Ssa(\rho)=0$.
\item[(D)] With collective representations, $\Ssa$ is \textit{not additive in general} for product states. In other words, there exist a group $G$, two systems $A$ and $B$, unitary representations $U^A$ and $U^B$ of $G$, and states $\rho_A$ and $\sigma_B$ satisfying
\eq
{&\Ssa(\rho_A\otimes\sigma_B\|U^A\otimes U^B)\nonumber\\
&\ne\Ssa(\rho_A\|U^A)+\Ssa(\sigma_B\|U^B).}
\item[(E)] When $S$ is a finite-dimension system and $G$ is a compact Lie group or a finite group, for any $\rho$ on $S$ and any representation $U$ of $G$ acting on $S$, 
\eq{
\lim_{n\rightarrow \infty}\frac{\Ssa(\rho^{\otimes n}\|U^{\otimes n})}{n}=0.
}
\end{description}
\end{theorem}

The second one is given by the logarithmic characteristic function.
\begin{definition}[Averaged logarithmic characteristic function]
Let $G$ be a group, and let $U$ be its (projective) unitary representation acting on a Hilbert space $\calH$.
We assume that $G$ has Haar measure $\int_G dg=1$ and that $U_g$ has the irreducible decomposition \eqref{eq:decom_U}.
Then, we define $L(\rho)$ as
\eq{
L(\rho):=\int_GdgL_{g,\mathrm{strong}}(\rho),\label{def_SLCF_a}
}
where 
\eq{\label{eq:logCh2}
L_{g,\mathrm{strong}}(\rho):=-\log|\Tr[U_g\rho]|,
}
and when $G$ is finite, $\int_Gdg$ is replaced by $\frac{1}{|G|}\sum_g$.
\end{definition}
\medskip
\noindent
\textbf{Examples:}
{\color{black} Let us go back to the one qubit example: $\rho=\frac{1}{2}\left(\mathbf{1}+r_xX+r_yY+r_zZ\right)$. One can compute
\begin{equation}
\begin{aligned}
    L^{\mathbb{Z}_2}(\rho)&=-\frac{1}{2}\log |r_z|\,,\\
    L^{U(1)}(\rho)&=- \log \frac{1+|r_z|}{2}\,.
\end{aligned}
\end{equation}
We can see that $\rho$ is strong-symmetric state with respect to $\mathbb{Z}_2$ or/and $U(1)$  iff $r_z=\pm 1$.

Now we consider the asymmetry due to $SU(2)$. Earlier we have pointed out that the $\vec{r}=0$ state is not strong symmetric even if 
$\Ssa^{SU(2)}=0$ for this state. In contrast, now we have
\begin{equation}
    L^{SU(2)}\left(\frac{1}{2}\mathbf{1}\right)=\frac{1}{2}+\log 2>0\,,
\end{equation}
indicating this is indeed not strong-symmetric. In the next theorem, we are going to show that this resource monotone is indeed faithful.}

\begin{theorem}\label{thm:SLCF}
The quantity $L(\rho)$ satisfies the following features
\begin{description}
\item[(A)] $L(\rho)$ is a resource measure. Namely, (i-a) it is non-negative and when $\rho$ is strong symmetric, it is zero, and (i-b) when $\Lambda$ is $(U,U')$-strong covariant, the inequality $L(\rho)\ge L(\Lambda(\rho))$ holds. More precisely, $L(\rho)=L(\Lambda(\rho))$.
\item[(B)] It is faithful for $\calF_{G,\mathrm{strong}}$. In other words, $L(\rho)=0$ if and only if $\rho$ is strong symmetric.
\item[(C)] When we employ collective representation, it is additive for product states. In other words, for two systems $A$ and $B$ and (projective) unitary representations $U^{A}$ and $U^{B}$ on them, any states $\rho_A$ and $\sigma_B$ satisfy
\eq{
L(\rho_A\otimes\sigma_B\|U^A\otimes U^B)=L(\rho_A\|U^A)+L(\sigma_B\|U^B).
}

Consequently, 
\eq{
\lim_{n\rightarrow \infty} \fr{L(\rho^{\otimes n}\|U^{\otimes n})}{n}=L(\rho)
}
holds.
\end{description}
\end{theorem}

We remark that each $L_\rho(g)$ is also a resource measure and additive, but it is not faithful.

\begin{table*}[htbp]
\centering
\begin{tabular}{l||c|c|c|c|c|c|c|}
    & Group & (i-a) for $\calF_{G,\mathrm{strong}}$ & (i-a) for $\calF_{G,\mathrm{single}}$ & (i-b)  & (ii) for $\calF_{G,\mathrm{strong}}$ & (ii) for $\calF_{G,\mathrm{single}}$ & (v) \\\hline\hline
$\Ssa$ & arbitrary & \checkmark & \checkmark & \checkmark & $\times$ & \checkmark  & $\times$ \\
$L$ & arbitrary & \checkmark & $\times$ & \checkmark & \checkmark  & $\times$  & \checkmark \\
$\VNS$ & compact Lie & \checkmark & $\times$ & \checkmark & \checkmark (for connected $G$)  & $\times$  & \checkmark \\
$\VS$ &compact Lie & \checkmark & $\times$ & \checkmark & \checkmark (for connected $G$)  & $\times$  & \checkmark \\\hline
\end{tabular}
\caption{Table of measures of strong symmetry breaking
}
\end{table*}

\subsubsection{Measures for compact-Lie-group symmetry}
Next, we introduce a resource measure that works particularly well when the group $G$ is a compact Lie group:
\begin{definition}[Non-symmetrized and symmetrized covariance matrices]
Let $G$ be a compact Lie group, and let $d_G:=\dim G$ denote the dimension of $G$ as a smooth manifold. Then, elements in the neighborhood of the identity $e\in G$ can be parametrized as $g(\bm{\lambda})=e^{\ii \sum_{j=1}^{d_G}\lambda^j A_j}$ with a basis $\{A_j\}_{j=1}^{d_G}$ of the Lie algebra $\mathfrak{g}$. 
Let $U$ be a unitary representation of $G$ acting on a Hilbert space $\calH$.
We introduce Hermitian operators 
\eq{
    X_j:=  -\ii \frac{\partial}{\partial \lambda^j} U(g(\bm{\lambda}))\biggl|_{\bm{\lambda}=\bm{0}}\quad \label{eq:hermitian_operators}
}
for $j=1,\cdots,d_G$, which corresponds to $L(A_j)$, where $L$ is the Lie algebra representation defined as $L(A):= -\ii \left.\frac{d}{dt }U(e^{\ii t A})\right|_{t=0}$.
Then, we define two types of covariance matrices associated with $\rho$.
\begin{enumerate}
    \item \textbf{Non-symmetrized covariance matrix.}
    The non-symmetrized covariance matrix $\VNS(\rho)$ is defined by
    \begin{equation}
        (\VNS)_{i,j}
        := \Tr[\rho\, X_i X_j]
           - \langle X_i \rangle_\rho \langle X_j \rangle_\rho ,
    \end{equation}
    where $\langle X \rangle_\rho := \Tr[\rho X]$.

    \item \textbf{Symmetrized covariance matrix.}
    The symmetrized covariance matrix $\VS(\rho)$ is defined by
    \begin{equation}
        (\VS(\rho))_{i,j}
        := \frac{1}{2}\Tr[\rho\,\{X_i, X_j\}]
           - \langle X_i \rangle_\rho \langle X_j \rangle_\rho ,
    \end{equation}
    where $\{X_i, X_j\} := X_i X_j + X_j X_i$ denotes the anti-commutator.
\end{enumerate}
\end{definition}

We remark that when the symmetry is $U(1)$, both covariance matrix becomes the variance.

\begin{theorem}\label{thm:VM}
The covariance matrices $\VNS(\rho)$ and $\VS(\rho)$ satisfy the following features
\begin{description}
\item[(A)] $\VNS(\rho)$ and $\VS(\rho)$ are resource measures. Namely, (i-a) they are positive-semidefinite matrices, and when $\rho$ is strong symmetric, they are zero, and (i-b) when $\Lambda$ is $(U,U')$-strong covariant, the inequalities $\VNS(\rho)\ge \VNS(\Lambda(\rho))$ and $\VS(\rho)\ge \VS(\Lambda(\rho))$ hold. More precisely, $\VNS(\rho)=\VNS(\Lambda(\rho))$ and $\VS(\rho)= \VS(\Lambda(\rho))$ are valid.
\item[(B)] When $G$ is connected, they are faithful for $\calF_{G,\mathrm{strong}}$. In other words, 
\eq{
\VNS(\rho)=0&\Leftrightarrow\VS(\rho)=0\nonumber\\
&\Leftrightarrow\rho\in\calF_{G,\mathrm{strong}}.
}
Using (i-a) of $\VS$ and $\VNS$, we can equivalently state
\eq{
\sum_i\VNS(\rho)_{ii}=0\Leftrightarrow\sum_i\VS(\rho)_{ii}=0\Leftrightarrow\rho\in\calF_{G,\mathrm{strong}}.
}
\item[(C)] When we employ collective representation, they are additive for product states. In other words, for two systems $A$ and $B$ and unitary representations $U^{A}$ and $U^{B}$ on them, any states $\rho_A$ and $\sigma_B$ satisfy 
\eq{
&\VNS(\rho_A\otimes\sigma_B\|U^A\otimes U^B)\nonumber\\
&=\VNS(\rho_A\|U^A)+\VNS(\sigma_B\|U^B),\\
&\VS(\rho_A\otimes\sigma_B\|U^A\otimes U^B)\nonumber\\
&=\VS(\rho_A\|U^A)+\VS(\sigma_B\|U^B)
}
Consequently, $\lim_{n\rightarrow \infty}\VNS(\rho^{\otimes n}\|U^{\otimes n})/n=\VNS(\rho)$ and $\lim_{n\rightarrow \infty}\VS(\rho^{\otimes n}\|U^{\otimes n})/n=\VS(\rho)$ hold.
\end{description}
\end{theorem}

\subsubsection{i.i.d.\ complete measures for compact-Lie group symmetry}\label{sec:strong_iid}

We next consider i.i.d.\ conversion of strong symmetry breaking, in analogy with entanglement distillation and dilution.
As discussed in Section~\ref{sec:RT_basic}, the question is at what rate copies of one state can be produced from many identical copies of another.
In this subsection, we give necessary and sufficient conditions for such i.i.d.\ asymptotic conversion between arbitrary states for a general compact Lie group.

Let $G$ be a compact Lie group, and fix continuous unitary representations $U,U'$ on the finite-dimensional input and output Hilbert spaces $\calH,\calH'$.
In this subsection, the representations are non-projective, and $G$ need not be connected or Abelian.
The collective representations on $N$ input copies and $M$ output copies are $U^{\otimes N}$ and $U'^{\otimes M}$, respectively.
Strong covariant operations cannot freely discard charge, and without an auxiliary system there may be no strong covariant CPTP map at all between systems with different copy numbers (Lemma~\ref{lemm:no_exist}).
We therefore allow an auxiliary system that retains a strong symmetric state after conversion, as specified below. We also adopt a definition under which the conversion error need only converge to zero along a suitably chosen subsequence of input copy numbers, as in Ref.~\cite{Kuroiwa2020generalquantum}. As explained below, these definitions give the same i.i.d.\ conversion rate in many standard resource theories, including the resource theory of asymmetry for weak symmetry.

More precisely, for a rate $r>0$, set $M_N:=\lfloor rN\rfloor$.
For each $N$, consider a finite-dimensional auxiliary system $A_N$, a continuous unitary representation $W_N$ on it, a state $\eta_{A_N}$ that is strong symmetric with respect to $W_N$, and a $(U^{\otimes N},U'^{\otimes M_N}\otimes W_N)$-strong covariant CPTP map
\eq{
\Phi_N:\calB(\calH^{\otimes N})
\longrightarrow\calB(\calH'^{\otimes M_N}\otimes\calH_{A_N})
}
We call $r$ achievable if there is a sequence $\{A_N,W_N,\eta_{A_N},\Phi_N\}_{N\in\mathbb{N}}$ such that
\eq{
\liminf_{N\to\infty}\Bigl\|\Phi_N(\rho^{\otimes N})-\sigma^{\otimes M_N}\otimes\eta_{A_N}\Bigr\|_1=0
\label{eq:strong_iid_conversion}
}
holds.
This is equivalent to the existence of operations and auxiliary states for which the actual conversion error converges to zero along a suitably chosen subsequence of input copy numbers.
The dimension and representation of the auxiliary system may depend on $N$, but the error is evaluated on the entire output, including the auxiliary system.
Thus, arbitrary resources cannot be left in the auxiliary system: the target output is a product with a strong symmetric state.
Writing $\calR_+(\rho\to\sigma)$ for the set of positive achievable rates, we define the optimal conversion rate by
\eq{
R(\rho\to\sigma)
:=\sup\bigl(\{0\}\cup\calR_+(\rho\to\sigma)\bigr).
\label{eq:strong_iid_rate}
}

Including $0$ ensures that $R=0$ when no positive rate is achievable.

In standard resource theories where appending a free state and taking a partial trace are both free operations, allowing a free auxiliary output state does not change the conversion rate.
Moreover, in many standard resource-conversion problems, including asymmetry for weak symmetry and entanglement conversion between bipartite pure states, allowing subsequences as above does not change the optimal rate relative to the definition with $\lfloor rN\rfloor$ output copies for every sufficiently large integer $N$~\cite{YT2026_1,Bennett1996}.
In these cases, the necessary conditions that bound the optimal rate apply along any subsequence, and every rate below that bound is achievable for all sufficiently large $N$.

\medskip
\noindent\textbf{$U(1)$ symmetry: Conversion rates determined by the variance ratio}

First, consider representations $U_t=e^{-iHt}$ and $U'_t=e^{-iH't}$ of $U(1)$.
Writing the auxiliary representation as $W_N(t)=e^{-iH_{A_N}t}$, note that $H_{A_N}$ may have negative eigenvalues.
Using $2\pi$-periodic parametrizations of $U_t,U'_t$ with real $t$, define the strong-symmetry period of $\rho$ by
\eq{
\taust(\rho\|H)
:=\inf\{t>0:\exists\theta\in\mathbb R,\ e^{-iHt}\rho=e^{i\theta}\rho\}.
}

For the output state, define $\taust(\sigma\|H')$ using $H'$.
For a pure state, this agrees with the usual state period $\inf\{t>0:e^{-iHt}\rho e^{iHt}=\rho\}$.

\begin{theorem}\label{thm:U1_variance}
For $U(1)$ symmetry, suppose that $\rho,\sigma$ satisfy
\eq{
V_H(\rho)>0,\qquad V_{H'}(\sigma)>0,\\
\taust(\rho\|H)=\taust(\sigma\|H').
}

Here, $V_H(\rho):=\Tr[\rho H^2]-\Tr[\rho H]^2$.
If the input $\rho$ is pure or the output $\sigma$ is weak symmetric, i.e., $[\sigma,H']=0$, a positive rate $r$ is achievable if and only if
\eq{
r=\frac{V_H(\rho)}{V_{H'}(\sigma)}.
}
Consequently,
\eq{
R(\rho\to\sigma)=\frac{V_H(\rho)}{V_{H'}(\sigma)}.
\label{eq:U1_variance_rate}
}
\end{theorem}

This theorem covers both conversion between pure states and conversion between weak symmetric states.
Moreover, a pure input may be converted to a mixed output, and for a weak symmetric output the input need not be pure or weak symmetric.
For these conversions, the variance completely characterizes i.i.d.\ convertibility once the periods agree.
For pure-state conversion, its role is analogous to that of entanglement entropy in entanglement theory.
It is interesting that the asymptotic conversion rate is determined by the variance alone, although strong covariant operations preserve the entire charge distribution.
For general mixed states, however, the variance and period are insufficient even for $U(1)$.
The following theorem gives the conversion conditions for a general compact Lie group, including this case.

\medskip
\noindent\textbf{Compact Lie groups: Conversion conditions for arbitrary states}

For a general group, we need information about charges that cannot be adjusted by strong symmetric auxiliary states, in addition to the non-symmetrized covariance matrix $\VNS$, which describes total charge fluctuations including classical fluctuations, and the metric adjusted quantum geometric tensor $\calQ^{\rho,q}$, which describes quantum fluctuations and phase information.
More specifically, building on the asymptotic bounds for metric adjusted quantum geometric tensors and the methods for comparing Gaussian models developed in Ref.~\cite{YT2026_1}, we develop a theory that gives necessary and sufficient conditions for i.i.d.\ asymptotic conversion of strong symmetry breaking.
To state the result, first define the stabilizer subgroup for strong symmetry by
\eq{
\mathrm{Sym}_{G,\mathrm{strong}}(\rho):=\{g\in G:\exists c\in U(1),\ U_g\rho=c\rho\}.
\label{eq:strong_iid_K}
}

This is the closed subgroup of elements that act with the same phase on the entire support of $\rho$, and it satisfies $\mathrm{Sym}_{G,\mathrm{strong}}(\rho)\subseteq\mathrm{Sym}_G(\rho)$.
Define $\mathrm{Sym}_{G,\mathrm{strong}}(\sigma)$ analogously using $U'$.

Next, define the residual mean charge $m_G(\rho)$, which describes the component of the mean charge that cannot be adjusted by strong symmetric auxiliary states, by
\eq{
m_G(\rho)_i
:=\Tr\bigl[(\rho-\calG_U(\rho))X_i\bigr].
\label{eq:strong_iid_residual_mean}
}

Here, $\calG_U$ is the twirling operation defined in \eqref{eq:twir}, with the subscript $U$ making its representation dependence explicit. On the output side, we use $U',X_i'$.
The real vector $m_G$ extracts the component of the mean charge removed by twirling and is unchanged by appending a strong symmetric auxiliary state.
Twirling is used here to define this quantity; it is not assumed to be a free operation.

\begin{theorem}\label{thm:compact_Lie_strong_iid}
Consider finite-dimensional non-projective continuous unitary representations $U,U'$ of a compact Lie group $G$ and arbitrary states $\rho,\sigma$.
A rate $r>0$ is achievable under the $\liminf$ definition in \eqref{eq:strong_iid_conversion} if and only if both of the following sets of conditions hold.
\begin{description}
\item[(W)] Conditions already required by weak covariance:
\eq{
\calQ^{\rho,q}&\ge r\calQ^{\sigma,q}
\quad(\forall q\in[1/2,1]),
\label{eq:strong_iid_QGT}\\
\mathrm{Sym}_G(\rho)&\subseteq\mathrm{Sym}_G(\sigma).
\label{eq:strong_iid_weak_stabilizer}
}
\item[(S)] Additional conditions imposed by strong covariance:
\eq{
\VNS(\rho)&=r\VNS(\sigma),
\label{eq:strong_iid_covariance}\\
m_G(\rho)&=r m_G(\sigma),
\label{eq:strong_iid_mean}\\
\mathrm{Sym}_{G,\mathrm{strong}}(\rho)&=\mathrm{Sym}_{G,\mathrm{strong}}(\sigma).
\label{eq:strong_iid_strong_stabilizer}
}
\end{description}
\end{theorem}

Condition (W) requires the metric adjusted quantum geometric tensors not to increase and the symmetry subgroups to satisfy the stated inclusion; it characterizes the optimal rate of i.i.d.\ asymptotic conversion under weak covariant operations~\cite{YT2026_1}.
In addition, strong covariance imposes condition (S): proportionality of the fluctuations and residual mean charges, together with equality of the strong-symmetry subgroups, as required by exact charge conservation.

If $\VNS(\sigma)\ne0$, at most one positive rate satisfies \eqref{eq:strong_iid_covariance}, and the remaining conditions determine whether that rate is achievable.
Thus, rates below an achievable rate need not all be achievable.
If $G$ is disconnected, $\VNS(\sigma)=0$ does not imply that $\sigma$ is strong symmetric, so the subgroup conditions in (W) and (S) cannot be omitted.
For a finite group, the conditions involving generators are vacuous. A positive rate $r$ is then achievable if and only if the symmetry-subgroup inclusion in (W) and the strong-symmetry-subgroup equality in (S) hold.

Finally, we verify that Theorem~\ref{thm:U1_variance} follows from this general result.
\begin{proof}[Proof of Theorem~\ref{thm:U1_variance}]
For $U(1)$, we have $m_G=0$, so the residual-mean-charge condition in (S) is automatic.
The covariance matrix reduces to the variance, and its proportionality in (S) gives $r=V_H(\rho)/V_{H'}(\sigma)$.
When the variances are positive, $\mathrm{Sym}_{G,\mathrm{strong}}(\rho)=\mathrm{Sym}_{G,\mathrm{strong}}(\sigma)$ is equivalent to equality of the strong-symmetry periods.
With a single generator, the metric adjusted quantum geometric tensor is a scalar and satisfies $0\le\calQ^{\omega,q}\le V_L(\omega)$ for any state $\omega$ and its generator $L$.
This upper bound follows from the convexity of $\calQ^{\omega,q}$, its equality with the variance on pure states, and the concavity of the variance.
If the input $\rho$ is pure, \eqref{eq:metric_adjusted_QGT_pure} gives $\calQ^{\rho,q}=V_H(\rho)$, so the tensor inequality in (W) holds at the above rate $r$.
Moreover, for pure states $\mathrm{Sym}_G(\rho)=\mathrm{Sym}_{G,\mathrm{strong}}(\rho)$, so $\mathrm{Sym}_{G,\mathrm{strong}}(\rho)=\mathrm{Sym}_{G,\mathrm{strong}}(\sigma)\subseteq\mathrm{Sym}_G(\sigma)$ also gives the inclusion in (W).
If the output $\sigma$ is weak symmetric, then $\calQ^{\sigma,q}=0$ and $\mathrm{Sym}_G(\sigma)=G$, so these conditions are again automatic.
Thus, in either case the variance ratio is achievable, and the covariance proportionality in (S) excludes every other positive rate.
\end{proof}

For general mixed states, even under $U(1)$ symmetry, the tensor inequalities and symmetry-subgroup inclusion in (W) are not automatic.
The variance describes total charge fluctuations, but the part associated with quantum information must also not increase under conversion.
Thus, general i.i.d.\ conversion is characterized by both conserved fluctuations and quantum information that may decrease.

\subsection{Conversion from weak symmetry breaking to strong symmetry breaking}
Reviewing the results obtained so far, we find that both the variance and the Fisher information
serve as good measures of strong and weak symmetry breaking, respectively.
Both quantities provide i.i.d.\ complete measures for a rather broad class of states,
including, at least, all pure states.
These two quantities admit a clear physical interpretation:
they quantify, respectively, the total amount of fluctuations of the conserved quantity
and the portion of fluctuations that originates purely from quantum superposition.

To make this distinction explicit, let us introduce a property of the SLD skew informatation (=a quarter of QFI)~\cite{Yu2013}:
\eq{
\tilde{I}_{H}(\rho)=\min_{\{q_j,\ket{\phi_j}\}}\sum_jq_jV_{\phi_j}(H),\label{eq:Yu}
}
where $\{q_j\ket{\phi_j}\}$ runs over all possible decomposition $\sigma=\sum_jq_j\phi_j$ including non-orthogonal ones.
The relation \eqref{eq:Yu} implies that the skew information quantifies the ``quantum part" of the variance, i.e., the fluctuation of $H$ by quantum superposition.
Furthermore, together with the concavity of the variance, the relation \eqref{eq:Yu} immediately implies that
the skew information is always upper bounded by the variance.
Indeed, the variance can be decomposed as \cite{Gibilisco2007,Gibilisco2011}
\eq{
V_H(\rho)=\tilde{I}_H(\rho)+C_H(\rho).\label{eq:decom_QC}
}
where the term $C$ may be regarded as the contribution arising from classical fluctuations.
This term can be defined as a variance with respect to the right logarithmic derivative (RLD) metric,
\eq{
C_H(\rho):=\sum_{k,l}p_{k}f_{RLD}(p_l/p_k)\left|\bra{k}H_\rho\ket{l}\right|^2,
}
where $\{p_k\}$ and $\{\ket{k}\}$ are eigenvalues and eigenvectors of $\rho$,  $f_{RLD}(x):=2x/(x+1)$ and $H_\rho:=H-I\Tr[H\rho]$,
whereas the usual variance corresponds to the variance defined with respect to the SLD metric:
\eq{
V_H(\rho)=\sum_{k,l}p_{k}f_{SLD}(p_l/p_k)\left|\bra{k}H_\rho\ket{l}\right|^2,
}
where $f_{SLD}(x):=(1+x)/2$.

Using this decomposition, we can quantitatively evaluate how weak symmetry breaking is converted into strong symmetry breaking over time in open-system dynamics that do not exchange conserved quantities with the environment.
First, the variance is invariant under strong covariant operations (Theorem \ref{thm:VM}).
By contrast, the quantum Fisher information monotonically decreases under strong covariant operations.
In other words, under strong covariant time evolution, the overall amount of symmetry breaking is conserved,
while it is irreversibly transformed from weak symmetry breaking into strong symmetry breaking.

Note that, by Theorem~\ref{thm:SD_strong_StoS}, any discrete time evolution that does not exchange the conserved quantity $H$ with the environment
is always described by a strong covariant operation.
This includes strong-covariant dynamics generated by a time-independent GKSL equation.
Our framework, therefore, provides a quantitative characterization of the transition
from weak symmetry breaking to strong symmetry breaking under such dynamics.
In particular, the question of
“how much of the strong symmetry breaking is purely strong, i.e. not weak symmetry breaking”
can be quantified by the following quantity:
\eq{
\calR_H(\rho):=\frac{C_H(\rho)}{V_H(\rho)}.
}

A similar behavior can be observed for symmetries described by an arbitrary compact Lie group.
The SLD quantum Fisher information matrix (SLD-QFIM) serves as a measure of weak symmetry breaking~\cite{Kudo2023}.
On the other hand, the symmetrized covariance matrix serves as a measure of strong symmetry breaking.
These two quantities are related by a decomposition analogous to \eqref{eq:decom_QC}, namely,
\eq{
\VS(\rho)=\tilde{\calI}^{\rho}+C^{\rho}
}
Here, $\tilde{\calI}^{\rho}:=\calI^{\rho}/4$, and the matrix $C$ is the covariance matrix associated with the RLD metric,
whose elements $C^{\rho}_{ij}$ are defined as
\eq{
C^{\rho}_{ij}=\sum_{k,l}p_{k}f_{RLD}(p_l/p_k)\bra{k}X_{i,\rho}\ket{l}\bra{l}X_{j,\rho}\ket{k},
}
where $X_{i,\rho}:=X_i-I\Tr[\rho X_i]$.

As before, the symmetrized covariance matrix is conserved under strong covariant operations,
whereas the SLD-QFIM decreases monotonically under such operations.
Therefore, when symmetry is described by a compact Lie group,
strong covariant time evolution preserves the total amount of symmetry breaking
while irreversibly converting weak symmetry breaking into strong symmetry breaking.

\section{Resource theory for the quasi-strong covariant operations}\label{sec:qstrong}

As in Section~\ref{sec:RStrong}, we restrict attention to quantum systems with finite-dimensional Hilbert spaces.
In this section, we introduce quasi-strong covariant operations by allowing the output representation in the definition of a strong covariant operation to be multiplied by a one-dimensional unitary representation.
This defines a new class of free operations and an associated resource theory.
As we will see, this class is a union of classes of strong covariant operations with different output representations, so each component is described by the strong covariance condition.
Although the class of free operations is enlarged, the following features of the resource theory are preserved:

(i) The minimal requirements of a resource theory are satisfied when either strong symmetric states or single-sector states are chosen as free states.

(ii) The main resource measures introduced in this paper remain resource measures for their respective choices of free states, and the conservation laws already established continue to hold.
More generally, resource measures that are invariant under character twists carry over to the theory with quasi-strong covariant operations.

(iii) The optimal i.i.d.\ conversion rate for $U(1)$ symmetry, as defined in this paper, is unchanged when strong covariant operations are replaced by quasi-strong covariant operations.
For the state conversions considered above, the variance remains the i.i.d.\ complete resource measure.

(iv) For finite-dimensional dynamics generated by a time-independent GKSL equation, the time-evolution maps are quasi-strong covariant at every time if and only if they are strong covariant at every time.

We write $\CharG$ for the set of continuous one-dimensional unitary representations of $G$.
For (projective) unitary representations $U^X$ and $U^Y$ of $G$, we denote the set of $(U^X,U^Y)$-strong covariant operations by
$\calO_{G,\mathrm{strong}}(U^X,U^Y)$.
For $\xi\in\CharG$, define
\eq{
(\xi U^Y)_g:=\xi(g)U_g^Y.
}
Then $\xi U^Y$ is a (projective) unitary representation with the same cocycle as $U^Y$.

\begin{definition}[Quasi-strong covariant operation]\label{def:qstrong}
A CPTP map $\Lambda:\calB(\calH_X)\to\calB(\calH_Y)$ is said to be $(U^X,U^Y)$-quasi-strong covariant if there exists
$\xi\in\CharG$ such that $\Lambda$ is $(U^X,\xi U^Y)$-strong covariant, namely,
\eq{
\Lambda(U_g^X A)=\xi(g)U_g^Y\Lambda(A),
\qquad \forall A\in\calB(\calH_X),\quad \forall g\in G.
\label{eq:qstrong_def}
}
We denote the set of these operations by $\calO_{G,\mathrm{qstrong}}(U^X,U^Y)$, and the class of all quasi-strong covariant operations by $\calO_{G,\mathrm{qstrong}}$.
\end{definition}

By definition,
\eq{
\calO_{G,\mathrm{qstrong}}(U^X,U^Y)
=\bigcup_{\xi\in\CharG}
\calO_{G,\mathrm{strong}}(U^X,\xi U^Y).
\label{eq:qstrong_union}
}
This relation allows us to establish the results outlined above for free states, resource measures, and i.i.d.\ complete measures.
We begin with two lemmas.
The first shows that the character associated with a quasi-strong covariant operation is unique.

\begin{lemma}\label{lem:qstrong_character_unique}
Let $X$ and $Y$ be quantum systems with (projective) unitary representations $U^X$ and $U^Y$.
For any $(U^X,U^Y)$-quasi-strong covariant channel $\Lambda$, the character $\xi\in\CharG$ satisfying \eqref{eq:qstrong_def} is unique.
\end{lemma}

\begin{proof}
Fix any state $\rho\in\calS(\calH_X)$.
Multiplying both sides of \eqref{eq:qstrong_def} from the left by $(U_g^Y)^\dagger$ and taking the trace, we obtain, for every $g\in G$,
\eq{
\xi(g)=\Tr\!\left[(U_g^Y)^\dagger\Lambda(U_g^X\rho)\right],
}
where we used $\Tr[\Lambda(\rho)]=1$.
The right-hand side is determined by $U^X$, $U^Y$, $\Lambda$, and the fixed state $\rho$, so $\xi$ is unique.
\end{proof}

The second lemma shows that replacing $U$ by $\xi U$ does not change the free states.

\begin{lemma}\label{lem:qstrong_free_states_twist}
For any $\xi\in\CharG$, the sets of strong symmetric states with respect to $U$ and $\xi U$ coincide.
If $U$ admits an irreducible decomposition, the sets of single-sector states with respect to $U$ and $\xi U$ also coincide.
\end{lemma}

\begin{proof}
If $U_g\rho=e^{i\theta_g}\rho$, then $(\xi U)_g\rho=\xi(g)e^{i\theta_g}\rho$.
The converse follows by multiplying by $\xi^{-1}$, so the sets of strong symmetric states coincide.
Moreover, each irreducible representation $U^{(\nu)}$ is replaced by $\xi U^{(\nu)}$ on the same subspace.
A state is weak symmetric with respect to $U$ if and only if it is weak symmetric with respect to $\xi U$.
Therefore, the sets of single-sector states also coincide.
\end{proof}

We now examine the resource-theoretic structure.
The following theorem shows that both $(\calF_{G,\mathrm{strong}},\calO_{G,\mathrm{qstrong}})$ and $(\calF_{G,\mathrm{single}},\calO_{G,\mathrm{qstrong}})$ satisfy the minimal requirements of a resource theory.

\begin{theorem}\label{thm:qstrong_resource_theory}
Let $G$ be a compact group. The following statements hold.
\begin{description}
\item[(i)] The identity channel is quasi-strong covariant.
\item[(ii)] Let $U^X$, $U^Y$ and $U^Z$ be (projective) unitary representations of $G$ on the systems $X$, $Y$ and $Z$. If $\Lambda_1:X\to Y$ is $(U^X,U^Y)$-quasi-strong covariant and $\Lambda_2:Y\to Z$ is $(U^Y,U^Z)$-quasi-strong covariant, then $\Lambda_2\circ\Lambda_1$ is $(U^X,U^Z)$-quasi-strong covariant.
\item[(iii)] Quasi-strong covariant operations map strong symmetric states to strong symmetric states and single-sector states to single-sector states.
\end{description}
Thus, both $(\calF_{G,\mathrm{strong}},\calO_{G,\mathrm{qstrong}})$ and $(\calF_{G,\mathrm{single}},\calO_{G,\mathrm{qstrong}})$ satisfy the minimal requirements of a resource theory.
\end{theorem}

\begin{proof}
We first prove (i). By \eqref{eq:qstrong_union}, $\calO_{G,\mathrm{qstrong}}(U^X,U^X)$ contains $\calO_{G,\mathrm{strong}}(U^X,U^X)$.
Since the latter contains the identity channel, so does $\calO_{G,\mathrm{qstrong}}(U^X,U^X)$.

Next, we prove (ii). By the definition of quasi-strong covariance, there exists a continuous one-dimensional unitary representation $\xi^Y$ of $G$ such that $\Lambda_1\in\calO_{G,\mathrm{strong}}(U^X,\xi^Y U^Y)$.
Similarly, there exists a continuous one-dimensional unitary representation $\xi'^Z$ of $G$ such that $\Lambda_2\in\calO_{G,\mathrm{strong}}(U^Y,\xi'^Z U^Z)$.
Strong covariant operations are closed under composition, and linearity gives $\Lambda_2(\xi^Y_g A)=\xi^Y_g\Lambda_2(A)$ for all $g\in G$ and $A\in\calB(\calH_Y)$.
Consequently, $\Lambda_2\circ\Lambda_1\in\calO_{G,\mathrm{strong}}(U^X,\xi^Y\xi'^Z U^Z)$.
Since $\xi^Y\xi'^Z\in\CharG$, the composition $\Lambda_2\circ\Lambda_1$ is $(U^X,U^Z)$-quasi-strong covariant.

Finally, (iii) follows from \eqref{eq:qstrong_union}, the fact that strong covariant operations map strong symmetric states to strong symmetric states and single-sector states to single-sector states, and Lemma~\ref{lem:qstrong_free_states_twist}.
\end{proof}

Next, we show that $\calO_{G,\mathrm{qstrong}}$, unlike $\calO_{G,\mathrm{strong}}$, need not be convex.

\begin{lemma}\label{lem:qstrong_nonconvex}
Each $\calO_{G,\mathrm{strong}}(U^X,\xi U^Y)$ is convex, whereas
$\calO_{G,\mathrm{qstrong}}(U^X,U^Y)$ need not be convex.
\end{lemma}

\begin{proof}
For fixed $\xi$, the strong covariance condition is linear in $\Lambda$, so $\calO_{G,\mathrm{strong}}(U^X,\xi U^Y)$ is convex.
We next show by an example that $\calO_{G,\mathrm{qstrong}}(U^X,U^Y)$ need not be convex.
Let $G=\mathbb Z_2=\{0,1\}$, and take both the input and output systems to be qubits.
Using the Pauli matrices
\eq{
\sigma_z&:=\ket{0}\bra{0}-\ket{1}\bra{1},\notag\\
\sigma_x&:=\ket{0}\bra{1}+\ket{1}\bra{0},
}
we define the input and output representations by
\eq{
U_g^X=U_g^Y=\sigma_z^g,\qquad g\in\mathbb Z_2.
}
Consider the identity and bit-flip channels on this qubit:
\eq{
\Lambda_0(A)=A,\qquad \Lambda_1(A)=\sigma_x A\sigma_x.
}
Clearly, $\Lambda_0\in\calO_{G,\mathrm{strong}}(U^X,U^Y)$.
Moreover, since $\sigma_x\sigma_z^g=(-1)^g\sigma_z^g\sigma_x$, we have
\eq{
\Lambda_1(U_g^X ...)=(-1)^g U_g^Y\Lambda_1(...).
}
Thus, with $\xi_g:=(-1)^g$, we have $\Lambda_1\in\calO_{G,\mathrm{strong}}(U^X,\xi U^Y)$.
It follows that
\eq{
\Lambda_0,\enskip\Lambda_1\in\calO_{G,\mathrm{qstrong}}(U^X,U^Y).
}
However, consider the bit-flip noise $\Lambda_p:=p\Lambda_0+(1-p)\Lambda_1$ obtained by probabilistically mixing these channels.
For $0<p<1$, it maps the strong symmetric state $\ket{0}\bra{0}$ to
\eq{
\Lambda_p(\ket{0}\bra{0})=p\ket{0}\bra{0}+(1-p)\ket{1}\bra{1}.
}
The output state has support in both eigenspaces of $\sigma_z$ and is therefore not strong symmetric.
By Theorem~\ref{thm:qstrong_resource_theory}(iii), this mixture is not quasi-strong covariant.
\end{proof}

\subsection{Kraus representations and unitary realizations}

By \eqref{eq:qstrong_union}, the Kraus characterization of strong covariant operations extends directly to quasi-strong covariant operations.

\begin{theorem}\label{thm:qstrong_kraus}
For a CPTP map $\Lambda:\calB(\calH_X)\to\calB(\calH_Y)$, the following statements are equivalent.
\begin{description}
\item[(i)] $\Lambda$ is $(U^X,U^Y)$-quasi-strong covariant.
\item[(ii)] There exists $\xi\in\CharG$ such that every Kraus representation $\{K_m\}$ of $\Lambda$ satisfies
\eq{
K_mU_g^X=\xi(g)U_g^YK_m,
\qquad \forall m,\quad\forall g\in G.
\label{eq:qstrong_kraus}
}
\item[(iii)] There exist $\xi\in\CharG$, a Hilbert space $\calH_E$, and a Stinespring isometry
$W:\calH_X\to\calH_Y\otimes\calH_E$ such that
\eq{
WU_g^X=\xi(g)(U_g^Y\otimes I_E)W,
\qquad\forall g\in G,
\label{eq:qstrong_isometry}
}
and $\Lambda(A)=\Tr_E[WAW^\dagger]$.
\end{description}
\end{theorem}

\begin{proof}
By Definition~\ref{def:qstrong}, (i) is equivalent to $\Lambda$ being $(U^X,\xi U^Y)$-strong covariant for some $\xi$.
Applying Theorem~\ref{thm:Kraus_strong} with output representation $\xi U^Y$ therefore gives (i)$\Leftrightarrow$(ii).
If \eqref{eq:qstrong_kraus} holds, then $W:=\sum_mK_m\otimes\ket{m}_E$ satisfies \eqref{eq:qstrong_isometry} and realizes $\Lambda$, proving (ii)$\Rightarrow$(iii).
Conversely, given (iii), choose an orthonormal basis $\{\ket{m}_E\}$ of $\calH_E$ and set $K_m:=(I_Y\otimes\bra{m}_E)W$.
These are Kraus operators of $\Lambda$ and satisfy \eqref{eq:qstrong_kraus}.
Thus $\Lambda$ satisfies (i), and the equivalence already proved gives (ii).
\end{proof}

For channels on a fixed system, we also obtain a unitary realization analogous to that in Theorem~\ref{thm:SD_strong_StoS}.

\begin{theorem}\label{thm:qstrong_unitary}
Let $\Lambda:\calB(\calH)\to\calB(\calH)$ be a $(U,U)$-quasi-strong covariant CPTP map with character $\xi$.
Then there exist an auxiliary Hilbert space $\calH_E$, a unitary representation $R:G\to U(\calH_E)$, a unit vector $\ket{0}_E$, and a unitary $V$ on $\calH\otimes\calH_E$ such that
\eq{
[\,V,U_g\otimes R_g\,]&=0,\qquad \forall g\in G,\label{eq:qstrong_unitary_cov}\\
R_g\ket{0}_E&=\ket{0}_E,\label{eq:qstrong_env_initial}\\
V(\calH\otimes\ket{0}_E)&\subset
\calH\otimes\calH_{E,\xi},\label{eq:qstrong_env_final}\\
\Lambda(\rho)&=\Tr_E\!\left[V(\rho\otimes\ket{0}\bra{0}_E)V^\dagger\right],
\label{eq:qstrong_unitary_channel}
}
where $\calH_{E,\xi}$ is a subspace of $\calH_E$ satisfying
\eq{
R_g\ket{\eta}_E&=\xi(g)\ket{\eta}_E,\enskip\forall\ket{\eta}_E\in\calH_{E,\xi},\enskip
\forall g\in G.
}
Conversely, a realization satisfying \eqref{eq:qstrong_unitary_cov}--\eqref{eq:qstrong_unitary_channel} defines a $(U,U)$-quasi-strong covariant operation with character $\xi$.
\end{theorem}

\begin{proof}
Since $\Lambda$ is $(U,\xi U)$-strong covariant, Theorem~\ref{thm:Kraus_strong} gives
$K_mU_g=\xi(g)U_gK_m$ for its Kraus operators $\{K_m\}_{m=1}^{d_\Lambda}$.
Take an auxiliary system $E$ with orthonormal basis $\{\ket{m}\}_{m=0}^{d_\Lambda}$, and define
\eq{
R_g&:=\ket{0}\bra{0}+\xi(g)\sum_{m=1}^{d_\Lambda}\ket{m}\bra{m},\\
W'&:=\sum_{m=1}^{d_\Lambda}K_m\otimes\ket{m}\bra{0}.
}
The relation for the Kraus operators implies $[W',U_g\otimes R_g]=0$.
Moreover, $W'^\dagger W'=I\otimes\ket{0}\bra{0}$, and $W'^\dagger W'$ and $W'W'^\dagger$ are orthogonal projections with mutually orthogonal ranges.
We can therefore use the unitary extension constructed in the proof of Theorem~\ref{thm:SD_strong_StoS}, replacing $U_g\otimes I_E$ by $U_g\otimes R_g$.
This gives a unitary $V$ satisfying \eqref{eq:qstrong_unitary_cov} that agrees with $W'$ on the initial subspace:
\eq{
V(\ket{\psi}\otimes\ket{0})
=\sum_mK_m\ket{\psi}\otimes\ket{m},
\qquad\forall\ket{\psi}\in\calH.
}
Taking $\calH_{E,\xi}:=\operatorname{span}\{\ket{1},\ldots,\ket{d_\Lambda}\}$ gives the remaining conditions \eqref{eq:qstrong_env_initial}--\eqref{eq:qstrong_unitary_channel}.

Conversely, choose an orthonormal basis $\{\ket{m}\}$ of $\calH_{E,\xi}$ and define $K_m:=\bra{m}V\ket{0}$.
By \eqref{eq:qstrong_unitary_cov}--\eqref{eq:qstrong_env_final},
$K_mU_g=\xi(g)U_gK_m$.
Theorem~\ref{thm:qstrong_kraus} then shows that $\Lambda$ is quasi-strong covariant.
\end{proof}

This theorem has a direct physical interpretation.
In the strong covariant realization, there is no exchange of conserved charge with the environment.
In the quasi-strong covariant realization, the environment starts in the trivial sector and always ends in the same character sector $\xi$.
Thus, the allowed transfer of conserved charge is fixed and does not depend on the initial state of the system.

\subsection{Partial trace and appending strong symmetric states}
We now revisit the partial trace and the appending of strong symmetric states, which were shown in Section~\ref{sec:strong_weak_relation} not to be strong covariant in general.
We examine which of these conclusions change when quasi-strong covariant operations are allowed.
Throughout this subsection, we consider ordinary (non-projective) unitary representations.

\begin{description}
\item[The partial trace]{
The partial trace is not quasi-strong covariant in general.
Indeed, in the example given in Section~\ref{sec:strong_weak_relation}, the partial trace maps a strong symmetric state to a state that is not strong symmetric.
By Theorem~\ref{thm:qstrong_resource_theory}, it therefore cannot be quasi-strong covariant.

As in the case of strong covariant operations, however, the partial trace can be free for certain representations of the discarded system.
Let $U^A$ and $U^B$ be the representations on systems $A$ and $B$, and suppose that, for some $\chi\in\CharG$,
\eq{
U_g^B=\chi(g)I_B,\qquad\forall g\in G
}
holds on the whole Hilbert space $\calH_B$.
Then, for every $X\in\calB(\calH_A\otimes\calH_B)$ and $g\in G$,
\eq{
\Tr_B[(U_g^A\otimes U_g^B)X]
=\chi(g)U_g^A\Tr_B[X].
}
Thus, $\Tr_B$ is a $(U^A\otimes U^B,U^A)$-quasi-strong covariant operation with character $\chi$.
This is a condition on the representation of the discarded system over its entire Hilbert space, not a restriction on the input state.
Having $U_g^B$ act by a single character only on the support of a particular input state is not sufficient to make the partial trace quasi-strong covariant on the full input space.
}
\item[Appending a strong symmetric state]{
In contrast, appending a strong symmetric state is always quasi-strong covariant.
First, recall the operation
\eq{
\Lambda(\cdot):=(\cdot)\otimes\ket{1}\bra{1}_{S'}
}
from the earlier two-qubit example.
Although this operation is not strong covariant, it satisfies
\eq{
\Lambda(e^{-iH_St}(\cdot))
=e^{it}e^{-i(H_S+H_{S'})t}\Lambda(\cdot),
}
so it is quasi-strong covariant with character $\xi(t)=e^{it}$.

More generally, let $U^A$ and $U^B$ be the representations on systems $A$ and $B$, and let $\chi\in\CharG$ be the character associated with a strong symmetric state $\sigma_B$, namely,
\eq{
U_g^B\sigma_B=\chi(g)\sigma_B,\qquad\forall g\in G.
}
The appending operation
\eq{
\Lambda_{\mathrm{add}}(\cdot):=(\cdot)\otimes\sigma_B
}
then satisfies
\eq{
\Lambda_{\mathrm{add}}(U_g^A(\cdot))
&=(U_g^A(\cdot))\otimes\sigma_B\notag\\
&=\chi(g)^{-1}(U_g^A\otimes U_g^B)
\Lambda_{\mathrm{add}}(\cdot).
}
Hence, $\Lambda_{\mathrm{add}}$ is a $(U^A,U^A\otimes U^B)$-quasi-strong covariant operation with character $\chi^{-1}$.

Note that $\Lambda_{\mathrm{add}}$ is not $(U^A,U^A\otimes U^B)$-strong covariant unless $\chi=1$.
It is, however, $(U^A,U^A\otimes\chi^{-1}U^B)$-strong covariant.
Indeed, for every $g\in G$,
\eq{
(\chi^{-1}U^B)_g\sigma_B
=\chi(g)^{-1}U_g^B\sigma_B
=\sigma_B,
}
and thus, for every $X\in\calB(\calH_A)$ and $g\in G$,
\eq{
\Lambda_{\mathrm{add}}(U_g^A X)
&=(U_g^A X)\otimes\sigma_B\notag\\
&=\bigl[U_g^A\otimes(\chi^{-1}U^B)_g\bigr]
\Lambda_{\mathrm{add}}(X).
}
}
\end{description}

\subsection{Resource measures and asymptotic convertibility}

The resource monotones considered here carry over from strong covariant operations to quasi-strong covariant operations because they are invariant under character twists.
The following theorem states this general principle.

\begin{theorem}\label{thm:qstrong_transfer}
Let $G$ be a compact group, and fix the set of free states $\calF$ to be either $\calF_{G,\mathrm{strong}}$ or $\calF_{G,\mathrm{single}}$.
Suppose that $M(\rho\|U)$ is a resource monotone for $(\calF,\calO_{G,\mathrm{strong}})$ and that
\eq{
M(\rho\|\xi U)=M(\rho\|U)
\label{eq:character_invariant_measure}
}
holds for every (projective) unitary representation $U$ on any system, every state $\rho$ of that system, and every $\xi\in\CharG$.
Then $M$ is also a resource monotone for $(\calF,\calO_{G,\mathrm{qstrong}})$.
Moreover, if $M$ is conserved under strong covariant operations, it is also conserved under quasi-strong covariant operations.
\end{theorem}

\begin{proof}
By Theorem~\ref{thm:qstrong_resource_theory}, $(\calF,\calO_{G,\mathrm{qstrong}})$ is a resource theory for either choice of $\calF$.
Replacing $\calO_{G,\mathrm{strong}}$ by $\calO_{G,\mathrm{qstrong}}$ changes neither the set of free states $\calF$ nor $M$ itself.
Thus, nonnegativity and the requirement that $M$ vanish on free states still hold by assumption, and it remains to prove monotonicity.

Let $\Lambda$ be any $(U^X,U^Y)$-quasi-strong covariant operation with character $\xi$.
Since $\Lambda$ is $(U^X,\xi U^Y)$-strong covariant, every input state $\rho$ satisfies
\eq{
M(\rho\|U^X)
&\ge M(\Lambda(\rho)\|\xi U^Y)\notag\\
&=M(\Lambda(\rho)\|U^Y).
}
Here, the inequality follows from monotonicity under $(\calF,\calO_{G,\mathrm{strong}})$, and the equality follows from \eqref{eq:character_invariant_measure}.
Hence, $M$ is a resource monotone for $(\calF,\calO_{G,\mathrm{qstrong}})$.
If $M$ is conserved under strong covariant operations, the inequality above is also an equality, proving conservation under quasi-strong covariant operations.
\end{proof}

All the main resource measures introduced in this paper satisfy \eqref{eq:character_invariant_measure}.

\begin{theorem}\label{thm:qstrong_measures}
Let $G$ be a compact group, and consider each of the following quantities under the assumptions of its definition.
\begin{description}
\item[(i)] Strong entanglement asymmetry $\Ssa$ is a resource monotone for both $(\calF_{G,\mathrm{strong}},\calO_{G,\mathrm{qstrong}})$ and $(\calF_{G,\mathrm{single}},\calO_{G,\mathrm{qstrong}})$.
\item[(ii)] The averaged logarithmic characteristic function $L$ is a resource monotone for $(\calF_{G,\mathrm{strong}},\calO_{G,\mathrm{qstrong}})$ and is conserved under quasi-strong covariant operations.
\item[(iii)] If $G$ is a compact Lie group, the covariance matrices $\VNS$ and $\VS$ are matrix-valued resource monotones for $(\calF_{G,\mathrm{strong}},\calO_{G,\mathrm{qstrong}})$ and are conserved under quasi-strong covariant operations.
\end{description}
In general, $L$, $\VNS$, and $\VS$ are not resource monotones for $(\calF_{G,\mathrm{single}},\calO_{G,\mathrm{qstrong}})$, since they need not vanish on single-sector states.
If $\calF_{G,\mathrm{strong}}=\calF_{G,\mathrm{single}}$, then (ii) and (iii) of course also hold with $\calF_{G,\mathrm{single}}$ as the set of free states.
For a fixed choice of free states, all previously established properties that do not depend on the class of free operations---including whether positivity, faithfulness, and additivity hold---remain unchanged.
\end{theorem}

\begin{proof}
Replacing $U$ by $\xi U$ leaves the adjoint action unchanged, so $\calG_{\xi U}=\calG_U$.
Although each irreducible representation $U^{(\nu)}$ is replaced by $\xi U^{(\nu)}$, the corresponding subspace is unchanged, and hence so is $p_\nu(\rho)$.
It follows that $\Ssa(\rho\|\xi U)=\Ssa(\rho\|U)$.
Also, since $|\xi(g)|=1$,
\eq{
L_{g,\mathrm{strong}}(\rho\|\xi U)
=-\log|\xi(g)\Tr[U_g\rho]|
=L_{g,\mathrm{strong}}(\rho\|U),
}
and therefore $L(\rho\|\xi U)=L(\rho\|U)$.
Finally, when $G$ is a compact Lie group, writing
$\xi(g(\bm\lambda))=\exp(i\sum_jc_j\lambda^j+o(|\bm\lambda|))$ shows that the generators of $\xi U$ are $X_j+c_jI$.
These constant terms cancel in the covariances, so $\VNS$ and $\VS$ also satisfy \eqref{eq:character_invariant_measure}.

By Theorem~\ref{thm:strong_EA}, $\Ssa$ is a resource monotone under strong covariant operations with either $\calF_{G,\mathrm{strong}}$ or $\calF_{G,\mathrm{single}}$ as the set of free states.
By Theorems~\ref{thm:SLCF} and \ref{thm:VM}, $L$, $\VNS$, and $\VS$ are resource monotones with $\calF_{G,\mathrm{strong}}$ as the set of free states and are conserved under strong covariant operations.
Applying Theorem~\ref{thm:qstrong_transfer} with the corresponding choices of free states gives (i)--(iii).

Finally, $\rho=I_2/2$ is a single-sector state for the two-dimensional irreducible representation of $SU(2)$, but, as shown above,
\eq{
L(I_2/2)=\frac12+\log2>0.
}
Moreover, taking the generators to be $X_j=\sigma_j/2$, where $\sigma_j$ are the Pauli matrices, gives
\eq{
\VNS(I_2/2)=\VS(I_2/2)=\frac14I_3\ne0.
}
Thus, these quantities do not, in general, vanish on single-sector states.
\end{proof}

We next consider the i.i.d.\ conversion setting of Section~\ref{sec:strong_iid}, in which a strong symmetric ancillary state is allowed to remain in the output.
In this setting, replacing strong covariant operations with quasi-strong covariant operations leaves the entire set of achievable rates unchanged for any compact Lie group.

\begin{theorem}\label{thm:qstrong_iid}
Let $U,U'$ be non-projective continuous unitary representations of a compact Lie group $G$ on finite-dimensional Hilbert spaces, and let $\rho,\sigma$ be arbitrary states on the respective spaces.
For a rate $r>0$, set $M_N:=\lfloor rN\rfloor$.
We say that $r$ is achievable under quasi-strong covariant operations if, for each $N$, there exist a finite-dimensional auxiliary system $A_N$, a non-projective continuous unitary representation $W_N$ on $A_N$, a state $\eta_{A_N}$ that is strong symmetric with respect to $W_N$, and a
$(U^{\otimes N},U'^{\otimes M_N}\otimes W_N)$-quasi-strong covariant CPTP map $\Phi_N$ such that
\eq{
\liminf_{N\to\infty}
\Bigl\|
\Phi_N(\rho^{\otimes N})
-\sigma^{\otimes M_N}\otimes\eta_{A_N}
\Bigr\|_1=0.
}
As in Section~\ref{sec:strong_iid}, the dimension and representation of the auxiliary system may be chosen freely for each $N$, and the error is evaluated on the full output, including the auxiliary system.
Let $\calR_{+,\mathrm{qstrong}}(\rho\to\sigma)$ denote the set of all such positive achievable rates, and define
\eq{
R_{\mathrm{qstrong}}(\rho\to\sigma)
:=\sup\bigl(\{0\}\cup\calR_{+,\mathrm{qstrong}}(\rho\to\sigma)\bigr).
}
Then
\eq{
\calR_{+,\mathrm{qstrong}}(\rho\to\sigma)
&=\calR_+(\rho\to\sigma),
\label{eq:qstrong_same_rate_set}\\
R_{\mathrm{qstrong}}(\rho\to\sigma)
&=R(\rho\to\sigma).
\label{eq:qstrong_same_rate}
}
Consequently, for any $r>0$, conditions (W) and (S) of Theorem~\ref{thm:compact_Lie_strong_iid} are also necessary and sufficient for $r$ to be achievable under quasi-strong covariant operations.
In particular, if $G=U(1)$ and the assumptions of Theorem~\ref{thm:U1_variance} hold, then
\eq{
R_{\mathrm{qstrong}}(\rho\to\sigma)
=\frac{V_H(\rho)}{V_{H'}(\sigma)}.
}
\end{theorem}

\begin{proof}
Every strong covariant operation is also quasi-strong covariant with the trivial character $\xi=1$. Hence,
\eq{
\calR_+(\rho\to\sigma)
\subseteq\calR_{+,\mathrm{qstrong}}(\rho\to\sigma).
}

Conversely, let $r\in\calR_{+,\mathrm{qstrong}}(\rho\to\sigma)$, and choose sequences of operations $\Phi_N$, auxiliary representations $W_N$, and auxiliary states $\eta_{A_N}$ that achieve this rate.
Let $\xi_N\in\CharG$ be the character associated with the quasi-strong covariance of $\Phi_N$, and define a representation on the same auxiliary system by
\eq{
\widetilde W_N(g):=\xi_N(g)W_N(g).
}
Since $\xi_N$ is a continuous one-dimensional unitary representation, $\widetilde W_N$ is also a non-projective continuous unitary representation.
By the definition of quasi-strong covariance, for every input operator $X$ and $g\in G$,
\eq{
\Phi_N(U_g^{\otimes N}X)
&=\xi_N(g)\bigl(U_g'^{\otimes M_N}\otimes W_N(g)\bigr)\Phi_N(X)\nonumber\\
&=\bigl(U_g'^{\otimes M_N}\otimes\widetilde W_N(g)\bigr)\Phi_N(X).
}
Thus, the same $\Phi_N$ is a $(U^{\otimes N},U'^{\otimes M_N}\otimes\widetilde W_N)$-strong covariant CPTP map.

Moreover, since $\eta_{A_N}$ is strong symmetric with respect to $W_N$, Lemma~\ref{lem:qstrong_free_states_twist} implies that it is also strong symmetric with respect to $\widetilde W_N$.
The auxiliary representation can be chosen freely in the definition of the conversion, so replacing $W_N$ with $\widetilde W_N$ gives an admissible protocol.
Because $\Phi_N$ and $\eta_{A_N}$ themselves are unchanged, the sequence of conversion errors on the full output, including the auxiliary system, is also unchanged and therefore still has limit inferior zero.
Therefore, the same rate $r$ is achievable under strong covariant operations, giving
\eq{
\calR_{+,\mathrm{qstrong}}(\rho\to\sigma)
\subseteq\calR_+(\rho\to\sigma).
}
This proves \eqref{eq:qstrong_same_rate_set}; adjoining $\{0\}$ to both sets and taking their suprema yields   \eqref{eq:qstrong_same_rate}.
The general conversion conditions and the $U(1)$ variance-ratio formula follow from Theorems~\ref{thm:compact_Lie_strong_iid} and~\ref{thm:U1_variance}, respectively.
\end{proof}

Finally, we show that the freedom to choose the character $\xi$ disappears for the continuous-time Markovian dynamics considered below: quasi-strong covariance at every time implies strong covariance at every time.

\begin{theorem}\label{thm:qstrong_semigroup}
Let $S$ be a finite-dimensional quantum system whose dynamics obeys the following GKSL equation with time-independent $H$ and $\{L_k\}$:
\eq{
\partial_t \rho=-i[H,\rho]+\sum_k\left(L_k\rho L^\dagger_k-\frac{1}{2}\{L^\dagger_kL_k,\rho\}\right).
}
Let $G$ be a compact group and $U$ a (projective) unitary representation on the Hilbert space $\calH$ of $S$.
Assume that the CPTP map $\Lambda_t:\calB(\calH)\rightarrow\calB(\calH)$ generated by this master equation is $(U,U)$-quasi-strong covariant for every $t\ge0$, and that $[H,U_g]=0$ for every $g\in G$.
Then any jump operators $\{L_k\}$ in the GKSL equation satisfy
\eq{
[L_k,U_g]=0,\qquad\forall g\in G.
\label{eq:qstrong_GKSL_comm}
}
Moreover, $\Lambda_t$ is $(U,U)$-strong covariant for every $t\ge0$.
This strong covariance conclusion holds even without the assumption $[H,U_g]=0$.
\end{theorem}
If there are multiple choices of jump operators $\{L_k\}$ generating the same dynamics, \eqref{eq:qstrong_GKSL_comm} holds for every such choice.

\begin{proof}
We first show that each $\Lambda_t$ is strong covariant.
Let $d:=\dim\calH$, and write the right-hand side of the GKSL equation as $\calL(\rho)$.
Since $\Lambda_t=e^{t\calL}$, $\Lambda_t$ is invertible as a linear map on $\calB(\calH)$.
The linear map given by left multiplication by $U_g$, $\mathcal M_g(A):=U_gA$, is also invertible.

Let $\xi_t$ be the unique character of $\Lambda_t$ given by Lemma~\ref{lem:qstrong_character_unique}.
Quasi-strong covariance can be written as
\eq{
\Lambda_t\circ\mathcal M_g=\xi_t(g)\mathcal M_g\circ\Lambda_t.
}
Since $\calB(\calH)$ has dimension $d^2$, taking determinants and using $\det(\Lambda_t)\ne0$ gives
\eq{
\det(\mathcal M_g)=\xi_t(g)^{d^2}\det(\mathcal M_g).
}
As $\det(\mathcal M_g)\ne0$, we obtain $\xi_t(g)^{d^2}=1$.
On the other hand, taking $\rho=I/d$ in the proof of Lemma~\ref{lem:qstrong_character_unique} gives
\eq{
\xi_t(g)=\frac1d\Tr[U_g^\dagger\Lambda_t(U_g)].
}
Thus, for each $g$, $\xi_t(g)$ is continuous in $t$, and $\xi_0(g)=1$ because $\Lambda_0=\mathrm{id}$.
Since $\xi_t(g)$ takes values in the finite set of $d^2$th roots of unity, continuity implies $\xi_t(g)=1$ for every $t\ge0$ and $g\in G$.
Hence, each $\Lambda_t$ is $(U,U)$-strong covariant.
This argument does not use the assumption $[H,U_g]=0$.

Under the additional assumption $[H,U_g]=0$, \eqref{com_GKSL} for strong covariant GKSL dynamics gives \eqref{eq:qstrong_GKSL_comm} for any choice of jump operators.
\end{proof}

Thus, choosing quasi-strong covariant operations as free operations leaves unchanged the conclusions established here for the main resource measures, the $U(1)$ i.i.d.\ conversion rate, and continuous-time Markovian dynamics.

\subsection{Equality of the operation classes for connected groups}

Beyond the GKSL setting, the classes of quasi-strong covariant and strong covariant operations also coincide when the group is connected and the input and output carry the same representation on a fixed finite-dimensional system.

\begin{theorem}\label{thm:qstrong_connected}
Let $G$ be a connected topological group, let $\calH$ be a Hilbert space with $1\le d:=\dim\calH<\infty$, and let $U$ be a fixed (projective) unitary representation on $\calH$.
Then
\eq{
\calO_{G,\mathrm{qstrong}}(U,U)
=\calO_{G,\mathrm{strong}}(U,U).
\label{eq:qstrong_connected_equal}
}
In other words, quasi-strong covariance and strong covariance are equivalent when the same representation $U$ is used for both the input and the output.
\end{theorem}

\begin{proof}
It suffices to show that the character of any $(U,U)$-quasi-strong covariant operation is trivial.
Let $\Lambda$ be a $(U,U)$-quasi-strong covariant operation with character $\xi$.
Fix any $g\in G$, and denote the set of eigenvalues of $U_g$ by $S_g$.
For any $\lambda\in S_g$, choose a corresponding normalized eigenvector $\ket{v}$, and set $\rho_v:=\ket{v}\!\bra{v}$ and $\sigma_v:=\Lambda(\rho_v)$.
Trace preservation gives $\Tr\sigma_v=1$, so in particular $\sigma_v\ne0$.
Moreover, $U_g\rho_v=\lambda\rho_v$ and   \eqref{eq:qstrong_def} imply
\eq{
\lambda\sigma_v
=\Lambda(U_g\rho_v)
=\xi(g)U_g\sigma_v.
}
Thus, every nonzero vector in the support of $\sigma_v$ is an eigenvector of $U_g$ with eigenvalue $\xi(g)^{-1}\lambda$.
Hence, $\xi(g)^{-1}S_g\subseteq S_g$.
Since $S_g$ is a nonempty finite set and multiplication by $\xi(g)^{-1}$ is injective, we in fact have $\xi(g)^{-1}S_g=S_g$.

For any $\lambda\in S_g$, repeated multiplication by $\xi(g)^{-1}$ yields $\xi(g)^{-m_g}\lambda=\lambda$ for some $1\le m_g\le |S_g|\le d$.
Indeed, the finite dimensionality of $\calH$ ensures that $S_g$ is finite. Writing $n_g:=|S_g|$, the equality $\xi(g)^{-1}S_g=S_g$ implies that all of the following $n_g+1$ values belong to $S_g$:
\[
\lambda,\ \xi(g)^{-1}\lambda,\ \ldots,\ \xi(g)^{-n_g}\lambda.
\]
Since $S_g$ has only $n_g$ elements, the pigeonhole principle gives some $0\le i<j\le n_g$ such that
\[
\xi(g)^{-i}\lambda=\xi(g)^{-j}\lambda.
\]
Multiplying both sides by $\xi(g)^i$ gives $\xi(g)^{-(j-i)}\lambda=\lambda$, so choosing $m_g:=j-i$ proves the claimed equality with $1\le m_g\le n_g\le d$.
Since $U_g$ is unitary, $\lambda\ne0$, and hence $\xi(g)^{m_g}=1$.
Therefore, for every $g\in G$,
\eq{
\xi(g)\in\bigcup_{m=1}^{d}\{z\in U(1):z^m=1\}.
}
The set on the right-hand side is finite.
On the other hand, continuity of $\xi$ and connectedness of $G$ imply that the image $\xi(G)$ is connected and must therefore be a singleton.
Since a character satisfies $\xi(e)=1$, we conclude that $\xi(g)=1$ for all $g\in G$.
Thus, $\Lambda$ is $(U,U)$-strong covariant.
The reverse inclusion follows by choosing the trivial character in the definition.
\end{proof}

\subsection{A finite-group example with different state-conversion criteria}\label{sec:qstrong_finite_conversion}

When $G$ is not connected, not only the classes of operations but also the necessary and sufficient conditions for state conversion can differ.
As in Lemma~\ref{lem:qstrong_nonconvex}, let $G=\mathbb Z_2=\{0,1\}$, take the same qubit system as both input and output, and fix its representation to be $U_g=\sigma_z^g$.
Here, $\sigma_z:=\ket{0}\bra{0}-\ket{1}\bra{1}$ and $\sigma_x:=\ket{0}\bra{1}+\ket{1}\bra{0}$ are the Pauli matrices.
Consider the family of diagonal states
\eq{
\rho_p:=p\ket{0}\bra{0}+(1-p)\ket{1}\bra{1},
\qquad 0\le p\le1.
}
For any $p,q\in[0,1]$, the necessary and sufficient conditions for realizing $\rho_p\to\rho_q$ under the fixed representation pair $(U,U)$ differ as follows:
\eq{
\text{strong covariant:}\quad &q=p,\notag\\
\text{quasi-strong covariant:}\quad &q\in\{p,1-p\}.
\label{eq:qstrong_z2_conversion}
}

To prove necessity, let $\Lambda$ be a $(U,U)$-quasi-strong covariant operation satisfying $\Lambda(\rho_p)=\rho_q$, and let $\xi$ be its character.
A character of $\mathbb Z_2$ satisfies $\xi(1)^2=\xi(0)=1$, so $\xi(1)=\pm1$.
Trace preservation and quasi-strong covariance give
\eq{
2p-1
&=\Tr[\Lambda(\sigma_z\rho_p)]\notag\\
&=\xi(1)\Tr[\sigma_z\rho_q]
=\xi(1)(2q-1).
}
Thus, $q=p$ is necessary when $\xi(1)=1$, whereas $q=1-p$ is necessary when $\xi(1)=-1$.
In particular, a strong covariant operation has $\xi=1$, so it requires $q=p$.

For sufficiency, the identity channel realizes the conversion when $q=p$.
When $q=1-p$, use the bit-flip channel $\Lambda_{\mathrm{flip}}(A):=\sigma_x A\sigma_x$.
Since $\sigma_x\sigma_z^g=(-1)^g\sigma_z^g\sigma_x$, we have
\eq{
\Lambda_{\mathrm{flip}}(U_gA)
=(-1)^gU_g\Lambda_{\mathrm{flip}}(A)
}
for every $A$ and $g\in\mathbb Z_2$.
Hence, this channel is quasi-strong covariant with character $\xi(g)=(-1)^g$ and satisfies $\Lambda_{\mathrm{flip}}(\rho_p)=\rho_{1-p}$.
This proves the necessary and sufficient conditions in   \eqref{eq:qstrong_z2_conversion}.

In particular, taking $p=1$ and $q=0$, the conversion $\ket{0}\bra{0}\to\ket{1}\bra{1}$ is achievable by a quasi-strong covariant operation but by no $(U,U)$-strong covariant operation.
This counterexample remains valid when projective representations are allowed.
The bit-flip channel can be regarded as $(U,\xi U)$-strong covariant.
However, this replaces the output representation $U$ with $\xi U$ and therefore changes the conditions from those imposed by the fixed representation pair $(U,U)$ considered here.

Despite this difference, the main structural properties of the resource theory remain unchanged.
The resource theories with either strong symmetric states or single-sector states as free states remain valid, and the monotonicity and conservation laws established for the main resource measures introduced in this paper continue to hold (Theorems~\ref{thm:qstrong_resource_theory} and \ref{thm:qstrong_measures}).

\section{Generalization to non-invertible symmetries}\label{sec:general}

So far, our discussion has focused on ordinary symmetries described by a group $G$.
Recently, it has been proposed that the same strategy extends to generalized symmetries \cite{Ahmad2025,Benini2025}.
Their key observation is that, instead of a group, one can work with a finite-dimensional
$C^\ast$-algebra $\ca{A}$ of symmetry operators acting on the Hilbert space
$\ca{H}$.
Typical examples are the fusion algebra, tube algebra, or strip algebra associated with a
fusion category describing the generalized symmetry.

Given such an algebra $\ca{A}$, one chooses a basis $\{X_a\}$ and defines the symmetric
bilinear form
\begin{equation}
  K_{ab} := \tr\bigl( X_a X_b \bigr).
\end{equation}
Let $\widetilde K$ denote the inverse matrix of $K$, and define the {\it symmetrizer} $S_{\ca{A}} : \ca{B}(\ca{H}) \to \ca{B}(\ca{H})$,
\begin{equation}
  S_{\ca{A}}(\rho)
  := \sum_{a,b} \widetilde K_{ab}\, X_a \rho X_b ,
  \label{eq:generalized-symmetrizer}
\end{equation}
which has the following properties:
\begin{itemize}
  \item $S_{\ca{A}}$ is a projector: $S_{\ca{A}}^2 = S_{\ca{A}}$.
  \item $S_{\ca{A}}(\rho)$ commutes with all elements of $\ca{A}$, i.e. it is
  $\ca{A}$-symmetric.
  \item $S_{\ca{A}}(\rho) = \rho$ if and only if $\rho$ is $\ca{A}$-symmetric.
  \item $S_{\ca{A}}$ is a CPTP map.
\end{itemize}
In particular, when $\ca{A}$ is the group algebra of a finite group $G$,
$\ca{A} = \bb{C}[G]$, the symmetrizer \eqref{eq:generalized-symmetrizer} reduces to the
familiar $G$-twirling channel (\ref{eq:twir}).
This construction thus provides the natural generalization of $G$-twirling to generalized
symmetries.

Using $S_{\ca{A}}$, one defines the relative entropy of asymmetry for generalized symmetries as
\begin{equation}
  A_{\ca{A}}(\rho)
  := S\bigl(\rho \,\Vert\, S_{\ca{A}}(\rho)\bigr)
  = S\bigl(S_{\ca{A}}(\rho)\bigr) - S(\rho),
  \label{eq:AGA-def}
\end{equation}
in complete analogy with the group case.

The construction of strong entanglement asymmetry extends to this setting with minimal
changes.
Since $\ca{A}$ is finite-dimensional and semisimple, the $\ca{A}$-symmetric state
$S_{\ca{A}}(\rho)$ can be decomposed into orthogonal ``charge sectors''.
Concretely, there exists a family of mutually orthogonal projectors $\{P_r\}$ on $\ca{H}$ such
that
\begin{equation}
  \sum_r P_r = \mathbf{1}, \qquad
  S_{\ca{A}}(\rho)
  = \sum_r P_r S_{\ca{A}}(\rho) P_r .
\end{equation}
The labels $r$ correspond to the simple sectors (irreducible $\ca{A}$-modules) of the
generalized symmetry.
We can then define the sector probabilities and conditional states as
\begin{equation}
  p_r := \tr\bigl( P_r S_{\ca{A}}(\rho) \bigr), \qquad
  \rho_{\mathrm{sym},r}^{(\ca{A})} := \frac{P_r S_{\ca{A}}(\rho) P_r}{p_r}.
\end{equation}
States with support entirely in a single sector $r$ behave as ``strong symmetric'' states
for the generalized symmetry, in direct analogy with the group case.

Motivated by our definition for ordinary groups, we introduce the strong entanglement
asymmetry for a generalized symmetry algebra $\ca{A}$ by
\begin{equation}
  S_{s\text{-asym}}^{(\ca{A})}(\rho)
  := - \sum_r p_r \log p_r + A_{\ca{A}}(\rho).
  \label{eq:ssasym-A}
\end{equation}
The first term quantifies how much the state spreads over different strong symmetry sectors
of $\ca{A}$, while $A_{\ca{A}}(\rho)$ encodes the weak symmetry breaking measured by the
relative entropy of asymmetry.
When $\ca{A} = \bb{C}[G]$ is the group algebra of an ordinary symmetry group, the
symmetrizer $S_{\ca{A}}$ reduces to the $G$-twirling channel and the projectors $\{P_r\}$ reduce to the
usual projectors onto irreducible $G$-representations, so that
$S_{s\text{-asym}}^{(\ca{A})}$ reproduces our strong entanglement asymmetry
$S_{s\text{-asym}}$ defined in the previous section.

At present, the resource-theoretic understanding of non-invertible asymmetry is still incomplete, even in the case of weak symmetries.
We plan to present a rigorous formulation of non-invertible asymmetry within the framework of resource theory in a forthcoming paper.

We can also extend the averaged logarithmic characteristic function \eqref{def_SLCF_a}
to the generalized symmetries. A natural extension is
\begin{equation}\label{eq:def_LA_Noninv}
L^{(\ca{A})}(\rho) := \fr{1}{\abs{\ca{A}}} \sum_a \bigl(-\log|\Tr[X_a \rho]|\bigr).
\end{equation}
In the same spirit as \eqref{eq:ssasym-A}, the definition \eqref{eq:def_LA_Noninv}
provides a natural generalization of the averaged logarithmic characteristic function.
In practice, this quantity is tractable both from a theoretical and a computational viewpoint. 
Hence, regardless of whether one considers the weak or strong symmetry,
it is an interesting direction to study symmetry breaking of generalized symmetries using the proposed quantity.
A systematic resource-theoretic justification of this quantity will be presented elsewhere.

\section{Examples}\label{sec:ex}
We have already seen a few illustrative examples following the definitions of various asymmetry monotones e.g.\, in sec.~\!\ref{sec:RStrong}. 
In this section, we include more examples, illustrating the computation of  strong asymmetry measure $\Ssa$ and other asymmetry monotones in a few analytically
tractable setups. Finally, we discuss the strong-to-weak spontaneous symmetry breaking and how to capture that using these resource monotones.

\subsection{Vacuum reduced density matrix in CFT}
Consider the vacuum state of a CFT reduced to a subsystem $A$,
\begin{equation}
  \rho_A = \tr_{\bar{A}} \ket{0}\bra{0},
\end{equation}
where $\bar{A}$ is the complement of $A$.
The vacuum is weak symmetric, so $\Ssa(\rho_A)$ reduces to the
Shannon entropy of the strong-charge distribution $\{p_\alpha\}$.

Following the same strategy as in the computation of symmetry-resolved entanglement entropy
\cite{Kusuki2023},
the probability of finding the subsystem in the irreducible representation $\alpha$ of $G$ is
\begin{equation}
  p_\alpha =  \frac{d_{\alpha}^2}{\abs{G}},
\end{equation}
where $d_\alpha$ is the quantum dimension of $\alpha$.
Hence
\begin{equation}
  \Ssa(\rho_A)
  = - \sum_{\alpha} \frac{d_{\alpha}^2}{\abs{G}}
    \log \frac{d_{\alpha}^2}{\abs{G}}.
\end{equation}
In particular, for $G=\bb{Z}_n$ one obtains
\begin{equation}
  \Ssa(\rho_A) = \log n.
\end{equation}
Thus, while the vacuum respects weak symmetry, it maximally breaks strong symmetry in the sense
of our resource measure.

\subsection{Global quantum quench}
As a dynamical example, consider a global quench in a $(1+1)$-dimensional CFT
\cite{Calabrese2016}.
The initial state is prepared from a conformal boundary state $\ket{B}$ \cite{Cardy2004} as
\begin{equation}
  \ket{\varphi_0} := \ex{-\fr{\beta}{4}H} \ket{B},
\end{equation}
and then evolved with a $G$-symmetric Hamiltonian $H$,
\begin{equation}
  \rho_A(t) = \tr_{\bar{A}}
  \bigl[\ex{-iHt}  \ket{\varphi_0}\bra{\varphi_0}  \ex{iHt}\bigr],
\end{equation}
where $A$ is an interval of length $l$.
We take the regularization parameter $\beta$ to be small.
The entanglement asymmetry for weak symmetry in this setup was computed in \cite{Kusuki2024}.

For a symmetry-breaking conformal boundary, one finds
\begin{equation}
  \Ssa\bigl(\rho_A(t)\bigr)
  =
  \begin{cases}
    2\log \abs{G}, & 0<t<\fr{l}{2}, \\[0.2em]
    \log \abs{G},  & t>\fr{l}{2},
  \end{cases}
\end{equation}
while for a symmetry-preserving boundary, one obtains
\begin{equation}
  \Ssa\bigl(\rho_A(t)\bigr) = 0.
\end{equation}
In the former case, the system starts from a state that maximally breaks both strong and weak
symmetries and, upon local equilibration, relaxes to a state in which only weak symmetry is
restored.
In the latter case, strong symmetry is preserved throughout the time evolution.

\subsection{ Thermal density matrix and strong asymmetry for rotation}

In this subsection we consider thermal state of a $(d-1)+1$-dimensional CFT on a spatial manifold $\bb{S}^{d-1}$.   The thermal density matrix 
\begin{equation}
  \rho_{th}(\beta):=\frac{\sum e^{-\beta E} |E\rangle\langle E|}{\sum e^{-\beta E}}  
\end{equation}
is naturally weak symmetric under rotational subgroup of the conformal group. On the hand, strong rotational symmetry is broken, since $\rho_{th}$ contains states transforming under various irreps of the rotation group. In what follows, we will quantify the amount of strong symmetry breaking. 

\medskip
\noindent{\bf Strong entaglement-asymmetry:} To be concrete, we will consider a $\mathbb{Z}_2$ subgroup of the rotation group $SO(d-1)$. Now the Hilbert space of states can be grouped into even spins and odd spins under $\mathbb{Z}_2$. In order to compute the asymmetry, let us first define some quantities, which will be useful later. We first define partition function $Z(\beta)$ as
\begin{equation}
Z(\beta):=\sum_E e^{-\beta E}\,,
\end{equation}
and we further define two refined partition functions, $Z_+$ and $Z_-$, out of states with even spin and odd spin respectively:
\be
Z_{+}(\beta):=\sum_{\text{even spin}} e^{-\beta E}\,, \quad Z_{-}(\beta):=\sum_{\text{odd spin}} e^{-\beta E}\,.
\ee

Clearly, we have $Z(\beta)=Z_+(\beta)+Z_{-}(\beta)$, and

\begin{equation}
\rho_{th}(\beta)= \frac{Z_{+}(\beta)}{Z(\beta)} \rho_{th,+}(\beta)+\frac{Z_{-}(\beta)}{Z(\beta)} \rho_{th,-}(\beta)\,,
\end{equation}
where we have
\be
\begin{aligned}
&\rho_{th,+}(\beta):=\frac{\sum_{\text{even spin}} e^{-\beta E} |E\rangle\langle E|}{Z_{+}(\beta)}\,,\\
&\rho_{th,-}(\beta):=\frac{\sum_{\text{odd spin}} e^{-\beta E} |E\rangle\langle E|}{Z_{-}(\beta)}\,.
\end{aligned}
\ee

The entanglement asymmetry of strong symmetry corresponding to this $\mathbb{Z}_2$ is given by 
\be
\Ssa=- \frac{Z_{+}(\beta)}{Z(\beta)} \log  \frac{Z_{+}(\beta)}{Z(\beta)} -  \frac{Z_{-}(\beta)}{Z(\beta)} \log  \frac{Z_{-}(\beta)}{Z(\beta)}\,,
\ee
where we have used that for weak asymmetry, $S_{\text{asym}}=0$ for $\rho_{th}$. Now we can use thermal effective field theory \cite{Benjamin:2024kdg} (originally initiated in \cite{Banerjee:2012iz,Bhattacharyya:2007vs,Jensen:2012jh}, recently reincarnated in \cite{Benjamin:2023qsc}, see also \cite{Bhattacharyya:2012nq,Bhattacharya:2012zx, Shaghoulian:2015lcn,Kang:2022orq,Anand2025}; see \cite{Kusuki:2025pgx} for the application of thermal EFT to entanglement entropy) to derive 
\begin{equation}\label{eq:thermalEFT}
\begin{aligned}
\frac{Z_+(\beta)-Z_-(\beta)}{Z_+(\beta)+Z_-(\beta)}&=\frac{\mathrm{Tr}\ e^{-\beta H} (-1)^J}{\mathrm{Tr}\ e^{-\beta H} } \\
&\underset{\beta\to 0}{\sim} \exp\left[-\frac{f}{\beta^{d-1}}\left(1-\frac{1}{2^d}\right)\right]\,,
\end{aligned}
\end{equation}
Here $f$ is  the leading Wilson coefficient appearing in the high temperature expansion of thermal partition function i.e.  $$f:=\lim_{\beta\to 0}\beta^{d-1}\log Z(\beta)\geqslant 0\,.$$
Physically, $f$ quantifies the energy density in a thermal state. It follows from   \!\eqref{eq:thermalEFT} that the strong asymmetry in given by  
\begin{equation}
\label{eq:Z2}
\begin{aligned}
\Ssa & \underset{\beta\to 0}{\sim} \log 2 - \frac{1}{2} \exp \left[ -\frac{2\left(1-2^{-d}\right)f}{\beta^{d-1}}\right]\,,
\end{aligned}
\end{equation}

We end this subsection with four remarks:
\begin{enumerate}
    \item In $1+1$ dimensions,  $f$ is related to the central charge and given by $f=\frac{\pi c L}{6}$ where $L$ is the length of spatial circle and $c$ is the central charge of the CFT.
    \item The   \!\eqref{eq:Z2} can easily be generalized to the $\mathbb{Z}_p$ (with $p$ being prime) subgroup of the rotational group following \cite{Benjamin:2024kdg}:
\begin{equation}
    \label{eq:Zq}
\begin{aligned}
\Ssa & \underset{\beta\to 0}{\sim} \log p - \frac{p-1}{2}\exp \left[ -\frac{2\left(1-p^{-d}\right)f}{\beta^{d-1}}\right]\,.
\end{aligned}
\end{equation}
It is possible to derive a formula for the strong asymmetry, given an arbitrary subgroup $Z_m$, where $m$ is not necessarily restricted to being a prime. 
\item In the $\beta\to 0$ limit, the   \!\eqref{eq:Z2} and \eqref{eq:Zq} saturate the bound appropriate for weak symmetric density matrices.
\item Since the system has conformal symmetry, here the relevant dimensional quantity is $\beta/R$, where $R$ is the size of the system i.e. radius of $\bb{S}^{d-1}$. Hence, both the thermodynamic limit and the high-temperature limit correspond to $\beta/R\to 0$. For notational simplicity, we have set $R$ such that the volume of $\bb{S}^{d-1}$ is $1$, and $R$ does not appear explicitly in the above equations.
\end{enumerate}

The physical upshot is that as we take the high temperature/thermodynamic limit, the Boltzmann weight approaches $1$, all states become equally likely in thermal ensemble. Hence, one expects that $\rho_{th}$ approaches $\mathbb{I}$, leading to breaking of strong symmetry. While this is indeed true for a finite dimensional quantum system, for a quantum field theory with infinite degrees of freedom, the state $\mathbb{I}$ cannot be properly normalized. Albeit, the strong-asymmetry is still a meaningful quantity and can be computed reliably, confirming the intuition that strong symmetry is indeed broken in this limit.

\medskip
\noindent{\bf Averaged logarithmic characteristic function:} Given $\rho_{th}$, we can easily compute $\mathcal{L}(
\rho_{th})$ for the aforementioned $\mathbb{Z}_2$ symmetry using \eqref{eq:thermalEFT}:

\be
\begin{aligned}\label{z2}
L(
\rho_{th})(\beta)&=- \frac{1}{2} \log \frac{|Z_+(\beta)- Z_-(\beta)|}{Z(\beta)}\\
&\underset{\beta\to 0}{\sim}\frac{f}{\beta^{d-1}}\left(\frac{2^d-1}{2^{d+1}}\right)\,.
\end{aligned} 
\ee

For $\mathbb{Z}_p$ (with $p$ being prime), the above formula is generalized to 
\be\label{zp}
L(
\rho_{th})(\beta)
\underset{\beta\to 0}{\sim}(p-1)\frac{f}{\beta^{d-1}}\left(\frac{p^d-1}{p^{d+1}}\right) \,.
\ee

Once again, in the high temperature limit, the strong symmetry is broken, as evident from \eqref{z2} and \eqref{zp}. Note that in the $\beta\to0$ limit, the averaged logarithmic characteristic function diverges, reflecting the fact that we are dealing with a quantum field theory with infinite degrees of freedom.

\medskip
\noindent{\bf Variance \& strong asymmetry w.r.t $U(1)$:}
Variance is a resource measure appropriate for compact Lie group $G$. Here we consider a CFT with the simplest Lie group $G=U(1)$, generated by one of the generators of rotational subgroup. Using thermal EFT \cite{Benjamin:2023qsc}(originally initiated in \cite{Banerjee:2012iz,Bhattacharyya:2007vs,Jensen:2012jh}, recently reincarnated in \cite{Benjamin:2023qsc}, see also \cite{Bhattacharyya:2012nq,Bhattacharya:2012zx, Shaghoulian:2015lcn,Kang:2022orq,Anand2025}), we can compute
\begin{equation}
\mathrm{Tr}\ \rho_{th}\ e^{-\beta\omega J}\underset{\beta\to 0}{\sim} \exp\left[\frac{f}{\beta^{d-1}}\frac{\omega^2}{(1-\omega^2)}\right]
\end{equation}

The variance is given by 
\begin{equation}
    \langle J^2 \rangle -\left(\langle J\rangle \right)^2 \underset{\beta\to0}{\sim} 2f\beta^{-(d+1)}\,. 
\end{equation}

We see that the variance blows up as we take the high temperature/thermodynamic limit. While non-zero variance implies breaking of strong symmetry, the divergent behavior is an imprint of the presence of infinite degrees of freedom in QFT.

\subsection{ Thermal density matrix and strong asymmetry for global symmetry}
{\color{black}
In this subsection, we consider a CFT with a global $U(1)$ symmetry. We will show that thermal density matrix breaks the strong $U(1)$ symmetry as the temperature becomes large.

We use thermal EFT \cite{Benjamin:2024kdg,Kang:2022orq} a.k.a hydrodynamics \cite{Banerjee:2012iz,Bhattacharyya:2007vs,Bhattacharyya:2012nq,Bhattacharya:2012zx,Jensen:2012jh, Shaghoulian:2015lcn,Anand2025} to show that (in particular, see Eq.$1.7$ of \cite{Kang:2022orq})
\begin{equation}
 Z(\beta, g=e^{i\phi})\sim \exp\left(f\beta^{-d+1}-\frac{b}{4}\phi^2\beta^{-d+1}\right)\,,
\end{equation}
where $f$ quantifies the energy density in a thermal state and $b\geq 0$ quantifies the tension of the domain wall, that generates the $g$-twisted sector. They should be thought of as Wilson coefficients of the effective field theory.

The variance is given by 
\begin{equation}
    \langle Q^2\rangle -\left(\langle Q\rangle \right)^2 \underset{\beta\to0}{\sim} \frac{b}{2}\beta^{-d+1}\,.
\end{equation}

Once again, we see that the variance blows up as we take the high temperature/thermodynamic limit. The situation is exactly similar to that of spin. In the case $U(1)$, generated by spin, $b$-coefficient is in fact related to $f$. In contrast, for the global symmetry, $b$ is a genuinely new quantity. This is expected because, unlike the global $U(1)$ symmetry, the $U(1)$ related to spin is part of space-time symmetry which includes Hamiltonian. }

\subsection{Strong-to-weak SSB in a $\bb{Z}_2$ spin chain}
{\it Spontaneous symmetry breaking} (SSB) is among the most fundamental phenomena in many-body physics.
Applications of entanglement asymmetry to SSB have been worked out in \cite{Capizzi2023} (see also \cite{Ahmad2025,Benini2025} for extensions to generalized symmetries).
They discuss the extent to which the ``vacuum'' satisfying cluster decomposition spontaneously breaks a given symmetry.  
In other words, instead of the ``cat state'' $\ket{0\cdots 0}+\ket{1\cdots 1}$, they study the physically realized vacua $\ket{0\cdots 0}$ or $\ket{1\cdots 1}$, which are selected by an SSB, by means of the entanglement asymmetry.
We note that the cat state itself is symmetric,
and therefore its entanglement asymmetry vanishes.

The same line of reasoning extends naturally to mixed states. A phenomenon particular to
mixed states is {\it strong-to-weak spontaneous symmetry breaking} (SW-SSB) \cite{Lee2023,Ogunnaike2023,Lessa2024,Sala2024}. As a concrete
example, let us consider a one-dimensional spin chain with a global
$\bb{Z}_2$ symmetry generated by the spin-flip operator
\begin{equation}
  X := \prod_{i} X_i,
\end{equation}
where $X_i$ is the Pauli operator acting at site $i$.
A simple mixed state exhibiting SW-SSB is
\begin{equation}
  \rho_0 \;\propto\; \bb{I} + X.
\end{equation}
This state is strong symmetric, therefore,
\begin{equation}
  \Ssa(\rho_0) = 0.
\end{equation}
However, this state is unstable
against arbitrarily small local perturbations that break strong symmetry in close analogy with standard SSB \cite{Lessa2024}.

Since an SW-SSB state is strong symmetric, the strong entanglement asymmetry vanishes for it, just as the entanglement asymmetry is zero for an SSB state (e.g., a cat state). Therefore, following \cite{Capizzi2023}, let us consider the strong entanglement asymmetry of the post-collapse state.
Performing a $Z$-measurement at a single site $i$ drives the state to the maximally mixed state through the SW-SSB,
\begin{equation}
\rho_0 \;  \to \; \rho_{\mathrm{mix}} \;\propto\; \bb{I}\,.
\end{equation}
In the decomposition into strong-charge sectors, the post-measurement state has equal weights in the two strong symmetric states,
\begin{equation}\label{eq:rhopm}
\rho_{\mathrm{mix}}  \propto \rho_+ + \rho_-,
\end{equation}
where $X \rho_\pm = \pm \rho_\pm$.
By our definition, the strong entanglement asymmetry reduces in this case to the Shannon entropy of the strong-charge distribution, and hence
\begin{equation}
\Ssa(\rho_{\mathrm{mix}})\;=\;-\sum_{\alpha=\pm} \tfrac{1}{2}\log\tfrac{1}{2}\;=\;\log 2\,.
\end{equation}

The above result holds quite generally.  
For any mixed state obtained after SW-SSB,  
the measure of SW-SSB is determined entirely by the statistical uncertainty over the strong-charge sectors.  
In other words, the strong entanglement asymmetry is given by the Shannon entropy of the corresponding sector probabilities,
\begin{equation}
\Ssa(\rho)
= -\sum_{\alpha} p_\alpha \log p_\alpha\,,
\qquad 
p_\alpha = \tr\!\left(P_\alpha\,\rho\right),
\end{equation}
where $\{P_\alpha\}$ are the projectors onto the distinct strong symmetry sectors.

\section{Cross-over: Strong-Mpemba}\label{sec:so}
The quantum Mpemba effect refers to the following scenario: A system with more broken symmetry evolves more quickly towards a state where the symmetry gets restored. The Quantum Mpemba effect is nicely captured by the weak entanglement asymmetry, which, as a function of time shows a cross-over, followed by eventual settling down at zero. 

A natural question is whether such a phenomenon exists for strong-symmetry as well. However, this endeavor has an immediate conceptual obstruction. While restoration of weak symmetry is linked with system reaching the equilibrium, strong symmetry is not naturally connected to such a process. In fact, strong symmetry prohibits the exchange of charge between the system and the environment. Hence, the natural intuition would be that with generic interaction, strong symmetry gets broken rather than restored. 

Given the fact that we should not expect that strong entanglement asymmetry asymptotes to $0$ as the system relaxes due to generic interaction, the question remains regarding how to define Mpemba like effect for strong symmetry.

A natural way to define strong Mpemba effect is to look for 
\emph{cross-over} in the evolution of $\Ssa(t)$, as the system under consideration evolves towards the equilibrium configuration i.e.\ towards $\Sa=0$. If such a cross-over happens, we will say ``Strong-Mpemba" has happened. We note that Strong-Mpemba may happen even if the usual Quantum Mpemba is absent. In what follows, we will show a simple example of Strong Mpemba phenomenon without the presence of usual Quantum Mpemba effect. 

We consider a single qubit system,  evolving  under de-phasing channel
\begin{equation}\label{master}
    \dot{\rho}= -\frac{k}{4} [Z,[Z,\rho(\vec{r}(t))]]
\end{equation}
Here $Z$ is the Pauli $Z$ matrix, generating the $\mathbb{Z}_2$/ $U(1)$ symmetry.

The density matrix $\rho$ is parameterized by a Bloch vector $\vec{r}$, which moves inside the Bloch sphere as a function of time i.e.\ we have
\begin{equation}
\rho(\vec{r}(t))=\frac{1}{2}\left(\mathbf{1}+r_x(t)
X+r_y(t)Y+r_zZ\right)\,.
\end{equation}
Solving the governing   \!\eqref{master},we find that the Bloch vector at time $t$ is given by
\begin{equation}
    \frac{r_x(t)}{r_x(0)}= \frac{r_y(t)}{r_y(0)}=e^{-kt}\,,\  r_z(t)=r_z(0)\,.
\end{equation}
In the infinite time limit, the off-diagonal pieces die off:
\begin{equation}
    \rho(\vec{r}(\infty))=\begin{pmatrix}
        \frac{1+r_z}{2}& 0\\
        0 & \frac{1-r_z}{2}
    \end{pmatrix}\,.
\end{equation}
One can easily compute entanglement asymmetry and strong entanglement asymmetry as a function of time:
\begin{equation}
\begin{aligned}\label{eq:Asss}
\Sa(t) &=f(|\vec{r}(t)|)-f(r_z)\,,\\
    \Ssa(t)&=\log 2-2f(r_z)+f(|\vec{r}(t)|)\,,
    \end{aligned}
\end{equation}
where $f:[-1,1]\to \mathbb{R}_{\geq 0}$ is given by 
\begin{equation}
    \label{eq:defF}f(x)=\frac{1+x}{2}\log (1+x) +\frac{1-x}{2}\log (1-x)\,.
\end{equation}

Now we choose two density matrices $\rho_1(0)$ and $\rho_2(0)$, parametrized by the Bloch-vectors at time $t=0$:
\begin{equation}
   \label{eq:r} \vec{r}_1(0)=(1/\sqrt{2},0,1/\sqrt{2})\,,\ \vec{r}_2(0)=(1/2,0,1/2)\,.
\end{equation}
We let them evolve under \eqref{master} and observe that eventually cross-over in $\Ssa(t)$ does happen (see fig.~\!\ref{fig:strong})
\begin{figure}
    \centering
    \includegraphics[scale=0.5]{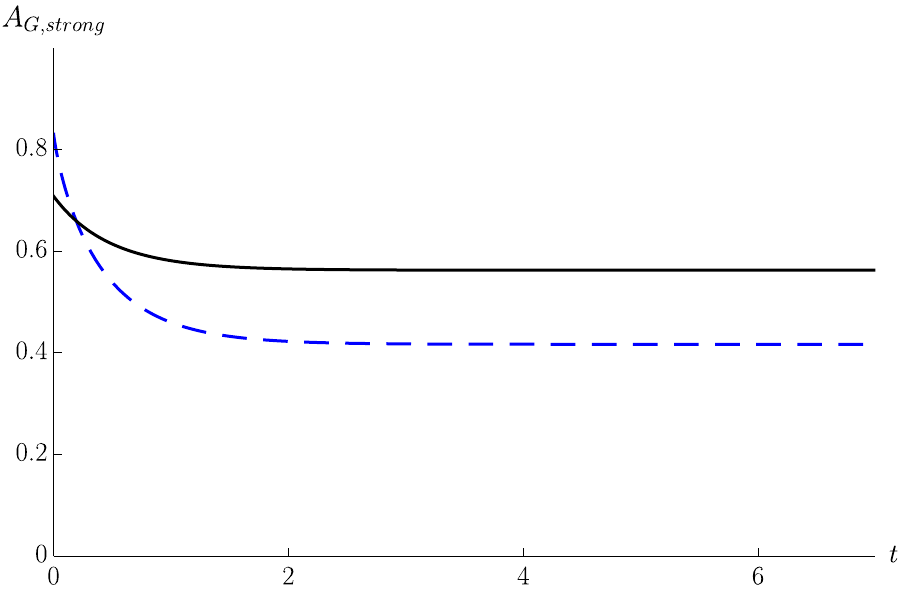}
    \caption{Dashed blue line denotes the evolution of strong asymmetry for $\rho_1$ and the thick black line denotes the evolution of strong asymmetry for $\rho_2$. There is a cross-over under evolution by \eqref{master} (here $k=1$). See eqs.~\!\eqref{eq:r}, \eqref{eq:Asss} and \eqref{eq:defF}.}
    \label{fig:strong}
\end{figure}
i.e.\ we have
\begin{equation}
    \begin{aligned}
      0.832\approx\Ssa(\rho_1(0))&>\Ssa(\rho_2(0))\approx 0.708\\  
      0.416\approx\Ssa(\rho_1(\infty))&<\Ssa(\rho_2(\infty))\approx 0.562\,,
    \end{aligned}
\end{equation}
which shows the ordering of $\Ssa$ for these density matrices gets flipped eventually.  Note that the weak asymmetry $\Sa(t)$ asymptotes to $0$ for both density matrices without any Mpemba-like effect for weak symmetry in this example, see fig.~\!\ref{fig:weak}.
\begin{figure}[!ht]
    \centering
    \includegraphics[scale=0.5]{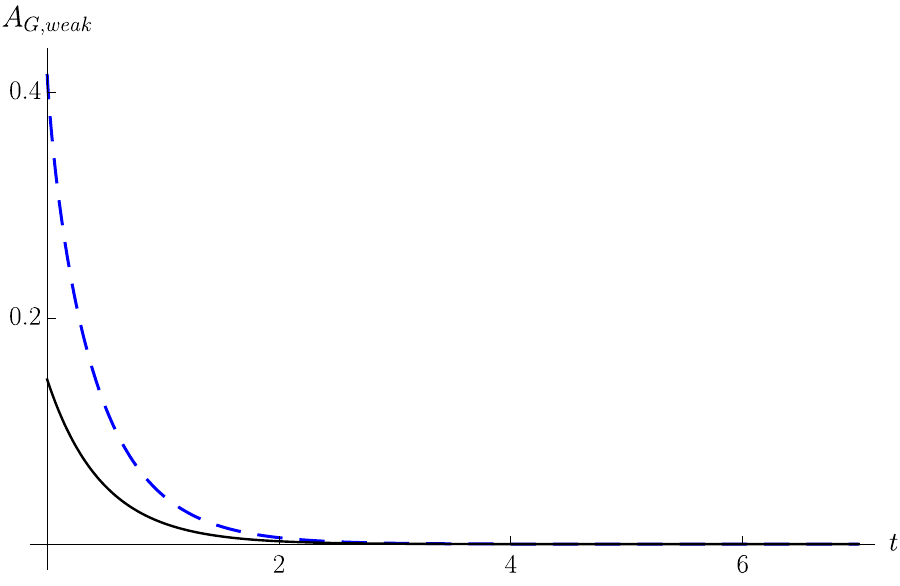}
    \caption{Dashed blue line denotes the evolution of weak asymmetry for $\rho_1$ and the thick black line denotes the evolution of weak asymmetry for $\rho_2$. The evolution happens according to \eqref{master} (here $k=1$). Both asymptotes to $0$ without having any cross-over. See eqs.~\!\eqref{eq:r}, \eqref{eq:Asss} and \eqref{eq:defF}.  }
    \label{fig:weak}
\end{figure}

\section{Discussion}
In this work, we established a method that systematically introduces new quantifiers of symmetry breaking under strong symmetry by constructing a new resource theory. We believe this represents an important next step beyond simply exploring connections between resource theory and condensed-matter physics in general. Since links between resource-theoretic ideas and condensed-matter physics have also begun to be explored in other contexts \cite{Styliaris2019,Liu2020,Oliviero2022,Aditya2025,Balasubramanian2025, Aditya2025a}, constructing resource theories tailored to specific physical settings and developing the corresponding appropriate measures may serve as a promising direction for future work.

\section*{Acknowledgments}
YK is supported by the INAMORI Frontier Program at Kyushu University and JSPS KAKENHI Grant Number 23K20046.
HT was supported by JSPS Grants-in-Aid for Scientific Research 
No. JP25K00924, and MEXT KAKENHI Grant-in-Aid for Transformative
Research Areas B ``Quantum Energy Innovation'' Grant Numbers 24H00830 and 24H00831, JST MOONSHOT No. JPMJMS2061, and JST FOREST No. JPMJFR2365. SP thanks all the participants of the workshop \textit{Entanglementasymmetry across energy scales} at LPTHE, especially Francesco Benini, Giuseppe Di Giulio, Victor Godet, Sara Murciano for discussion on various topics related to entanglement asymmetry. SP thanks Shu-heng Shao for introducing him to the notion of strong symmetry.

\clearpage
\appendix

\section{Properties of the logarithmic fidelity-based characteristic functions for weak symmetry}\label{sec:prop_LCF}
Here we show that the logarithmic fidelity-based characteristic functions
\eq{
LF_{G,g}(\rho):=-\log F(\rho,U_g\rho U^\dagger_g)
}
satisfy (i) \emph{the minimal requirements (i-a) and (i-b) as a resource monotone}, (iii) \emph{Convexity}, and (v) \emph{Additivity} for every $g\in G$.
We also show that (iv) \emph{Strong monotonicity} fails in general.
And the set $\{LF_{G,g}(\rho)\}_{g\in G}$ satisfies (ii) \emph{Faithfulness}: $\rho\in\calF_{G}$ iff $LF_{G,g}(\rho)=0$ for any $g\in G$.

\begin{proof}
(i-a): By definition, each $LF_{G,g}(\rho)$ is non-negative. And when $\rho\in\calF_{G}$, $U_g\rho U^\dagger_g=\rho$ holds, and thus $LF_{G,g}(\rho)=0$.

(i-b): For any $(U,U')$-covariant operation $\Lambda$, the following relation holds
\eq{
F(\rho,U_g\rho U^\dagger_g)&\stackrel{(a)}{\le} F(\Lambda(\rho),\Lambda(U_g\rho U^\dagger_g))\nonumber\\
&\le F(\Lambda(\rho),U'_g\Lambda(\rho) U'^\dagger_g),
}
which implies the monotonicity of $LF_{G,g}(\rho)$. Here, we used the monotonicity of fidelity in (a).

(ii): 
Since $F(\rho,\sigma)=1$ iff $\rho=\sigma$, if $LF_{G,g}(\rho)=0$ for any $g\in G$, $\rho$ is symmetric. The converse is guaranteed by (i-a).

(iii):
Due to $F(\sum_jp_j\rho_j,\sum_jp_j\sigma_j)\ge\sum_jp_jF(\rho_j,\sigma_j)$, we obtain
\eq{
F(\sum_jp_j\rho_j,U_g\sum_jp_j\rho_jU^\dagger_g)\ge\sum_jp_jF(\rho_j,U_g\rho_jU^\dagger_g).
}
Therefore, due to the convexity of the function $-\log$, we obtain the convexity of $LF_{G,g}(\rho)$.

(iv): Strong monotonicity does not hold in general, as the following example shows.
Let $G=\mathbb Z_2=\{0,1\}$, and consider
\eq{
\calH&=(\mathbb C^2)_1\oplus(\mathbb C^2)_2,\notag\\
U_g&=Z^g\oplus Z^g,\qquad g\in\mathbb Z_2,
}
where $Z:=\ket{0}\bra{0}-\ket{1}\bra{1}$ acts on each qubit.
Let $P_1$ and $P_2$ be the projections onto the two blocks by
\eq{
P_1:=I_2\oplus0,\qquad P_2:=0\oplus I_2,
}
and define
\eq{
\Psi_i(X):=P_iXP_i,\qquad i=1,2.
}
Since $[P_i,U_g]=0$, each $\Psi_i$ is covariant.
The associated QC channel with a trivial representation on the classical register is therefore covariant, so $\{\Psi_1,\Psi_2\}$ is a free measurement.

For $0<a<1$, define the qubit state
\eq{
\ket{\psi(a)}&:=\sqrt{\frac{1+a}{2}}\ket{0}
+\sqrt{\frac{1-a}{2}}\ket{1}.
}
With the fidelity convention used here,
\eq{
F(\psi(a),Z\psi(a)Z)
=|\bra{\psi(a)}Z\ket{\psi(a)}|=a,
}
where we used the abbreviation $\psi(a):=\ket{\psi(a)}\bra{\psi(a)}$.
Set $a_1=1/10$ and $a_2=9/10$, and take
\eq{
\rho:=\frac12\psi(a_1)\oplus\frac12\psi(a_2).
}
The outcomes have probabilities $p_1=p_2=1/2$, and their normalized states are
\eq{
\rho_1=\psi(a_1)\oplus0,\qquad
\rho_2=0\oplus\psi(a_2).
}
Evaluating the fidelity block by block gives
\eq{
F(\rho,U_1\rho U_1^\dagger)
&=\frac{a_1+a_2}{2}=\frac12,\notag\\
LF_{G,1}(\rho)&=\log2.
}
On the other hand,
\eq{
\sum_{i=1}^2p_iLF_{G,1}(\rho_i)
&=-\frac12\log(a_1a_2)\notag\\
&=\log\frac{10}{3}>\log2=LF_{G,1}(\rho).
}
Thus, this free measurement increases $LF_{G,1}$ on average, disproving strong monotonicity.

\vspace{0.1cm}
(v): Due to $F(\rho\otimes\sigma,\rho'\otimes\sigma')=F(\rho,\rho')F(\sigma,\sigma')$,we obtain the additivity of $LF_{G,g}(\rho)$ for product states.

\end{proof}

\section{Rényi-2 proxies are not resource monotones}
\label{app:Renyi2-not-monotone}

In this appendix, we give an explicit counterexample showing that the
second R\'enyi version of the relative entropy of asymmetry is not
a resource monotone under $G$-covariant channels, already in the
simplest setting of a single qubit with a $\bb{Z}_2$ symmetry.

\subsection{Setup: $\mathbb{Z}_2$ symmetry generated by $X$}

Consider a single qubit with Hilbert space $\ca{H} \simeq \bb{C}^2$.
We take the symmetry group to be $G = \bb{Z}_2 = \{e,g\}$ with
unitary representation
\begin{equation}
  U_e = \mathbf{1}, \qquad U_g = X,
\end{equation}
where $X$ is the Pauli-$X$ matrix.
The corresponding twirling channel is
\begin{equation}
  \ca{G}_X(\rho)
  := \frac{1}{2}\bigl(\rho + X\rho X\bigr).
\end{equation}
The second R\'enyi entropy of a state $\rho$ is
\begin{equation}
  S^{(2)}(\rho) := - \log \tr(\rho^2),
\end{equation}
and we define the second R\'enyi asymmetry with respect to $X$ as
\begin{equation}
  A_X^{(2)}(\rho)
  := S^{(2)}\!\bigl(\ca{G}_X(\rho)\bigr) - S^{(2)}(\rho)
  = \log\frac{\tr(\rho^2)}{\tr(\ca{G}_X(\rho)^2)}.
\end{equation}
By construction $A_X^{(2)}(\rho) \ge 0$ and $A_X^{(2)}(\rho) = 0$ if and only if
$\rho$ is invariant under conjugation by $X$.
The question is whether $A_X^{(2)}$ is non-increasing under all
$\bb{Z}_2$-covariant CPTP maps.

\subsection{$Z$-dephasing as an $\bb{Z}_2$-covariant channel}

Let us consider the dephasing channel in the $Z$ basis,
\begin{equation}
  \Delta_Z(\rho)
  := \sum_{s=\pm 1} P_s \rho P_s, \qquad
  P_{\pm} := \frac{1}{2}(\mathbf{1} \pm Z),
\end{equation}
where $Z$ is the Pauli-$Z$ matrix.
This is a CPTP map that removes the off-diagonal elements in the eigenbasis of $Z$.
We claim that $\Delta_Z$ is $\bb{Z}_2$-covariant.

Indeed, using $XZX = -Z$ one checks that
\begin{equation}
  X P_{\pm} X = P_{\mp}.
\end{equation}
Then, for any state $\rho$,
\begin{equation}
\begin{aligned}
  \Delta_Z\bigl( X\rho X \bigr)
  &= \sum_{s=\pm 1} P_s X\rho X P_s\\
  &= X\, \Delta_Z(\rho)\, X,
\end{aligned}
\end{equation}
i.e.\ it is $\bb{Z}_2$-covariant and hence a free operation in the resource theory of asymmetry for this symmetry.

\subsection{Explicit counterexample}

We now exhibit a state $\rho$ such that
\begin{equation}
  A_X^{(2)}\bigl( \Delta_Z(\rho) \bigr)
  > A_X^{(2)}(\rho),
\end{equation}
showing that $A_X^{(2)}$ is not a resource monotone.

It is convenient to use the Bloch-sphere representation
\begin{equation}
  \rho = \frac{1}{2}\bigl( \mathbf{1} + r_x X + r_y Y + r_z Z \bigr)
\end{equation}
with Bloch vector $\vec{r} = (r_x,r_y,r_z)$ satisfying $\abs{\vec{r}} \le 1$.
Here, $X,Y,Z$ are the Pauli matrices.
In this parametrization, one has
\begin{equation}
  \tr(\rho^2) = \frac{1 + \abs{\vec{r}}^2}{2}.
\end{equation}
The twirling channel $\ca{G}_X$ implements the average over conjugation by $X$,
under which the Bloch vector transforms as
\begin{equation}
  (r_x,r_y,r_z) \mapsto (r_x,-r_y,-r_z).
\end{equation}
Thus,
\begin{equation}
  \ca{G}_X(\rho) = \frac{1}{2}\bigl( \mathbf{1} + r_x X \bigr),
\end{equation}
so that
\begin{equation}
  \tr\bigl(\ca{G}_X(\rho)^2\bigr)
  = \frac{1 + r_x^2}{2}.
\end{equation}
Therefore
\begin{equation}
  A_X^{(2)}(\rho)
  = \log\frac{\tr(\rho^2)}{\tr(\ca{G}_X(\rho)^2)}
  = \log\frac{1 + r_x^2 + r_y^2 + r_z^2}{1 + r_x^2}.
\end{equation}

Next, we apply the $Z$-dephasing channel $\Delta_Z$.
This map sets the $X$ and $Y$ components of the Bloch vector to zero,
\begin{equation}
  \Delta_Z(\rho) = \frac{1}{2}\bigl( \mathbf{1} + r_z Z \bigr),
\end{equation}
so that
\begin{equation}
  \tr\bigl(\Delta_Z(\rho)^2\bigr) = \frac{1 + r_z^2}{2}.
\end{equation}
Twirling this dephased state with respect to $X$ simply removes the $Z$ component,
\begin{equation}
\begin{aligned}
  \ca{G}_X\bigl(\Delta_Z(\rho)\bigr)
  &= \frac{\mathbf{1}}{2}, \\
  \tr\Bigl(\ca{G}_X\bigl(\Delta_Z(\rho)\bigr)^2\Bigr) &= \frac{1}{2}.
\end{aligned}
\end{equation}
Hence
\begin{equation}
  A_X^{(2)}\bigl(\Delta_Z(\rho)\bigr)
  = \log\frac{\tr\bigl(\Delta_Z(\rho)^2\bigr)}{\tr\bigl(\ca{G}_X(\Delta_Z(\rho))^2\bigr)}
  = \log(1 + r_z^2).
\end{equation}

The condition
\begin{equation}
  A_X^{(2)}\bigl(\Delta_Z(\rho)\bigr) > A_X^{(2)}(\rho)
\end{equation}
is therefore equivalent to
\begin{equation}
  \log(1 + r_z^2)
  > \log\frac{1 + r_x^2 + r_y^2 + r_z^2}{1 + r_x^2},
\end{equation}
or, after exponentiating and simplifying,
\begin{equation}
  r_x^2 r_z^2 > r_y^2.
\end{equation}
This inequality is satisfied for a large set of Bloch vectors.
As a simple explicit choice, take
\begin{equation}
  \vec{r} = (r,0,r), \qquad 0 < r < \frac{1}{\sqrt{2}},
\end{equation}
so that $\abs{\vec{r}}^2 = 2r^2 < 1$ and $\rho$ is a valid quantum state.
Then
\begin{equation}
  r_x^2 r_z^2 = r^4 > 0 = r_y^2,
\end{equation}
and therefore
\begin{equation}
  A_X^{(2)}\bigl(\Delta_Z(\rho)\bigr) > A_X^{(2)}(\rho).
\end{equation}

For instance, choosing $r = 1/2$ one obtains
\begin{equation}
  \rho = \frac{1}{2}\bigl(\mathbf{1} + \tfrac{1}{2} X + \tfrac{1}{2} Z\bigr)
  = \begin{pmatrix}
      0.75 & 0.25 \\
      0.25 & 0.25
    \end{pmatrix},
\end{equation}
for which
\begin{equation}
\begin{aligned}
  A_X^{(2)}(\rho) &= \log\frac{6}{5} \approx 0.1823, \\
  A_X^{(2)}\bigl(\Delta_Z(\rho)\bigr) &= \log\frac{5}{4} \approx 0.2231.
\end{aligned}
\end{equation}
Thus the $\bb{Z}_2$-covariant dephasing channel $\Delta_Z$ strictly \emph{increases}
the second R\'enyi asymmetry.

This provides a concrete counterexample to the monotonicity of $A_X^{(2)}$
under $G$-covariant operations.
In particular, the second R\'enyi version of the relative entropy of asymmetry
cannot be regarded as a resource monotone for weak symmetry breaking,
even in the simplest qubit example.

Note that for the same $\rho$,\
the relative entropy of asymmetry is given by
\begin{equation}
\begin{aligned}
A_X(\rho) &\approx 0.1458,\\
A_X(\Delta_Z(\rho)) &\approx 0.1308,
\end{aligned}
\end{equation}
and hence
\begin{equation}
A_X(\Delta_Z(\rho))<A_X(\rho).
\end{equation}
Therefore, $A_X$ indeed satisfies monotonicity, unlike $A_X^{(2)}$. 

{\color{black}
In fact, for a general density matrix $\rho$, one can show that 
\begin{equation}
   A_X(\Delta_Z(\rho))- A_X(\rho)=f(r_x)+f(r_z)-f\left(\sqrt{r_x^2+r_y^2+r_z^2}\right)\,,
\end{equation}
where $f:[-1,1]\to \mathbf{R}_{\geq0}$ is given by 
\begin{equation}
  \label{eq0}  f(x)=\fr{1+x}{2}\log(1+x)+\fr{1-x}{2}\log(1-x).
\end{equation}
First of all note that $f(x)$ is an even function of $x$. So it suffices to restrict our attention to $x\geq 0$.

Now note that 
\begin{equation}\label{eq1}
\frac{f(x)}{x^2}=\sum_{n=0}^{\infty} \frac{x^{2 n}}{2 (n+1) (2 n+1)}\,,
\end{equation}
from which it follows that
$\frac{f(x)}{x^2}$ is an increasing function of $x$ for $x\geq 0$. Thus we have 
\begin{equation}\label{ineq1}
\begin{aligned}
    &\frac{f(r_x)}{r_x^2} \leq \frac{f\left(\sqrt{r_x^2+r_y^2+r_z^2}\right)}{r_x^2+r_y^2+r_z^2}\,,\\
    &\frac{f(r_z)}{r_z^2} \leq \frac{f\left(\sqrt{r_x^2+r_y^2+r_z^2}\right)}{r_x^2+r_y^2+r_z^2}\,.
    \end{aligned}
\end{equation}
Thus we have 
\begin{equation}
\begin{aligned}
    &f(r_x)+f(r_z)\\
    &\leq \frac{r_x^2+r_z^2}{r_x^2+r_y^2+r_z^2}f\left(\sqrt{r_x^2+r_y^2+r_z^2}\right)\\
    &\leq f\left(\sqrt{r_x^2+r_y^2+r_z^2}\right).
\end{aligned}
\end{equation}
Hence it follows that 
\begin{equation}
\begin{aligned}
    A_X(\Delta_Z(\rho))- A_X(\rho)\leq 0\end{aligned}.
\end{equation}
}
It is important to stress that this ad hoc R\'enyi asymmetry $A_G^{(n)}$ should \emph{not}
be confused with the $\alpha$-R\'enyi asymmetry measure of relative entropy introduced in
~\cite{Gao2017}.  
The latter is a genuine resource monotone under $G$-covariant operations, but completely different from $A_G^{(n)}$.

\clearpage
\bibliographystyle{JHEP}
\bibliography{main}

\widetext

\clearpage

\begin{center}
{\large \bf Supplemental Material for \protect \\ 
``Resource-Theoretic Quantifiers of Weak and Strong Symmetry Breaking:\\ Strong Entanglement Asymmetry and Beyond''}\\
\vspace*{0.3cm}
Yuya Kusuki$^{1,2,3}$, Sridip Pal$^{4}$ and Hiroyasu Tajima$^{5,6}$ \\
\vspace*{0.1cm}
$^{1}${\small \em Institute for Advanced Study, 
Kyushu University, Fukuoka 819-0395, Japan.
}
\\ 
$^{2}${\small \em Department of Physics, 
Kyushu University, Fukuoka 819-0395, Japan.
}
\\
$^{3}${\small \em RIKEN Interdisciplinary Theoretical and Mathematical Sciences (iTHEMS),
Wako, Saitama 351-0198, Japan.
}
\\
$^{3}${\small \em Institut des Hautes Études Scientifiques (IHES), 91440 Bures-sur-Yvette, France.
}
\\
$^{4}${\small \em Department of Informatics, Faculty of Information Science and Electrical Engineering, 
Kyushu University, Fukuoka 819-0395, Japan.
}
\\
$^{5}${\small \em JST, FOREST, 4-1-8 Honcho, Kawaguchi, Saitama, 332-0012, Japan.
}
\end{center}

\section{Preliminaries: Groups and (projective) unitary representations}

Let $G$ be a group. A \emph{unitary representation} of $G$ on a Hilbert space $\mathcal{H}$ is a homomorphism
\begin{equation}
    U: G \to \mathcal{U}(\mathcal{H}), \qquad
    g \mapsto U_g,
\end{equation}
satisfying the exact group composition rule
\begin{equation}
    U_g U_h = U_{gh}, \quad
    U_e = \mathbb{I},
\end{equation}
where $\mathcal{U}(\mathcal{H})$ denotes the unitary group acting on $\mathcal{H}$ and $e$ is the identity element of $G$.

In quantum theory, symmetries are often represented not by exact unitary representations but by \emph{projective unitary representations}.  
A projective unitary representation $U: G\rightarrow\calU(\calH)$ satisfies
\begin{equation}
    U_g U_h = \omega(g,h)\, U_{gh},
    \label{eq:projective_rep}
\end{equation}
where $\omega(g,h)$ is a phase factor known as a \emph{2-cocycle}, obeying the associativity condition
\begin{equation}
    \omega(g,h)\, \omega(gh,k)
    = \omega(h,k)\, \omega(g,hk),
    \label{eq:cocycle_condition}
\end{equation}
for all $g,h,k \in G$.  

When $\{U_g\}$ is a projective unitary representation, $U_e=I$ does not always hold.
In general, $U_e=cI$ is always valid, where $c$ is a complex number whose absolute value is 1.
However, taking $\alpha(g):G\rightarrow U(1)$ satisfying $\alpha(e)=c^{-1}$, we can always redefine $\{U_g\}$ as $\tilde{U}_g:=\alpha(g)U_g$. 
Then, $\tilde{U}_g$ becomes another projective unitary representation since it satisfies
\eq{
\tilde{U}_g\tilde{U}_h&=\alpha(g)\alpha(h)\omega(g,h)U_{gh}\nonumber\\
&=\tilde{\omega}(g,h)\tilde{U}_{gh},\\
\tilde{\omega}(g,h)&:=\alpha(g)\alpha(h)\alpha(gh)^{-1}\omega(g,h).
}
By definition, $\{\tilde{U}_g\}$ satisfies $\tilde{U}_e=I$ which implies that its cocycle satisfies $\tilde{\omega}(g,e)=\tilde{\omega}(e,g)=1$. Therefore, in this supplementary materials, we assume that $U_e=I$ (and $\omega(g,e)=\omega(e,g)=1$).

In this Supplementary Material, we assume that the group $G$ is compact. We also focus on (projective) unitary representations which have the following irreducible decomposition:
\eq{
U_g=\bigoplus_{\nu}U^{\nu}_g\otimes I^{m_{\nu}}\label{decom_U_a}.
}
Here $m_{\nu}$ denotes the multiplicity of the irrep $\nu$.
We also use the irreducible decomposition of the Hilbert space $\calH$ induced by $\{U_g\}_{g\in G}$:
\eq{
\calH=\bigoplus_{\nu}\calH_{\nu}\otimes \mathbb{C}^{m_{\nu}}.\label{decom_H}
}

In the decomposition of a projective unitary representation $U_g=\bigoplus_\nu U^{(\nu)}_g\otimes I^{m_{\nu}}$, the cocycle of each irreducible representation is the same. To be concrete, for any irreducible representation $U^{(\mu)}_g$ in the decomposition, its cocycle $\omega^{(\mu)}$ satisfies 
\eq{
\omega^{(\mu)}(g,h)=\omega(g,h),
}
where $\omega$ is the cocycle of $\{U_g\}$.

\begin{proof}
\eq{
\omega(g,h)U_{gh}&=U_gU_h\nonumber\\
&=\bigoplus_{\mu}U^{(\mu)}_gU^{(\mu)}_h\otimes I^{m_{\mu}}\nonumber\\
&=\bigoplus_{\mu}\omega^{(\mu)}(g,h) U^{(\mu)}_{gh}\otimes I^{m_{\mu}}
}
Multiplying the both sides by $P^{(\mu)}$ that is the projection to $\calH_{\mu}\otimes \mathbb{C}^{m_{\mu}}$ on the right, we obtain
\eq{
\omega(g,h)U^{(\mu)}_{gh}\otimes I^{m_{\mu}}=\omega^{(\mu)}(g,h)U^{(\mu)}_{gh}\otimes I^{m_{\mu}}.
}
Since $U^{(\mu)}_{gh}\otimes I^{m_{\mu}}$ is an invertible matrix, we obtain $\omega(g,h)=\omega^{(\mu)}(g,h)$.
\end{proof}

Even for such projective unitary representations, Schur’s lemma and Schur's orthogonality theorem remain valid.
\begin{lemma}[Schur's lemma]\label{SLemm}
Let $G$ be a group, and let $\{U_g\}$ and $\{U'_g\}$ be irreducible (projective) unitary representations acting on $\calH$ and $\calH'$. Let $T$ be a linear map from $\calH$ to $\calH'$ satsfying 
\eq{
TU_g=U'_gT,\enskip \forall g\in G.\label{eq:Schur}
}
Then, (a) $T$ is a zero map or (b) $T$ is bijection and $T=cV$, where $c$ is a scalar satisfying $|c|>0$ and $V$ is a unitary map from $\calH$ to $\calH'$, and the irreducible representations $\{U_g\}$ and $\{U'_g\}$ satisfy $U'_g=VU_gV^\dagger$, i.e. they are equivalent to each other.
\end{lemma}

\begin{proof}
When $T=0$ is valid, \eqref{eq:Schur} is also clearly valid.
Therefore, below we assume that $T$ is a non-zero map, and show that (b) is valid.
Note that
\eq{
\omega(g,h)TU_{gh}&=TU_gU_h\nonumber\\
&=U'_gU'_hT\nonumber\\
&=\omega'(g,h)U'_{gh}T\nonumber\\
&=\omega'(g,h)TU_{gh}.\label{schur1}
}
Since $T$ is a non-zero map, there exists $\ket{\psi}$ s.t. $T\ket{\psi}\ne0$.
Then, multiplying the both sides of \eqref{schur1} by $U^\dagger_{gh}\ket{\psi}$ on the right, we obtain
$\omega(g,h)T\ket{\psi}=\omega'(g,h)T\ket{\psi}$, which implies
\eq{
\omega(g,h)=\omega'(g,h).\label{omega=omega'}
}
Due to \eqref{omega=omega'}, we can define a group $\tilde{G}:=U(1)\times G$ whose element $(z,g)$ satisfies
\eq{
(z,g)\cdot(z',g')=(zz'\omega(g,g'),gg')
}
and non-projective unitary representations $\tilde{U}$ and $\tilde{U}'$ of $\tilde{G}$ acting on $\calH$ and $\calH'$
\eq{
\tilde{U}_{z,g}:=zU_g,\enskip \tilde{U}'_{z,g}:=zU'_g.
}
Then, $\tilde{U}$ is irreducible iff $U$ is irreducible, and $\tilde{U}'$ is irreducible iff $U'$ is irreducible.
Therefore, \eqref{eq:Schur} implies
\eq{
T\tilde{U}_{z,g}=U'_{z,g}T,\enskip\forall(z,g)\in\tilde{G}.
}
Therefore, because of Schur's lemma for non-projective unitary representation~\cite{hayashi2017group}, (b) must be valid.
\end{proof}

\begin{lemma}[Schur's orthogonality theorem]
Let $G$ be a group, and let $U^{(\mu)}$ and $U^{(\lambda)}$ be irreducible (projective) unitary representations of $G$ acting on $\calH_\mu$ and $\calH_\lambda$, respectively.
If the cocycles of $U^{(\mu)}$ and $U^{(\lambda)}$ are the same, the following relations hold:
\begin{description}
\item[(a)]When $U^{(\mu)}$ and $U^{(\lambda)}$ are not equivalent,
\eq{
\int_G dg\bra{j}_\mu (U^{(\mu)}_g)^\dagger\ket{i}_{\mu}\bra{k}_\lambda U^{(\lambda)}_{g}\ket{l}_{\lambda}=0,\label{eq:orthogonality1}
}
where $\{\ket{i}_\mu\}$ and $\{\ket{k}_\lambda\}$ are arbitrary orthonormal bases of $\calH_\mu$ and $\calH_\lambda$, respectively.
\item[(b)]When $U^{(\mu)}$ and $U^{(\lambda)}$ are equivalent, 
\eq{
\int_G dg\bra{j}_\mu (U^{(\mu)}_g)^\dagger\ket{i}_{\mu}\bra{k}_\lambda U^{(\lambda)}_{g}\ket{l}_{\lambda}=\frac{\delta_{i,k}\delta_{j,l}}{d_{\mu}},\label{eq:orthogonality2}
}
where $d_\mu$ is the dimension of $\calH_\mu$, and $\{\ket{i}_\mu\}$ and $\{\ket{k}_\lambda\}:=\{V\ket{k}_{\mu}\}$ are orthonormal bases of $\calH_\mu$ and $\calH_\lambda$, respectively, where $V$ is a unitary satisfying $U^{(\lambda)}_g=VU^{(\mu)}_gV^\dagger$.
\end{description}
\end{lemma}

\begin{proof}
Since cocycles of $U^{(\mu)}$ and $U^{(\lambda)}$ are the same, again we can define the group $\tilde{G}$ and its non-projective unitary representations $\tilde{U}^{(\mu)}$ and $\tilde{U}^{(\lambda)}$ as $\tilde{G}$, $\tilde{U}$ and $\tilde{U}'$ in the proof of Lemma \ref{SLemm}, respectively.
Note that
\eq{
\int_{U(1)}dz\int_G dg\bra{j}_\mu (U^{(\mu)}_{z,g})^\dagger\ket{i}_{\mu}\bra{k}_\lambda U^{(\lambda)}_{z,g}\ket{l}_{\lambda}&=\int_{U(1)}dz |z|^2\int_G dg\bra{j}_\mu (U^{(\mu)}_{g})^\dagger\ket{i}_{\mu}\bra{k}_\lambda U^{(\lambda)}_{g}\ket{l}_{\lambda}\nonumber\\
&=\int_G dg\bra{j}_\mu (U^{(\mu)}_g)^\dagger\ket{i}_{\mu}\bra{k}_\lambda U^{(\lambda)}_{g}\ket{l}_{\lambda}.
}
Therefore, from Schur's orthogonality theorem for non-projective unitary representations~\cite{hayashi2017group}, we obtain (a) and (b).
\end{proof}

\section{Properties of strong symmetric states, single-sector states and strong covariant operations}
In this section, we give the basic properties of strong symmetric states, single-sector states and strong covariant operations.
For the readers' convenience, we repeat the definitions:

\medskip
\noindent\textbf{Free states: Strong symmetric states.}
Let $S$ be a quantum system, and let $G$ be a group with a (projective) unitary representation $U$ acting on $S$. 
A state $\rho$ on $S$ is said to be strong symmetric with respect to $U$ if it satisfies
\eq{
U_g\rho=e^{i\theta_{g,\rho}}\rho,\enskip\forall g\in G,
}
where $\theta_{g,\rho}$ is a real-valued function of $g$ and $\rho$. We refer to the whole set of strong symmetric states as ${\cal F}_{G,\mathrm{strong}}$.

\medskip
\noindent\textbf{Additional class of states: Single-sector states.}
Let $S$ be a quantum system, and let $G$ be a group with a (projective) unitary representation $U$ acting on $S$. 
We assume that $U_{g\in G}$ admits the following irreducible decomposition:
\eq{
U_g = \bigoplus_{\nu} U_g^{(\nu)} \otimes I_{m_\nu},
}
where $m_{\nu}$ denotes the multiplicity of the irrep labeled by $\nu$.  
Correspondingly, the Hilbert space $\calH$ can be decomposed as
\eq{\label{sum_b}
\calH = \bigoplus_{\nu} \calH_{\nu} \otimes \mathbb{C}^{m_\nu}.
}
A state $\rho$ on $S$ is said to be \textit{single-sector} with respect to $U$ if it is weak symmetric with respect to $U$ and satisfies
\eq{
\exists\, \nu \text{ such that } P_\nu \rho P_\nu = \rho,
}
where $P_\nu$ is the projection operator onto $\calH_{\nu} \otimes \mathbb{C}^{m_\nu}$.
We refer to the whole set of single-sector states as ${\cal F}_{G,\mathrm{single}}$.

\medskip
\noindent
\textbf{Free operations: strong covariant operations.}
Let $G$ be a group, and let $U$ and $U'$ be (projective) unitary  representations of $G$ acting on two Hilbert spaces $\calH$ and $\calH'$, respectively.
A CPTP map $\Lambda: \calB(\calH)\to \calB(\calH')$ is said to be $(U,U')$-strong covariant if it satisfies
\begin{equation}
    \Lambda\!\left(U_g...\right)
    = U'_g \Lambda(...) 
    \quad \forall g \in G.
    \label{eq_SM:s-covariance_def}
\end{equation}
We refer to the whole set of strong covariant operations as ${\cal O}_{G,\mathrm{strong}}$.

For strong covariant operations, the following lemma holds:
\begin{lemma}\label{lemm_SM:cocycle_SC}
Let $G$ be a (compact) group, and let $U:G\rightarrow U(\calH)$ and $U':G\rightarrow U(\calH')$ be (projective) unitary representations acting on Hilbert spaces $\calH$ and $\calH'$, respectively.
Suppose that there is a $(U,U')$-strong covariant operation $\Lambda:\calB(\calH)\rightarrow\calB(\calH')$.
Then, the cocycles $\omega(g,h)$ and $\omega'(g,h)$ of $U$ and $U'$ coincide:
\eq{
\omega(g,h)=\omega'(g,h),\enskip\forall g,h\in G
}
\end{lemma}

\begin{proof}
Since $\Lambda$ is $(U,U')$-strong covariant, the following relations hold for an arbitrary state $\rho\in\calS(\calH)$, and thus the cocycles of $U$ and $U'$ coincide: 
\eq{
\omega'(g,h)&=\Tr[U'^\dagger_{gh}U'_{g}U'_h\Lambda(\rho)]\nonumber\\
&=\Tr[\Lambda(U^\dagger_{gh}U_{g}U_h\rho)]\nonumber\\
&=\omega(g,h).
}
\end{proof}

\subsection{Properties as free states and free operations}
We firstly show that either the combination of $(\calF_{G,\mathrm{strong}},\calO_{G,\mathrm{strong}})$ or $(\calF_{G,\mathrm{single}},\calO_{G,\mathrm{strong}})$ satisfies the minimal requirements of the resource theory.
\begin{theorem}\label{thm_SM:GR_strong}
Let $G$ be a (compact) group, and let $U$, $U'$ and $U''$ be (projective) unitary representations acting on Hilbert spaces $\calH$, $\calH'$ and $\calH''$, respectively.
Then, the following three are valid:$\\$
(i) The identity operation on $\calB(\calH)$ is $(U,U)$-strong covariant. $\\$
(ii) If $\Lambda:\calB(\calH)\rightarrow \calB(\calH')$ and $\Lambda':\calB(\calH')\rightarrow \calB(\calH'')$ are $(U,U')$- and $(U',U'')$-strong covariant, respectively, $\Lambda'\circ\Lambda$ is also $(U,U'')$-strong covariant. $\\$
(iii) If $\Lambda:\calB(\calH)\rightarrow \calB(\calH')$ is a $(U,U')$-strong covariant CPTP map, the following relations hold:
\eq{
\rho\in{\cal F}_{G,\mathrm{strong}}&\Rightarrow\Lambda(\rho)\in{\cal F}_{G,\mathrm{strong}},\label{eq_SM:GR_strong}\\
\rho\in{\cal F}_{G,\mathrm{single}}&\Rightarrow\Lambda(\rho)\in{\cal F}_{G,\mathrm{single}}.\label{eq_SM:GR_single}
}
\end{theorem}

\begin{proofof}{Theorem \ref{thm_SM:GR_strong}}
The properties (i) and (ii) directly follow from the definition of strong covariant operations.

Let us prove (iii).
We first prove \eqref{eq_SM:GR_strong}. 
Let $\rho$ be a state on $S$ that is strong symmetric with respect to $U$. 
Then, there exists a real-valued function $\theta_{g,\rho}$ such that 
$U_g \rho = e^{i\theta_{g,\rho}} \rho$ holds for all $g \in G$. 
Since $\Lambda$ is $(U,U')$-strong covariant, 
we have
\eq{
U'_g \Lambda(\rho) 
= \Lambda(U_g \rho) 
= e^{i\theta_{g,\rho}} \Lambda(\rho),
}
which shows that $\Lambda(\rho)$ is also strong symmetric with respect to $\{U'_g\}$.

Next, let us prove \eqref{eq_SM:GR_single}.
Due to Lemma \ref{lemm_SM:cocycle_SC}, the cocycles of $U$ and $U'$ coincide, and thus we can take the irreducible decompositions of $U$, $U'$, $\calH$ and $\calH'$ as
\eq{
U_g&=\bigoplus_\nu U^{(\nu)}_g\otimes I_{m_\nu},\\
U'_g&=\bigoplus_{\nu'} U'^{(\nu')}_g\otimes I_{m'_{\nu'}},\\
\calH&=\bigoplus_\nu \calH_{\nu}\otimes \mathbb{C}^{m_\nu},\\
\calH'&=\bigoplus_{\nu'} \calH'_{\nu'}\otimes \mathbb{C}^{m'_{\nu'}},
}
where, irreducible representations labeled by the same index $\nu$ are understood to be identical, i.e., $U^{(\nu)} = U'^{(\nu)}$, and the cocycles of irreducible representations $U^{(\nu)}$ ($U'^{(\nu)}$) in the decomposition of $U$ ($U'$) coincide with each others.
Here, using Schur's orthogonality theorem, we obtain
\eq{
\int\chi_\nu(g)^* U^{(\nu')}_gdg=\frac{\delta_{\nu,\nu'}I_\nu}{d_\nu},
}
where $d_\nu:=\mathrm{dim}\calH_\nu$, $\chi_\nu(g):=\Tr[U^{(\nu)}_g]$ and $I_\nu$ is the identity operator on $\calH_\nu$.
Using this theorem and $U^{(\nu)} = U'^{(\nu)}$, we can rewrite the projection $P_\nu$ to $\calH_{\nu}\otimes\mathbb{C}^{m_\nu}$ and $P'_{\nu'}$ to $\calH_{\nu'}\otimes\mathbb{C}^{m'_{\nu'}}$ as
\eq{
P_\nu&=I_\nu\otimes I_{m_\nu}\nonumber\\
&=d_\nu\int \chi_\nu(g)^* \bigoplus_{\nu'}\left(U^{(\nu')}_g\otimes I_{m_{\nu'}}\right)dg\nonumber\\
&=d_\nu \int\chi_\nu(g)^* U_gdg,\\
P'_{\nu'}&=d_{\nu'} \int\chi_{\nu'}(g)^* U'_gdg.
}
Therefore, for any $\rho$ on $S$, we obtain
\eq{
p'_\nu(\rho)&:=\Tr[P'_\nu\Lambda(\rho)]\nonumber\\
&=\Tr[d_\nu\int_Gdg\chi_{\nu}(g)^*U'_g\Lambda(\rho)]\nonumber\\
&=\Tr[d_\nu\int_Gdg\chi_{\nu}(g)^*\Lambda(U_g\rho)]\nonumber\\
&=\Tr[\Lambda(P_\nu \rho)]\nonumber\\
&=\Tr[P_\nu\rho]=p_\nu(\rho).\label{eq_SM:preserve_probability},
}
which is also given by Ref.~\cite{Zapusek2025} in a different context.
Therefore, when $\rho$ is a single sector state, the support of $\Lambda(\rho)$ is in an irreducible component $\calH'_\nu\otimes\mathbb{C}^{m'_\nu}$.
Since $\Lambda$ is strong covariant,  $\Lambda$ is also weak covariant, and thus $\Lambda(\rho)$ is weak symmetric if $\rho$ is a single sector state.
Therefore, $\Lambda(\rho)$ is a single sector state.
\end{proofof}

\subsection{Stinespring representations and Kraus representations of the strong-covariant operations}
In this subsection, we clarify the properties of Stinespring representations and Kraus representations of strong covariant operations.
\begin{theorem}\label{thm_SM:Kraus_strong}
Let $G$ be a (compact) group, and let $U$ and $U'$ be (projective) unitary representations acting on Hilbert spaces $\calH$ and $\calH'$, respectively.
Let $\Lambda:\calB(\calH)\rightarrow \calB(\calH')$ be a $(U,U')$-strong covariant CPTP map.
Then, any Kraus representation $\{K_m\}$ of $\Lambda$ satisfies
\eq{
    K_mU_g=U'_gK_m,\enskip \forall g\in G.\label{eq_SM:K_comm}
}
Conversely, when a CPTP map $\Lambda:\calB(\calH)\rightarrow \calB(\calH')$ has a Kraus representation satisfying \eqref{eq_SM:K_comm}, it is $(U,U')$-strong covariant.
\end{theorem}

\begin{proofof}{Theorem \ref{thm_SM:Kraus_strong}}
Due to the Stinespring's dilation theorem, there exists a \textit{minimal} Stinespring representation of $\Lambda^\dagger$~\cite{Stinespring1955,davies1976quantum}. Namely, there exists a quantum system $E'$ with Hilbert space $\calH_{E'}$ and an isometry $W:\calH\rightarrow \calH'\otimes\calH_{E'}$ such that
\begin{equation}
    \begin{aligned}
\Lambda^\dagger(Y)&=W^\dagger (Y\otimes 1_{E'})W,\quad\quad \forall Y\in \calB(\calH_{S'})\\
\calH_{S'}\otimes\calH_{E'}&=\overline{\mathrm{span}}\{(Y\otimes1_{E'})W\ket{\psi}|\quad Y\in \calB(\calH_{S'}),\enskip\ket{\psi}\in\calH_{S}\}\,
\end{aligned}
\end{equation}
where $\calB(\calH_{S'})$ is the whole set of bounded operators of $\calH_{S'}$.

Since $\Lambda$ is $(U,U')$-strong covariant, for any operator $Y$, 
\eq{
\Lambda^\dagger(U'^\dagger_gY)=U^\dagger_g\Lambda^\dagger(Y), \enskip \forall g\in G.
}
Therefore,
\eq{
U^\dagger_g W^\dagger (Y\otimes 1_{E'})W=W^\dagger (U'^\dagger_gY\otimes 1_{E'})W,\enskip \forall g\in G.
}
By multiplying the equation from the right by $\ket{\psi}$, 
allowing $Y$ and $\ket{\psi}$ to vary freely, and then taking linear combinations, we obtain the following:
\eq{
U^\dagger_gW^\dagger=W^\dagger(U'^\dagger_g\otimes1_{E'})\label{com_W_U}
}
Let us take an orthonormal basis $\{\ket{m}\}$ of $E'$ and define $\{K_m\}$ as
\eq{
K_m:=\bra{m}W\label{def_K}
}
We remark that since $W$ is a map from $\calH$ to $\calH'\otimes\calH_{E'}$, and thus $\bra{m}W$ is a map from $\calH$ to $\calH'$.

The operators $\{K_m\}$ satisfy
\eq{
\Lambda(...)&=\Tr_{E'}[W...W^\dagger]=\sum_m\bra{m}W...W^\dagger\ket{m}=\sum_mK_m....K^\dagger_m.\\
\sum_mK^\dagger_mK_m&=W^\dagger1_{S'}\otimes\sum_m\ket{m}\bra{m}W=W^\dagger1_{S'}\otimes1_{E'}W=1_S.
}
Therefore, $\{K_m\}$ is a Kraus representation of $\Lambda$.
And due to \eqref{com_W_U} and \eqref{def_K}, we obtain $K_mU_g=U'_gK_m$ for any $g\in G$ as follows:
\eq{
K_mU_g&=\bra{m}WU_g\nonumber\\
&=\bra{m}1_{E'}\otimes U'_gW\nonumber\\
&=U'_g\bra{m}W\nonumber\\
&=U'_gK_m.
}
Therefore, $\{K_m\}$ satisfies \eqref{eq_SM:K_comm}.
Furthermore, any other Kraus representation $\{K'_{l}\}$ of $\Lambda$ can be written as
\eq{
K'_l=\sum_{m}u_{l,m}K_m,
}
$\{K'_l\}$ also satisfies \eqref{eq_SM:K_comm}.

The converse part is obvious.
When a CPTP map $\Lambda$ from $S$ to $S'$ admits a Kraus representation $\{K_m\}$ satisfying \eqref{eq_SM:K_comm}, 
the map $\Lambda$ is strong covariant, since
\eq{
U'_g\Lambda(...)&=U'_g\sum_mK_m...K^\dagger_m=\sum_mK_mU_g...K^\dagger_m=\Lambda(U_g...).
}
\end{proofof}

\begin{theorem}\label{thm_SM:SD_strong_StoS}
Let $G$ be a (compact) group, and let $U$ be a (projective) unitary representation acting on a Hilbert space $\calH$.
Let $\Lambda:\calB(\calH)\rightarrow \calB(\calH)$ be a $(U,U)$-strong covariant CPTP map.
Then, there exists an auxiliary Hilbert space $\mathcal{H}_E$, a unitary operator $V$ on $\calH\otimes\calH_{E}$, and a state $\sigma_E\in\calS(\calH_E)$ such that
\eq{
    (U_g \otimes I^{(E)})\, V
    &= V\, (U_g \otimes I^{(E)})
    \quad \forall g \in G
\label{eq_SM:covariant_unitary_general_strong},\\
    \Lambda(\rho)
    &= \Tr_E\!\left[\,V (\rho \otimes \sigma_E) V^\dagger \,\right].\label{eq_SM:covariant_unitary_general_strong2}
}
Conversely, when a CPTP map $\Lambda:\calB(\calH)\rightarrow \calB(\calH)$ can be realized by $(V,\sigma_E)$ satisfying \eqref{eq_SM:covariant_unitary_general_strong}  and \eqref{eq_SM:covariant_unitary_general_strong2}, the map is $(U,U)$-strong covariant.
\end{theorem}

\begin{proofof}{Theorem \ref{thm_SM:SD_strong_StoS}}
Due to the proof of Theorem \ref{thm_SM:Kraus_strong}, when $\Lambda$ is strong covariant, there is a Kraus representation $\{K_m\}_{m=1,...,d_{\Lambda}}$ satisfying \eqref{eq_SM:K_comm}.
Let us take quantum systems $E$ satisfying $\mathrm{dim}\calH_{E}= d_{\Lambda}+1$, and take an orthonormal basis $\{\ket{m}\}_{m=0,1,...,d_{\Lambda}}$.
Using them, we define a partial isometry $W'$ from $\calH\otimes\calH_E$ to $\calH\otimes\calH_E$ as
\eq{
W':=\sum^{d_\Lambda}_{m=1}K_m\otimes \ket{m}\bra{0}_E
}
This $W'$ is indeed a partial isometry, since $W'^\dagger W' = 1_S \otimes \ket{0}_E\bra{0}_E$ and $(W'W'^\dagger)^2 = W'W'^\dagger$, 
i.e., both $W'^\dagger W'$ and $W'W'^\dagger$ are projections.

Since $W'W'^\dagger$ and $W'^\dagger W'$ are orthogonal to each other, the following $Q$ is also a projection:
\eq{
Q:=I_{SE}-W'^\dagger W'-W'W'^\dagger,
}
where $I_{SE}$ is the identity of $\calH\otimes\calH_E$.
Because of $[U_g\otimes1_E,W'^\dagger W']=0$ and $[U_g\otimes1_E,W'W'^\dagger]=0$, the projection $Q$ satisfies
\eq{
[U_g\otimes1_E,Q]=0.
}
It is also a projection which is orthogonal to $W'^\dagger W'$ and $W'W'^\dagger$, and thus it satisfies
\eq{
\|QW'\|^2_2&=\Tr[W'^\dagger QW']=\Tr[QW'W'^\dagger]=0,\\
\|W'Q\|^2_2&=\Tr[QW'^\dagger W'Q]=\Tr[QW'^\dagger W']=0,
}
which imply $QW'=W'Q=0$.

Now, let us define 
\eq{
V:=W'+W'^\dagger+Q.
}
By definition, 
\eq{
V^\dagger V=(W'+W'^\dagger+Q)^2=VV^\dagger. 
}
$V$ is a unitary, since
\eq{
V^\dagger V&=(W'^\dagger+W'+Q)^2\nonumber\\
&=W'^\dagger W'+W'W'^\dagger+Q=I_{SE}.
}
Here we used $QW'=W'Q=0$ and $W'W'=0$.

Furthermore, due to \eqref{eq_SM:K_comm} and $[U_g\otimes1_E,Q]=0$, we obtain
\eq{
[U_g\otimes1_E,V]=0.
}
By defining $\sigma_E:=\ket{0}\bra{0}_E$,  
\eq{
Tr_E{V...\otimes\sigma_EV^\dagger}=\sum_mK_m...K^\dagger_m=\Lambda(...).
}
Therefore, we have proved the former part of the theorem.
The latter part is straightforward. When a Kraus repretentation $\{K_m\}$ of $\Lambda$ satisfies \eqref{eq_SM:K_comm}, 
\eq{
U_g\Lambda(...)=U_g\sum_mK_m...K^\dagger_m=\sum_mK_mU_g...K^\dagger_m=\Lambda(U_g...).
}
Therefore, $\Lambda$ is strong covariant.
\end{proofof}

\subsection{GKSL equation realizing strong covariant operations}

\begin{theorem}
Let $S$ be a quantum system whose dynamics obey the following GKSL equation:
\eq{
\partial_t \rho=-i[H,\rho]+\sum_k\left(L_k\rho L^\dagger_k-\frac{1}{2}\{L^\dagger_kL_k,\rho\}\right).
}
Let $G$ and $U$ be a group and its (projective) unitary representation on $\calH$, the Hilbert space of $S$.
We also assume that the CPTP map $\Lambda_t:\calB(\calH)\rightarrow\calB(\calH)$ realized by the master equation is $(U,U)$-strong covariant for any $t$, and that $[H,U_g]=0$.
Then, any jump operators $\{L_k\}$ of the GKSL equation satisfy
\eq{
[L_k,U_g]=0, \enskip\forall g\in G.\label{eq_SM:com_GKSL}
}
\end{theorem}

\begin{proof}
Let $\calL$ denote the generator given by the right-hand side of the GKSL equation.
For any $g\in G$, strong covariance gives $\Lambda_t(U_g)=U_g\Lambda_t(I)$.
Differentiating at $t=0$ yields
\eq{
\calL(U_g)=U_g\calL(I).
}
Since trace preservation implies $\Tr[\calL(I)]=0$, we obtain
$\Tr[U_g^\dagger\calL(U_g)]=0$.
Substituting the GKSL form of $\calL$ and using cyclicity of the trace, we find
\eq{
0&=\operatorname{Re}\Tr[U_g^\dagger\calL(U_g)]\notag\\
&=\sum_k\left(
\operatorname{Re}\Tr[U_g^\dagger L_kU_gL_k^\dagger]
-\Tr[L_k^\dagger L_k]\right)\notag\\
&=-\frac12\sum_k\Tr\!\left([L_k,U_g]^\dagger[L_k,U_g]\right).
}
Here, the Hamiltonian term vanishes because
\eq{
\Tr[U_g^\dagger[H,U_g]]
=\Tr[U_g^\dagger H U_g-H]=0.
}
Each term in the last sum is non-negative and vanishes only when $[L_k,U_g]=0$.
Hence, $[L_k,U_g]=0$ for every $k$ and $g\in G$.
This argument applies to any choice of jump operators representing the same GKSL generator.
\end{proof}

\subsection{The proof of Lemma \ref{lemm:no_exist}}

When the system $S$ has a larger Hilbert-space dimension than $S'$, 
a strong covariant operation from $S$ to $S'$ does not necessarily exist. 
This already becomes evident in the case of $U(1)$ symmetry. 
For example, the following lemma holds.
\begin{lemma}\label{lemm_SM:no_exist}
Let $\calH$ and $\calH'$ be Hilbert spaces whose dimensions are $d$ and $d'$, respectively, 
and assume $d > d'$. 
Let $U:U(1)\rightarrow U(\calH)$ and $U':U(1)\rightarrow U(\calH')$ be unitary representations which are defined as $U_t:=e^{-iHt}$ and $U'_t:=e^{-iH't}$, respectively. 
If $H$ has strictly more distinct eigenvalues than $H'$, then there exists no $(U,U')$-strong covariant operation 
$\Lambda: \calB(\calH) \rightarrow \calB(\calH')$.
\end{lemma}
\begin{proof}
By assumption, there exists a real number $E$ that is an eigenvalue of $H$ but not an eigenvalue of $H'$.  
Let $\ket{E}$ be an eigenstate of $H$ whose eigenvalue is $E$.  
Suppose, for the sake of contradiction, that there exists a strong covariant map 
$\Lambda : \calB(\calH) \rightarrow \calB(\calH')$, i.e.,
\[
\Lambda(e^{-iHt}...) = e^{-iH't} \Lambda(...).
\]
Then in particular,
\begin{equation}
e^{-iH't}\Lambda(\ket{E}\!\bra{E}) = e^{-iEt}\Lambda(\ket{E}\!\bra{E}).
\label{eq_SM:evol_e}
\end{equation}
Differentiating \eqref{eq_SM:evol_e} at $t=0$ yields
\eq{
H' \Lambda(\ket{E}\!\bra{E}) = E\,\Lambda(\ket{E}\!\bra{E}).
\label{eq_SM:evol_e2}
}
Let $\ket{\psi}$ be any eigenstate of $\Lambda(\ket{E}\!\bra{E})$ with a nonzero eigenvalue.
Right-multiplying both sides of \eqref{eq_SM:evol_e2} by $\ket{\psi}$ gives
\eq{
H'\ket{\psi} = E\ket{\psi}.
}
This implies that $E$ is an eigenvalue of $H'$, contradicting the assumption.  
Therefore, no such strong covariant operation $\Lambda$ can exist.
\end{proof}

\section{resource measures of symmetry-breaking for strong symmetry}

In this section, we introduce resource measures of strong symmetry breaking.
Unless otherwise stated, in this section, we treat the strong symmetric states as the free states. 
When the single-sector states are taken as the free states, it will be explicitly mentioned.

Since the symmetry depends on the choice of the (projective) unitary representation, 
all quantities introduced in this section are functions of the representation. 
For simplicity, we omit the explicit dependence on the unitary representation 
whenever it is clear from the context, and specify it only when necessary. 
In particular, we write a measure $M$ as $M(\rho)$ when the representation is omitted, 
and write it as $M(\rho\,\|\{U_g\})$ when the representation is explicitly indicated.

\subsection{measure for general-group symmetry}
In this subsection, we introduce two measures which are applicable to a general group $G$.

\begin{definition}[Entanglement asymmetry of strong symmetry breaking]
Let $G$ be a group, and let $U$ be its (projective) unitary representation acting on a system $S$.
We assume that $G$ has Haar measure $\int_G dg=1$ and that $U$ has the irreducible decomposition 
\eq{
U_g = \bigoplus_{\nu} U_g^{(\nu)} \otimes I_{m_\nu}.\label{decom_U_b}
}
Then, we define $\Ssa(\rho)$ as
\eq{
\Ssa(\rho)&:=H\{p_\nu(\rho)\}+S(\ca{G}(\rho))-S(\rho),\nonumber\\
&=H\{p_\nu(\rho)\}+\Sa(\rho),
}
where
\eq{
S(\rho)&:=-\Tr[\rho\log\rho],\\
\ca{G}(\rho)&:=\int_Gdg U_g\rho U^\dagger_g,\\
p_{\nu}(\rho)&:=\Tr[P_\nu\ca{G}(\rho)]=\Tr[P_\nu\rho],\\
H\{p_\nu\}&:=-\sum_{\nu}p_\nu\log p_{\nu}\,,
}
and $P_\nu$ is the projection to $\calH_{\nu}\otimes \mathbb{C}^{m_\nu}$ in the decomposition of $\calH$ under $U$:
\eq{
\calH = \bigoplus_{\nu} \calH_{\nu} \otimes \mathbb{C}^{m_\nu}.
}
\end{definition}

\begin{theorem}\label{thm_SM:strong_EA}
The quantity $\Ssa(\rho)$ satisfies the following features
\begin{description}
\item[(A)] $\Ssa(\rho)$ is a resource measure. Namely, (i-a) it is non-negative, and when $\rho$ is strong symmetric, it is zero, and (i-b) when a CPTP map $\Lambda:\calB(\calH)\rightarrow\calB(\calH')$ is $(U,U')$-strong covariant, \eq{\Ssa(\rho\|U)\ge\Ssa(\Lambda(\rho)\|U')
}
\item[(B)] It is faithful when we employ $\calF_{G,\mathrm{single}}$ as free states. In other words, $\Ssa(\rho)=0$ if and only if $\rho$ is a single-sector state.
\item[(C)] It is \textit{not always} faithful when we employ $\calF_{G,\mathrm{strong}}$ as free states. In other words, there exists a resource state $\rho$ (=state which is not strong symmetric) satisfying $\Ssa(\rho)=0$. 
\item[(D)] With collective representations, $\Ssa$ is \textit{not additive in general} for product states. In other words, there exist a group $G$, two systems $A$ and $B$, unitary representations $U^A$ and $U^B$ of $G$, and states $\rho_A$ and $\sigma_B$ satisfying
\eq
{&\Ssa(\rho_A\otimes\sigma_B\|U^A\otimes U^B)\nonumber\\
&\ne\Ssa(\rho_A\|U^A)+\Ssa(\sigma_B\|U^B).}
\item[(E)] When $S$ is a finite-dimension system and $G$ is a compact Lie group or a finite group, for any $\rho$ on $S$ and any representation $U$ of $G$ acting on $S$, 
\eq{
\lim_{n\rightarrow \infty}\frac{\Ssa(\rho^{\otimes n}\|U^{\otimes n})}{n}=0.
}
\end{description}
\end{theorem}

\begin{proofof}{Theorem \ref{thm_SM:strong_EA}}
(A): We first show (i-a). The term $S(\calT(\rho))-S(\rho)$ is equal to the relative entropy of asymmetry (=entanglement asymmetry) $\Sa(\rho):=D(\rho\|\calT(\rho))$ \cite{Gour2009}.
Hence, $\Ssa(\rho)$ can be written as 
\eq{
\Ssa(\rho)=H\{p_\nu(\rho)\}+\Sa(\rho).\label{decom_Ssa}
}
Due to \eqref{decom_Ssa}, $\Ssa(\rho)$ is clearly nonnegative. 

To show $\rho\in\calF_{G,\mathrm{strong}}$, note that any strong symmetric state $\rho$ is weak symmetric (i.e. $\Sa(\rho)=0$) and is single-sector (i.e. $H\{p_\nu(\rho)\}=0$). Therefore, for any strong symmetric state $\rho$, $\Ssa(\rho)=0$.

For (i-b), let $\Lambda:\calB(\calH)\rightarrow \calB(\calH')$ be $(U,U')$-strong covariant.
Then, due to \eqref{eq_SM:preserve_probability}, $p_\mu(\rho)=p_{\mu}(\Lambda(\rho))$ holds, and thus the term $H\{p_\nu\}$  does not change via $\Lambda$.
Since any strong covariant operation is also weak covariant, $\Lambda$ is weak covariant, and thus $\Sa(\rho)\ge\Sa(\Lambda(\rho))$.
Therefore, we obtain $\Ssa(\rho)\ge\Ssa(\Lambda(\rho))$.

(B): For any single-sector state $\rho$, $H\{p_\mu(\rho)\}=0$. Since any single-sector state is also a weak symmetric state, $\Sa(\rho)=0$ is also valid. Therefore, for any single-sector state $\rho$, $\Ssa(\rho)=0$.
Conversely, since $H\{p_\mu(\rho)\}$ and $\Sa(\rho)$ are non-negative, if $\Ssa(\rho)=0$ is valid, $H\{p_\mu(\rho)\}=0$ and $\Sa(\rho)=0$ also hold, which imply $\rho$ is a single-sector state.

(C): Because of (B), $\Ssa(\rho)=0$ for any single-sector state. Therefore, we only have to show that there is a single-sector state which is not strong symmetric.
Such an example has already been given in the main text: we consider the situation where the system of interest is a qubit, $G$ is $SU(2)$, 
and $U$ is the natural irreducible unitary representation of $SU(2)$ acting on the qubit:  
$U_g := e^{-i\frac{\theta_g}{2}\vec{n}_g\cdot\vec{\sigma}}$, 
where $\theta_g$ is a real parameter describing the rotation angle, $\vec{n}_g$ is a real vector parameter defined as 
$\vec{n}_g := (n_x(g), n_y(g), n_z(g))$ describing the rotation axis, and $\vec{\sigma} := (\sigma_x, \sigma_y, \sigma_z)$ 
denotes the Pauli operators. 
Then, for any strong symmetric state $\eta$, the condition $U_g\eta U_g^\dagger = \eta$ must hold 
for all $g\in G$. Since the  representation $U$ is irreducible, it follows that 
$\eta = I/2$. However, since the representation $U$ includes $\sigma_x$, we have 
$\sigma_x \eta = \sigma_x/2 \not\propto \eta$. 
This leads to a contradiction, and hence there is no strong symmetric state in this case.
And clearly, $I/2$ is a single-sector state in this case. Therefore, $I/2$ is single-sector but not strong symmetric, and thus $\Ssa$ is not faithful for $\calF_{G,\mathrm{strong}}$.

(D): Let $G=U(1)$, and consider a qubit with
\eq{
Q:=\ket{1}\bra{1},\qquad U_\theta:=e^{-i\theta Q}.
}
Take the pure state
\eq{
\ket{\psi}:=\frac{\ket{0}+\ket{1}}{\sqrt2},
\qquad \psi:=\ket{\psi}\bra{\psi}.
}
In general, for a pure state $\omega$ with $U(1)$ symmetry, twirling produces a mixture of mutually orthogonal pure charge-sector states. Consequently,
\eq{
\Sa(\omega)&=H\{p_q(\omega)\},\notag\\
\Ssa(\omega)&=2H\{p_q(\omega)\}.\label{eq:A=H}
}
Hence,
\eq{
\Ssa(\psi\|U)=2H(1/2,1/2)=2\log2.
}
For two copies with the collective representation $U^{\otimes2}$, the total charge is $Q\otimes I+I\otimes Q$.
The charge probabilities and the corresponding normalized sector vectors are
\eq{
(p_0,p_1,p_2)&=\left(\frac14,\frac12,\frac14\right),\notag\\
\ket{\psi_0}&=\ket{00},\qquad
\ket{\psi_1}=\frac{\ket{01}+\ket{10}}{\sqrt2},\notag\\
\ket{\psi_2}&=\ket{11}.
}
Thus, the twirled state is a mixture of three orthogonal pure states with these probabilities.
Using $H(1/4,1/2,1/4)=\frac32\log2$ and \eqref{eq:A=H}, we obtain
\eq{
\Ssa(\psi^{\otimes2}\|U^{\otimes2})
&=2H\left(\frac14,\frac12,\frac14\right)\notag\\
&=3\log2\ne2\Ssa(\psi\|U).
}
This proves (D).

(E): We first establish a bound for pure states.
For any pure state $\omega$ and (projective) unitary representation $W$, write
\eq{
\calG_W(\omega)=\sum_\nu p_\nu(\omega)\omega_\nu,
}
where, for each $\nu$ with $p_\nu(\omega)>0$, $\omega_\nu$ is the normalized state obtained by projecting $\calG_W(\omega)$ onto the sector $\nu$:
\eq{
\omega_\nu:=
\frac{P_\nu\calG_W(\omega)P_\nu}{p_\nu(\omega)}.
}
Here, $P_\nu$ is the projection onto the sector $\nu$ in the irreducible decomposition of $W$.
The states $\{\omega_\nu\}$ have mutually orthogonal supports, so
\eq{
S(\calG_W(\omega))
&=H\{p_\nu(\omega)\}
+\sum_\nu p_\nu(\omega)S(\omega_\nu)\notag\\
&\ge H\{p_\nu(\omega)\}.
}
Since $S(\omega)=0$, it follows that
\eq{
\Ssa(\omega\|W)
&=H\{p_\nu(\omega)\}+S(\calG_W(\omega))\notag\\
&\le2S(\calG_W(\omega))
=2\Sa(\omega\|W).
\label{eq:ssa_pure_bound}
}

Now let $\rho$ be an arbitrary state on the finite-dimensional system $S$.
Choose a finite-dimensional reference system $R$ and a purification $\ket{\Psi}\in\calH_S\otimes\calH_R$, and set
\eq{
\Psi:=\ket{\Psi}\bra{\Psi},\qquad \Tr_R[\Psi]=\rho.
}
Equip $R$ with the trivial unitary representation $I_R$ and write
\eq{
V_g:=U_g\otimes I_R.
}
We group together the $n$ copies of each system and write the total Hilbert space as $\calH_S^{\otimes n}\otimes\calH_R^{\otimes n}$.
With this ordering, the (projective) unitary representation $V^{\otimes n}$ takes the form
\eq{
V_g^{\otimes n}
=U_g^{\otimes n}\otimes I_{R^{\otimes n}},
\qquad \forall g\in G.
}
For every operator $X$ on this space and every $g\in G$,
\eq{
\Tr_{R^{\otimes n}}[V_g^{\otimes n}X]
=U_g^{\otimes n}\Tr_{R^{\otimes n}}[X].
}
Thus, $\Tr_{R^{\otimes n}}$ is $(V^{\otimes n},U^{\otimes n})$-strong covariant.
This is precisely the case discussed in Section~\ref{sec:strong_weak_relation}: the partial trace is strong covariant when the discarded system carries the trivial unitary representation.
Using (A) and then \eqref{eq:ssa_pure_bound}, we obtain
\eq{
0\le\Ssa(\rho^{\otimes n}\|U^{\otimes n})
&\le\Ssa(\Psi^{\otimes n}\|V^{\otimes n})\notag\\
&\le2S\!\left(\calG_{V^{\otimes n}}(\Psi^{\otimes n})\right).
\label{eq:ssa_purification_bound}
}

For every $g\in G$, the vector
\eq{
V_g^{\otimes n}\ket{\Psi}^{\otimes n}
=(V_g\ket{\Psi})^{\otimes n}
}
is invariant under any permutation of the $n$ copies of $S\otimes R$, since all tensor factors are identical.
Thus, it lies in the symmetric subspace, which we denote by $\operatorname{Sym}^n(\calH_S\otimes\calH_R)$.
Consequently, the twirled state $\calG_{V^{\otimes n}}(\Psi^{\otimes n})$ is supported on this subspace.
Its dimension is
\eq{
\dim\operatorname{Sym}^n(\calH_S\otimes\calH_R)
&=\binom{n+D-1}{D-1},\\
D&:=\dim(\calH_S\otimes\calH_R),
}
because a basis is indexed by $D$ nonnegative occupation numbers whose sum is $n$.
The entropy of a state is at most the logarithm of the dimension of its support, so
\eq{
S\!\left(\calG_{V^{\otimes n}}(\Psi^{\otimes n})\right)
\le\log\binom{n+D-1}{D-1}.
}
Combining this with \eqref{eq:ssa_purification_bound} gives
\eq{
0\le\frac{\Ssa(\rho^{\otimes n}\|U^{\otimes n})}{n}
\le\frac{2}{n}\log\binom{n+D-1}{D-1}.
}
Here, $D$ is independent of $n$.
Using the elementary bound
\eq{
\binom{n+D-1}{D-1}\le(n+D-1)^{D-1},
}
we obtain
\eq{
0\le\frac{2}{n}\log\binom{n+D-1}{D-1}
&\le\frac{2(D-1)\log(n+D-1)}{n}\notag\\
&\longrightarrow0,
}
in the limit of $n\rightarrow\infty.$
This proves (E) for arbitrary mixed states as well.
\end{proofof}

The second one is given by the logarithmic characteristic function, which is the complete measure of the weak symmetry breaking for the finite-group case~\cite{Shitara2023}.

\begin{definition}[Averaged logarithmic characteristic function]
Let $G$ be a group, and let $\{U_g\}$ be its (projective) unitary representation acting on a Hilbert space $\calH$.
We assume that $G$ has Haar measure $\int_G dg=1$ and that $U_g$ has the irreducible decomposition \eqref{decom_U_b}.
Then, we define $L(\rho)$ as
\eq{
L(\rho):=\int_GdgL_{g,\mathrm{strong}}(\rho),\label{def_SLCF_b}
}
where 
\eq{
L_{g,\mathrm{strong}}(\rho):=-\log|\Tr[U_g\rho]|,
}
and when $G$ is finite, $\int_Gdg$ becomes $\frac{1}{|G|}\sum_g$.
\end{definition}

\begin{theorem}\label{thm_SM:SLCF}
The quantity $L(\rho)$ satisfies the following features
\begin{description}
\item[(A)] $L(\rho)$ is a resource measure. Namely, (i-a) it is non-negative, and when $\rho$ is strong symmetric, it is zero, and (i-b) when $\Lambda$ is $(U,U')$-strong covariant, the inequality $L(\rho)\ge L(\Lambda(\rho))$ holds. More precisely, $L(\rho)=L(\Lambda(\rho))$.
\item[(B)] It is faithful for $\calF_{G,\mathrm{strong}}$. In other words, $L(\rho)=0$ if and only if $\rho$ is strong symmetric.
\item[(C)] When we employ collective representation, it is additive for product states. In other words, for two systems $A$ and $B$ and (projective) unitary representations $U^{A}$ and $U^{B}$ on them, any states $\rho_A$ and $\sigma_B$ satisfy
\eq{
L(\rho_A\otimes\sigma_B\|U^A\otimes U^B)=L(\rho_A\|U^A)+L(\sigma_B\|U^B).
}

Consequently, 
\eq{
\lim_{n\rightarrow \infty} \fr{L(\rho^{\otimes n}\|U^{\otimes n})}{n}=L(\rho)
}
holds.
\end{description}
\end{theorem}

\textbf{Remark:} Each $L_{g,\mathrm{strong}}(\rho)$ is also a resource measure and additive, but it is not faithful.

\begin{proofof}{Theorem \ref{thm_SM:SLCF}}
(A): To show the non-negativity, we only have to show
\eq{
|\Tr[U_g\rho]|\le1\label{ia-1}.
}
Let $\rho=\sum_jp_j\ket{\psi_j}\bra{\psi_j}$ be a spectral decomposition of $\rho$. We obtain \eqref{ia-1} as follows:
\eq{
|\Tr[U_g\rho]|&=|\sum_jp_j\Tr[U_g\ket{\psi_j}\bra{\psi_j}]|\nonumber\\
&\le\sum_jp_j|\bra{\psi_j}U_g\ket{\psi_j}|\nonumber\\
&\le1.\label{ia-2}
}

Also, since any strong symmetric state $\rho$ satisfies $|\Tr[U_g\rho]|=1$, any strong symmetric state $\rho$ satisfies $L(\rho)=0$.

For (i-b), let a CPTP map $\Lambda:\calB(\calH)\rightarrow \calB(\calH')$ be $(U,U')$-strong covariant.
Then, any Kraus representation $\{K_m\}$ of $\Lambda$ satisfies $K_mU_g=U'_gK_m$.
Therefore,
\eq{
\Tr[U'_g\Lambda(\rho)]&=\sum_m\Tr[U'_gK_m\rho K^\dagger_m]\nonumber\\
&=\sum_m\Tr[K^\dagger_mK_mU_g\rho]\nonumber\\
&=\Tr[U_g\rho].
}
Therefore, $L(\rho)=L(\Lambda(\rho))$ and thus (i-b) holds.

(B): We only have to show $L(\rho)=0$ implies that $\rho$ is strong symmetric.
Let us assume that $L(\rho)=0$ holds. Then,
\eq{
|\Tr[U_g\rho]|=1,\enskip\forall g\in G.\label{ii-1}
}
Due to \eqref{ia-2}, the eqaution \eqref{ii-1} requires
\eq{
U_g\ket{\psi_j}=c_{g,j}\ket{\psi_j}
}
where $C_{g,j}$ is a complex number satisfying $|c_{g,j}|=1$.
Then, $Tr[U_g\rho]=\sum_jp_jc_{j,g}$. Due to \eqref{ii-1}, $c_{j,g}=c_{j',g}=:c_g$ for any $j$ and $j'$. Therefore, we obtain $U_g\rho=c_g\rho$, and thus $\rho$ is strong symmetric.

(C): the additivity immediately follows from
\eq{
|\Tr[\rho\otimes\sigma U^A_g\otimes U^B_g]|
&=|\Tr[\rho U^A_g]| |\Tr[\sigma U^B_g]|.
}
\end{proofof}

\subsection{measure for compact-Lie-group symmetry}
Next, we introduce a resource measure that works particularly well when the group $G$ is a compact Lie group:
\begin{definition}[Non-symmetrized and symmetrized covariance matrices]
Let $G$ be a compact Lie group, and let $d_G:=\dim G$ denote the dimension of $G$ as a smooth manifold. Then, elements in the neighborhood of the identity $e\in G$ can be parametrized as $g(\bm{\lambda})=e^{\ii \sum_{j=1}^{d_G}\lambda^j A_j}$ with a basis $\{A_j\}_{j=1}^{d_G}$ of the Lie algebra $\mathfrak{g}$. 
Let $U$ be a (projective) unitary representation of $G$ acting on a Hilbert space $\calH$.
We assume that $U$ is  differentiable, and
introduce Hermitian operators 
\eq{
    X_j:=  -\ii \frac{\partial}{\partial \lambda^j} U(g(\bm{\lambda}))\biggl|_{\bm{\lambda}=\bm{0}}\quad \label{eq_SM:hermitian_operators}
}
for $j=1,\cdots,d_G$. Equivalently, $X_j=L(A_j)$, where $L(A):= -\ii \left.\frac{d}{dt }U(e^{\ii t A})\right|_{t=0}$ is the real-linear map obtained by differentiating $U$ at the identity.
For a projective representation, $L$ need not be a Lie algebra representation, since the corresponding commutation relations may contain additional scalar terms.
Then, we define two types of covariance matrices associated with $\rho$.
\begin{enumerate}
    \item \textbf{Non-symmetrized covariance matrix.}
    The non-symmetrized covariance matrix $\VNS(\rho)$ is defined by
    \begin{equation}
        (\VNS)_{i,j}
        := \Tr[\rho\, X_i X_j]
           - \langle X_i \rangle_\rho \langle X_j \rangle_\rho ,
    \end{equation}
    where $\langle X \rangle_\rho := \Tr[\rho X]$.

    \item \textbf{Symmetrized covariance matrix.}
    The symmetrized covariance matrix $\VS(\rho)$ is defined by
    \begin{equation}
        (\VS(\rho))_{i,j}
        := \frac{1}{2}\Tr[\rho\,\{X_i, X_j\}]
           - \langle X_i \rangle_\rho \langle X_j \rangle_\rho ,
    \end{equation}
    where $\{X_i, X_j\} := X_i X_j + X_j X_i$ denotes the anti-commutator.
\end{enumerate}
\end{definition}

\begin{theorem}\label{thm_SM:VM}
The covariance matrices $\VNS(\rho)$ and $\VS(\rho)$ satisfy the following features
\begin{description}
\item[(A)] $\VNS(\rho)$ and $\VS(\rho)$ are resource measures. Namely, (i-a) they are positive-semidefinite matrices, and when $\rho$ is strong symmetric, it is zero, and (i-b) when $\Lambda$ is $(U,U')$-strong covariant, the inequalities $\VNS(\rho)\ge \VNS(\Lambda(\rho))$ and $\VS(\rho)\ge \VS(\Lambda(\rho))$ hold. More precisely, $\VNS(\rho)=\VNS(\Lambda(\rho))$ and $\VS(\rho)= \VS(\Lambda(\rho))$ are valid.
\item[(B)] When $G$ is connected, they are faithful for $\calF_{G,\mathrm{strong}}$. In other words, 
\eq{
\VNS(\rho)=0\Leftrightarrow\VS(\rho)=0\Leftrightarrow\rho\in\calF_{G,\mathrm{strong}}.
}
Using (i-a) of $\VS$ and $\VNS$, we can equivalently state
\eq{
\sum_i\VNS(\rho)_{ii}=0\Leftrightarrow\sum_i\VS(\rho)_{ii}=0\Leftrightarrow\rho\in\calF_{G,\mathrm{strong}}.
}
\item[(C)] When we employ collective representation, they are additive for product states. In other words, for two systems $A$ and $B$ and unitary representations $U^{A}$ and $U^{B}$ on them, any states $\rho_A$ and $\sigma_B$ satisfy 
\eq{
&\VNS(\rho_A\otimes\sigma_B\|U^A\otimes U^B)\nonumber\\
&=\VNS(\rho_A\|U^A)+\VNS(\sigma_B\|U^B),\\
&\VS(\rho_A\otimes\sigma_B\|U^A\otimes U^B)\nonumber\\
&=\VS(\rho_A\|U^A)+\VS(\sigma_B\|U^B)
}
Consequently, $\lim_{n\rightarrow \infty}\VNS(\rho^{\otimes n}\|U^{\otimes n})/n=\VNS(\rho)$ and $\lim_{n\rightarrow \infty}\VS(\rho^{\otimes n}\|U^{\otimes n})/n=\VS(\rho)$ hold.
\end{description}
\end{theorem}

\begin{proofof}{Theorem \ref{thm_SM:VM}}
(A): By definition, $\VNS$ and $\VS$ are always positive-semidefinite, and hence satisfy the positivity requirement in (i-a).

Next, let us assume $\rho$ to be strong symmetric, i.e. $U_g\rho=e^{i\theta_g}\rho$. Then, 
\eq{
X_i\rho&=-i\left.\frac{\partial}{\partial \lambda^{i}}U(g(\bm{\lambda}))\right|_{\bm{\lambda}=0}\rho\nonumber\\
&=\left.\frac{\partial \theta_g}{\partial \lambda^{i}}\right|_{\bm{\lambda}=0}\rho.
}
Therefore, we obtain
\eq{
(\VNS(\rho))_{i,j}=(\VS(\rho))_{i,j}=\left.\frac{\partial \theta_g}{\partial \lambda^{i}}\right|_{\bm{\lambda}=0}\left.\frac{\partial \theta_g}{\partial \lambda^{j}}\right|_{\bm{\lambda}=0}-\left.\frac{\partial \theta_g}{\partial \lambda^{i}}\right|_{\bm{\lambda}=0}\left.\frac{\partial \theta_g}{\partial \lambda^{j}}\right|_{\bm{\lambda}=0}=0.
}

Next let us show the property (i-b), i.e., monotonicity.
Let $\Lambda$ be a $(U,U')$-strong covariant operation, and let $\{K_m\}$ be the Kraus representation of it.
Then, due to $U'_gK_m=K_mU_g$, we obtain
\eq{
X'_iK_m=K_mX_i,\enskip\forall m,\enskip\forall i,
}
where $\{X'_i\}$ and $\{X_i\}$ are generators of $U'$ and $U$, respectively.
Furthermore, since $\Lambda$ is CPTP, $\sum_mK^\dagger_mK_m=I$ is valid.
Therefore, we obtain 
\eq{
\Tr[\Lambda(\rho)X'_i]=\Tr[\sum_{m}K^\dagger_mK_m X_i\rho ]=\Tr[X_i\rho ]=\Tr[\rho X_i],\\
\Tr[\Lambda(\rho)X'_i X'_j]=\Tr[\rho X_i \sum_{m}K^\dagger_mK_m X_j]=\Tr[\rho X_i X_j],
}
which hold when $\Lambda$ is strong covariant.

(B): Due to (A), we only have to show that when either $\VNS(\rho)=0$ or $\VS(\rho)=0$ holds, $\rho$ is strong symmetric.
Let us define $\Delta X_i:=X_i-\Tr[X_i\rho]I$.
Note that either $\VNS(\rho)=0$ or $\VS(\rho)=0$ holds,
\eq{
\Tr[\rho(\Delta X_i)^2]=0.
}
Hence, we obtain $\Delta X_i \sqrt{\rho}=0$, and thus $\Delta X_i \rho=0$. Therefore, we have
\eq{
X_i\rho=\Tr[\rho X_i]\rho.
\label{eq_SM:zero_variance_scalar_action}
}
For $\bm\lambda\in\mathbb R^{d_G}$, write $X(\bm\lambda):=\sum_i\lambda^iX_i$ and use the parametrization $g(\bm\lambda)=\exp(i\sum_i\lambda^iA_i)$ introduced above.
For fixed $\bm\lambda$, the conjugation maps $B\mapsto U_{g(t\bm\lambda)}B U_{g(t\bm\lambda)}^\dagger$ form a one-parameter group, since the phase factors in \eqref{eq:projective_rep} cancel.
Differentiating at $t=0$ gives the generator $B\mapsto i[X(\bm\lambda),B]$. Consequently,
\eq{
U_{g(t\bm\lambda)}B U_{g(t\bm\lambda)}^\dagger
=e^{itX(\bm\lambda)}B e^{-itX(\bm\lambda)},\qquad B\in\calB(\calH).
\label{eq_SM:projective_conjugation_flow}
}
Thus, the unitary $e^{-itX(\bm\lambda)}U_{g(t\bm\lambda)}$ commutes with every operator on $\calH$ and must be a scalar phase times the identity. We can therefore write
\eq{
U_{g(t\bm\lambda)}=z_{\bm\lambda}(t)e^{itX(\bm\lambda)},\qquad |z_{\bm\lambda}(t)|=1.
\label{eq_SM:projective_one_parameter_phase}
}
Using \eqref{eq_SM:zero_variance_scalar_action}, we obtain
\eq{
U_{g(t\bm\lambda)}\rho
=z_{\bm\lambda}(t)e^{it\sum_i\lambda^i\Tr[\rho X_i]}\rho.
\label{eq_SM:zero_variance_projective_phase}
}
Since $G$ is connected, every $g\in G$ is a finite product of elements of the form $g(\bm\lambda)$.
By \eqref{eq:projective_rep}, $U_g$ is the corresponding product of unitaries multiplied by a scalar phase. Each factor acts on $\rho$ as a scalar phase by \eqref{eq_SM:zero_variance_projective_phase}, so their product does as well. Hence, for every $g\in G$, there is a real number $\theta_g$ such that
\eq{
U_g\rho=e^{i\theta_g}\rho,
\label{eq_SM:zero_variance_strong_symmetry}
}
which implies that $\rho$ is strong symmetric.

(C): By definition, the variances $\VNS(\rho)$ and $\VS(\rho)$ are additive for product states.
\end{proofof}


\section{Proof of the necessary and sufficient conditions for i.i.d.\ conversion under strong covariant operations for compact Lie groups}
\label{app:compact_Lie_strong_iid}

In this appendix, we prove Theorem~\ref{thm:compact_Lie_strong_iid}.
We first specify the conversion framework and notation, then state the lemmas used in the proof, and finally establish necessity and sufficiency in turn.

\subsection{Framework for i.i.d.\ state conversion under strong covariant operations for compact Lie groups}
\label{app:strong_iid_framework}

Consider a compact Lie group $G$ and continuous non-projective unitary representations $U,U'$ on finite-dimensional Hilbert spaces $\calH,\calH'$.
Fix states $\rho,\sigma$ and a rate $r>0$, and let $M_N:=\lfloor rN\rfloor$.
For each $N$, choose a finite-dimensional auxiliary system $A_N$, a continuous non-projective unitary representation $W_N$ on it, a state $\eta_{A_N}$ that is strong symmetric with respect to $W_N$, and a $(U^{\otimes N},U'^{\otimes M_N}\otimes W_N)$-strong covariant CPTP map $\Phi_N$.

Achievability of the rate $r$ is defined by
\eq{
\epsilon_N&:=\bigl\|\Phi_N(\rho^{\otimes N})
-\sigma^{\otimes M_N}\otimes\eta_{A_N}\bigr\|_1,\\
\liminf_{N\to\infty}\epsilon_N&=0
\label{eq:app_strong_iid_error}
}
This is the same formulation as \eqref{eq:strong_iid_conversion}.
Below, we represent a subsequence of input copy numbers by an infinite subset $\mathcal N\subseteq\mathbb N$, using its elements $N$ as the copy-number indices.
The condition in \eqref{eq:app_strong_iid_error} is equivalent to the existence of some $\mathcal N$ along which
\eq{
\lim_{\substack{N\to\infty\\N\in\mathcal N}}\epsilon_N=0
\label{eq:app_strong_iid_subsequence_error}
}
holds.
The sequences of operations and auxiliary states themselves are nevertheless chosen for all $N\in\mathbb N$, as in the main text.

Writing the character of the auxiliary state as $\gamma_N\in\CharG$, we have $W_N(g)\eta_{A_N}=\gamma_N(g)\eta_{A_N}$.

\subsection{Notation}
\label{app:strong_iid_preliminary_0}

We introduce the notation used in this appendix.
We first describe symmetry transformations near the identity using the generators of the unitary representations.
As in the main text, let $d_G:=\dim G$ be the dimension of $G$ as a smooth manifold, and let $\mathfrak g$ be its Lie algebra.
Choose a common basis $\{A_j\}_{j=1}^{d_G}$ for the input and output, and denote the generators defined in \eqref{eq:hermitian_operators} by $X_j$ and those obtained by applying the same definition to $U'$ by $X'_j$.
We continue to use the local coordinates $\bm{\lambda}$ of the main text and write
\eq{
g(\bm{\lambda})&:=\exp\!\left(i\sum_{j=1}^{d_G}\lambda^jA_j\right),\\
X(\bm{\lambda})&:=\sum_j\lambda^jX_j,\qquad
X'(\bm{\lambda}):=\sum_j\lambda^jX'_j,\\
U_{g(t\bm{\lambda})}&=e^{itX(\bm{\lambda})},\qquad
U'_{g(t\bm{\lambda})}=e^{itX'(\bm{\lambda})}
\label{eq:app_strong_iid_generators}
}

In terms of the Lie algebra representation $L$ in the main text, $X(\bm{\lambda})$ is $L(\sum_j\lambda^jA_j)$; it is shorthand for a linear combination of the generators already introduced.
Here, $\mathfrak g$ is a $d_G$-dimensional real vector space describing infinitesimal transformation directions from the identity $e$.
The transformation in the direction $A=\sum_j\lambda^jA_j$ is the one-parameter subgroup $t\mapsto\exp(itA)=g(t\bm{\lambda})$, and $X(\bm{\lambda})$ generates its non-projective representation. Thus, \eqref{eq:app_strong_iid_generators} holds for every $t\in\mathbb R$.
Although the exponential map itself is defined for every $\bm{\lambda}$, when using $\bm{\lambda}$ as unique local coordinates, we restrict them to a sufficiently small neighborhood of the origin, as in the main text.
Moreover, since $g(t\bm{\lambda})$ is a continuous curve starting at $e$, if $G$ is disconnected, $g(\bm{\lambda})$ cannot represent elements in connected components other than the identity component.

We choose the common basis above so that conjugation, the map sending a group element $h$ to $ghg^{-1}$ for fixed $g\in G$, is represented by real orthogonal matrices.
Let $\operatorname{Ad}_g$ denote the real linear transformation of directions induced by applying conjugation to one-parameter subgroups. That is,
\eq{
g\exp(itA)g^{-1}=\exp\!\left(it\operatorname{Ad}_g A\right)
}
so that the direction $A$ is mapped to the direction $\operatorname{Ad}_g A$.
Using the compactness of $G$, we choose an inner product for which these transformations preserve lengths and angles between directions.
Specifically, starting from any positive-definite real inner product $\langle\cdot,\cdot\rangle_0$, define
\eq{
\langle A,B\rangle_{\mathfrak g}
:=\int_G\langle\operatorname{Ad}_h A,\operatorname{Ad}_h B\rangle_0\,dh
}
Here, $dh$ is the normalized Haar measure used for twirling in the main text.
Since $\operatorname{Ad}_h$ is invertible, this inner product is also positive definite.
Moreover, $\operatorname{Ad}_h\operatorname{Ad}_g=\operatorname{Ad}_{hg}$ and the right invariance of the Haar measure give
\eq{
\langle\operatorname{Ad}_g A,\operatorname{Ad}_g B\rangle_{\mathfrak g}
=\langle A,B\rangle_{\mathfrak g}
}
In this sense, the inner product is invariant under conjugation.
Because the average is taken over all of $G$, the same construction applies when $G$ is disconnected.
Since this inner product is positive definite, applying Gram--Schmidt orthonormalization to any basis yields a basis satisfying $\langle A_j,A_k\rangle_{\mathfrak g}=\delta_{jk}$. We use such a basis below.
Under a change of basis, the covariance matrices and metric adjusted quantum geometric tensors in the main text undergo congruence transformations by the same real invertible matrix, so the equalities and matrix inequalities in the theorem are unchanged.

With respect to this basis, define the $d_G\times d_G$ real matrix $O_g$ by
\eq{
\operatorname{Ad}_g A_j=\sum_k(O_g)_{kj}A_k
}
For a general direction, this gives
\eq{
\operatorname{Ad}_g\!\left(\sum_j\lambda^jA_j\right)
=\sum_k(O_g\bm{\lambda})^kA_k
}
Thus, $O_g$ acts on the vector $\bm{\lambda}$ of real coefficients specifying the transformation direction.
The invariance of the inner product and the orthonormality of the basis give
\eq{
(O_g^{\mathsf T}O_g)_{ij}
&=\sum_k(O_g)_{ki}(O_g)_{kj}\nonumber\\
&=\langle\operatorname{Ad}_g A_i,\operatorname{Ad}_g A_j\rangle_{\mathfrak g}
=\delta_{ij}
}
and hence $O_g$ is a real orthogonal matrix.
Furthermore, applying the representation $U$ to the conjugation action and differentiating at $t=0$, we obtain
\eq{
U_gX_jU_g^\dagger=\sum_k(O_g)_{kj}X_k
\label{eq:app_strong_iid_adjoint}
}
The output generators transform by the same $O_g$.

We write the mean charge vector collecting the expectation values $\langle X_j\rangle_\rho=\Tr[\rho X_j]$ from the main text, and the characteristic function corresponding to the trace in \eqref{eq:logCh2}, as
\eq{
(\bm\mu_\rho)_j&:=\Tr[\rho X_j],\qquad
\chi_\rho(g):=\Tr[\rho U_g],\\
(\bm\mu_\sigma)_j&:=\Tr[\sigma X'_j],\qquad
\chi_\sigma(g):=\Tr[\sigma U'_g]
\label{eq:app_strong_iid_mean_character}
}
respectively.
The logarithmic characteristic function in the main text is $-\log|\chi_\rho(g)|$; below, we use the complex-valued function $\chi_\rho(g)$, which retains its phase as well.

Elements of the strong symmetry subgroup defined in \eqref{eq:strong_iid_K} in the main text impart a common phase to the entire support of the state.
For each $g\in\mathrm{Sym}_{G,\mathrm{strong}}(\rho)$, denote this phase by $c_\rho(g)$, so that
\eq{
U_g\rho=c_\rho(g)\rho
}

Taking the trace of both sides and using $\Tr\rho=1$, we obtain
\eq{
c_\rho(g)=\chi_\rho(g)
\qquad\bigl(g\in\mathrm{Sym}_{G,\mathrm{strong}}(\rho)\bigr)
}
Thus, $c_\rho$ is the restriction of the characteristic function to the strong symmetry subgroup, where it is a uniquely determined continuous $U(1)$-valued function.
Furthermore, for any $g,h\in\mathrm{Sym}_{G,\mathrm{strong}}(\rho)$, the non-projective representation property $U_{gh}=U_gU_h$ gives
\eq{
c_\rho(gh)=c_\rho(g)c_\rho(h)
}
so $c_\rho$ is a character on this subgroup.
For the output, we similarly define $c_\sigma$ using $U'$.
The domains of $c_\rho,c_\sigma$ are their respective strong symmetry subgroups and need not be all of $G$, unlike the auxiliary-state character $\gamma_N$.

For a character $\gamma\in\CharG$ on all of $G$, define its charge vector $\bm a(\gamma)\in\mathbb R^{d_G}$ by
\eq{
\gamma(g(t\bm{\lambda}))=e^{it\bm a(\gamma)\cdot\bm{\lambda}}
\label{eq:app_strong_iid_character_charge}
}

This definition is possible because restricting a character to a one-parameter subgroup gives a continuous homomorphism from the additive group $\mathbb R$ to $U(1)$.
Each component $a_j(\gamma)$ corresponds to the constant $c_j$ describing the character-induced shift of a generator in the proof of Theorem~\ref{thm:qstrong_measures}.
When the group acts on a strong symmetric auxiliary state through $\gamma$, the vector $\bm a(\gamma)$ represents the definite charge of that state.

We next extract phase information that remains unchanged when a strong symmetric auxiliary state is appended.
Let $D_G:=[G,G]$ be the subgroup generated by all commutators $ghg^{-1}h^{-1}$ in the group.
For a compact Lie group, $D_G$ is a closed subgroup~\cite[Theorem~6.11]{  hofmann2006structure}, so $\mathrm{Sym}_{G,\mathrm{strong}}(\rho)\cap D_G$ is also a compact Lie group.
For any character $\gamma\in\CharG$ on $G$,
\eq{
\gamma(ghg^{-1}h^{-1})
&=\gamma(g)\gamma(h)\gamma(g)^{-1}\gamma(h)^{-1}\nonumber\\
&=1
}
and hence $\gamma$ is identically $1$ on $D_G$.
Thus, appending a strong symmetric auxiliary state does not change the phases associated with elements of $D_G$.
By definition, the strong symmetry subgroup of the representation restricted to $D_G$ is
\eq{
\mathrm{Sym}_{D_G,\mathrm{strong}}(\rho)
=\mathrm{Sym}_{G,\mathrm{strong}}(\rho)\cap D_G
}
and the same holds for the output.
To record the phase that each element of this subgroup imparts to the entire support of the state, we define the residual characters by
\eq{
\zeta_\rho:=c_\rho\big|_{\mathrm{Sym}_{D_G,\mathrm{strong}}(\rho)},\qquad
\zeta_\sigma:=c_\sigma\big|_{\mathrm{Sym}_{D_G,\mathrm{strong}}(\sigma)}
\label{eq:strong_iid_residual_character}
}

For the mean charge, the components that cannot be adjusted by a strong symmetric auxiliary state are likewise described by the residual mean charge from the main text.
In the notation above, \eqref{eq:strong_iid_residual_mean} reads
\eq{
m_G(\rho)&=\bm\mu_\rho-\bm\mu_{\calG_U(\rho)},\\
m_G(\sigma)&=\bm\mu_\sigma-\bm\mu_{\calG_{U'}(\sigma)}
}
Here, $\calG_U,\calG_{U'}$ are the twirling operation in \eqref{eq:twir} applied to the respective representations.
A strong symmetric auxiliary state is also weak symmetric, so it is unchanged by twirling and has zero residual mean charge.
Consequently, appending an auxiliary state leaves the residual mean charge unchanged.
The residual character records the definite phases associated with elements of the subgroup $\mathrm{Sym}_{D_G,\mathrm{strong}}(\rho)$, whereas the residual mean charge records the components of the generator expectation values that vanish under twirling.
They describe charge information that cannot be adjusted by a strong symmetric auxiliary state, through phases associated with group elements and through generator expectation values, respectively.

We now apply this notation to systems consisting of multiple copies.
Write the generators of the collective representations $U^{\otimes N}$ and $U'^{\otimes M}$ from the main text as
\eq{
X_j^{(N)}&:=\sum_{a=1}^N I^{\otimes(a-1)}\otimes X_j\otimes I^{\otimes(N-a)},\\
X_j'^{(M)}&:=\sum_{a=1}^M I^{\otimes(a-1)}\otimes X'_j\otimes I^{\otimes(M-a)}
}

The generators of the output representation $\gamma U'^{\otimes M}$ multiplied by a character are $X_j'^{(M)}+a_j(\gamma)I$.
Thus, the characteristic functions of the input and target output are $\chi_\rho(g)^N$ and $\gamma_N(g)\chi_\sigma(g)^{M_N}$, respectively, and the difference in their mean charges is
\eq{
N\bm\mu_\rho-M_N\bm\mu_\sigma-\bm a(\gamma_N)
}

For a finite group, $d_G=0$: sums over generators are zero, mean charge vectors are zero-dimensional vectors, and the metric adjusted quantum geometric tensors and covariance matrices are $0\times0$ matrices.
Consequently, the conditions on the metric adjusted quantum geometric tensors \eqref{eq:strong_iid_QGT}, covariance matrices \eqref{eq:strong_iid_covariance}, and residual mean charges \eqref{eq:strong_iid_mean} are all automatically satisfied.
The subgroup conditions \eqref{eq:strong_iid_weak_stabilizer} and \eqref{eq:strong_iid_strong_stabilizer}, however, are also required for finite groups.

We next introduce purifications to describe correlations between a mixed state and a reference system.
For a state $\tau$ on a space $\calK$ carrying a general finite-dimensional unitary representation $\widehat{U}$, fix a basis $\{\ket a\}$ and define the canonical purification by
\eq{
\ket{\Psi_\tau}:=(\sqrt\tau\otimes I)\sum_a\ket a\otimes\ket a
\label{eq:app_strong_iid_purification}
}

We call the second system the reference system and denote its marginal state and support by
\eq{
\omega_\tau:=\tau^{\mathsf T},\qquad
E_\tau:=\supp\omega_\tau
}
Below, $\mathsf T$ denotes the transpose, and an overline on an operator denotes entrywise complex conjugation, both taken in the fixed basis $\{\ket{a}\}$.
Thus, for an operator $A$,
\eq{
(A^{\mathsf T})_{ab}=A_{ba},\qquad
(\overline A)_{ab}=\overline{A_{ab}},\qquad
A^\dagger=(\overline A)^{\mathsf T}
\label{eq:app_strong_iid_conjugation_convention}
}
An overline on a matrix entry on the right-hand side denotes complex conjugation of a complex number. We use the same notation for density operators.
We have $\langle\Psi_\tau|\Psi_\tau\rangle=\Tr\tau=1$, and the marginal states of the physical and reference systems are $\tau$ and $\omega_\tau$, respectively.
For tensor powers, we use the product basis built from this single-copy basis.
This reference system is a mathematical device for purification and is distinct from the auxiliary system $A_N$ that holds a strong symmetric state after the conversion.

To describe the relation between this purification and the characteristic functions, we introduce the operator-valued functions
\eq{
B_\rho(g)&:=\sqrt{\omega_\rho}\,\overline{U_g}\sqrt{\omega_\rho},\\
B_\sigma(g)&:=\sqrt{\omega_\sigma}\,\overline{U'_g}\sqrt{\omega_\sigma}
\label{eq:app_strong_iid_B}
}

These are operators on $E_\rho,E_\sigma$, respectively.
Since the support need not be invariant under the whole group, the group action on the right-hand side is evaluated in the original space before restriction to the support.
At the identity, $B_\rho(e)=\omega_\rho$, and taking the trace gives
\eq{
\Tr B_\rho(g)
=\Tr[\rho^{\mathsf T}\overline{U_g}]
=\overline{\chi_\rho(g)}
}
The same holds for the output.
Thus, $B_\rho(g)$ retains, as an operator, the reference-system information that is traced out to obtain the characteristic function; for a general $g$, it need not be a state.

These operators can also be obtained by multiplying the density operator of the purification from the left by $\widehat{U}_g^\dagger\otimes I$ and taking the partial trace over the physical system.
For a general state $\tau$ and representation $\widehat{U}$, evaluate $\Tr_{\calK}[(\widehat{U}_g^\dagger\otimes I)\ket{\Psi_\tau}\bra{\Psi_\tau}]$ in the reference-system basis used in \eqref{eq:app_strong_iid_purification}. Its $(b,d)$ entry is
\eq{
&\sum_{a,c}(\widehat{U}_g^\dagger)_{ca}(\sqrt\tau)_{ab}
\overline{(\sqrt\tau)_{cd}}\nonumber\\
&\quad=\sum_{a,c}(\sqrt\tau)_{dc}(\widehat{U}_g^\dagger)_{ca}(\sqrt\tau)_{ab}\nonumber\\
&\quad=(\sqrt\tau \widehat{U}_g^\dagger\sqrt\tau)_{db}
}
Therefore,
\eq{
\Tr_{\calK}\!\left[(\widehat{U}_g^\dagger\otimes I)
\ket{\Psi_\tau}\bra{\Psi_\tau}\right]
&=(\sqrt\tau \widehat{U}_g^\dagger\sqrt\tau)^{\mathsf T}\nonumber\\
&=\sqrt{\tau^{\mathsf T}}\,\overline{\widehat{U}_g}\sqrt{\tau^{\mathsf T}}.
\label{eq:app_strong_iid_purification_character}
}
The right-hand side is $B_\rho(g)$ when $\tau=\rho,\widehat{U}=U$, and $B_\sigma(g)$ when $\tau=\sigma,\widehat{U}=U'$.

Next, we describe the same purification in terms of irreducible sectors.
Apply the irreducible decomposition in \eqref{eq:decom_U} of the main text to the representation $\widehat{U}$ and write
\eq{
\calK=\bigoplus_\lambda\calV_\lambda\otimes\calM_\lambda,
\qquad
\widehat{U}_g=\bigoplus_\lambda U_g^\lambda\otimes I_{\calM_\lambda}
}
where $P_\lambda$ denotes the projection onto each component.
Here, $\calV_\lambda$ is the irreducible representation space corresponding to $\calH_\nu$ in the main text, $\calM_\lambda\cong\mathbb C^{m_\lambda}$ is the multiplicity space, and $\lambda$, like $\nu$ in the main text, labels irreducible representations.
For the two representations being compared, we choose the same $\calV_\lambda$ and basis for each common irreducible label. The multiplicity spaces and projections $P_\lambda$ depend on the respective representations.

Using this decomposition, project the purification onto each sector and then trace out the multiplicity space, on which the group acts trivially:
\eq{
\Omega_{\tau,\widehat{U}}&:=\bigoplus_\lambda\Omega_{\tau,\widehat{U}}^\lambda,\\
\Omega_{\tau,\widehat{U}}^\lambda
&:=\Tr_{\calM_\lambda}\!\left[
(P_\lambda\otimes I)\ket{\Psi_\tau}\bra{\Psi_\tau}
(P_\lambda\otimes I)\right].
\label{eq:app_strong_iid_Omega}
}
Each block is a positive operator on $\calV_\lambda\otimes E_\tau$, and its trace equals the sector weight $p_\lambda(\tau)$ from the main text:
\eq{
\Tr\Omega_{\tau,\widehat{U}}^\lambda
&=\langle\Psi_\tau|(P_\lambda\otimes I)|\Psi_\tau\rangle\nonumber\\
&=\Tr[P_\lambda\tau]=p_\lambda(\tau).
}
Thus, $\Tr\Omega_{\tau,\widehat{U}}=\sum_\lambda p_\lambda(\tau)=1$.
As a direct-sum state with distinct irreducible labels, $\Omega_{\tau,\widehat{U}}$ records the correlations between each irreducible representation space and the reference system.
We take a common direct-sum space on the irreducible-representation side, assigning zero to components absent from a representation.
The reference spaces generally differ for different states.

We verify the relation between this direct-sum state and the operator-valued function $B$ for a general state $\tau$ and representation $\widehat{U}$.
For these states and representations, the block-diagonal form of $\widehat{U}_g$ in the irreducible decomposition gives
\eq{
&\sum_\lambda\Tr_{\calV_\lambda}\!\left[
((U_g^\lambda)^\dagger\otimes I)\Omega_{\tau,\widehat{U}}^\lambda\right]\nonumber\\
&\quad=\Tr_{\calK}\!\left[(\widehat{U}_g^\dagger\otimes I)
\ket{\Psi_\tau}\bra{\Psi_\tau}\right]
\label{eq:app_strong_iid_partial_trace_Fourier}
}

The right-hand side is evaluated in \eqref{eq:app_strong_iid_purification_character}: it is $B_\rho(g)$ for $\tau=\rho,\widehat{U}=U$ and $B_\sigma(g)$ for $\tau=\sigma,\widehat{U}=U'$.
Thus, $B$ is expressed as a Fourier transform that combines the blocks of $\Omega$ into a function on the group.

For the input consisting of multiple copies and the output representation multiplied by a character, we apply the same definition and use the abbreviations
\eq{
\Omega_{\rho,N}&:=\Omega_{\rho^{\otimes N},U^{\otimes N}},\\
\Omega_{\sigma,M}^{\gamma}&:=\Omega_{\sigma^{\otimes M},\gamma U'^{\otimes M}}
\label{eq:app_strong_iid_Omega_iid}
}

Applying the above relation to these tensor-product states and representations, we obtain
\eq{
\sum_\lambda\Tr_{\calV_\lambda}\!\left[
((U_g^\lambda)^\dagger\otimes I)\Omega_{\rho,N}^\lambda\right]
&=B_\rho(g)^{\otimes N},\nonumber\\
\sum_\lambda\Tr_{\calV_\lambda}\!\left[
((U_g^\lambda)^\dagger\otimes I)(\Omega_{\sigma,M_N}^{\gamma_N})^\lambda\right]
&=\overline{\gamma_N(g)}B_\sigma(g)^{\otimes M_N}
\label{eq:app_strong_iid_Omega_Fourier}
}

The factor $\overline{\gamma_N(g)}$ in \eqref{eq:app_strong_iid_Omega_Fourier} appears because the physical representation is $\gamma_N(g)U_g'^{\otimes M_N}$, whose adjoint is used in \eqref{eq:app_strong_iid_partial_trace_Fourier}.

Finally, we summarize the notation for norms used below.
For state vectors and real or complex coordinate vectors, we use the norm induced by the inner product without a subscript. Namely,
\eq{
\|v\|&:=\sqrt{\langle v,v\rangle},
\label{eq:app_strong_iid_vector_norm}\\
\|\bm\lambda\|&:=\left(\sum_j|\lambda^j|^2\right)^{1/2}.
\label{eq:app_strong_iid_coordinate_norm}
}
Upon identifying the real coefficients $\bm\lambda$ with the Lie algebra element $\sum_j\lambda^jA_j$, the latter norm agrees with that induced by the fixed inner product $\langle\cdot,\cdot\rangle_{\mathfrak g}$.
For an operator $B$, the trace norm $\|B\|_1$, Hilbert--Schmidt norm $\|B\|_2$, and operator norm $\|B\|_\infty$ used in the main text are written as
\eq{
\|B\|_1&:=\Tr\sqrt{B^\dagger B},
\label{eq:app_strong_iid_trace_norm}\\
\|B\|_2&:=\left(\Tr B^\dagger B\right)^{1/2},
\label{eq:app_strong_iid_hilbert_schmidt_norm}\\
\|B\|_\infty&:=\sup_{\|v\|\le1}\|Bv\|
\label{eq:app_strong_iid_operator_norm}
}
respectively. In the last expression, $v$ is a vector in the input space of $B$.
For a scalar-valued function $f$, the same subscript $\infty$ denotes the supremum norm
\eq{
\|f\|_\infty:=\sup_x|f(x)|
\label{eq:app_strong_iid_function_sup_norm}
}
where $x$ ranges over the entire domain of $f$.

\subsection{Lemmas}
\label{app:strong_iid_preliminary}

We first state a lemma relating the residual charge conditions to the choice of auxiliary characters.
In this lemma, we assume that the input and output strong symmetry subgroups coincide and denote their common subgroup by
\eq{
S^{\mathrm{strong}}_{G,\rho,\sigma}
&:=\mathrm{Sym}_{G,\mathrm{strong}}(\rho)\nonumber\\
&=\mathrm{Sym}_{G,\mathrm{strong}}(\sigma)
\label{eq:app_strong_iid_common_subgroup}
}

Likewise, we denote the common strong symmetry subgroup of the representations restricted to $D_G$ by
\eq{
S^{\mathrm{strong}}_{D_G,\rho,\sigma}
&:=\mathrm{Sym}_{D_G,\mathrm{strong}}(\rho)\nonumber\\
&=\mathrm{Sym}_{D_G,\mathrm{strong}}(\sigma)\nonumber\\
&=S^{\mathrm{strong}}_{G,\rho,\sigma}\cap D_G
\label{eq:app_strong_iid_common_derived_subgroup}
}
The last equality follows because the strong symmetry subgroup for $D_G$ is the intersection of the strong symmetry subgroup for $G$ with $D_G$.
We first fix an infinite subset $\mathcal N\subseteq\mathbb N$ of input copy numbers and compare three conditions on this same set.
We then consider the case in which the set itself can be chosen.

\begin{lemma}\label{lem:app_strong_iid_character_lift}
Under the above setup, the following three conditions are equivalent.
\begin{description}
\item[(i)] The equality $m_G(\rho)=r m_G(\sigma)$ holds, and, for every sufficiently large $N\in\mathcal N$, on the common subgroup $S^{\mathrm{strong}}_{D_G,\rho,\sigma}$, we have
\eq{
\zeta_\rho^{N}=\zeta_\sigma^{M_N},
\qquad M_N=\lfloor rN\rfloor
\label{eq:strong_iid_character}
}

\item[(ii)] There exists a sequence of characters $\gamma_N\in\CharG$ indexed by $\mathcal N$ such that, for every sufficiently large $N\in\mathcal N$,
\eq{
\gamma_N|_{S^{\mathrm{strong}}_{G,\rho,\sigma}}&=c_\rho^{\,N}c_\sigma^{-M_N},
\label{eq:app_strong_iid_auxiliary_character}\\
\bm d_N&:=N\bm\mu_\rho-M_N\bm\mu_\sigma-\bm a(\gamma_N),
\label{eq:app_strong_iid_d}\\
\|\bm d_N\|&\le D_0
\label{eq:app_strong_iid_bounded_mean}
}
where $D_0$ is a finite constant independent of $N$.
\item[(iii)] There exists a sequence of characters $\gamma_N\in\CharG$ indexed by $\mathcal N$ satisfying \eqref{eq:app_strong_iid_auxiliary_character} for every sufficiently large $N\in\mathcal N$, and
\eq{
\lim_{\substack{N\to\infty\\N\in\mathcal N}}
\frac{\|\bm d_N\|}{\sqrt N}=0
\label{eq:app_strong_iid_small_mean}
}

\end{description}
Furthermore, if $m_G(\rho)=r m_G(\sigma)$ holds under the above setup, there exists an infinite set $\mathcal N\subseteq\mathbb N$ satisfying these equivalent conditions.
In particular, a sequence of auxiliary characters $\gamma_N$ satisfying condition (ii) can be chosen on this set.
\end{lemma}

\begin{proofof}{Lemma~\ref{lem:app_strong_iid_character_lift}}
We first prove the equivalence of (i)--(iii) for a fixed $\mathcal N$.
The implication from (ii) to (iii) follows from $D_0/\sqrt N\to0$.

Next, assume (iii) and prove (i). Since a character equals $1$ on every commutator, restricting \eqref{eq:app_strong_iid_auxiliary_character} to $S^{\mathrm{strong}}_{D_G,\rho,\sigma}$ first gives
$\zeta_\rho^N=\zeta_\sigma^{M_N}$.

We next establish the equality of the residual mean charges. To this end, we examine the action of group averaging on the mean charge.
Using the real orthogonal matrices in \eqref{eq:app_strong_iid_adjoint}, let
\eq{P:=\int_G O_g\,dg}
The invariance of the Haar measure and $O_g^{\mathsf T}=O_{g^{-1}}$ give
\eq{P^{\mathsf T}=P,\qquad
P^2=\int_G\!\int_G O_{gh}\,dg\,dh=P}
Thus, $P$ is an orthogonal projection and $\|I-P\|_\infty\le1$.
Furthermore, averaging \eqref{eq:app_strong_iid_adjoint} over the group gives
\eq{
\bm\mu_{\calG_U(\tau)}=P\bm\mu_\tau,
\qquad m_G(\tau)=(I-P)\bm\mu_\tau
\label{eq:app_strong_iid_mean_projection}
}
The same $P$ is used for the output.
For any $\gamma\in\CharG$, differentiating $\gamma(hgh^{-1})=\gamma(g)$ along a one-parameter subgroup yields
$O_h^{\mathsf T}\bm a(\gamma)=\bm a(\gamma)$, and hence
\eq{P\bm a(\gamma)=\bm a(\gamma)
\label{eq:app_strong_iid_character_projection}}

Applying $I-P$ to \eqref{eq:app_strong_iid_d} and using \eqref{eq:app_strong_iid_mean_projection} and \eqref{eq:app_strong_iid_character_projection}, we obtain
\eq{
\frac{(I-P)\bm d_N}{N}
=m_G(\rho)-\frac{M_N}{N}m_G(\sigma)
}
The norm of the left-hand side is at most $\|\bm d_N\|/N$, which converges to zero along $N\in\mathcal N$.
Together with $M_N/N\to r$, this implies $m_G(\rho)=r m_G(\sigma)$.
This proves the implication from (iii) to (i). It therefore remains to prove the implication from (i) to (ii) to establish the equivalence of (i)--(iii).

We now prove that (i) implies (ii).
We first use determinants of irreducible representations to match the mean charges up to a bounded error, and then correct the phases on the strong symmetry subgroup.
Let $\mathcal I$ be the finite union of the irreducible representation labels appearing in $U,U'$. For a representative $U^\lambda$ of each equivalence class, let
\eq{
d_\lambda:=\dim\calV_\lambda,\qquad
\kappa_\lambda(g):=\det U_g^\lambda
}
By the multiplicativity of the determinant, $\kappa_\lambda$ is a character on $G$.
Writing $X_{\lambda,j}$ for the generators of the irreducible representation $U^\lambda$, we have
\eq{
\kappa_\lambda(g(t\bm\lambda))
&=\det\exp\!\left(it\sum_j\lambda^jX_{\lambda,j}\right)\nonumber\\
&=\exp\!\left(it\sum_j\lambda^j\Tr X_{\lambda,j}\right)
}
and therefore \eqref{eq:app_strong_iid_character_charge} gives
\eq{
\bigl(\bm a(\kappa_\lambda)\bigr)_j=\Tr X_{\lambda,j}.
\label{eq:app_strong_iid_determinant_charge}
}
Let $P_\lambda,P'_\lambda$ be the projections onto the input and output irreducible components, respectively, and write their weights as
\eq{
p_\lambda(\rho):=\Tr(P_\lambda\rho),\qquad
p_\lambda(\sigma):=\Tr(P'_\lambda\sigma)
}
Components absent from the respective representations are assigned zero weight.
The twirling formula obtained from Schur's lemma gives, for the input,
\eq{
\Tr[\calG_U(\rho)X_j]
&=\sum_{\lambda\in\mathcal I}
\frac{p_\lambda(\rho)}{d_\lambda}\Tr X_{\lambda,j}\nonumber\\
&=\sum_{\lambda\in\mathcal I}
\frac{p_\lambda(\rho)}{d_\lambda}
\bigl(\bm a(\kappa_\lambda)\bigr)_j
}
The same holds for the output, so \eqref{eq:app_strong_iid_mean_projection} gives
\eq{
P\bm\mu_\rho
&=\sum_{\lambda\in\mathcal I}
\frac{p_\lambda(\rho)}{d_\lambda}\bm a(\kappa_\lambda),\\
P\bm\mu_\sigma
&=\sum_{\lambda\in\mathcal I}
\frac{p_\lambda(\sigma)}{d_\lambda}\bm a(\kappa_\lambda)
\label{eq:app_strong_iid_mean_determinants}
}

For each $\lambda\in\mathcal I$, consider the real number
\eq{
z_{N,\lambda}:=
\frac{Np_\lambda(\rho)-M_Np_\lambda(\sigma)}{d_\lambda}
}
Let $n_{N,\lambda}$ be a nearest integer to it and define
\eq{
\chi_N:=\prod_{\lambda\in\mathcal I}
\kappa_\lambda^{\,n_{N,\lambda}}
\label{eq:app_strong_iid_rounded_character}
}
Since a product of integer powers of characters is again a character, $\chi_N\in\CharG$.
Denote the residual after subtracting its charge by
\eq{
\bm e_N:=N\bm\mu_\rho-M_N\bm\mu_\sigma-\bm a(\chi_N)
\label{eq:app_strong_iid_rounded_residual}
}
Using $\bm\mu_\tau=m_G(\tau)+P\bm\mu_\tau$ and the assumption $m_G(\rho)=r m_G(\sigma)$, we obtain
\eq{
\bm e_N
&=Nm_G(\rho)-M_Nm_G(\sigma)\nonumber\\
&\quad+\sum_{\lambda\in\mathcal I}
(z_{N,\lambda}-n_{N,\lambda})\bm a(\kappa_\lambda)\nonumber\\
&=(rN-M_N)m_G(\sigma)\nonumber\\
&\quad+\sum_{\lambda\in\mathcal I}
(z_{N,\lambda}-n_{N,\lambda})\bm a(\kappa_\lambda).
}
Since $0\le rN-M_N<1$ and $|z_{N,\lambda}-n_{N,\lambda}|\le1/2$,
\eq{
\|\bm e_N\|
\le\|m_G(\sigma)\|
+\frac12\sum_{\lambda\in\mathcal I}\|\bm a(\kappa_\lambda)\|
=:C
\label{eq:app_strong_iid_rounded_residual_bound}
}
so the mean charge difference is bounded independently of $N$.

Next, express the difference between $\chi_N$ and the desired phase by the character on $S^{\mathrm{strong}}_{G,\rho,\sigma}$ defined as
\eq{
\eta_N(k):=
c_\rho(k)^N c_\sigma(k)^{-M_N}\chi_N(k)^{-1}
\qquad(k\in S^{\mathrm{strong}}_{G,\rho,\sigma})
\label{eq:app_strong_iid_phase_correction}
}

Write $\mathfrak{s}^{\mathrm{strong}}_{G,\rho,\sigma}:=\operatorname{Lie}(S^{\mathrm{strong}}_{G,\rho,\sigma})$.
Below, using the correspondence associated with the orthonormal basis $\{A_j\}$ chosen above,
\eq{
\sum_{j=1}^{d_G}\lambda^jA_j
\quad\longleftrightarrow\quad
\bm\lambda=(\lambda^1,\ldots,\lambda^{d_G})
\label{eq:app_strong_iid_coordinate_identification}
}
we regard this Lie algebra as a subspace of $\mathbb R^{d_G}$.
If $\bm\lambda\in\mathfrak{s}^{\mathrm{strong}}_{G,\rho,\sigma}$, then $g(t\bm\lambda)\in S^{\mathrm{strong}}_{G,\rho,\sigma}$ for every $t\in\mathbb R$.
Differentiating the strong symmetry relation $e^{itX(\bm\lambda)}\rho=c_\rho(g(t\bm\lambda))\rho$ and its trace at $t=0$, we obtain $X(\bm\lambda)\rho=(\bm\mu_\rho\cdot\bm\lambda)\rho$.
Using this relation in the power series of the exponential, taking the trace, and carrying out the same calculation for $\sigma$, we obtain
\eq{
c_\rho(g(t\bm\lambda))=e^{it\bm\mu_\rho\cdot\bm\lambda},\qquad
c_\sigma(g(t\bm\lambda))=e^{it\bm\mu_\sigma\cdot\bm\lambda}
}
Therefore,
\eq{
\eta_N(g(t\bm\lambda))
=e^{it\bm e_N\cdot\bm\lambda},
\qquad\bm\lambda\in\mathfrak{s}^{\mathrm{strong}}_{G,\rho,\sigma}
\label{eq:app_strong_iid_phase_correction_derivative}
}
and the derivative of $\eta_N$ at the identity is uniformly bounded by \eqref{eq:app_strong_iid_rounded_residual_bound}.

Using this boundedness and the orthogonality of characters, we show that only finitely many distinct phase differences $\eta_N$ occur.
Choose $\delta>0$ sufficiently small that $C\delta\le\pi/3$ and
\eq{
V:=\{g(\bm\lambda):
\bm\lambda\in\mathfrak{s}^{\mathrm{strong}}_{G,\rho,\sigma},\ 
\|\bm\lambda\|<\delta\}
}
is an open neighborhood of the identity in $S^{\mathrm{strong}}_{G,\rho,\sigma}$.
By \eqref{eq:app_strong_iid_phase_correction_derivative} and $\|\bm e_N\|\le C$, for every $N$,
\eq{
\operatorname{Re}\eta_N(k)\ge\frac12\qquad(k\in V)
}
In the following integrals, we use the normalized Haar measure $dk$ of $S^{\mathrm{strong}}_{G,\rho,\sigma}$ itself and let $v:=\int_Vdk>0$.
Then,
\eq{
\left|\int_V\eta_N(k)\,dk\right|
\ge\operatorname{Re}\int_V\eta_N(k)\,dk
\ge\frac v2.
}
Choose any $J$ distinct characters that occur and denote them by $\eta_1,\ldots,\eta_J$.
Here, $L^2(S^{\mathrm{strong}}_{G,\rho,\sigma},dk)$ is the usual $L^2$ space: the Hilbert space of complex-valued functions on the subgroup $S^{\mathrm{strong}}_{G,\rho,\sigma}$ whose squared absolute values have finite integral with respect to $dk$.
Its inner product is defined by
\eq{
\langle f,h\rangle
:=\int_{S^{\mathrm{strong}}_{G,\rho,\sigma}}
\overline{f(k)}h(k)\,dk
\label{eq:app_strong_iid_L2_inner_product}
}

By the orthogonality of characters, $\eta_1,\ldots,\eta_J$ form an orthonormal system in $L^2(S^{\mathrm{strong}}_{G,\rho,\sigma},dk)$.
Consider the indicator function $1_V$ of $V$, namely the function on $S^{\mathrm{strong}}_{G,\rho,\sigma}$ defined by
\eq{
1_V(k):=\begin{cases}
1 & (k\in V),\\
0 & (k\notin V)
\end{cases}
\label{eq:app_strong_iid_indicator}
}
Apply Bessel's inequality, which states that the sum of the squared absolute values of the inner products of a function with the members of an orthonormal system does not exceed the squared norm of that function, to obtain
\eq{
\frac{Jv^2}{4}
&\le\sum_{j=1}^J\left|\int_V\eta_j(k)\,dk\right|^2\nonumber\\
&\le\int_{S^{\mathrm{strong}}_{G,\rho,\sigma}}|1_V(k)|^2\,dk
=v
\label{eq:app_strong_iid_finite_phase_bessel}
}
Thus, $J\le4/v$, so the number of distinct characters is finite.
The property $v>0$ used here also holds when the subgroup is disconnected.
Since each connected component of a Lie group is open, the neighborhood $V$ chosen within the identity component is also open in $S^{\mathrm{strong}}_{G,\rho,\sigma}$.
By compactness, finitely many translates $k_1V,\ldots,k_mV$ cover the entire subgroup.
The invariance and normalization of the Haar measure therefore give
\eq{
1=\int_{S^{\mathrm{strong}}_{G,\rho,\sigma}}dk
\le\sum_{a=1}^m\int_{k_aV}dk
=mv
\label{eq:app_strong_iid_neighborhood_positive_measure}
}
and hence $v\ge1/m>0$. Thus, the argument applies without change to the disconnected case.
For a finite group, simply take $V=\{e\}$.

Denote the finite set of characters $\eta_N$ occurring for sufficiently large $N\in\mathcal N$ by $\mathcal J_\eta$.
By (i), $c_\rho^Nc_\sigma^{-M_N}$ is trivial on $S^{\mathrm{strong}}_{D_G,\rho,\sigma}$, meaning that it always takes the value 1. Since $\chi_N$ is a character on $G$, it is also trivial on $D_G$.
Consequently, each $\eta\in\mathcal J_\eta$ is trivial on $S^{\mathrm{strong}}_{D_G,\rho,\sigma}$.
Using this property, we extend each $\eta$ to a character on $G$.
To this end, consider the quotient group $A:=G/D_G$ of $G$ by $D_G$.
Specifically, use the continuous quotient map sending each $g\in G$ to the set $gD_G$,
\eq{
\pi:G\to A,\qquad
\pi(g):=gD_G=\{gd:d\in D_G\}
\label{eq:app_strong_iid_quotient_map}
}
Since $D_G$ is a closed normal subgroup, $A$ is a compact Lie group; it is abelian because all commutators are identified with the identity.
The product in the quotient group is defined by $\pi(g)\pi(g')=\pi(gg')$, so $\pi$ preserves products.
Thus, $B:=\pi(S^{\mathrm{strong}}_{G,\rho,\sigma})$ is a subgroup of $A$.
Moreover, $S^{\mathrm{strong}}_{G,\rho,\sigma}$ is compact, and hence its continuous image $B$ is compact and therefore a closed subgroup of $A$.

To define a function on $B$ consistently using the values of $\eta$, it suffices to verify that $\eta$ takes the same value on elements mapped to the same point by $\pi$.
If $k,k'\in S^{\mathrm{strong}}_{G,\rho,\sigma}$ satisfy $\pi(k)=\pi(k')$, then there exists $d\in D_G$ such that $k'=kd$.
Moreover, $k^{-1}k'$ also belongs to $S^{\mathrm{strong}}_{G,\rho,\sigma}$, so
\eq{
k^{-1}k'
\in S^{\mathrm{strong}}_{G,\rho,\sigma}\cap D_G
=S^{\mathrm{strong}}_{D_G,\rho,\sigma}.
\label{eq:app_strong_iid_equal_cosets}
}
Since $\eta$ takes the value $1$ on this intersection,
\eq{
\eta(k')=\eta(k)\eta(k^{-1}k')=\eta(k).
}
Thus, defining
\eq{
\eta_B(\pi(k)):=\eta(k)
\qquad(k\in S^{\mathrm{strong}}_{G,\rho,\sigma})
\label{eq:app_strong_iid_descended_character}
}
does not assign different values to the same $\pi(k)$.
Furthermore, since $\pi$ and $\eta$ preserve products,
\eq{
\eta_B(\pi(k)\pi(k'))
&=\eta_B(\pi(kk'))=\eta(kk')\nonumber\\
&=\eta(k)\eta(k')
=\eta_B(\pi(k))\eta_B(\pi(k'))
}
so $\eta_B$ also preserves products.
Its continuity can be verified as follows.
For any closed set $F\subseteq U(1)$, continuity of $\eta$ implies that the preimage $\eta^{-1}(F)$ is
a closed subset of the compact group $S^{\mathrm{strong}}_{G,\rho,\sigma}$ and is therefore compact.
Consequently,
\eq{
\eta_B^{-1}(F)=\pi\bigl(\eta^{-1}(F)\bigr)
}
is also compact and hence closed in $B$.
Since the preimage of every closed set is closed, $\eta_B$ is continuous.
Thus, $\eta_B$ is a continuous character on $B$.
For a compact abelian group $A$ and a closed subgroup $B\subseteq A$,
the character extension theorem~\cite[Corollary~2.33(ii)]{hofmann2006structure}
guarantees that every continuous character $\chi:B\to U(1)$ extends to
a continuous character $\widetilde\chi:A\to U(1)$ such that
\eq{
\widetilde\chi|_B=\chi
\label{eq:app_strong_iid_character_extension}
}
holds.
For each $\eta\in\mathcal J_\eta$, fix one extension $\eta_A$ of $\eta_B$ to $A$ and let
\eq{
\widetilde\eta(g):=\eta_A(\pi(g))\qquad(g\in G)
\label{eq:app_strong_iid_character_pullback}
}
As a composition of continuous group homomorphisms, this is a character on $G$, and for every $k\in S^{\mathrm{strong}}_{G,\rho,\sigma}$ it satisfies
\eq{
\widetilde\eta(k)
=\eta_A(\pi(k))
=\eta_B(\pi(k))=\eta(k)
}

Since $\mathcal J_\eta$ is finite and each extension is fixed once and for all,
\eq{
C':=\max_{\eta\in\mathcal J_\eta}\|\bm a(\widetilde\eta)\|<\infty.
}
Now define
\eq{
\gamma_N:=\chi_N\widetilde{\eta_N}
}
By \eqref{eq:app_strong_iid_phase_correction},
\eq{
\gamma_N|_{S^{\mathrm{strong}}_{G,\rho,\sigma}}
=c_\rho^Nc_\sigma^{-M_N}
}
holds exactly. Moreover, charges add under multiplication of characters, so
\eq{
\bm d_N&=\bm e_N-\bm a(\widetilde{\eta_N}),\\
\|\bm d_N\|&\le C+C'.
}
For the finitely many excluded $N\in\mathcal N$, choose the characters arbitrarily.
Thus, (ii) holds with $D_0=C+C'$.
This proves that (i) implies (ii), completing the equivalence of (i)--(iii) for a fixed $\mathcal N$.

Finally, we prove the additional assertion concerning the existence of an infinite set $\mathcal N$.
Assume $m_G(\rho)=r m_G(\sigma)$ and construct an infinite set $\mathcal N$ satisfying \eqref{eq:strong_iid_character}.
Equality of the strong symmetry subgroups remains an assumption of the lemma.

First, consider a direction $\bm{\lambda}$ satisfying $\sum_j\lambda^jA_j\in\operatorname{Lie}(D_G)$.
Define the generator of the irreducible representation $U^\lambda$ by $U^\lambda_{g(t\bm{\lambda})}=e^{itX_\lambda(\bm{\lambda})}$. The map $g\mapsto\det U_g^\lambda$ is a character on $G$.
Since $g(t\bm{\lambda})\in D_G$,
\eq{
1=\det U^\lambda_{g(t\bm{\lambda})}
=e^{it\Tr X_\lambda(\bm{\lambda})}\quad(\forall t\in\mathbb R)
}
implies $\Tr X_\lambda(\bm{\lambda})=0$.
Since a twirled state is proportional to the identity on each irreducible representation space, writing $p_\lambda(\tau)$ for the weight of each irreducible component of a state $\tau$ and $d_\lambda$ for its dimension gives
\eq{
\Tr[\calG_U(\tau)X(\bm{\lambda})]
=\sum_\lambda\frac{p_\lambda(\tau)}{d_\lambda}
\Tr X_\lambda(\bm{\lambda})=0.
}
Thus, in this direction, $\bm\mu_\tau\cdot\bm{\lambda}=m_G(\tau)\cdot\bm{\lambda}$, and the same holds for the output.
If, in addition, $\sum_j\lambda^jA_j\in\operatorname{Lie}(S^{\mathrm{strong}}_{D_G,\rho,\sigma})$, applying the definition of strong symmetry to the one-parameter subgroup gives
\eq{
\zeta_\rho(g(t\bm{\lambda}))&=e^{it\bm\mu_\rho\cdot\bm{\lambda}},\nonumber\\
\zeta_\sigma(g(t\bm{\lambda}))&=e^{it\bm\mu_\sigma\cdot\bm{\lambda}}
}
Taking the derivatives of the residual characters along this one-parameter subgroup at $t=0$, we obtain
\eq{
\left.\frac{d}{dt}\zeta_\rho(g(t\bm{\lambda}))\right|_{t=0}
&=i\bm\mu_\rho\cdot\bm{\lambda},
\label{eq:app_strong_iid_residual_character_derivative_rho}\\
\left.\frac{d}{dt}\zeta_\sigma(g(t\bm{\lambda}))\right|_{t=0}
&=i\bm\mu_\sigma\cdot\bm{\lambda}
\label{eq:app_strong_iid_residual_character_derivative_sigma}
}
Thus, $\bm\mu_\rho\cdot\bm{\lambda}$ and $\bm\mu_\sigma\cdot\bm{\lambda}$ describe the respective rates of change of the phases in this direction.
By $m_G(\rho)=r m_G(\sigma)$, for every $\bm{\lambda}$ such that $\sum_j\lambda^jA_j\in\operatorname{Lie}(S^{\mathrm{strong}}_{D_G,\rho,\sigma})$, they satisfy
\eq{
\bm\mu_\rho\cdot\bm{\lambda}=r\bm\mu_\sigma\cdot\bm{\lambda}
\label{eq:app_strong_iid_residual_differential}
}

For each integer $n\ge0$, consider the character on $S^{\mathrm{strong}}_{D_G,\rho,\sigma}$ defined by
\eq{
\theta_n:=\zeta_\rho^n\zeta_\sigma^{-\lfloor rn\rfloor},
\qquad\theta_0=1
\label{eq:app_strong_iid_subsequence_phase}
}
By \eqref{eq:app_strong_iid_residual_differential},
in every direction satisfying $\sum_j\lambda^jA_j\in\operatorname{Lie}(S^{\mathrm{strong}}_{D_G,\rho,\sigma})$, we have
\eq{
\theta_n(g(t\bm\lambda))
&=\exp\!\left[it(n\bm\mu_\rho-\lfloor rn\rfloor\bm\mu_\sigma)
\cdot\bm\lambda\right]\nonumber\\
&=\exp\!\left[it(rn-\lfloor rn\rfloor)
\bm\mu_\sigma\cdot\bm\lambda\right]
\label{eq:app_strong_iid_subsequence_phase_derivative}
}
Since $0\le rn-\lfloor rn\rfloor<1$, the absolute value of the derivative of its phase is
at most $\|\bm\mu_\sigma\|\|\bm\lambda\|$ and is bounded independently of $n$.
Thus, applying the argument that led to \eqref{eq:app_strong_iid_finite_phase_bessel}
to the characters $\theta_n$ on $S^{\mathrm{strong}}_{D_G,\rho,\sigma}$ shows that only finitely many distinct characters occur.
Here, we use the normalized Haar measure of this subgroup itself.
Denote the number of distinct characters by
\eq{
q:=\bigl|\{\theta_n:n=0,1,2,\ldots\}\bigr|<\infty
\label{eq:app_strong_iid_subsequence_phase_count}
}
Since $\theta_0=1$ is included, $q\ge1$.

We next verify that the rounding error can satisfy
\eq{
0\le rn-\lfloor rn\rfloor<\frac1q
\label{eq:app_strong_iid_subsequence_small_fraction}
}
for arbitrarily large positive integers $n$.
If $r$ is rational, taking $n$ to be a multiple of its denominator makes the rounding error zero.
If $r$ is irrational, for any integer $Q>1$,
place the $Q+1$ fractional parts of $0,r,\ldots,Qr$ into the $Q$ equal half-open intervals partitioning $[0,1)$.
Choosing two that lie in the same interval gives $0<|ar-z|<1/Q$ for some $1\le a\le Q$ and integer $z$.
If $ar-z>0$, set $n=a$ to obtain $0<rn-\lfloor rn\rfloor<1/Q$.
If $ar-z<0$, let $\delta:=z-ar$, $k:=\lfloor1/\delta\rfloor$, and $n:=ka$.
Since $\delta$ is irrational, $k\delta<1<(k+1)\delta$; together with $rn=kz-k\delta$, this gives
\eq{
0<rn-\lfloor rn\rfloor=1-k\delta<\delta<\frac1Q
}
As $Q\to\infty$, this fractional part can be made arbitrarily small, while $rn-\lfloor rn\rfloor>0$ for every fixed $n$, so the resulting values of $n$ are unbounded.
Thus, in the irrational case as well, arbitrarily large $n$ satisfying \eqref{eq:app_strong_iid_subsequence_small_fraction} can be chosen.

Fix such an $n$. For $j=0,\ldots,q$, we have $j(rn-\lfloor rn\rfloor)<1$, so multiplying by these integers causes no carry in the rounding, and
\eq{
\lfloor rjn\rfloor
&=j\lfloor rn\rfloor
+\left\lfloor j(rn-\lfloor rn\rfloor)\right\rfloor\nonumber\\
&=j\lfloor rn\rfloor.
\label{eq:app_strong_iid_subsequence_no_carry}
}
Therefore, \eqref{eq:app_strong_iid_subsequence_phase} gives
\eq{
\theta_{jn}
=\zeta_\rho^{jn}\zeta_\sigma^{-j\lfloor rn\rfloor}
=\theta_n^j
\qquad(j=0,\ldots,q)
\label{eq:app_strong_iid_subsequence_phase_powers}
}
Since these $q+1$ characters include at most $q$ distinct characters, there exist $0\le j_1<j_2\le q$ such that
$\theta_n^{j_1}=\theta_n^{j_2}$.
Let $L_n:=j_2-j_1$. Then $1\le L_n\le q$ and $\theta_n^{L_n}=1$.
Choosing $N:=L_n n$, we obtain from \eqref{eq:app_strong_iid_subsequence_no_carry} and
\eqref{eq:app_strong_iid_subsequence_phase_powers} that
\eq{
M_N&=\lfloor rL_n n\rfloor=L_n\lfloor rn\rfloor,\\
\zeta_\rho^N\zeta_\sigma^{-M_N}
&=\theta_{L_n n}=\theta_n^{L_n}=1.
\label{eq:app_strong_iid_subsequence_exact_phase}
}
The resulting $N$ satisfies $N\ge n$, and $n$ can be arbitrarily large, so there are infinitely many such $N$.
Let $\mathcal N$ be the set of these values. Then \eqref{eq:strong_iid_character} holds exactly for every $N\in\mathcal N$.
Thus, under $m_G(\rho)=r m_G(\sigma)$, for any rate $r>0$ we can choose an infinite set $\mathcal N$ satisfying condition (i).
Applying the already established implication from (i) to (ii) to this set gives a sequence of auxiliary characters $\gamma_N$ satisfying condition (ii).
This proves the additional assertion.
\end{proofof}

Next, we state a lemma using an operator-valued Fourier transform. We first explain its meaning in the case of $U(1)$.

In this case, the irreducible representations are labeled by integer charges $q$, and $U_\theta^q=e^{iq\theta}$.
For two probability distributions $p_q,p'_q$ supported on finitely many charges, define
\eq{
F(\theta):=\sum_q e^{-iq\theta}(p_q-p'_q)
\label{eq:app_strong_iid_fourier_U1_example}
}
This is the Fourier transform of the difference between the charge distributions.
By the orthogonality of the exponential functions, each charge component can be recovered as
\eq{
p_q-p'_q=\frac1{2\pi}\int_0^{2\pi}e^{iq\theta}F(\theta)\,d\theta
\label{eq:app_strong_iid_fourier_U1_recovery}
}

For a general compact group $G$, let $\lambda$ label the irreducible representations, and denote their representation spaces and dimensions by $\calV_\lambda$ and $d_\lambda$, respectively.
Using the notation of the following lemma, consider finitely many irreducible components and expand the difference between two operators $A^\lambda,A'^\lambda$ on $\calV_\lambda\otimes\calR$ in an orthonormal basis of $\calV_\lambda$ as
\eq{
\Delta^\lambda:=A^\lambda-A'^\lambda
=\sum_{a,b}\ket a\bra b\otimes\Delta^\lambda_{ab}
\label{eq:app_strong_iid_fourier_block_example}
}
Here, $\calR$ is a reference system, and each coefficient $\Delta^\lambda_{ab}$ is an operator on it.
The counterpart of the above $U(1)$ transform is then
\eq{
F(g)&:=\sum_{\lambda,a,b}\overline{(U_g^\lambda)_{ab}}\,\Delta^\lambda_{ab}\nonumber\\
&=\sum_\lambda\Tr_{\calV_\lambda}\!\left[
((U_g^\lambda)^\dagger\otimes I)\Delta^\lambda\right]
\label{eq:app_strong_iid_fourier_general_example}
}
That is, $e^{-iq\theta}$ is replaced by the matrix element $\overline{(U_g^\lambda)_{ab}}$, and the scalar coefficient $p_q-p'_q$ by the operator $\Delta^\lambda_{ab}$.
Substituting $\Delta^\lambda=A^\lambda-A'^\lambda$ into \eqref{eq:app_strong_iid_fourier_general_example} gives exactly \eqref{eq:app_strong_iid_fourier_function} in the lemma.
Conversely, the Schur orthogonality relations for the matrix elements give the recovery formula
\eq{
\Delta^\lambda=d_\lambda\int_G U_g^\lambda\otimes F(g)\,dg
\label{eq:app_strong_iid_fourier_general_recovery}
}
Here, $dg$ is the normalized Haar measure; the componentwise calculation is given in \eqref{eq:app_strong_iid_operator_fourier_inverse} in the proof.
Thus, knowing $F(g)$ for every $g$ recovers the difference in each block, including its correlations with the reference system.
The error between the direct sums $A=\bigoplus_\lambda A^\lambda$ and $A'=\bigoplus_\lambda A'^\lambda$ is $\|A-A'\|_1=\sum_\lambda\|\Delta^\lambda\|_1$, whereas the error after the Fourier transform is measured by $\int_G\|F(g)\|_1\,dg$.
The following lemma bounds the former by the latter under symmetry and second-moment conditions.

Let $X_{\lambda,j}$ be the generators of the irreducible representation $U^\lambda$, and for $\bm b\in\mathbb R^{d_G}$ define
\eq{
\mathsf C_\lambda(\bm b):=\sum_{j=1}^{d_G}(X_{\lambda,j}-b_jI)^2
\label{eq:app_strong_iid_quadratic_charge}
}
This is a positive operator measuring the deviation of the charges from $\bm b$.

\begin{lemma}\label{lem:app_strong_iid_fourier_bound}
Let $G$ be a compact Lie group, let $H\subseteq G$ be a closed subgroup, and set $s:=\dim G-\dim H$.
Let $\calR$ be an arbitrary finite-dimensional space, and let $V_h$ be a continuous unitary representation of $H$ on it.
For finitely many irreducible unitary representations $U^\lambda$ of $G$, take positive operators $A^\lambda,A'^\lambda$ on $\calV_\lambda\otimes\calR$, and set
\eq{
A:=\bigoplus_\lambda A^\lambda,\qquad
A':=\bigoplus_\lambda A'^\lambda,\qquad
\Tr A=1,\qquad \Tr A'\le1
}
Assume that, for all $h\in H$,
\eq{
(U_h^\lambda\otimes V_h)A^\lambda&=A^\lambda,\nonumber\\
(U_h^\lambda\otimes V_h)A'^\lambda&=A'^\lambda
\label{eq:app_strong_iid_fourier_support}
}
That is, the supports of both operators lie in the subspace in which every vector is invariant under $U_h^\lambda\otimes V_h$.
Further, suppose that $\bm b$ satisfies $O_h\bm b=\bm b$ for all $h\in H$, and that, for some nonnegative real number $E\ge0$,
\eq{
E_A:=\sum_\lambda\Tr[(\mathsf C_\lambda(\bm b)\otimes I)A^\lambda]\le E
\label{eq:tb_one_sided_moment}
}
For the operator-valued function
\eq{
F(g):=\sum_\lambda\Tr_{\calV_\lambda}\!\left[
((U_g^\lambda)^\dagger\otimes I)(A^\lambda-A'^\lambda)
\right]
\label{eq:app_strong_iid_fourier_function}
}
there exist constants $\delta_0>0$ and $C_{G,H}>0$ such that, for $0<\delta\le\delta_0$, setting
\eq{
\eta_\delta:=C_{G,H}\delta^{-s}\int_G\|F(g)\|_1\,dg
\label{eq:tb_one_sided_eta}
}
gives
\eq{
\|A-A'\|_1
\le(2+2\sqrt2)\delta\sqrt E+2\sqrt{\eta_\delta}+\eta_\delta
\label{eq:tb_one_sided_fourier}
}

The constants $\delta_0,C_{G,H}$ depend only on $G,H$ and the fixed inner product, and are independent of the dimension and representation of $\calR$ and of $\bm b,E,\delta$.
\end{lemma}
\begin{proofof}{Lemma~\ref{lem:app_strong_iid_fourier_bound}}
The basic strategy is to construct, for $0<\delta\le\delta_0$, a positive operator $0\preceq R\preceq I$ satisfying the following three estimates:
\eq{
\|R(A-A')R\|_1&\le\eta_\delta,
\label{eq:tb_fourier_strategy_smoothed}\\
\|A-RAR\|_1&\le2\delta\sqrt E,
\label{eq:tb_fourier_strategy_input_recovery}\\
\|A'-RA'R\|_1&\le2\sqrt{2\delta^2E+\eta_\delta}.
\label{eq:tb_fourier_strategy_output_recovery}
}
Combining these estimates using the triangle inequality and $\sqrt{a+b}\le\sqrt a+\sqrt b$ ($a,b\ge0$) gives \eqref{eq:tb_one_sided_fourier}.
As shown below, such an operator is obtained by averaging phase-adjusted unitary operators near the identity to form $T$, and then setting $R=T^2$.

To define $T$ and $R$, we first use local coordinates near the identity to specify the small ball over which the average is taken.
Suppose that $d_G>0$, and set $d_H:=\dim H$ and $\mathfrak h:=\operatorname{Lie}(H)$.
As in \eqref{eq:app_strong_iid_coordinate_identification}, identify each Lie algebra element with its real coefficient vector in the basis $\{A_j\}$.
The exponential map $g(x)=\exp(i\sum_jx^jA_j)$ gives one-to-one smooth coordinates near the origin, so choose a sufficiently small $r_0>0$ and, for $0<t\le r_0$, define
\eq{
 B_G(t):=\{g(x):\|x\|<t\}
 \label{eq:tb_local_ball}
}

For the closed subgroup $H$, taking $r_0$ sufficiently small ensures that the unique expression $g=g(x)$ for a group element in $B_G(r_0)$ satisfies
\eq{
 g(x)\in H\quad\Longleftrightarrow\quad x\in\mathfrak h
 \label{eq:app_strong_iid_subgroup_coordinate_correspondence}
}
~\cite[Proposition~20.8(f) and the proof of Theorem~20.12, Eq.~(20.6)]{Lee2013Smooth}.
In other words, the exponential map sends small vectors in $\mathfrak h$ to elements of $H$, and conversely, every element of $H$ sufficiently close to the identity is obtained uniquely from such a vector.

Thus, for $0<t\le r_0$, we can write
\eq{
 H\cap B_G(t)=\{g(x):x\in\mathfrak h,\ \|x\|<t\}
 \label{eq:app_strong_iid_subgroup_coordinates}
}

Let $dg$ be the normalized Haar measure on $G$, and denote this measure by $|\cdot|$.
Since $B_G(t)$ is a nonempty open set, $0<|B_G(t)|\le1$.
Choose $0<\delta_0\le\min\{r_0,1\}$ and carry out the following construction for every $0<\delta\le\delta_0$.
If required by the volume estimates below, we further decrease $r_0,\delta_0$ and use the construction in the final range thus chosen.

To average group transformations near the identity, introduce a probability density $q_\delta$ on $B_G(\delta)$.
Using this density, define the average $T_\lambda$ of phase-adjusted unitary operators as follows:
\eq{
 q_\delta(g)&:=\frac{\mathbf1_{B_G(\delta)}(g)}{|B_G(\delta)|},
 \label{eq:tb_smoothing_density}\\
 f(g)&:=\begin{cases}
 e^{i\bm b\cdot x}q_\delta(g),&g=g(x)\in B_G(\delta),\\
 0,&g\notin B_G(\delta),
 \end{cases}
 \label{eq:tb_smoothing_function}\\
 T_\lambda&:=\int_G f(g)(U_g^\lambda)^\dagger dg.
 \label{eq:tb_smoothing_operator}
}
Here, $\mathbf1_{B_G(\delta)}$ is the indicator function that equals $1$ inside $B_G(\delta)$ and $0$ outside it.
The function $q_\delta$ is a probability density with respect to Haar measure and satisfies $\int_Gq_\delta(g)\,dg=1$.
That is, $q_\delta(g)\,dg$ is the normalized restriction of Haar measure to the small ball $B_G(\delta)$, and $T_\lambda$ is an average with respect to this measure.

In terms of $T_\lambda$, the desired operators $T$ and $R$ are written as
\eq{
T&:=\bigoplus_\lambda T_\lambda\otimes I_{\cal R},
\label{eq:tb_fourier_averaged_operator}\\
R&:=T^2,\qquad 0\preceq R\preceq I.
\label{eq:tb_fourier_positive_operator}
}
We collect the properties of $T_\lambda$ and the weight function that will be used below:
\eq{
 T_\lambda^\dagger&=T_\lambda,\qquad \|T_\lambda\|_\infty\le1,
 \label{eq:tb_smoothing_operator_properties}\\
 \bigl[T_\lambda,U_h^\lambda\bigr]&=0\qquad(h\in H),
 \label{eq:tb_smoothing_commutation_summary}\\
 |f(g)|&=q_\delta(g),\qquad q_\delta(g^{-1})=q_\delta(g)\qquad(g\in G).
 \label{eq:tb_smoothing_weight_properties}
}
We first prove the properties of the weight function in \eqref{eq:tb_smoothing_weight_properties}. In local coordinates, inversion and conjugation are expressed as
\eq{
 g(x)^{-1}=g(-x),\qquad kg(x)k^{-1}=g(O_kx)\quad(k\in G)
 \label{eq:tb_smoothing_ball_symmetry}
}
Since $O_k$ is orthogonal with respect to the $\operatorname{Ad}_G$-invariant inner product, $\|-x\|=\|O_kx\|=\|x\|$.
Therefore, $g\in B_G(\delta)$ is equivalent to each of $g^{-1}\in B_G(\delta)$ and $kgk^{-1}\in B_G(\delta)$. This is the invariance of the small ball under inversion and conjugation, and \eqref{eq:tb_smoothing_density} gives
\eq{
 q_\delta(g^{-1})=q_\delta(g),\qquad q_\delta(kgk^{-1})=q_\delta(g)
 \label{eq:tb_smoothing_density_symmetry}
}

Inversion sends $x$ to $-x$ and complex conjugates the phase factor $e^{i\bm b\cdot x}$ in $f$. Together with $f=0$ outside the small ball, this implies that, for every $g\in G$,
\eq{
 f(g^{-1})=f(g)^*,\qquad |f(g)|=q_\delta(g)
 \label{eq:tb_smoothing_function_inverse}
}
Here, $*$ denotes complex conjugation. Together with \eqref{eq:tb_smoothing_density_symmetry}, this proves \eqref{eq:tb_smoothing_weight_properties}.

Next, we prove \eqref{eq:tb_smoothing_operator_properties}.
Since normalized Haar measure on a compact group is also invariant under inversion, the change of variables $k=g^{-1}$ and the identity $U_{k^{-1}}^\lambda=(U_k^\lambda)^\dagger$ give
\eq{
 T_\lambda^\dagger
 &=\int_G f(g)^*U_g^\lambda\,dg\nonumber\\
 &=\int_G f(k^{-1})^*(U_k^\lambda)^\dagger\,dk=T_\lambda
 \label{eq:tb_smoothing_self_adjoint}
}
This proves the self-adjointness in \eqref{eq:tb_smoothing_operator_properties}. Moreover, a unitary operator has operator norm $1$, and $q_\delta$ is a probability density. Hence, the triangle inequality and $|f(g)|=q_\delta(g)$ in \eqref{eq:tb_smoothing_weight_properties} give
\eq{
 \|T_\lambda\|_\infty
 \le\int_G|f(g)|\,\|(U_g^\lambda)^\dagger\|_\infty\,dg
 =\int_Gq_\delta(g)\,dg=1
 \label{eq:tb_smoothing_contraction}
}
which also proves the norm bound in \eqref{eq:tb_smoothing_operator_properties}.

Finally, we prove the commutation relation \eqref{eq:tb_smoothing_commutation_summary}. For $h\in H$, the assumption $O_h\bm b=\bm b$ and the orthogonality of $O_h$ also imply $O_h^{\mathsf T}\bm b=\bm b$.
Thus, inside the small ball,
\eq{
 f(hg(x)h^{-1})
 &=e^{i\bm b\cdot O_hx}q_\delta(g(x))\nonumber\\
 &=e^{i\bm b\cdot x}q_\delta(g(x))=f(g(x))
 \label{eq:tb_smoothing_function_conjugation}
}
Since both sides vanish outside the small ball, $f(hgh^{-1})=f(g)$ holds for every $g\in G$.
Using the conjugation invariance of Haar measure and changing variables to $k=hgh^{-1}$, we obtain
\eq{
 U_h^\lambda T_\lambda(U_h^\lambda)^\dagger
 &=\int_G f(g)(U_{hgh^{-1}}^\lambda)^\dagger\,dg\nonumber\\
 &=\int_G f(h^{-1}kh)(U_k^\lambda)^\dagger\,dk=T_\lambda
 \label{eq:tb_smoothing_commutation}
}
which proves \eqref{eq:tb_smoothing_commutation_summary}.
Furthermore, \eqref{eq:tb_smoothing_operator_properties} and the direct-sum expression for $T$ give $T=T^\dagger$ and $\|T\|_\infty\le1$, so $0\preceq R=T^2\preceq I$ in \eqref{eq:tb_fourier_positive_operator} indeed holds.

We now prove \eqref{eq:tb_fourier_strategy_smoothed}.
To this end, we first show that
\eq{
\|T(A-A')T\|_1\le \eta_\delta
\label{eq:tb_fourier_intermediate_target}
}
This estimate, together with $R=T^2$ and $\|T\|_\infty\le1$, immediately yields \eqref{eq:tb_fourier_strategy_smoothed} as follows:
\eq{
\|R(A-A')R\|_1
&=\|T[T(A-A')T]T\|_1\nonumber\\
&\le\|T\|_\infty^2\|T(A-A')T\|_1\nonumber\\
&\le\|T(A-A')T\|_1\le\eta_\delta.
}

To prove \eqref{eq:tb_fourier_intermediate_target}, it suffices to construct a positive operator $W$ on $\calR$ satisfying the following two properties:
\eq{
\|T(A-A')T\|_1&\le\|W\|_\infty\int_G\|F(g)\|_1dg,
\label{eq:tb_fourier_W_reconstruction_target}\\
\|W\|_\infty&\le C_{G,H}\delta^{-s}.
\label{eq:tb_fourier_W_norm_target}
}
We construct $W$ explicitly and first prove \eqref{eq:tb_fourier_W_reconstruction_target}. Set $\Delta^\lambda:=A^\lambda-A'^\lambda$ and define the orthogonal projection onto the invariant vectors by
$Q_\lambda:=\int_H U_h^\lambda\otimes V_h\,dh$.
Here, $dh$ is the normalized Haar measure on $H$.
Integrating the support condition \eqref{eq:app_strong_iid_fourier_support} over $H$ gives
\eq{
 Q_\lambda A^\lambda=A^\lambda,\qquad Q_\lambda A'^\lambda=A'^\lambda
 \label{eq:tb_fourier_supported_blocks}
}
The invariance of Haar measure under inversion implies $Q_\lambda^\dagger=Q_\lambda$. Since $A^\lambda,A'^\lambda$ are also self-adjoint, taking adjoints in \eqref{eq:tb_fourier_supported_blocks} gives $A^\lambda Q_\lambda=A^\lambda$ and $A'^\lambda Q_\lambda=A'^\lambda$ as well.
Therefore,
\eq{
 Q_\lambda\Delta^\lambda Q_\lambda
 =Q_\lambda A^\lambda Q_\lambda-Q_\lambda A'^\lambda Q_\lambda
 =A^\lambda-A'^\lambda=\Delta^\lambda
 \label{eq:tb_fourier_projected_difference}
}

Define $S_\lambda$ by
\eq{
S_\lambda:=(T_\lambda\otimes I)Q_\lambda
\label{eq:tb_fourier_supported_operator}
}
Combining the definition \eqref{eq:tb_fourier_supported_operator} with \eqref{eq:tb_fourier_projected_difference} gives
\eq{
 (T_\lambda\otimes I)\Delta^\lambda(T_\lambda\otimes I)
 &=(T_\lambda\otimes I)Q_\lambda\Delta^\lambda Q_\lambda(T_\lambda\otimes I)\nonumber\\
 &=\bigl[(T_\lambda\otimes I)Q_\lambda\bigr]\Delta^\lambda
   \bigl[(T_\lambda\otimes I)Q_\lambda\bigr]^\dagger\nonumber\\
 &=S_\lambda\Delta^\lambda S_\lambda^\dagger
 \label{eq:tb_fourier_smoothed_block}
}

Next, we recover each block from the Fourier transform.
Substituting the component expansion \eqref{eq:app_strong_iid_fourier_block_example} into \eqref{eq:app_strong_iid_fourier_function} gives
\eq{
 F(g)
 &=\sum_{\mu,c,d}\Tr_{\cal V_\mu}\!\left[(U_g^\mu)^\dagger\ket c\bra d\right]\Delta^\mu_{cd}\nonumber\\
 &=\sum_{\mu,c,d}\overline{(U_g^\mu)_{cd}}\Delta^\mu_{cd}
 \label{eq:tb_fourier_component_expansion}
}
For labels $\lambda,\mu$ denoting equivalence classes of irreducible representations, the Schur orthogonality relations \eqref{eq:orthogonality1} and \eqref{eq:orthogonality2} together read
\eq{
 \int_G(U_g^\lambda)_{ab}\overline{(U_g^\mu)_{cd}}\,dg
 =\frac{\delta_{\lambda\mu}\delta_{ac}\delta_{bd}}{d_\lambda}
 \label{eq:tb_fourier_matrix_orthogonality}
}
Multiplying \eqref{eq:tb_fourier_component_expansion} by $(U_g^\lambda)_{ab}$, integrating, and applying \eqref{eq:tb_fourier_matrix_orthogonality}, we obtain
\eq{
 d_\lambda\int_G(U_g^\lambda)_{ab}F(g)\,dg
 &=d_\lambda\sum_{\mu,c,d}
 \left[\int_G(U_g^\lambda)_{ab}\overline{(U_g^\mu)_{cd}}\,dg\right]\Delta^\mu_{cd}\nonumber\\
 &=\sum_{\mu,c,d}\delta_{\lambda\mu}\delta_{ac}\delta_{bd}\Delta^\mu_{cd}
 =\Delta^\lambda_{ab}
 \label{eq:tb_fourier_component_recovery}
}
Taking the tensor product of both sides with $\ket a\bra b$ on the left and summing over $a,b$ gives
\eq{
 \Delta^\lambda=d_\lambda\int_G U_g^\lambda\otimes F(g)\,dg
 \label{eq:app_strong_iid_operator_fourier_inverse}
}

For $\|Z_\lambda\|_\infty\le1$, the cyclicity of the trace then gives
\eq{
 \sum_\lambda\Tr[Z_\lambda S_\lambda\Delta^\lambda S_\lambda^\dagger]
 &=\int_G\Tr[B_gF(g)]\,dg,\nonumber\\
 B_g&:=\sum_\lambda d_\lambda\Tr_{\cal V_\lambda}
 [S_\lambda^\dagger Z_\lambda S_\lambda(U_g^\lambda\otimes I)].
 \label{eq:tb_fourier_test_identity}
}
The desired operator $W$ is defined by
\eq{
W:=\sum_\lambda d_\lambda\Tr_{\cal V_\lambda}(S_\lambda^\dagger S_\lambda)
\label{eq:tb_fourier_control_operator}
}
Each $S_\lambda^\dagger S_\lambda$ is positive, and the partial trace preserves positivity, so $W\succeq0$.
Since $\{|a\rangle\}$ and $\{U_g^\lambda|a\rangle\}$ are both orthonormal bases,
\eq{
 |\langle v,B_gw\rangle|
 &\le\left(\sum_{\lambda,a}d_\lambda\|S_\lambda(a\otimes v)\|^2\right)^{1/2}
       \left(\sum_{\lambda,a}d_\lambda\|S_\lambda(U_g^\lambda a\otimes w)\|^2\right)^{1/2}\nonumber\\
 &=\langle v,Wv\rangle^{1/2}\langle w,Ww\rangle^{1/2}.
 \label{eq:tb_fourier_test_matrix_bound}
}
Taking the supremum in \eqref{eq:tb_fourier_test_matrix_bound} over $\|v\|=\|w\|=1$ gives $\|B_g\|_\infty\le\|W\|_\infty$.
Although $B_g$ depends on the choice of $\{Z_\lambda\}$, this bound holds for every $g\in G$ and every choice satisfying $\|Z_\lambda\|_\infty\le1$.

We use the following relations, valid for any operator $C$:
\eq{
 \|C\|_1=\sup_{\|Z\|_\infty\le1}|\Tr[ZC]|,\qquad
 |\Tr[ZC]|\le\|Z\|_\infty\|C\|_1
 \label{eq:tb_trace_norm_duality}
}
If $C$ is self-adjoint, choosing $Z$ to act on each eigenspace as the sign of the corresponding eigenvalue makes $\Tr[ZC]$ equal to the sum of the absolute values of the eigenvalues, attaining the supremum in the first relation.
Here, $T(A-A')T$ is the direct sum of $S_\lambda\Delta^\lambda S_\lambda^\dagger$, so the operators $Z_\lambda$ can be chosen independently in each block.
Using \eqref{eq:tb_trace_norm_duality} and \eqref{eq:tb_fourier_test_identity}, and applying the bound $\|B_g\|_\infty\le\|W\|_\infty$ obtained above, we find
\eq{
 \|T(A-A')T\|_1
 &=\sup_{\|Z_\lambda\|_\infty\le1\ (\forall\lambda)}
 \left|\sum_\lambda\Tr[Z_\lambda S_\lambda\Delta^\lambda S_\lambda^\dagger]\right|\nonumber\\
 &=\sup_{\|Z_\lambda\|_\infty\le1\ (\forall\lambda)}
 \left|\int_G\Tr[B_gF(g)]\,dg\right|\nonumber\\
 &\le\|W\|_\infty\int_G\|F(g)\|_1dg.
 \label{eq:tb_smoothed_fourier}
}
This proves \eqref{eq:tb_fourier_W_reconstruction_target}.

Next, we prove the norm bound \eqref{eq:tb_fourier_W_norm_target} for the same $W$.
By \eqref{eq:tb_smoothing_commutation_summary}, $T_\lambda\otimes I$ commutes with $Q_\lambda$, so we can write $S_\lambda=Q_\lambda(T_\lambda\otimes I)$.
Substitute the definition of $Q_\lambda$ and \eqref{eq:tb_smoothing_operator} into this expression, and change variables to $g=yh$ with $h$ fixed.
Right invariance of Haar measure gives $dg=dy$, and the multiplication law for the representation gives
\eq{
 U_h^\lambda(U_{yh}^\lambda)^\dagger
 =U_h^\lambda(U_h^\lambda)^\dagger(U_y^\lambda)^\dagger
 =(U_y^\lambda)^\dagger
 \label{eq:tb_fourier_right_translation}
}
For $v\in\cal R$, set $p_v(y):=\int_Hf(yh)V_hv\,dh$ and interchange the order of integration to obtain
\eq{
 S_\lambda(|a\rangle\otimes v)
 &=\int_Hdh\int_Gdg\,f(g)
   \bigl(U_h^\lambda(U_g^\lambda)^\dagger|a\rangle\bigr)\otimes V_hv\nonumber\\
 &=\int_Gdy\,\bigl((U_y^\lambda)^\dagger|a\rangle\bigr)
   \otimes\left[\int_Hf(yh)V_hv\,dh\right]\nonumber\\
 &=\sum_b|b\rangle\otimes\int_G\overline{(U_y^\lambda)_{ab}}p_v(y)\,dy
 \label{eq:tb_fourier_supported_vector}
}
In the last equality, we used the expansion $(U_y^\lambda)^\dagger|a\rangle=\sum_b\overline{(U_y^\lambda)_{ab}}|b\rangle$.

By \eqref{eq:tb_fourier_matrix_orthogonality}, the functions $\sqrt{d_\lambda}(U_y^\lambda)_{ab}$ form an orthonormal system.
Resolve $p_v(y)$ into components in an orthonormal basis of $\cal R$, apply Bessel's inequality for this orthonormal system to each component, and sum.
Using the definition of $W$ in \eqref{eq:tb_fourier_control_operator} and \eqref{eq:tb_fourier_supported_vector}, we obtain
\eq{
 \langle v,Wv\rangle
 &=\sum_{\lambda,a}d_\lambda\|S_\lambda(|a\rangle\otimes v)\|^2\nonumber\\
 &=\sum_{\lambda,a,b}d_\lambda
   \left\|\int_G\overline{(U_y^\lambda)_{ab}}p_v(y)\,dy\right\|^2\nonumber\\
 &\le\int_G\|p_v(y)\|^2dy
 \label{eq:tb_fourier_bessel_bound}
}
Here, the second equality uses the orthonormality of $\{|b\rangle\}$, and the last inequality is the application of Bessel's inequality described above.

Further, set $r_\delta(y):=\int_Hq_\delta(yh)\,dh$.
The triangle inequality, $|f(g)|=q_\delta(g)$ in \eqref{eq:tb_smoothing_weight_properties}, and the unitarity of $V_h$ give, for each $y\in G$,
\eq{
 \|p_v(y)\|
 &\le\int_H|f(yh)|\,\|V_hv\|\,dh\nonumber\\
 &=\|v\|\int_Hq_\delta(yh)\,dh=\|v\|r_\delta(y)
 \label{eq:tb_fourier_reference_vector_bound}
}
Squaring this inequality, integrating over $G$, and combining it with \eqref{eq:tb_fourier_bessel_bound}, we obtain
\eq{
 \langle v,Wv\rangle
 \le\int_G\|p_v(y)\|^2dy
 \le\|v\|^2\int_G r_\delta(y)^2dy
 \label{eq:tb_fourier_W_density_bound}
}
Moreover, the right invariance of Haar measure and the normalization of $q_\delta$ and $dh$ imply
\eq{
 \int_G r_\delta(y)\,dy
 =\int_Hdh\int_Gq_\delta(yh)\,dy
 =\int_H1\,dh=1
 \label{eq:tb_fourier_coset_density_normalization}
}

To bound $r_\delta$, we use the local descriptions \eqref{eq:tb_local_ball} and \eqref{eq:app_strong_iid_subgroup_coordinates} to examine the volume of the small ball itself and of its intersection with the subgroup $H$.
Haar measure has a positive smooth density in local coordinates~\cite[p.~45]{Folland2015}.
More precisely, there exist smooth positive real-valued functions $J_G(x)$ and $J_H(x)$ in the respective local coordinates of $G$ and $H$.
Take coordinates on a slightly larger neighborhood containing the closed ball $\|x\|\le r_0$, decreasing $r_0$ if necessary. The continuity of the densities and compactness of the closed ball then give
\eq{
 dg&=J_G(x)\,dx,\qquad 0<c_1\le J_G(x)\le C_1,
 \label{eq:app_strong_iid_local_haar_density_G}\\
 dh&=J_H(x)\,dx\big|_{\mathfrak h},\qquad 0<J_H(x)\le C_2
 \label{eq:app_strong_iid_local_haar_densities}
}

The integral with respect to $dx\big|_{\mathfrak h}$ in the second relation is taken in coordinates identifying $\mathfrak h$ with $\mathbb R^{d_H}$ through an orthonormal basis for the fixed inner product.
Thus, denoting the normalized Haar measure on $H$ by $|\cdot|_H$, we have
\eq{
 |B_G(t)|&=\int_{\|x\|<t}J_G(x)\,dx,\\
 |H\cap B_G(t)|_H
 &=\int_{\substack{x\in\mathfrak h\\\|x\|<t}}J_H(x)\,dx\big|_{\mathfrak h}
 \label{eq:tb_local_volume_integrals}
}
Changing variables to $x=ty$ in the first relation multiplies the volume element by $t^{d_G}$, and the upper and lower bounds on $J_G$ bound $|B_G(t)|$ above and below by positive constant multiples of $t^{d_G}$.
In the second relation, the same change of variables multiplies the volume element by $t^{d_H}$, and the upper bound on $J_H$ bounds $|H\cap B_G(t)|_H$ above by a constant multiple of $t^{d_H}$.
If $d_H=0$, then $H$ is finite, and $H\cap B_G(t)=\{e\}$ in a sufficiently small neighborhood.
In this case, the second relation is interpreted as a zero-dimensional integral giving the mass $1/|H|$ of the identity.
These local descriptions remain valid when $G,H$ are disconnected, provided that the neighborhood of the identity is chosen sufficiently small.

We also bound the size of the coordinates after multiplication.
For $x,y$ sufficiently close to the origin, define the real vector-valued function $\bm\Phi(x,y)$ by
\eq{
 g(\bm\Phi(x,y))=g(x)^{-1}g(y)
 \label{eq:app_strong_iid_local_product_coordinates}
}
Since group multiplication, inversion, and the inverse of the exponential coordinate map are smooth, $\bm\Phi$ is smooth near the origin and satisfies $\bm\Phi(0,0)=0$.
On a sufficiently small fixed neighborhood, bound the norms of the two partial derivatives with respect to $x,y$ by a constant $C_{\mathrm{loc}}\ge1$. Integrating the derivative along a line segment then gives
\eq{
 \|\bm\Phi(x,y)\|
 &=\left\|\int_0^1\frac{d}{dt}\bm\Phi(tx,ty)\,dt\right\|\nonumber\\
 &\le C_{\mathrm{loc}}(\|x\|+\|y\|)
 \label{eq:app_strong_iid_local_product_bound}
}

Decrease $\delta_0$ further if necessary so that this estimate holds for $\|x\|,\|y\|\le\delta_0$ and $2C_{\mathrm{loc}}\delta_0<r_0$.
Then, for $0<\delta\le\delta_0$, the condition $\|x\|,\|y\|<\delta$ implies $\|\bm\Phi(x,y)\|<2C_{\mathrm{loc}}\delta$, and hence
\eq{
 B_G(\delta)^{-1}B_G(\delta)
 \subseteq B_G(2C_{\mathrm{loc}}\delta)
 \label{eq:app_strong_iid_local_product_support}
}

Applying the preceding volume estimates with $t=\delta$ and $t=2C_{\mathrm{loc}}\delta$, and absorbing the fixed constant $(2C_{\mathrm{loc}})^{d_H}$ into the upper-bound constant, we obtain, for constants $c,C,C_{G,H}>0$ independent of $\delta$,
\eq{
 c\delta^{d_G}\le |B_G(\delta)|\le C\delta^{d_G},\qquad
 |H\cap B_G(2C_{\mathrm{loc}}\delta)|_H\le C_{G,H}\delta^{d_H}
 \label{eq:tb_local_volumes}
}

Substituting the definition of $q_\delta$ in \eqref{eq:tb_smoothing_density} into $r_\delta$ gives
\eq{
 r_\delta(y)
 =\frac{|\{h\in H:yh\in B_G(\delta)\}|_H}{|B_G(\delta)|}
 \label{eq:tb_coset_density_ratio}
}
The numerator is the normalized Haar measure on $H$ of the set of $h$ for which $yh$ lies in the small ball.
If this set is empty, then $r_\delta(y)=0$. Otherwise, fix one $h_0\in H$ such that $yh_0\in B_G(\delta)$.
For every $h\in H$ satisfying $yh\in B_G(\delta)$,
\eq{
 h_0^{-1}h=(yh_0)^{-1}(yh)
 \in H\cap B_G(\delta)^{-1}B_G(\delta)
 \label{eq:tb_coset_intersection_translation}
}
Since $h\mapsto h_0^{-1}h$ preserves Haar measure on $H$, \eqref{eq:app_strong_iid_local_product_support} and \eqref{eq:tb_local_volumes} give
\eq{
 r_\delta(y)
 &\le\frac{|H\cap B_G(2C_{\mathrm{loc}}\delta)|_H}{|B_G(\delta)|}\nonumber\\
 &\le C_{G,H}\delta^{d_H-d_G}=C_{G,H}\delta^{-s}
 \label{eq:tb_coset_density_bound}
}
Here, $C_{G,H}$ has been increased if necessary to include the constants in the upper and lower volume bounds.
Moreover, since $W\succeq0$, taking the supremum over $\|v\|=1$ in \eqref{eq:tb_fourier_W_density_bound} gives $\|W\|_\infty\le\int_G r_\delta(y)^2dy$.
Using $r_\delta\ge0$ and the normalization \eqref{eq:tb_fourier_coset_density_normalization}, we obtain
\eq{
 \|W\|_\infty
 &\le\int_G r_\delta(y)^2dy
 \le\|r_\delta\|_\infty\int_G r_\delta(y)dy\nonumber\\
 &=\|r_\delta\|_\infty\le C_{G,H}\delta^{-s}
 \label{eq:tb_coset_bound}
}
This also proves \eqref{eq:tb_fourier_W_norm_target}.
Combining \eqref{eq:tb_fourier_W_reconstruction_target} and \eqref{eq:tb_fourier_W_norm_target}, and using the definition of $\eta_\delta$ in \eqref{eq:tb_one_sided_eta}, gives \eqref{eq:tb_fourier_intermediate_target}.
Furthermore, since $R(A-A')R=T[T(A-A')T]T$, the bound $\|T\|_\infty\le1$ following from \eqref{eq:tb_smoothing_operator_properties} and \eqref{eq:tb_smoothed_fourier} give
\eq{
\|R(A-A')R\|_1\le\|T(A-A')T\|_1\le\eta_\delta
\label{eq:tb_one_sided_smoothed_difference}
}
which proves \eqref{eq:tb_fourier_strategy_smoothed}.

Next, to prove \eqref{eq:tb_fourier_strategy_input_recovery} and \eqref{eq:tb_fourier_strategy_output_recovery}, we use the second moment of $A$ to bound the differences from the unaveraged operators.
Set $Y_x:=\sum_jx_j(X_{\lambda,j}-b_jI)$. For any $\psi\in\cal V_\lambda$, applying the Cauchy--Schwarz inequality to the vectors $(X_{\lambda,j}-b_jI)\psi$ gives
\eq{
 \langle\psi,Y_x^2\psi\rangle
 &=\left\|\sum_jx_j(X_{\lambda,j}-b_jI)\psi\right\|^2\nonumber\\
 &\le\|x\|^2\sum_j\|(X_{\lambda,j}-b_jI)\psi\|^2\nonumber\\
 &=\|x\|^2\langle\psi,\mathsf C_\lambda(\bm b)\psi\rangle
 \label{eq:tb_generator_quadratic_form}
}
The last equality uses the definition \eqref{eq:app_strong_iid_quadratic_charge}. Since this holds for every $\psi$, we have the operator inequality
\eq{
Y_x^2\preceq\|x\|^2\mathsf C_\lambda(\bm b)
\label{eq:tb_generator_quadratic_bound}
}

Next, for $g=g(x)$ in the small ball, \eqref{eq:tb_smoothing_function} and the definition of the generators give
\eq{
 f(g(x))(U_{g(x)}^\lambda)^\dagger
 &=q_\delta(g(x))e^{i\bm b\cdot x}e^{-i\sum_jx_jX_{\lambda,j}}\nonumber\\
 &=q_\delta(g(x))e^{-iY_x}
 \label{eq:tb_smoothing_centered_exponential}
}
Thus, $T_\lambda$ in \eqref{eq:tb_smoothing_operator} is the average of $e^{-iY_x}$.
Using $g(x)^{-1}=g(-x)$ in \eqref{eq:tb_smoothing_ball_symmetry}, $q_\delta(g^{-1})=q_\delta(g)$ in \eqref{eq:tb_smoothing_weight_properties}, and the invariance of Haar measure under inversion, this is also the average of $e^{iY_x}$.
Taking half the sum of these two expressions gives
\eq{
 T_\lambda
 &=\int_{B_G(\delta)}q_\delta(g)\frac{e^{-iY_x}+e^{iY_x}}2\,dg\nonumber\\
 &=\int_{B_G(\delta)}q_\delta(g)\cos Y_x\,dg
 \label{eq:tb_smoothing_cosine_average}
}
Here, $x$ in the integral is the local coordinate satisfying $g=g(x)$.
Since each $Y_x$ is self-adjoint, applying $0\le1-\cos t\le t^2/2$ to each eigenvalue in its spectral decomposition gives
\eq{
 0\preceq I-\cos Y_x\preceq\tfrac12Y_x^2
 \label{eq:tb_smoothing_cosine_order}
}
Averaging using \eqref{eq:tb_smoothing_cosine_average} and the normalization of $q_\delta$, and applying \eqref{eq:tb_generator_quadratic_bound} and $\|x\|<\delta$ on the small ball, we obtain
\eq{
 2(I-T_\lambda)
 &\preceq\int_{B_G(\delta)}q_\delta(g)Y_x^2\,dg\nonumber\\
 &\preceq\delta^2\mathsf C_\lambda(\bm b)
 \label{eq:tb_smoothing_block_moment_bound}
}
On the other hand, \eqref{eq:tb_smoothing_operator_properties} implies that $T$ is self-adjoint and $\|T\|_\infty\le1$, so $I-T^2\succeq0$. Moreover,
\eq{
 2(I-T)-(I-T^2)=(I-T)^2\succeq0
 \label{eq:tb_smoothing_contraction_order}
}
Combining these facts with the blockwise estimate \eqref{eq:tb_smoothing_block_moment_bound} gives
\eq{
0\preceq I-T^2\preceq2(I-T)
\preceq\delta^2\bigoplus_\lambda\mathsf C_\lambda(\bm b)\otimes I
\label{eq:tb_one_sided_smoothing_order}
}

For $R=T^2$ defined in \eqref{eq:tb_fourier_positive_operator}, set
\eq{
u:=\Tr[A(I-R)]\le\delta^2 E_A\le\delta^2 E
\label{eq:tb_one_sided_positive_contraction}
}
The positive operator $I-R$ then allows us to bound the loss of weight due to averaging.
Using the normalization of $A$, we can also bound the weight lost from $A'$ as follows:
\eq{
\Tr[A'(I-R^2)]
&=\Tr A'-\Tr RA'R\nonumber\\
&\le1-\Tr RAR+\Tr[R(A-A')R]\nonumber\\
&\le\Tr[A(I-R^2)]+\eta_\delta\nonumber\\
&\le2\Tr[A(I-R)]+\eta_\delta=2u+\eta_\delta.
\label{eq:tb_one_sided_discarded_mass}
}
The last inequality uses $I-R^2\preceq2(I-R)$, which follows from $0\preceq R\preceq I$ in \eqref{eq:tb_fourier_positive_operator}.
This estimate bounds $A'-RA'R$ without using the second moment of $A'$.
Indeed, using the square root of $A\succeq0$, we factor the terms in $A-RAR=(I-R)A+RA(I-R)$ as
\eq{
 (I-R)A&=[(I-R)\sqrt A]\sqrt A,
 \label{eq:tb_recovery_first_factorization}\\
 RA(I-R)&=(R\sqrt A)[\sqrt A(I-R)]
 \label{eq:tb_recovery_factorization}
}
Applying $\|BC\|_1\le\|B\|_2\|C\|_2$ to each product, the definition $\|B\|_2^2=\Tr[B^\dagger B]$ of the Hilbert--Schmidt norm and the cyclicity of the trace give
\eq{
 \|(I-R)A\|_1
 &\le\|(I-R)\sqrt A\|_2\|\sqrt A\|_2\nonumber\\
 &=\sqrt{\Tr[A(I-R)^2]}\sqrt{\Tr A},
 \label{eq:tb_recovery_first_term}\\
 \|RA(I-R)\|_1
 &\le\|R\sqrt A\|_2\|\sqrt A(I-R)\|_2\nonumber\\
 &\le\sqrt{\Tr A}\sqrt{\Tr[A(I-R)^2]}
 \label{eq:tb_recovery_second_term}
}
The last inequality uses $\|R\sqrt A\|_2\le\|\sqrt A\|_2$, which follows from $0\preceq R\preceq I$ in \eqref{eq:tb_fourier_positive_operator}.
Adding these two terms using the triangle inequality yields the same upper bound twice. The same factorization and estimates also hold for $A'$.
Furthermore, using $(I-R)^2\preceq I-R\preceq I-R^2$, which follows from $0\preceq R\preceq I$, and applying \eqref{eq:tb_one_sided_positive_contraction}, \eqref{eq:tb_one_sided_discarded_mass}, and $\Tr A'\le1$, we obtain
\eq{
\|A-RAR\|_1
&\le2\sqrt{\Tr A}\sqrt{\Tr[A(I-R)^2]}
\le2\sqrt u,\\
\|A'-RA'R\|_1
&\le2\sqrt{\Tr A'}\sqrt{\Tr[A'(I-R)^2]}\nonumber\\
&\le2\sqrt{\Tr[A'(I-R^2)]}
\le2\sqrt{2u+\eta_\delta}.
\label{eq:tb_one_sided_gentle_bounds}
}
Substituting $u\le\delta^2E$ here gives \eqref{eq:tb_fourier_strategy_input_recovery} and \eqref{eq:tb_fourier_strategy_output_recovery}.
Therefore, the triangle inequality gives
\eq{
\|A-A'\|_1
&\le\|A-RAR\|_1+\|R(A-A')R\|_1+\|RA'R-A'\|_1\nonumber\\
&\le\|R(A-A')R\|_1+2\sqrt u+2\sqrt{2u+\eta_\delta}\nonumber\\
&\le\|R(A-A')R\|_1+(2+2\sqrt2)\delta\sqrt E+2\sqrt{\eta_\delta}
\label{eq:tb_fourier_strategy_recovery}
}
Combining this estimate with \eqref{eq:tb_fourier_strategy_smoothed} proves \eqref{eq:tb_one_sided_fourier}.
If $G$ is finite, set $\delta_0=1$ and take $q_\delta=f=|G|\mathbf1_{\{e\}}$.
Then $T=R=I$ and $r_\delta=|G|\mathbf1_H/|H|$, so the preceding estimate using the Schur orthogonality relations and Bessel's inequality gives
\eq{
\|A-A'\|_1\le\frac{|G|}{|H|}\int_G\|F(g)\|_1dg
}
Thus, \eqref{eq:tb_one_sided_fourier} holds with $s=0$ and $C_{G,H}=|G|/|H|$.
\end{proofof}

The following lemma constructs the map on the reference system used in the sufficiency proof.

We now combine the comparison of local state families, the treatment of the group action on the environment, and the Fourier estimates over the entire group.

\begin{lemma}\label{lem:app_strong_iid_reference_implementation}
Let $H:=\mathrm{Sym}_G(\rho)\subseteq\mathrm{Sym}_G(\sigma)$ and $\mathrm{Sym}_{G,\mathrm{strong}}(\rho)=\mathrm{Sym}_{G,\mathrm{strong}}(\sigma)=:S^{\mathrm{strong}}_{G,\rho,\sigma}$, and assume that
\eq{
\VNS(\rho)=r\VNS(\sigma),\qquad
\calQ^{\rho,q}\ge r\calQ^{\sigma,q}\quad(q\in[1/2,1]).
}

Furthermore, suppose that a sequence of characters $\gamma_N$ on a given infinite subset $\mathcal N\subseteq\mathbb N$ satisfies \eqref{eq:app_strong_iid_auxiliary_character} and \eqref{eq:app_strong_iid_bounded_mean} for all sufficiently large $N\in\mathcal N$.
Then, for each sufficiently large $N\in\mathcal N$, we can choose a trace-nonincreasing completely positive map
\eq{
\calD_N:\calB(E_\sigma^{\otimes M_N})\longrightarrow\calB(E_\rho^{\otimes N}),
\qquad M_N=\lfloor rN\rfloor,
}
such that, for the same set $\mathcal N$ and the same sequence of characters $\gamma_N$,
\eq{
\bigl\|\Omega_{\rho,N}-(\mathrm{id}\otimes\calD_N)
\Omega_{\sigma,M_N}^{\gamma_N}\bigr\|_1
\longrightarrow0\qquad(N\to\infty,\ N\in\mathcal N).
\label{eq:app_strong_iid_reference_Omega_convergence}
}

\end{lemma}

\begin{proofof}{Lemma~\ref{lem:app_strong_iid_reference_implementation}}
\begingroup
The basic strategy of this proof is to construct trace-nonincreasing completely positive maps
\eq{
 \mathcal D_N^\varepsilon:\calB(E_\sigma^{\otimes M_N})
 \longrightarrow\calB(E_\rho^{\otimes N}),\qquad M_N=\lfloor rN\rfloor
 \label{eq:tb_reference_strategy_map}
}
that satisfy the following conditions for each fixed $0<\varepsilon<1/2$ and all sufficiently large $N\in\mathcal N$.
The first condition is that each Kraus operator $D_{N,\ell}^\varepsilon$ defining $\calD_N^\varepsilon$ satisfies
\eq{
 D_{N,\ell}^\varepsilon\overline{U'_{h,E_{\sigma}}}^{\otimes M_N}
 =\gamma_N(h)\overline{U_{h,E_{\rho}}}^{\otimes N}D_{N,\ell}^\varepsilon
 \qquad(h\in H).
 \label{eq:tb_reference_strategy_intertwiner}
}
Here, we define $\overline{U'_{h,E_{\sigma}}}$ and $\overline{U_{h,E_{\rho}}}$ as the conjugate representations on the supports of the reference systems:
\eq{
 \overline{U'_{h,E_{\sigma}}}
 &:=\left.\overline{U'_h}\right|_{E_\sigma}
 :E_\sigma\longrightarrow E_\sigma,
 \label{eq:tb_reference_restriction_sigma}\\
 \overline{U_{h,E_{\rho}}}
 &:=\left.\overline{U_h}\right|_{E_\rho}
 :E_\rho\longrightarrow E_\rho.
 \label{eq:tb_reference_restriction_rho}
}
Since $h\in H$ preserves both states, the conjugate representations preserve their respective supports.
Thus, for example, $\overline{U'_{h,E_{\sigma}}}v=\overline{U'_h}v\in E_\sigma$ for $v\in E_\sigma$, so we use the original action directly on the support. The same applies to $\rho$.
The second condition concerns the difference between the Fourier transforms of $\Omega_{\rho,N}$ and $(\mathrm{id}\otimes\mathcal D_N^\varepsilon)\Omega_{\sigma,M_N}^{\gamma_N}$, which, by \eqref{eq:app_strong_iid_Omega_Fourier}, is given by
\eq{
 F_N^\varepsilon(g):=B_\rho(g)^{\otimes N}
 -\overline{\gamma_N(g)}\mathcal D_N^\varepsilon(B_\sigma(g)^{\otimes M_N}).
 \label{eq:app_strong_iid_F}
}
We require this difference to satisfy
\eq{
 \limsup_{\substack{N\to\infty\\N\in\mathcal N}}
 N^{s/2}\int_G\|F_N^\varepsilon(g)\|_1\,dg
 \le C_1\sqrt\varepsilon
 \label{eq:tb_reference_strategy_integral}
}
for a constant $C_1$ independent of $\varepsilon$.
Here, $s:=\dim G-\dim H$.

We first show that constructing maps satisfying these conditions yields the desired $\mathcal D_N$.
Let $\widetilde\Omega_N^\varepsilon:=(\mathrm{id}\otimes\mathcal D_N^\varepsilon)
\Omega_{\sigma,M_N}^{\gamma_N}$. Then, by \eqref{eq:app_strong_iid_Omega_Fourier},
the Fourier transform of $\Omega_{\rho,N}-\widetilde\Omega_N^\varepsilon$ is $F_N^\varepsilon$.
The support condition in Lemma~\ref{lem:app_strong_iid_fourier_bound} also holds.
Indeed, the standard purification of $\rho$ is invariant under $U_h^{\otimes N}\otimes\overline{U_{h,E_{\rho}}}^{\otimes N}$.
On the other hand, rewriting \eqref{eq:tb_reference_strategy_intertwiner} gives
\eq{
 \overline{U_{h,E_{\rho}}}^{\otimes N}D_{N,\ell}^\varepsilon
 =\overline{\gamma_N(h)}D_{N,\ell}^\varepsilon\overline{U'_{h,E_{\sigma}}}^{\otimes M_N}.
 \label{eq:tb_reference_strategy_phase_cancellation}
}
Hence, in the state obtained by applying each reference Kraus operator to the purification on the output side, $\gamma_N(h)$ on the physical system cancels $\overline{\gamma_N(h)}$ on the reference system.
This invariant support is preserved when we take irreducible components and trace out the multiplicity spaces.
Thus, we can use the common reference representation $\overline{U_{h,E_{\rho}}}^{\otimes N}$ for both states.

We now apply Lemma~\ref{lem:app_strong_iid_fourier_bound} with $A=\Omega_{\rho,N}$ and $A'=\widetilde\Omega_N^\varepsilon$.
The state $A$ is normalized, and $\Tr A'\le1$ because $\mathcal D_N^\varepsilon$ is trace-nonincreasing.
The center $\bm b_N=N\bm\mu_\rho$ is $H$-invariant, and the second moment of $A$ is
\eq{
 \sum_\lambda\Tr[(\mathsf C_\lambda(\bm b_N)\otimes I)\Omega_{\rho,N}^\lambda]
 =N\Tr \VNS(\rho)=:C_0N.
 \label{eq:tb_reference_strategy_moment}
}
This equality uses the fact that the second moments of the generators on the physical system are preserved when the multiplicity spaces are traced out after the irreducible decomposition.
Moreover, in the product state $\rho^{\otimes N}$, the expectation of a product of centered generators acting on different copies vanishes, so the variance is $N$ times its single-copy value.
We therefore set $E=C_0N$, fix $0<a\le\delta_0$, and choose $\delta=a/\sqrt N$.
The Fourier error in \eqref{eq:tb_one_sided_eta} is then
\eq{
 \eta_{a/\sqrt N}
 =C_{G,H}a^{-s}N^{s/2}\int_G\|F_N^\varepsilon(g)\|_1\,dg.
 \label{eq:tb_reference_strategy_scaled_eta}
}
Thus, \eqref{eq:tb_one_sided_fourier} and \eqref{eq:tb_reference_strategy_integral} give, with $C_2:=C_{G,H}C_1$,
\eq{
 \limsup_{\substack{N\to\infty\\N\in\mathcal N}}\|\Omega_{\rho,N}-\widetilde\Omega_N^\varepsilon\|_1
 \le(2+2\sqrt2)a\sqrt{C_0}
 +2\sqrt{C_2}\,a^{-s/2}\varepsilon^{1/4}
 +C_2a^{-s}\sqrt\varepsilon.
 \label{eq:tb_final_master}
}
Taking $a_*:=\min\{\delta_0,1\}$, $a_j=a_*/j$, and $\varepsilon_j=1/(4j^{2(s+2)})$, we can bound the right-hand side by $C_3/j$ for a constant $C_3$ independent of $j$.
We can therefore choose strictly increasing integers $N_j$ such that, whenever $N\in\mathcal N$ and $N\ge N_j$, the map $\mathcal D_N^{\varepsilon_j}$ is defined and
\eq{
 \|\Omega_{\rho,N}-\widetilde\Omega_N^{\varepsilon_j}\|_1
 \le\frac{C_3+1}{j}.
 \label{eq:tb_reference_strategy_threshold}
}

We then define
\eq{
 \mathcal D_N:=\mathcal D_N^{\varepsilon_j}
 \qquad(N_j\le N<N_{j+1},\ N\in\mathcal N).
 \label{eq:tb_reference_strategy_diagonal}
}
As $N\in\mathcal N$ tends to infinity, the index $j$ satisfying $N_j\le N<N_{j+1}$ also tends to infinity. Hence,
\eq{
 \bigl\|\Omega_{\rho,N}-(\mathrm{id}\otimes\mathcal D_N)
 \Omega_{\sigma,M_N}^{\gamma_N}\bigr\|_1
 \le\frac{C_3+1}{j}\longrightarrow0
 \label{eq:tb_reference_strategy_convergence}
}
which proves \eqref{eq:app_strong_iid_reference_Omega_convergence}.
This selects a map for every sufficiently large $N$ in the fixed set $\mathcal N$, without changing either the subsequence or the sequence of characters.

We now construct maps $\mathcal D_N^\varepsilon$ satisfying \eqref{eq:tb_reference_strategy_intertwiner} and \eqref{eq:tb_reference_strategy_integral}.
We first define the maps and then verify the relation between their Kraus operators in \eqref{eq:tb_reference_strategy_intertwiner}.
We subsequently prove the integral estimate \eqref{eq:tb_reference_strategy_integral}.
In the construction below, we denote the total numbers of copies by $n\in\mathcal N$ and $m_n=\lfloor rn\rfloor$.
We also set $K:=S^{\mathrm{strong}}_{G,\rho,\sigma}$, $\mathfrak h:=\operatorname{Lie}(H)$, and $\mathfrak k:=\operatorname{Lie}(K)$.
We use the common notation
\eq{
 (\tau^\circ,U^\circ,\{X_i^\circ\})
 \in\{(\rho,U,\{X_i\}),(\sigma,U',\{X'_i\})\},
 \label{eq:app_strong_iid_common_notation}\\
 X^\circ(x):=\sum_i x_iX_i^\circ
 \label{eq:app_strong_iid_common_generator}
}
and write $\tau=\rho,\sigma$ for short within the proof.
In this common notation, in accordance with \eqref{eq:tb_reference_restriction_sigma} and \eqref{eq:tb_reference_restriction_rho}, we also write
\eq{
 \overline{U^\circ_{h,E_{\tau}}}
 :=\left.\overline{U_h^\circ}\right|_{E_\tau}
 \qquad(h\in H).
 \label{eq:tb_reference_restriction_common}
}
When specifying the representation itself in a subscript, as in the compression and recovery maps, we use $\overline{U^\circ_{E_{\tau}}}$ to denote the representation $h\mapsto\overline{U^\circ_{h,E_{\tau}}}$ of $H$.
We write $\overline{U_{E_{\rho}}}$ on the $\rho$ side and $\overline{U'_{E_{\sigma}}}$ on the $\sigma$ side.
We also write $\mu_\tau(x):=\bm\mu_\tau\cdot x$.

We first define the local state family $\omega_{\tau,x}$ on the reference system used to define $\mathcal D_n^\varepsilon$:
\eq{
 \omega_{\tau,x}:=M_\tau(x)/Z_\tau(x) \label{eq:ref_tilt}}
where
 \eq{
 M_\tau(z)&:=\sqrt{\omega_\tau}e^{D_\tau(z)}\sqrt{\omega_\tau},\quad
  Z_\tau(x):=\Tr M_\tau(x),\quad
  D_\tau(z):=\sum_i z_i(Y_{\tau,i}+\mu_{\tau,i}I),\\
  Y_\tau(x)&:=\sum_ix_iY_{\tau,i},\quad
  Y_{\tau,i}:=-(X_i^\circ)^{\mathsf T}.
}
Here, the exponential $e^{D_\tau(z)}$ is evaluated on the full Hilbert space of the reference system, which is a copy of $\calH$ for $\tau=\rho$ and a copy of $\calH'$ for $\tau=\sigma$.
Since the factors $\sqrt{\omega_\tau}$ on both sides vanish outside $E_\tau$, $M_\tau(z)$ is naturally defined as an operator on $E_\tau$.
For real $x$, $\omega_{\tau,x}$ is positive definite on $E_\tau$.
We choose this local state family so that real parameters allow us to compare states, while imaginary parameters connect it to the Fourier quantities. Indeed, for real vectors $x$ and complex vectors $z$,
\eq{
 M_\tau(ix)=e^{i\mu_\tau(x)}B_\tau(g(x)),\qquad
 \Tr[\omega_\tau D_\tau(z)]=0.
 \label{eq:tb_local_model_fourier_centering}
}
The first identity follows from $e^{iY_\tau(x)}=\overline{U^\circ_{g(x)}}$, and the second follows from $\Tr(\omega_\tau Y_{\tau,i})=-\mu_{\tau,i}$.
Thus, the term $+\mu_{\tau,i}I$ centers the generator at the reference state.

Next, to state an existence condition for the CPTP map $\mathcal L_n$ used in constructing $\mathcal D_n^\varepsilon$, we define, for $0<q<1$,
\eq{
 C_\tau^q:=q\VNS(\tau)+(1-q)[\VNS(\tau)]^{\mathsf T}-\calQ^{\tau,q}.
 \label{eq:tb_reference_C_definition}
}
The assumed equality of the covariance matrices and inequality between the metric adjusted quantum geometric tensors give
\eq{
 rC_\sigma^q-C_\rho^q=\calQ^{\rho,q}-r\calQ^{\sigma,q}\succeq0.
 \label{eq:ref_C_identity}
}
Since $C_\tau^{1-q}=(C_\tau^q)^{\mathsf T}$ and $\calQ^{\tau,1-q}=(\calQ^{\tau,q})^{\mathsf T}$, this inequality holds for all $0<q<1$.

Using this matrix, we state the following local comparison that supplies the CPTP map $\mathcal L_n$ for our construction.
If $N_n/n\to\alpha>0$, $M_n/n\to\beta>0$, and $\beta C_\sigma^q\succeq\alpha C_\rho^q$ for all $0<q<1$, then we can choose $H$-weak covariant CPTP maps $\mathcal L_n$ such that
\eq{
 \sup_{\|x\|\le R}
 \|\mathcal L_n(\omega_{\sigma,x/\sqrt n}^{\otimes M_n})
       -\omega_{\rho,x/\sqrt n}^{\otimes N_n}\|_1\longrightarrow0
 \qquad(R<\infty).
 \label{eq:ref_local}
}
We prove this statement using QLAN at the end of the proof.

We first present the composition defining $\mathcal D_n^\varepsilon$ from $\calL_n$, and then define its constituent maps in the order in which they act.
Fix $0<\varepsilon<1/2$ and split the copies as
\eq{
 k_\rho:=\lfloor\varepsilon n\rfloor,\quad k_\sigma:=\lfloor\varepsilon m_n\rfloor,
 \qquad \ell_{\rho,n}:=n-k_\rho,\quad \ell_{\sigma,n}:=m_n-k_\sigma.
}
We call the tensor factors in the following decompositions of the input and output reference systems of $\mathcal D_n^\varepsilon$ the main block and the small block:
\eq{
 E_\sigma^{\otimes m_n}
 &=\underbrace{E_\sigma^{\otimes \ell_{\sigma,n}}}_{\text{input main block}}
 \otimes\underbrace{E_\sigma^{\otimes k_\sigma}}_{\text{input small block}},
 \label{eq:tb_reference_input_blocks}\\
 E_\rho^{\otimes n}
 &=\underbrace{E_\rho^{\otimes \ell_{\rho,n}}}_{\text{output main block}}
 \otimes\underbrace{E_\rho^{\otimes k_\rho}}_{\text{output small block}}.
 \label{eq:tb_reference_output_blocks}
}
The small block receives the group action that remains in the environment after the main block is converted, and is used to recover the required output representation, including the character $\gamma_n$.
We construct $\mathcal D_n^\varepsilon$ in three stages, as a composition of three CP maps.
Specifically, ordering the input tensor factors as the main block followed by the small block, we define
\eq{
 \mathcal D_n^\varepsilon
 &:=(\mathrm{id}_{E_\rho^{\otimes\ell_{\rho,n}}}\otimes\mathcal R_n)
 \circ\left((\mathrm{id}_{E_\rho^{\otimes\ell_{\rho,n}}}\otimes\mathcal C_{R_{\sigma,k_\sigma}\otimes E_n})
 \circ\operatorname{Ad}_{\Pi_n}\right)
 \circ(\mathcal W_n\otimes\mathcal C_{\sigma,k_\sigma}).
 \label{eq:tb_reference_channel_composition}
}
The second stage combines the rearrangement of tensor factors $\operatorname{Ad}_{\Pi_n}$ with the subsequent compression $\mathrm{id}_{E_\rho^{\otimes\ell_{\rho,n}}}\otimes\mathcal C_{R_{\sigma,k_\sigma}\otimes E_n}$.
The input and output spaces at each stage are
\eq{
 \mathcal W_n\otimes\mathcal C_{\sigma,k_\sigma}:\quad
 &\calB(E_\sigma^{\otimes \ell_{\sigma,n}}\otimes E_\sigma^{\otimes k_\sigma})\longrightarrow\calB(E_\rho^{\otimes\ell_{\rho,n}}\otimes\mathcal E_n\otimes\mathcal H_{R_{\sigma,k_\sigma}}),
 \label{eq:tb_reference_first_stage_spaces}\\
 (\mathrm{id}_{E_\rho^{\otimes\ell_{\rho,n}}}\otimes\mathcal C_{R_{\sigma,k_\sigma}\otimes E_n})
 \circ\operatorname{Ad}_{\Pi_n}:\quad
 &\calB(E_\rho^{\otimes\ell_{\rho,n}}\otimes\mathcal E_n\otimes\mathcal H_{R_{\sigma,k_\sigma}})\longrightarrow\calB(E_\rho^{\otimes\ell_{\rho,n}}\otimes\mathcal K_n),
 \label{eq:tb_reference_second_stage_spaces}\\
 \mathrm{id}_{E_\rho^{\otimes\ell_{\rho,n}}}\otimes\mathcal R_n:\quad
 &\calB(E_\rho^{\otimes\ell_{\rho,n}}\otimes\mathcal K_n)\longrightarrow\calB(E_\rho^{\otimes\ell_{\rho,n}}\otimes E_\rho^{\otimes k_\rho})
 =\calB(E_\rho^{\otimes n}).
 \label{eq:tb_reference_third_stage_spaces}
}
Here, $\mathcal E_n$ is the environment space, $\mathcal H_{R_{\sigma,k_\sigma}}$ is the space of the compressed small block, and $\mathcal K_n$ is the direct-sum space shared by the output of the second compression and the input of the recovery map.
The output into $\mathcal K_n$ includes the addition of zero blocks, as explained below.
We now define these spaces, representations, and maps explicitly, proceeding from right to left in the composition.

To define the first stage $\mathcal W_n\otimes\mathcal C_{\sigma,k_\sigma}$, we first define $\mathcal W_n$, which acts on the main block.
Since $\ell_{\rho,n}/n\to1-\varepsilon$ and $\ell_{\sigma,n}/n\to r(1-\varepsilon)$, \eqref{eq:ref_C_identity} allows us to apply \eqref{eq:ref_local}.
We use the resulting $H$-weak covariant CPTP map $\mathcal L_n$ to map the $\ell_{\sigma,n}$ copies in the main block to $\ell_{\rho,n}$ copies.
With environment space $\mathcal E_n$, take a minimal Stinespring representation of $\mathcal L_n$ satisfying
\eq{
 W_n\overline{U'_{h,E_{\sigma}}}^{\otimes \ell_{\sigma,n}}
   =(\overline{U_{h,E_{\rho}}}^{\otimes \ell_{\rho,n}}\otimes E_n(h))W_n,\qquad
 \mathcal L_n(A)=\Tr_E W_nAW_n^\dagger.
 \label{eq:tb_covariant_dilation}
}
Here, $W_n:E_\sigma^{\otimes \ell_{\sigma,n}}\to E_\rho^{\otimes\ell_{\rho,n}}\otimes\mathcal E_n$ is an isometry, and $E_n$ is a finite-dimensional nonprojective representation on $\mathcal E_n$.
By covariance, applying $h$ to the input before $W_n$, or applying $h$ to the main output of $W_n$, gives minimal Stinespring representations of the same channel.
The uniqueness of minimal Stinespring representations~\cite{Stinespring1955,davies1976quantum} therefore relates them by a unique unitary $E_n(h)$ on the environment.
Comparing the action of $h_1h_2$ at once with the successive actions of its two factors, this uniqueness gives $E_n(h_1h_2)=E_n(h_1)E_n(h_2)$, yielding a representation with no phase ambiguity.
Using $W_n$, we define $\mathcal W_n(A):=W_nAW_n^\dagger$, retaining the environment $\mathcal E_n$ together with its correlations with the main output $E_\rho^{\otimes\ell_{\rho,n}}$.

Next, we define $\mathcal C_{\sigma,k_\sigma}$, which acts on the small block. For this purpose, we define a map that retains the irreducible labels and irreducible representation spaces of a finite-dimensional representation space.
For any representation $V$ of $H$, take an irreducible decomposition $\mathcal H_V=\bigoplus_\nu\mathcal V_\nu\otimes\mathcal M_\nu$ and define
\eq{
 \mathcal C_V(A):=\bigoplus_\nu\Tr_{\mathcal M_\nu}(P_\nu AP_\nu),\qquad
 R_V(h):=\bigoplus_{\nu\in\operatorname{Irr}(V)}U_\nu(h).
 \label{eq:tb_charge_compression}
}
Thus, this map retains the irreducible labels and irreducible representation spaces while discarding the multiplicity spaces.
Here, $\operatorname{Irr}(V)$ is the finite set of irreducible labels that occur.
The Kraus operators are $(I_{\mathcal V_\nu}\otimes\langle m|)P_\nu$, so $\mathcal C_V$ is a $(V,R_V)$-strong covariant CPTP map.
Writing the output blocks as $A_\nu$, we have
\eq{
 \Tr[\mathcal C_V(A)R_V(h)]=\Tr[AV(h)],\qquad
 A_\nu=d_\nu\int_H\Tr[AV(h)]U_\nu(h)^\dagger dh.
 \label{eq:tb_charge_fourier}
}
The first identity expresses preservation of the charge characteristic function, and the second follows from the Schur orthogonality relations.
Thus, equal characteristic functions give equal compressed states.
When correlations with another system are present, the same identities hold in operator-valued form, without tracing out that system.

We henceforth write $\mathcal C_{\tau,k}:=\mathcal C_{\overline{U^\circ_{E_{\tau}}}^{\otimes k}}$ and apply $\mathcal C_{\sigma,k_\sigma}$ to the small block.
We denote its output representation by $R_{\sigma,k_\sigma}:=R_{\overline{U'_{E_{\sigma}}}^{\otimes k_\sigma}}$ and its representation space by $\mathcal H_{R_{\sigma,k_\sigma}}$.
Thus, the first-stage map is $\mathcal W_n\otimes\mathcal C_{\sigma,k_\sigma}$, with output on $E_\rho^{\otimes\ell_{\rho,n}}\otimes\mathcal E_n\otimes\mathcal H_{R_{\sigma,k_\sigma}}$.

For the second stage, first let $\Pi_n:E_\rho^{\otimes\ell_{\rho,n}}\otimes\mathcal E_n\otimes\mathcal H_{R_{\sigma,k_\sigma}}\to E_\rho^{\otimes\ell_{\rho,n}}\otimes\mathcal H_{R_{\sigma,k_\sigma}}\otimes\mathcal E_n$ be the unitary that exchanges the last two tensor factors, and define $\operatorname{Ad}_{\Pi_n}(A):=\Pi_n A\Pi_n^\dagger$.
We then apply $\mathrm{id}_{E_\rho^{\otimes\ell_{\rho,n}}}\otimes\mathcal C_{R_{\sigma,k_\sigma}\otimes E_n}$, keeping the main output $E_\rho^{\otimes\ell_{\rho,n}}$ while jointly compressing the compressed small block and the environment.
Here, $\mathcal C_{R_{\sigma,k_\sigma}\otimes E_n}$ is defined by applying \eqref{eq:tb_charge_compression} to $V=R_{\sigma,k_\sigma}\otimes E_n$.
To use a common set of irreducible labels for the representation being compressed and the representation to be recovered, set
\eq{
 \mathcal I_n&:=\operatorname{Irr}(R_{\sigma,k_\sigma}\otimes E_n)
 \cup\operatorname{Irr}(\gamma_n\overline{U_{E_{\rho}}}^{\otimes k_\rho}),\\
 \mathcal K_n&:=\bigoplus_{\nu\in\mathcal I_n}\mathcal V_\nu.
}
We choose the same irreducible representation and basis for each common label, and add zero blocks for labels absent from either representation. The output of the second compression and the input of the recovery map are both taken on $\mathcal K_n$.

Finally, we define the third-stage map $\mathrm{id}_{E_\rho^{\otimes\ell_{\rho,n}}}\otimes\mathcal R_n$. This map keeps $E_\rho^{\otimes\ell_{\rho,n}}$ and recovers the output small block $E_\rho^{\otimes k_\rho}$ from the compressed system.
For this purpose, we define a map that reverses compression on the reference state.
If $\omega$ is weak symmetric with respect to $V$ and positive definite, then
$\omega=\bigoplus_\nu(I_\nu/d_\nu)\otimes\omega_\nu$ and $p_\nu:=\Tr\omega_\nu>0$.
In the recovery map below, $P_\nu$ denotes the projection onto the $\nu$ component of the compressed space $\bigoplus_\nu\mathcal V_\nu$.
The state $\omega_\nu/p_\nu$ is the normalized state of the multiplicity space conditioned on the irreducible label $\nu$ in the reference state $\omega$.
The map reversing compression on the reference state is
\eq{
 \mathcal R_{V,\omega}(A):=\bigoplus_\nu(P_\nu AP_\nu)\otimes\frac{\omega_\nu}{p_\nu},
 \qquad \mathcal R_{V,\omega}\mathcal C_V(\omega)=\omega.
 \label{eq:tb_charge_recovery}
}
It restores the reference state by reattaching this multiplicity state to each input block.
Writing $\omega_\nu/p_\nu=\sum_m r_{\nu,m}|m\rangle\langle m|$, its Kraus operators are
$\sqrt{r_{\nu,m}}(I_\nu\otimes|m\rangle)P_\nu$.
Thus, this map is also strong covariant and is trace-preserving on the blocks present in $V$.
When applied from a larger direct-sum space, it is defined as a trace-nonincreasing completely positive map that discards the absent blocks.
We also write $\mathcal P_{\tau,k}:=\mathcal R_{\overline{U^\circ_{E_{\tau}}}^{\otimes k},\omega_\tau^{\otimes k}}\mathcal C_{\tau,k}$.

The recovery map used here is $\mathcal R_n:=\mathcal R_{\gamma_n\overline{U_{E_{\rho}}}^{\otimes k_\rho},\omega_\rho^{\otimes k_\rho}}$.
We apply it on $\mathcal K_n$, discarding blocks absent from the target representation $\gamma_n\overline{U_{E_{\rho}}}^{\otimes k_\rho}$.
The third-stage map is therefore $\mathrm{id}_{E_\rho^{\otimes\ell_{\rho,n}}}\otimes\mathcal R_n$, with output space $E_\rho^{\otimes\ell_{\rho,n}}\otimes E_\rho^{\otimes k_\rho}=E_\rho^{\otimes n}$.

This defines all the maps in \eqref{eq:tb_reference_channel_composition}.
The first and second stages are CPTP maps, and the third is a trace-nonincreasing completely positive map.
Hence, $\mathcal D_n^\varepsilon$ is trace-nonincreasing and completely positive, completing the definition of the map in \eqref{eq:tb_reference_strategy_map}.

Next, we prove \eqref{eq:tb_reference_strategy_intertwiner}.
The isometry $W_n$ in \eqref{eq:tb_covariant_dilation} and each Kraus operator of the compression and recovery maps in \eqref{eq:tb_charge_compression} and \eqref{eq:tb_charge_recovery} intertwine the input and output representations.
Since $\Pi_n$ aligns the order of the corresponding tensor factors, composing these maps in the order of \eqref{eq:tb_reference_channel_composition} gives the output representation after recovery as
$\overline{U_{h,E_{\rho}}}^{\otimes \ell_{\rho,n}}\otimes(\gamma_n(h)\overline{U_{h,E_{\rho}}}^{\otimes k_\rho})=\gamma_n(h)\overline{U_{h,E_{\rho}}}^{\otimes n}$. Therefore,
\eq{
 D_{n,a}^\varepsilon\overline{U'_{h,E_{\sigma}}}^{\otimes m_n}
 =\gamma_n(h)\overline{U_{h,E_{\rho}}}^{\otimes n}D_{n,a}^\varepsilon
 \quad(h\in H).
 \label{eq:app_strong_iid_reference_intertwiner}
}
This is \eqref{eq:tb_reference_strategy_intertwiner}, proving the first condition.

We next prove \eqref{eq:tb_reference_strategy_integral}.
We write the Fourier error in \eqref{eq:app_strong_iid_F} as $F_n^\varepsilon(g)$, in accordance with our current notation for the number of copies.
Also, set
\eq{
 \Delta_w(g):=\|U_g\rho U_g^\dagger-\rho\|_2.
 \label{eq:tb_weak_distance}
}
To obtain the integral estimate in this part of the proof, it suffices to prove the following two estimates.
The first is the global estimate
\eq{
 \|F_n^\varepsilon(g)\|_1
 \le2e^{-c_{\mathrm{dom}}\varepsilon n\Delta_w(g)^2}
 \qquad(g\in G)
 \label{eq:tb_integral_global_bound}
}
for each fixed $0<\varepsilon<1/2$ and all sufficiently large $n\in\mathcal N$. The second is
\eq{
 \limsup_{\substack{n\to\infty\\n\in\mathcal N}}
 \|F_n^\varepsilon(g(x/\sqrt n))\|_1
 \le C\sqrt\varepsilon\,\|x\|e^{-c\|x\|^2}
 \label{eq:tb_pointwise_outer}
}
for every fixed $x\in\mathfrak h^\perp$. Here, $c_{\mathrm{dom}},C,c>0$ are independent of $n,\varepsilon$.
The threshold in $n$ beyond which \eqref{eq:tb_integral_global_bound} holds may depend on the fixed $\varepsilon$.

We first show that these two estimates imply \eqref{eq:tb_reference_strategy_integral}.
For $h_1,h_2\in H$, we have $B_\tau(h_1gh_2)=\overline{U^\circ_{h_1,E_\tau}}B_\tau(g)\overline{U^\circ_{h_2,E_\tau}}$.
Thus, the already established relation \eqref{eq:app_strong_iid_reference_intertwiner} and the multiplicativity of $\gamma_n$ give
\eq{
 F_n^\varepsilon(h_1gh_2)
 =\overline{U_{h_1,E_\rho}}^{\otimes n}
 F_n^\varepsilon(g)\overline{U_{h_2,E_\rho}}^{\otimes n}
 \label{eq:tb_integral_H_invariance}
}
so $\|F_n^\varepsilon(g)\|_1$ is both left- and right-$H$-invariant.
Moreover, the zero set of $\Delta_w$ is $H$, and $\Delta_w$ is also both left- and right-$H$-invariant.
Expanding to second order near the origin gives
\eq{
 \Delta_w(g(x))^2=\|i[X(x),\rho]\|_2^2+O(\|x\|^3).
 \label{eq:tb_transverse_decay}
}
This quadratic form vanishes precisely when $[X(x),\rho]=0$.
Equivalently, the group action of the one-parameter subgroup obtained by exponentiating the generator preserves $\rho$ at all times:
\eq{
 [X(x),\rho]=0
 &\Longleftrightarrow U_{g(tx)}\rho U_{g(tx)}^\dagger=\rho\quad(\forall t\in\mathbb R),\nonumber\\
 &\Longleftrightarrow g(tx)\in H\quad(\forall t\in\mathbb R)
 \Longleftrightarrow x\in\mathfrak h.
 \label{eq:tb_weak_distance_kernel}
}
Thus, the kernel of this quadratic form is $\mathfrak h$, and it is positive definite on $\mathfrak h^\perp$.
Consequently, there exist $\delta,c_w,C_w>0$ such that
\eq{
 c_w\|x\|^2\le\Delta_w(g(x))^2\le C_w\|x\|^2
 \qquad(x\in\mathfrak h^\perp,\ \|x\|<\delta).
 \label{eq:tb_weak_transverse_bounds}
}

We estimate integrals near a subgroup using finitely many local coordinate charts.
For a general closed subgroup $K_0\subseteq L$, take $\eta\in\operatorname{Lie}(K_0)$ and $x\in\operatorname{Lie}(L)\cap\operatorname{Lie}(K_0)^\perp$.
Thus, within $\operatorname{Lie}(L)$, $\eta$ describes directions along the subgroup, and $x$ describes directions orthogonal to it.
The derivative at the origin of the map defining $g(\eta)g(x)$ is $(\eta,x)\mapsto\eta+x$, a linear isomorphism onto $\operatorname{Lie}(L)$.
By the inverse function theorem, $k_jg(\eta)g(x)$ provides local coordinates near each point $k_j$ of $K_0$.
Choose each coordinate domain to be bounded and sufficiently small that its closure lies inside a larger coordinate domain. Compactness of $K_0$ then gives a cover by finitely many such charts.
Since Haar measure has a smooth positive density in local coordinates, in each chart we can write
\eq{
 dl=J_j(\eta,x)\,d\eta\,dx,\qquad
 0<J_j(\eta,x)\le C_j.
 \label{eq:tb_local_haar_density}
}
Here, $d\eta\,dx$ is Lebesgue measure on the Euclidean coordinate space, and $J_j$ is the density arising from the change of variables in the integral.
The upper bound follows from continuity of the density and compactness of the closure of the chosen coordinate domain.
For a nonnegative left-$K_0$-invariant function $f$, the inclusion $k_jg(\eta)\in K_0$ gives
\eq{
 f(k_jg(\eta)g(x))=f(g(x)).
}
Thus, after using the upper bound on the density, integration over $\eta$ in each chart contributes only the volume of its domain.
Absorbing the contributions of the finitely many charts and the upper bounds on the Haar densities into a constant, we obtain, for some fixed neighborhood $\mathcal U\supseteq K_0$ and $\delta,C>0$,
\eq{
 \int_L f(l)dl\le
 C\int_{\substack{x\in\operatorname{Lie}(L)\cap\operatorname{Lie}(K_0)^\perp\\\|x\|<\delta}}f(g(x))dx
 +\sup_{L\setminus\mathcal U}f.
 \label{eq:tb_finite_chart_integral}
}
If the complement is empty, the second term is defined to be zero.
In dimension zero, the integrals are interpreted as evaluation at a point or finite sums, so finite groups are also included.

Apply \eqref{eq:tb_finite_chart_integral} with $(L,K_0)=(G,H)$ and $f(g)=\|F_n^\varepsilon(g)\|_1$.
If necessary, decrease $\delta$ so that \eqref{eq:tb_weak_transverse_bounds} applies, and let $\mathcal U$ be the corresponding fixed neighborhood of $H$.
Changing variables as $x\mapsto x/\sqrt n$ and using $s=\dim G-\dim H$, we obtain
\eq{
 n^{s/2}\int_G\|F_n^\varepsilon(g)\|_1dg
 &\le C_H\int_{\substack{x\in\mathfrak h^\perp\\\|x\|<\delta\sqrt n}}
       \|F_n^\varepsilon(g(x/\sqrt n))\|_1dx\nonumber\\
 &\quad+n^{s/2}\sup_{g\in G\setminus\mathcal U}\|F_n^\varepsilon(g)\|_1.
 \label{eq:tb_integral_scaled_chart}
}
The constant $C_H$ is independent of $n,\varepsilon$.
If $G\setminus\mathcal U$ is nonempty, compactness and the fact that the zero set of $\Delta_w$ is $H$ imply $\Delta_w\ge b>0$ on this set.
The first estimate \eqref{eq:tb_integral_global_bound} therefore gives, for each fixed $\varepsilon>0$,
\eq{
 n^{s/2}\sup_{g\in G\setminus\mathcal U}\|F_n^\varepsilon(g)\|_1
 \le2n^{s/2}e^{-c_{\mathrm{dom}}\varepsilon nb^2}\longrightarrow0.
 \label{eq:tb_integral_exterior_vanishes}
}
If the complement is empty, this term is zero from the outset.

On the other hand, extending the integrand in \eqref{eq:tb_integral_scaled_chart} by zero for $\|x\|\ge\delta\sqrt n$, \eqref{eq:tb_integral_global_bound} and \eqref{eq:tb_weak_transverse_bounds} give
\eq{
 \mathbf1_{\{\|x\|<\delta\sqrt n\}}
 \|F_n^\varepsilon(g(x/\sqrt n))\|_1
 \le2e^{-c_{\mathrm{dom}}c_w\varepsilon\|x\|^2}.
 \label{eq:tb_integral_gaussian_domination}
}
For fixed $\varepsilon>0$, the right-hand side is integrable on $\mathfrak h^\perp$.
Applying Fatou's lemma to this dominating function minus the integrand bounds the upper limit of the integrals by the integral of the pointwise upper limit.
Thus, applying the second estimate \eqref{eq:tb_pointwise_outer} to \eqref{eq:tb_integral_scaled_chart} and \eqref{eq:tb_integral_exterior_vanishes} gives
\eq{
 \limsup_{\substack{n\to\infty\\n\in\mathcal N}}
 n^{s/2}\int_G\|F_n^\varepsilon(g)\|_1dg
 &\le C_H\int_{\mathfrak h^\perp}
       \limsup_{\substack{n\to\infty\\n\in\mathcal N}}
       \|F_n^\varepsilon(g(x/\sqrt n))\|_1dx\nonumber\\
 &\le C_HC\sqrt\varepsilon
       \int_{\mathfrak h^\perp}\|x\|e^{-c\|x\|^2}dx
 \le C_1\sqrt\varepsilon.
 \label{eq:app_strong_iid_integrated_error}
}
The final Gaussian integral is finite, and $C_H,C,c$ are all independent of $\varepsilon$, so $C_1$ can also be chosen independently of $\varepsilon$.
Here, $\varepsilon$ is fixed first, so the dominating function in \eqref{eq:tb_integral_gaussian_domination} itself may depend on $\varepsilon$.
We have therefore shown that the two estimates \eqref{eq:tb_integral_global_bound} and \eqref{eq:tb_pointwise_outer} imply \eqref{eq:tb_reference_strategy_integral}.
We now prove these two estimates in turn.

We first prove the global estimate \eqref{eq:tb_integral_global_bound}.
For this purpose, we use the following two decay estimates for $g\in G$:
\eq{
 \|B_\rho(g)\|_1&\le e^{-\Delta_w(g)^2/8},\nonumber\\
 f_k(g):=\|\mathcal C_{\sigma,k}(B_\sigma(g)^{\otimes k})\|_1
 &\le e^{-ck\Delta_w(g)^2}\quad(k\ge k_0).
 \label{eq:tb_global_suppression}
}
Here, $c>0$ and $k_0$ are constants independent of $g,k,\varepsilon$.
The first inequality in \eqref{eq:tb_global_suppression} follows from $\|B_\rho(g)\|_1=\mathsf F(\rho,U_g\rho U_g^\dagger)$ and
$\frac12\|\alpha-\beta\|_1\le\sqrt{1-\mathsf F(\alpha,\beta)^2}$.
To prove the second inequality in \eqref{eq:tb_global_suppression}, we extend decay obtained at finitely many copy numbers to all sufficiently large copy numbers.
The identity $B_\sigma(hg)=\overline{U'_{h,E_{\sigma}}}B_\sigma(g)$, its counterpart for right multiplication, and strong covariance of $\mathcal C$ imply that $f_k$ is both left- and right-$H$-invariant.

The dual of a trace-nonincreasing completely positive map is contractive in operator norm, so the original map is contractive in trace norm also for complex operators.
First, since $\mathcal C$ preserves the characteristic function, compressing $j,k$ copies separately and then compressing their tensor product representation gives the same result.
Indeed, this two-stage compression and the direct compression of $j+k$ copies give the same characteristic function for all $h\in H$.
The inverse transform from the Schur orthogonality relations uniquely recovers the matrix in each irreducible block. Thus, after identifying the outputs using common representatives of the irreducible representations, the compressed results coincide.
By \eqref{eq:tb_charge_fourier} and norm contractivity,
\eq{
 f_{j+k}(g)\le f_j(g)f_k(g),\qquad f_k(g)\le1,\quad f_0(g):=1.
 \label{eq:tb_submultiplicative}
}
Also, let $d_\sigma:=\dim E_\sigma$.
Under the irreducible decomposition, the algebra generated by the $k$-copy representation is $\bigoplus_\nu\mathcal B(\mathcal V_\nu)\otimes I_{\mathcal M_\nu}$, of dimension $\sum_\nu d_\nu^2$.
This equality follows because the Schur orthogonality relations allow us to extract the matrix units on each $\mathcal V_\nu$ by averaging over the group with representation matrix elements as weights.
On the other hand, $\overline{U'_{h,E_{\sigma}}}^{\otimes k}$ is invariant under permutations of the copies, and its linear combinations have the same property.
Choosing $d_\sigma^2$ basis operators of $\mathcal B(E_\sigma)$, the permutation-invariant tensors are spanned by tensors indexed by tuples of nonnegative integers counting the occurrences of each basis operator.
These integers sum to $k$, so there are $\binom{k+d_\sigma^2-1}{d_\sigma^2-1}$ such tuples.
The dimension of the algebra above is therefore at most this number.
Furthermore, writing $A_{k,\nu}(g):=\mathcal C_{\sigma,k}(B_\sigma(g)^{\otimes k})_\nu$, \eqref{eq:tb_charge_fourier} gives
\eq{
 A_{k,\nu}(g)
 &=d_\nu\int_H\overline{\chi_\sigma(hg)}^{\,k}U_\nu(h)^\dagger dh,\nonumber\\
 \|A_{k,\nu}(g)\|_1
 &\le d_\nu^2\left(\max_{h\in H}|\chi_\sigma(hg)|\right)^k.
}
The second line uses $\|U_\nu(h)^\dagger\|_1=d_\nu$.
Since the trace norm of a direct sum is the sum of the trace norms of its blocks, combining these observations gives
\eq{
 \sum_\nu d_\nu^2\le\binom{k+d_\sigma^2-1}{d_\sigma^2-1},\qquad
 f_k(g)\le\binom{k+d_\sigma^2-1}{d_\sigma^2-1}
       \left(\max_{h\in H}|\chi_\sigma(hg)|\right)^k.
 \label{eq:tb_polynomial_suppression}
}
Since the set where $|\chi_\sigma(t)|=1$ is $K\subseteq H$, for $g\notin H$ there is some $k$ such that $f_k(g)<1$.

We also obtain quadratic decay near the identity.
Here, the superscript $(k)$ denotes the sum of the generators acting on the individual copies of a $k$-copy system. Specifically, define
\eq{
 Y_{\tau,i}^{(k)}:=\sum_{a=1}^{k}
 I^{\otimes(a-1)}\otimes Y_{\tau,i}\otimes I^{\otimes(k-a)}
 \label{eq:tb_reference_copy_generator}
}
and define $Y_\tau(x)^{(k)}$ and $D_\tau(x)^{(k)}$ by the same rule.
Set $\omega:=\omega_\sigma^{\otimes k}$ and $X:=Y_\sigma(x)^{(k)}$, and let $S$ be the support projection of $\omega$.
We extract the component of $X$ that lies in the algebra generated by the representation of $H$.
The following calculation shows that the squared deviation from this component determines the quadratic decay of $f_k$.
On the support, set
\eq{
 \mathcal E_\omega X:=\bigoplus_\nu x_\nu\otimes I_{\mathcal M_\nu},\qquad
 x_\nu:=\frac{\Tr_{\mathcal M_\nu}[(I\otimes\omega_\nu)P_\nu XP_\nu]}{p_\nu}.
}
Set it to zero outside the support; the right-hand side uses the supported part $SXS$ of $X$.
This map fixes every operator in the algebra generated by the representation of $H$.
Moreover, for any operator $A=\bigoplus_\nu A_\nu\otimes I_{\mathcal M_\nu}$ in this algebra, $\omega=\bigoplus_\nu(I_\nu/d_\nu)\otimes\omega_\nu$ gives
\eq{
 \Tr[\omega A^\dagger(X-\mathcal E_\omega X)]
 =\sum_\nu\frac1{d_\nu}\Tr\!\left[A_\nu^\dagger
 \left(\Tr_{\mathcal M_\nu}[(I\otimes\omega_\nu)P_\nu XP_\nu]-p_\nu x_\nu\right)\right]=0.
}
Thus, on the operator space of the support, this map is the orthogonal projection for the inner product $\Tr\omega A^\dagger B$.
If $X$ is Hermitian, then each $x_\nu$ is Hermitian, and $\mathcal E_\omega X$ commutes with $\omega$.
The $\nu$ block of the compressed reference state is $q_\nu I_\nu$, where $q_\nu=p_\nu/d_\nu$.
Set $B_\nu:=\mathcal C_{\sigma,k}(\sqrt\omega X^2\sqrt\omega)_\nu$, and write $Q_\nu(t)$ for the $\nu$ block of $\mathcal C_{\sigma,k}(B_\sigma(g(tx))^{\otimes k})$.
Expanding the exponential and multiplying by the adjoint gives
\eq{
 Q_\nu(t)&=q_\nu I+itq_\nu x_\nu-\tfrac12t^2B_\nu+O(t^3),\nonumber\\
 Q_\nu(t)^\dagger Q_\nu(t)
 &=q_\nu^2I+t^2(q_\nu^2x_\nu^2-q_\nu B_\nu)+O(t^3),\nonumber\\
 \|Q_\nu(t)\|_1
 &=q_\nu d_\nu+\tfrac12t^2(q_\nu\Tr x_\nu^2-\Tr B_\nu)+O(t^3).
}
Since $q_\nu>0$, the last line follows by expanding the square root around the positive number $q_\nu^2$ and taking its trace.
The definitions of the blocks and the orthogonality above give
\eq{
 \sum_\nu\Tr B_\nu&=\Tr(\omega X^2),\nonumber\\
 \sum_\nu q_\nu\Tr x_\nu^2
 &=\Tr[\omega(\mathcal E_\omega X)^2]
 =\Tr[\omega(\mathcal E_\omega X)X].
}
Thus, the sum of the quadratic coefficients is $-\Tr[\omega(X-\mathcal E_\omega X)^2]/2$, yielding
\eq{
 f_k(g(tx))=1-\frac{t^2}{2}\|(X-\mathcal E_\omega X)\sqrt\omega\|_2^2+O(t^3).
 \label{eq:tb_compression_hessian}
}
Here, $X^2$ is evaluated on the original space; $SX^2S$ has not been replaced by $(SXS)^2$.
Indeed, $SX^2S=(SXS)^2+SX(I-S)XS$, so this quadratic coefficient also includes transitions outside the support.

Denote this nonnegative quadratic form by $q_k(x):=\|(X-\mathcal E_\omega X)\sqrt\omega\|_2^2$.
If $q_k(x)=0$, then $XS=(\mathcal E_\omega X)S$, and $X$ preserves the support.
Since the generators agree exactly on the support, for every $t\in\mathbb R$,
\eq{
 \left.e^{itX}\right|_{\operatorname{ran}S}
 =\bigoplus_\nu e^{itx_\nu}\otimes I_{\mathcal M_\nu}.
 \label{eq:tb_zero_hessian_unitary}
}
On the support, this operator belongs to the algebra of the representation of $H$, so each compressed block is $q_\nu e^{itx_\nu}$.
Consequently, $f_k(g(tx))=\sum_\nu q_\nu d_\nu=1$ for all $t$.
Suppose this holds for every $k$.
If $g(tx)\notin H$ for some fixed $t$, then $\max_{h\in H}|\chi_\sigma(hg(tx))|<1$, so the right-hand side of \eqref{eq:tb_polynomial_suppression} tends to zero as $k\to\infty$, contradicting $f_k(g(tx))=1$.
Hence, $g(tx)\in H$ for every $t$.
That is, $x\in\mathfrak h$. Conversely, if $x\in\mathfrak h$, then $f_k(g(tx))=1$, so $q_k(x)=0$.
Thus, $\bigcap_{k\ge1}\ker q_k=\mathfrak h$.

By finite dimensionality, selecting one $q_k$ whenever the dimension of the intersection of the kernels decreases gives finitely many copy numbers $k_1,\ldots,k_a$ with the same common kernel.
Since these quadratic forms are nonnegative, their sum is positive definite on $\mathfrak h^\perp$.
By \eqref{eq:tb_compression_hessian}, for sufficiently small $x\in\mathfrak h^\perp$,
\eq{
 \prod_{j=1}^a f_{k_j}(g(x))
 =1-\tfrac12\sum_{j=1}^a q_{k_j}(x)+O(\|x\|^3)
 \le e^{-c_1\|x\|^2}
 \le e^{-c_2\Delta_w(g(x))^2}.
}
The last inequality uses the local estimate in \eqref{eq:tb_transverse_decay}.
By left-$H$-invariance of $f_k$ and $\Delta_w$, this estimate holds throughout some fixed neighborhood of $H$.
If $\mathfrak h^\perp$ has dimension zero, this neighborhood can be chosen inside $H$, and the same estimate follows directly from $f_k\le1$.
If necessary, add $k_1=1$ so that at least one copy number is selected.

Outside this neighborhood, for each $g$ there is some $\ell$ such that $f_\ell(g)<1$.
Continuity makes this inequality valid in a neighborhood of that point, and compactness gives finitely many copy numbers $\ell_1,\ldots,\ell_b$ whose neighborhoods cover the entire exterior.
Since all other factors are at most $1$, their product is bounded above by some $0<\theta<1$ throughout the exterior.
Using also $\Delta_w^2\le4$, decreasing the constant, and setting $k_*:=\sum_jk_j+\sum_j\ell_j$, we obtain from \eqref{eq:tb_submultiplicative}
\eq{
 f_{k_*}(g)\le\prod_{j=1}^a f_{k_j}(g)\prod_{j=1}^b f_{\ell_j}(g)
 \le e^{-c_*\Delta_w(g)^2}\qquad(g\in G).
}
If the exterior is empty, no additional $\ell_j$ are needed.
Finally, writing $k=qk_*+j$ with $0\le j<k_*$ and applying the same estimate $q$ times gives
\eq{
 f_k(g)\le f_{k_*}(g)^qf_j(g)
 \le e^{-c_*q\Delta_w(g)^2}
 \le e^{-c_*k\Delta_w(g)^2/(2k_*)}\qquad(k\ge2k_*).
}
The last inequality uses $q=\lfloor k/k_*\rfloor\ge k/(2k_*)$.
This is the second inequality in \eqref{eq:tb_global_suppression}.
If $H=G$, then $\Delta_w=0$, and the same conclusion follows directly from $f_k\le1$.
This proves both inequalities in \eqref{eq:tb_global_suppression}.
Since $\mathcal D_n^\varepsilon$ first applies $\mathcal C_{\sigma,k_\sigma}$ and the subsequent maps are also trace-norm contractive, the triangle inequality gives
\eq{
 \|F_n^\varepsilon(g)\|_1
 \le e^{-c'n\Delta_w(g)^2}+e^{-ck_\sigma\Delta_w(g)^2}.
 \label{eq:tb_global_fourier_domination}
}

Since $k_\sigma/n\to r\varepsilon$, for each fixed $\varepsilon>0$ we have $k_\sigma\ge r\varepsilon n/2$ for all sufficiently large $n\in\mathcal N$.
Using also $\varepsilon<1/2$ and taking $c_{\mathrm{dom}}:=\min\{2c',cr/2\}>0$, we obtain
\eq{
 e^{-c'n\Delta_w(g)^2}+e^{-ck_\sigma\Delta_w(g)^2}
 \le2e^{-c_{\mathrm{dom}}\varepsilon n\Delta_w(g)^2}.
 \label{eq:tb_integral_global_bound_derivation}
}
The constant $c_{\mathrm{dom}}$ is independent of $\varepsilon$, which proves the first estimate \eqref{eq:tb_integral_global_bound}.

Next, we prove the local estimate \eqref{eq:tb_pointwise_outer}.
We do so by comparing two compressed states using Lemma~\ref{lem:app_strong_iid_fourier_bound} and bounding the Fourier error after recovery by the error remaining in the small block.
For this purpose, we separately establish the local vanishing of the Fourier error, the conditions for applying the lemma, bounds on the tails of the integral, and the error estimate after recovery.

We first define the states and Fourier transforms to be compared.
Set $\alpha_{n,x}:=\omega_{\sigma,x/\sqrt n}^{\otimes\ell_{\sigma,n}}$,
$\beta_{n,x}:=\omega_{\rho,x/\sqrt n}^{\otimes\ell_{\rho,n}}$, and
$\Xi_{n,x}:=W_n\alpha_{n,x}W_n^\dagger$.
We define the state immediately after the second compression in \eqref{eq:tb_reference_channel_composition} and its target state.
Setting $S_{n,x}:=\mathcal C_{\sigma,k_\sigma}(\omega_{\sigma,x/\sqrt n}^{\otimes k_\sigma})$, the resulting direct-sum state is
\eq{
 Z_{n,x}:=(\mathrm{id}_{E_\rho^{\otimes\ell_{\rho,n}}}\otimes\mathcal C_{R_{\sigma,k_\sigma}\otimes E_n})
       (\Xi_{n,x}\otimes S_{n,x}).
 \label{eq:tb_actual_charge_state}
}
On the right-hand side, the tensor factors are reordered as $E_\rho^{\otimes\ell_{\rho,n}}$, the small block, and the environment before compression.
The target for comparison is
\eq{
 Z'_{n,x}:=\beta_{n,x}\otimes
   \mathcal C_{\gamma_n\overline{U_{E_{\rho}}}^{\otimes k_\rho}}
          (\omega_{\rho,x/\sqrt n}^{\otimes k_\rho}).
 \label{eq:tb_target_charge_state}
}
Both states are normalized. Adding zero blocks for absent irreducible components, we compare them on the common space $E_\rho^{\otimes\ell_{\rho,n}}\otimes\mathcal K_n$ defined above.
For the Fourier estimates, we rearrange the tensor factors and regard each block $Z_{n,x}^\nu,Z_{n,x}'{}^\nu$ as an operator on $\mathcal V_\nu\otimes E_\rho^{\otimes\ell_{\rho,n}}$.
Define the operator-valued characteristic function of their difference by
\eq{
 F^H_{n,x}(h):=\sum_\nu\Tr_{\mathcal V_\nu}
 \left[(U_\nu(h)\otimes I_{E_\rho^{\otimes\ell_{\rho,n}}})
       (Z_{n,x}^\nu-Z_{n,x}'{}^\nu)\right]
 \qquad(h\in H).
 \label{eq:tb_inner_fourier_definition}
}

For any fixed $0<\varepsilon<1/2$ and $R<\infty$, we establish the following four properties.

The first is local agreement of the charge characteristic functions:
\eq{
 \sup_{\substack{\|x\|,\|y\|\le R\\y\in\mathfrak h}}
 \|F^H_{n,x}(g(y/\sqrt n))\|_1\longrightarrow0.
 \label{eq:tb_inner_local_convergence}
}
The second consists of a common support condition and a second-moment bound on the target side, which allow us to apply Lemma~\ref{lem:app_strong_iid_fourier_bound}.
Set $\kappa_n(k):=\overline{c_\sigma(k)}^{m_n}\overline{c_\rho(k)}^{-\ell_{\rho,n}}$. Then, for all $k\in K$,
\eq{
 (U_\nu(k)\otimes\overline{\kappa_n(k)}I)Z_{n,x}^\nu&=Z_{n,x}^\nu,\nonumber\\
 (U_\nu(k)\otimes\overline{\kappa_n(k)}I)Z_{n,x}'{}^\nu&=Z_{n,x}'{}^\nu.
 \label{eq:tb_inner_support_goal}
}

Moreover, the center $\widetilde b_n:=a(\gamma_n)|_{\mathfrak h}-k_\rho\mu_\rho|_{\mathfrak h}$ is $\operatorname{Ad}_H$-invariant, and
\eq{
 \sup_{\|x\|\le R}\sum_\nu\Tr
 \left[(\mathsf C_\nu(\widetilde b_n)\otimes I)Z_{n,x}'{}^\nu\right]
 \le C_Rn.
 \label{eq:tb_inner_moment_goal}
}
Here, $\mathsf C_\nu$ is defined by \eqref{eq:app_strong_iid_quadratic_charge}, using an orthonormal basis of the Lie algebra of $H$ and the corresponding generators.

The third is an upper bound that allows the first local comparison to be integrated over $H$.
We can choose $\delta,c>0$ and a fixed neighborhood $\mathcal U_K$ of $K$ in $H$ such that, for sufficiently large $n\in\mathcal N$,
\eq{
 \sup_{\|x\|\le R}\|F^H_{n,x}(g(y/\sqrt n))\|_1
 &\le2e^{-c\varepsilon\|y\|^2}
 \quad\left(y\in\mathfrak h\cap\mathfrak k^\perp,\ \|y\|<\delta\sqrt n\right),
 \label{eq:tb_inner_scaled_domination}\\
 \sup_{\substack{\|x\|\le R\\h\in H\setminus\mathcal U_K}}
 \|F^H_{n,x}(h)\|_1
 &\le2e^{-c\varepsilon n}.
 \label{eq:tb_inner_exterior_decay}
}
We take $\delta$ and $\mathcal U_K$ to be the radius and neighborhood used when applying \eqref{eq:tb_finite_chart_integral} with $(L,K_0)=(H,K)$.
If the complement is empty, the left-hand side of the second inequality is defined to be zero.

The fourth is an estimate transferring the comparison of the compressed states to the Fourier error after recovery.
If $\|Z_{n,x}-Z'_{n,x}\|_1\to0$ uniformly on every fixed ball, then
\eq{
 \|F_n^\varepsilon(g(x/\sqrt n))\|_1
 \le o(1)+2\|B_\rho(g(x/\sqrt n))\|_1^{\ell_{\rho,n}}
       \sqrt{\frac{k_\rho}{n}\,x^{\mathsf T}\VNS(\rho)x}.
 \label{eq:tb_local_error_product}
}
Here, $o(1)$ tends to zero uniformly for $x$ in a fixed ball and fixed $\varepsilon$.
The right-hand side contains the product of the norm of the main block and the error remaining in the small block.

We first derive \eqref{eq:tb_pointwise_outer} from these four properties.
By \eqref{eq:tb_inner_support_goal}, $\|F^H_{n,x}\|_1$ is both left- and right-$K$-invariant, so \eqref{eq:tb_finite_chart_integral} applies to $(H,K)$.
Writing the transverse coordinate as $y/\sqrt n$, the first property \eqref{eq:tb_inner_local_convergence} shows that the integrand converges to zero.
For fixed $\varepsilon$, the third property, \eqref{eq:tb_inner_scaled_domination} and \eqref{eq:tb_inner_exterior_decay}, gives an integrable Gaussian bound and exponential decay outside the neighborhood. Hence,
\eq{
 \sup_{\|x\|\le R}
 n^{(\dim H-\dim K)/2}\int_H\|F^H_{n,x}(h)\|_1dh\longrightarrow0.
 \label{eq:tb_inner_integral}
}
Here, we apply dominated convergence to the supremum of the integrand over $\|x\|\le R$, which also gives the stated uniformity.

We now apply Lemma~\ref{lem:app_strong_iid_fourier_bound} to the group $H$, subgroup $K$, and operators $A=Z'_{n,x}$ and $A'=Z_{n,x}$.
The reference space is $E_\rho^{\otimes\ell_{\rho,n}}$, with reference representation $k\mapsto\overline{\kappa_n(k)}I$.
The second property, \eqref{eq:tb_inner_support_goal} and \eqref{eq:tb_inner_moment_goal}, supplies the support condition and the second-moment bound with $E=C_Rn$.
Denote the upper bound on the radius in the lemma for the group $H$ and subgroup $K$ by $\delta_{0,H,K}$. Fix $0<a\le\delta_{0,H,K}$ and set $\delta=a/\sqrt n$.
In this application, $F(h)$ in \eqref{eq:app_strong_iid_fourier_function} is $-F^H_{n,x}(h^{-1})$. The sign does not affect the trace norm, and Haar measure is invariant under inversion, so the integrals of their trace norms coincide.
Thus, in this application, $\eta_\delta$ is
\eq{
 \eta_\delta
 =C_{H,K}a^{-(\dim H-\dim K)}n^{(\dim H-\dim K)/2}
 \int_H\|F^H_{n,x}(h)\|_1dh
 \label{eq:tb_inner_fourier_scaled_eta}
}
which tends to zero uniformly on $\|x\|\le R$ for fixed $a$, by \eqref{eq:tb_inner_integral}.
Moreover, since $\delta\sqrt E=a\sqrt{C_R}$, \eqref{eq:tb_one_sided_fourier} gives
\eq{
 \|Z_{n,x}-Z'_{n,x}\|_1
 \le(2+2\sqrt2)a\sqrt{C_R}+2\sqrt{\eta_\delta}+\eta_\delta.
 \label{eq:tb_inner_trace_bound}
}
Taking first $n\to\infty$ and then $a\downarrow0$ yields
\eq{
 \sup_{\|x\|\le R}\|Z_{n,x}-Z'_{n,x}\|_1\longrightarrow0.
 \label{eq:tb_charge_state_convergence}
}

We now use the fourth property \eqref{eq:tb_local_error_product}. Since $k_\rho\le\varepsilon n$, the square root on the right-hand side is at most $C\sqrt\varepsilon\,\|x\|$.
For the main-block factor, we can use the already established estimates \eqref{eq:tb_global_suppression} and \eqref{eq:tb_weak_transverse_bounds}.
For any fixed $x\in\mathfrak h^\perp$ and sufficiently large $n$, the inequality $\ell_{\rho,n}\ge n/2$ gives
\eq{
 \|B_\rho(g(x/\sqrt n))\|_1^{\ell_{\rho,n}}
 \le e^{-\ell_{\rho,n}\Delta_w(g(x/\sqrt n))^2/8}
 \le e^{-c_w\|x\|^2/16}.
 \label{eq:tb_local_main_block_decay}
}
Substituting this into \eqref{eq:tb_local_error_product} and taking the upper limit along $n\in\mathcal N$ proves \eqref{eq:tb_pointwise_outer} with $C,c>0$ independent of $\varepsilon$.
Thus, the error remaining in the small block contributes $\sqrt\varepsilon$, while the main-block factor provides Gaussian decay in $\|x\|$.
Consequently, the four properties above imply the desired estimate \eqref{eq:tb_pointwise_outer}.
We now prove these four properties in turn.

We first prove the first property \eqref{eq:tb_inner_local_convergence}.
To transfer the local comparison \eqref{eq:ref_local} to a comparison of group actions, we explain how to transfer convergence from real to imaginary parameters.
The identity $\Tr\omega_\tau D_\tau(z)=0$ and the Cauchy--Schwarz inequality for Schatten norms give, on the complex ball $\|z\|\le R$,
\eq{
 \|M_\tau(z/\sqrt n)\|_1
 &\le\left(\Tr\omega_\tau e^{D_\tau(z)^\dagger/\sqrt n}e^{D_\tau(z)/\sqrt n}\right)^{1/2}
 =1+O_R(n^{-1}),\nonumber\\
 \log Z_\tau(x/\sqrt n)&=\frac{x^{\mathsf T}\VNS(\tau)x}{2n}+O_R(n^{-3/2}).
 \label{eq:tb_analytic_bounds}
}
Thus, even after taking $O(n)$ copies, the trace norm remains bounded on complex balls.
By the trace-norm contractivity of trace-nonincreasing completely positive maps established above, this boundedness is preserved under these maps.
We explain the fact about holomorphic functions used here.
Suppose that entire functions $f_n(z)$ taking values in finite matrices are bounded in trace norm independently of $n$ on each compact subset of the complex parameter space, and satisfy $\|f_n(x)\|_1\to0$ on a real open set.
Then $\|f_n(z)\|_1\to0$ uniformly on every compact subset of the complex parameter space, even if the matrix dimension grows with $n$.
Indeed, if this were false, there would be a sequence of points $z_n$ in some compact set and a constant $c>0$ such that $\|f_n(z_n)\|_1\ge c$ along a subsequence.
By trace-norm duality, choose $\|J_n\|_\infty\le1$ and form the scalar entire functions
\eq{
 h_n(z):=\Tr[J_nf_n(z)],\qquad
 |h_n(z_n)|=\|f_n(z_n)\|_1\ge c.
 \label{eq:tcm_scalar_analytic_test}
}
Since $|h_n(z)|\le\|f_n(z)\|_1$, this sequence is also uniformly bounded on every compact set.
Cauchy's integral formula bounds the derivatives using the bound on a slightly larger domain. This gives equicontinuity and hence a subsequence converging uniformly on compact sets.
Its limit $h$ is holomorphic and vanishes on the real open set.
Taking a real rectangle within that set, fixing all but one coordinate to be real, and applying the one-variable identity theorem successively shows that $h$ vanishes identically.
Passing to a further subsequence along which $z_n$ converges, uniform convergence and continuity give $h_n(z_n)\to0$, contradicting the lower bound above.
Using $\bm d_N$ from \eqref{eq:app_strong_iid_bounded_mean} with $N=n$, write
\eq{
 d_n(z):=\bm d_n\cdot z
 =\bigl(n\bm\mu_\rho-m_n\bm\mu_\sigma-\bm a(\gamma_n)\bigr)\cdot z.
 \label{eq:tb_local_mean_functional}
}
By assumption, $\|\bm d_n\|\le D_0$, and we also have $n\VNS(\rho)-m_n\VNS(\sigma)=O(1)$.
We first approximate the operator on the main output obtained by inserting the group action on the environment, using a known phase and the target state.
Denote the environment state by $\xi_{n,x}:=\Tr_{E_\rho^{\otimes\ell_{\rho,n}}}\Xi_{n,x}$.
Also write $b_n:=(\ell_{\rho,n}\mu_\rho-\ell_{\sigma,n}\mu_\sigma)|_{\mathfrak h}$ and
$A_{n,x}(h):=\Tr_E[\Xi_{n,x}(I\otimes E_n(h))]$.
For $y\in\mathfrak h$, we aim to prove
\eq{
 A_{n,x}(g(y/\sqrt n))=e^{ib_n(y)/\sqrt n}\beta_{n,x}+o(1).
 \label{eq:tb_environment_local}
}
Here, $o(1)$ denotes a term tending to zero in trace norm uniformly over $\|x\|,\|y\|\le R$ for every fixed $R<\infty$.

To prove \eqref{eq:tb_environment_local}, we apply the preceding holomorphic-function argument to \eqref{eq:ref_local}.
First, if the copy-number ratios in \eqref{eq:ref_local} satisfy $\alpha \VNS(\rho)=\beta \VNS(\sigma)$, then \eqref{eq:tb_analytic_bounds} gives
\eq{
 N_n\log Z_\rho(x/\sqrt n)-M_n\log Z_\sigma(x/\sqrt n)
 &=\frac12x^{\mathsf T}\left[\frac{N_n}{n}\VNS(\rho)
       -\frac{M_n}{n}\VNS(\sigma)\right]x+O_R(n^{-1/2})\nonumber\\
 &\longrightarrow0
 \label{eq:tb_local_normalization_ratio}
}
uniformly on fixed balls.
Thus, $Z_\rho(x/\sqrt n)^{N_n}/Z_\sigma(x/\sqrt n)^{M_n}\to1$, and each normalization factor raised to its corresponding copy number is bounded on fixed balls.
Consequently, replacing the states in \eqref{eq:ref_local} by the unnormalized operators $M_\tau$ still gives a difference tending to zero for real parameters.
Applying the holomorphic-function argument above yields
\eq{
 \mathcal L_n(M_\sigma(iy/\sqrt n)^{\otimes M_n})
       -M_\rho(iy/\sqrt n)^{\otimes N_n}\longrightarrow0
 \label{eq:tb_complex_local}
}
in trace norm, uniformly on fixed compact sets.

For the main blocks chosen here, $\alpha=1-\varepsilon$ and $\beta=r(1-\varepsilon)$, so the assumption of the lemma ensures this condition.

We use this comparison at imaginary parameters to estimate the group action on the environment.
By \eqref{eq:tb_covariant_dilation},
\eq{
 \Tr[\xi_{n,0}E_n(h)]
 =\Tr[\overline{U_{h,E_{\rho}}}^{\dagger\otimes \ell_{\rho,n}}
          \mathcal L_n(\overline{U'_{h,E_{\sigma}}}^{\otimes \ell_{\sigma,n}}\alpha_{n,0})].
}
For $y\in\mathfrak h$, we have $M_\tau(iy/\sqrt n)=e^{i\mu_\tau(y)/\sqrt n}\overline{U^\circ_{g(y/\sqrt n),E_{\tau}}}\omega_\tau$, so \eqref{eq:tb_complex_local} gives
\eq{
 \Tr[\xi_{n,0}V_n(y)]\longrightarrow1,\qquad
 V_n(y):=e^{-ib_n(y)/\sqrt n}E_n(g(y/\sqrt n)).
 \label{eq:tb_environment_character}
}
This convergence is uniform for $y$ in fixed compact sets.
Since $V_n$ is unitary, $\Tr\xi_{n,0}(V_n-I)^\dagger(V_n-I)\to0$.
This estimate transfers to the local inputs as well.
Since the single-copy difference is $O_R(n^{-1/2})$ and has trace zero,
\eq{
 \Tr(\alpha_{n,x}\alpha_{n,0}^{-1}\alpha_{n,x})\le C_R,\qquad
 (\Tr Q\alpha_{n,x})^2\le C_R\Tr Q\alpha_{n,0}\quad(0\preceq Q\preceq I).
 \label{eq:tb_contiguity}
}
The first inequality follows by raising the single-copy value $1+O_R(n^{-1})$ to the power $\ell_{\sigma,n}$.
For the second, applying the Hilbert--Schmidt Cauchy--Schwarz inequality to $\Tr Q\alpha_{n,x}=\Tr[(\alpha_{n,0}^{-1/2}\alpha_{n,x})(Q\alpha_{n,0}^{1/2})]$ gives
\eq{
 |\Tr Q\alpha_{n,x}|^2
 &\le\Tr(\alpha_{n,x}\alpha_{n,0}^{-1}\alpha_{n,x})
       \Tr(\alpha_{n,0}Q^2)\nonumber\\
 &\le C_R\Tr(\alpha_{n,0}Q).
 \label{eq:tb_contiguity_cauchy_schwarz}
}
The last step uses $Q^2\preceq Q$.
Thus, a positive operator whose expectation tends to zero at the reference input also has expectation tending to zero for local inputs in a fixed ball.
Since unitarity of $V_n$ gives $\|V_n-I\|_\infty\le2$, the operator $Q=W_n^\dagger[I\otimes(V_n-I)^\dagger(V_n-I)]W_n/4$ satisfies $0\preceq Q\preceq I$.
Substituting this operator gives
$\Tr\xi_{n,x}(V_n-I)^\dagger(V_n-I)\to0$ uniformly on fixed balls.
We further use the Cauchy--Schwarz bound
$\|\Tr_E[\Xi_{n,x}(I\otimes(V_n-I))]\|_1
\le\sqrt{\Tr\xi_{n,x}(V_n-I)^\dagger(V_n-I)}$.
The left-hand side equals the trace norm of the difference between $A_{n,x}(g(y/\sqrt n))$ and $e^{ib_n(y)/\sqrt n}\mathcal L_n(\alpha_{n,x})$.
Since \eqref{eq:ref_local} gives $\mathcal L_n(\alpha_{n,x})-\beta_{n,x}\to0$, this proves the environment estimate \eqref{eq:tb_environment_local}.

We use this environment estimate to obtain the local comparison for $F^H_{n,x}$.
Write $\vartheta_{\tau,x}(h):=\Tr[\omega_{\tau,x}\overline{U^\circ_{h,E_{\tau}}}]$ for the characteristic function of the local state.
By \eqref{eq:tb_charge_fourier}, the operator-valued characteristic function of the difference is
\eq{
 F^H_{n,x}(h)
 =\vartheta_{\sigma,x/\sqrt n}(h)^{k_\sigma}A_{n,x}(h)
  -\gamma_n(h)\vartheta_{\rho,x/\sqrt n}(h)^{k_\rho}\beta_{n,x}.
 \label{eq:tb_inner_fourier}
}
For $y\in\mathfrak h$, $Y_\tau(y)$ commutes with $\omega_\tau$, so the characteristic function to be expanded can be written as
\eq{
 \vartheta_{\tau,x/\sqrt n}(g(y/\sqrt n))
 =e^{-i\mu_\tau(y)/\sqrt n}
 \frac{\Tr[\omega_\tau e^{D_\tau(x)/\sqrt n}e^{iD_\tau(y)/\sqrt n}]}
      {Z_\tau(x/\sqrt n)}.
 \label{eq:tb_local_character_expansion_input}
}

Expanding the two exponentials in the numerator to second order, the linear terms vanish by centering.
Moreover, $[D_\tau(y),\omega_\tau]=0$ implies that the mixed quadratic term is $\Tr\omega_\tau D_\tau(x)D_\tau(y)=\operatorname{Re}[x^{\mathsf T}\VNS(\tau)y]$.
Combining this with the denominator expansion in \eqref{eq:tb_analytic_bounds} and taking the logarithm gives
\eq{
 \log\vartheta_{\tau,x/\sqrt n}(g(y/\sqrt n))
 =-\frac{i\mu_\tau(y)}{\sqrt n}-\frac{y^{\mathsf T}\VNS(\tau)y}{2n}
   +\frac{i\operatorname{Re}[x^{\mathsf T}\VNS(\tau)y]}{n}+O_R(n^{-3/2}).
}
Using \eqref{eq:tb_environment_local}, the difference between the first-order phases of the two terms in \eqref{eq:tb_inner_fourier} is
$b_n-k_\sigma\mu_\sigma-a(\gamma_n)|_{\mathfrak h}+k_\rho\mu_\rho=d_n|_{\mathfrak h}=O(1)$.
The quadratic terms also agree asymptotically because $k_\rho \VNS(\rho)-k_\sigma \VNS(\sigma)=O(1)$, so
$F^H_{n,x}(g(y/\sqrt n))\to0$ uniformly for $x,y$ in fixed compact sets.

This proves \eqref{eq:tb_inner_local_convergence}.

Next, we prove the second property, \eqref{eq:tb_inner_support_goal} and \eqref{eq:tb_inner_moment_goal}.
For $k\in K$, we have $\overline{U^\circ_{k,E_{\tau}}}=\overline{c_\tau(k)}I$, so \eqref{eq:tb_covariant_dilation} gives
\eq{
 \left(I\otimes\left[E_n(k)
 -\overline{c_\sigma(k)}^{\ell_{\sigma,n}}
  \overline{c_\rho(k)}^{-\ell_{\rho,n}}I\right]\right)W_n=0.
 \label{eq:tb_environment_scalar_on_dilation}
}

Minimality means that $(\langle u|\otimes I)W_nv$ spans the environment space $\mathcal E_n$ as the main-output vector $u$ and the input vector $v$ range over their entire spaces.
Hence, the operator in square brackets vanishes on the entire environment, and
$E_n(k)=\overline{c_\sigma(k)}^{\ell_{\sigma,n}}\overline{c_\rho(k)}^{-\ell_{\rho,n}}I$.
Thus, $k$ acts on the charge part of $Z_{n,x}$ as multiplication by $\kappa_n(k)$ defined above.
The coefficient $\gamma_n(k)\overline{c_\rho(k)}^{k_\rho}$ on the target side is the same by assumption \eqref{eq:app_strong_iid_auxiliary_character}.
Since the Kraus operators of each compression map intertwine the group actions, this scalar action before compression is preserved after compression.
Consequently, for either state, every vector $v\in\mathcal V_\nu\otimes E_\rho^{\otimes\ell_{\rho,n}}$ in the support of its $\nu$ block satisfies, for all $k\in K$,
\eq{
 (U_\nu(k)\otimes I_{E_\rho^{\otimes\ell_{\rho,n}}})v=\kappa_n(k)v,\qquad
 (U_\nu(k)\otimes\overline{\kappa_n(k)}I_{E_\rho^{\otimes\ell_{\rho,n}}})v=v.
 \label{eq:tb_charge_support}
}
The second identity is precisely the support condition in Lemma~\ref{lem:app_strong_iid_fourier_bound} when the reference space $E_\rho^{\otimes\ell_{\rho,n}}$ carries the scalar representation $k\mapsto\overline{\kappa_n(k)}I_{E_\rho^{\otimes\ell_{\rho,n}}}$.
This also makes $\|F^H_{n,x}\|_1$ both left- and right-$K$-invariant.
Next, the center $\widetilde b_n$ is $\operatorname{Ad}_H$-invariant because it is the difference between the character derivative $a(\gamma_n)$ and the mean $\mu_\rho|_{\mathfrak h}$ of an $H$-invariant state.
It suffices to check the second moment on the target side.
In this paragraph, $j$ indexes directions in an orthonormal basis of $\mathfrak h$, and $Y_{\rho,j}$ and $\mu_{\rho,j}$ denote the generator and mean in that direction.
The generator of the character-twisted representation $\gamma_n\overline{U_{E_\rho}}^{\otimes k_\rho}$ is $a_j(\gamma_n)I+Y_{\rho,j}^{(k_\rho)}$. Subtracting the center $\widetilde b_{n,j}=a_j(\gamma_n)-k_\rho\mu_{\rho,j}$ gives $Y_{\rho,j}^{(k_\rho)}+k_\rho\mu_{\rho,j}I$.
Strong covariance of compression preserves the moments of the generators, and the main-block state $\beta_{n,x}$ has trace $1$. Thus,
\eq{
 &\sum_\nu\Tr[(\mathsf C_\nu(\widetilde b_n)\otimes I)Z_{n,x}'{}^\nu]\nonumber\\
 &\qquad=\sum_{j=1}^{\dim H}\Tr\!\left[
 \omega_{\rho,x/\sqrt n}^{\otimes k_\rho}
 (Y_{\rho,j}^{(k_\rho)}+k_\rho\mu_{\rho,j}I)^2\right].
 \label{eq:tb_inner_moment_identity}
}
For each direction $j$ on the right-hand side,
\eq{
 \Tr\left[\omega_{\rho,x/\sqrt n}^{\otimes k_\rho}
        (Y_{\rho,j}^{(k_\rho)}+k_\rho\mu_{\rho,j}I)^2\right]
 \le C_R(k_\rho+k_\rho^2/n)=O_R(n).
}
Indeed, the mean of the single-copy centered generator is $O_R(n^{-1/2})$, its second moment is $O_R(1)$, and expectations factorize across distinct copies.
Summing over the finitely many basis directions and redefining $C_R$, we can choose the second-moment bound in the lemma as $E=C_Rn$.
This establishes the common support condition \eqref{eq:tb_inner_support_goal} and the target second-moment estimate \eqref{eq:tb_inner_moment_goal}.

To prove the third property, \eqref{eq:tb_inner_scaled_domination} and \eqref{eq:tb_inner_exterior_decay}, we estimate decay away from $K$.
Set $\Delta_s(g):=\sqrt{1-|\chi_\rho(g)|^2}$.
Its zero set is $K$, and it is both left- and right-$K$-invariant.
Expanding to second order gives
\eq{
 \Delta_s(g(x))^2=x^{\mathsf T}\VNS(\rho)x+O(\|x\|^3).
 \label{eq:tb_strong_transverse_decay}
}
The kernel of this quadratic form is $\mathfrak k$.
Indeed, $x^{\mathsf T}\VNS(\rho)x=\|(X(x)-\mu_\rho(x)I)\sqrt\rho\|_2^2$, so it vanishes only in directions whose generators act as a single scalar on the support of the state.
Hence, on $\mathfrak h\cap\mathfrak k^\perp$ near the origin, $\Delta_s(g(x))^2$ is bounded above and below by positive constant multiples of $\|x\|^2$.
Outside a fixed neighborhood of $K$, compactness also gives a positive lower bound for $\Delta_s$.

The characteristic functions of the local states satisfy, uniformly for $x$ in a sufficiently small fixed ball,
\eq{
 |\vartheta_{\tau,x}(h)|\le e^{-c_\tau\Delta_s(h)^2}\qquad(h\in H).
 \label{eq:tb_local_character_decay}
}
Indeed, positive definiteness allows us to ensure $\omega_{\tau,x}\succeq\omega_\tau/2$.
Thus, $\eta_{\tau,x}:=2\omega_{\tau,x}-\omega_\tau$ is a state, and we can write
\eq{
 \omega_{\tau,x}=\tfrac12\omega_\tau+\tfrac12\eta_{\tau,x},\qquad
 \eta_{\tau,x}\succeq0,\quad\Tr\eta_{\tau,x}=1.
 \label{eq:tb_local_character_mixture}
}
The bound $|\Tr[\eta_{\tau,x}\overline{U^\circ_{h,E_\tau}}]|\le1$ and the triangle inequality give
$|\vartheta_{\tau,x}(h)|\le(1+|\chi_\tau(h)|)/2\le e^{-(1-|\chi_\tau(h)|^2)/4}$.
For $\tau=\sigma$ as well, the common zero set $K$, the equality $\VNS(\rho)=r\VNS(\sigma)$, the second-order expansions, and compactness imply
$1-|\chi_\sigma(h)|^2\ge c\Delta_s(h)^2$.

We apply these observations to $F^H_{n,x}$. By $\|A_{n,x}(h)\|_1\le1$ and \eqref{eq:tb_local_character_decay},
\eq{
 \|F^H_{n,x}(h)\|_1\le e^{-ck_\sigma\Delta_s(h)^2}+e^{-c'k_\rho\Delta_s(h)^2}.
}
For each fixed $\varepsilon>0$, sufficiently large $n$ satisfy $k_\sigma\ge r\varepsilon n/2$ and $k_\rho\ge\varepsilon n/2$.
Choose a sufficiently small radius $\delta$ when applying \eqref{eq:tb_finite_chart_integral} with $(L,K_0)=(H,K)$, and let $\mathcal U_K$ be the corresponding fixed neighborhood of $K$.
By \eqref{eq:tb_strong_transverse_decay}, if $y\in\mathfrak h\cap\mathfrak k^\perp$ and $\|y\|<\delta\sqrt n$, then $\Delta_s(g(y/\sqrt n))^2\ge c_s\|y\|^2/n$.
Moreover, if $H\setminus\mathcal U_K$ is nonempty, $\Delta_s$ has a positive lower bound there.
Substituting these bounds into the two exponentials above and decreasing the constant $c>0$ gives \eqref{eq:tb_inner_scaled_domination} and \eqref{eq:tb_inner_exterior_decay}.
These constants and the neighborhood can be chosen independently of $\varepsilon$; only the starting threshold in $n$ depends on $R,\varepsilon$.

Finally, we prove the fourth property \eqref{eq:tb_local_error_product}.
Here, we assume the convergence of the compressed states in \eqref{eq:tb_charge_state_convergence}.
Apply the final recovery map in \eqref{eq:tb_reference_channel_composition} to both states.
Multiplication by a character changes an irreducible representation $U_\nu$ to $\gamma_nU_\nu$ and relabels the irreducible components as $\nu\mapsto\gamma_n\nu$.
This correspondence is one-to-one because it is inverted by the inverse character $\gamma_n^{-1}$. The subspaces in the irreducible decomposition of the original Hilbert space and their multiplicity spaces remain unchanged.
The unitaries identifying these components with fixed representatives of the irreducible representations also act in opposite directions during compression and recovery.
Therefore, compressing with the character-twisted representation and then recovering with the same representation gives the same map $\mathcal P_{\rho,k_\rho}$ as in the absence of the character.
Thus, recovering the target state gives
$\beta_{n,x}\otimes\mathcal P_{\rho,k_\rho}(\omega_{\rho,x/\sqrt n}^{\otimes k_\rho})$.
The assumed convergence \eqref{eq:tb_charge_state_convergence} and norm contractivity of the recovery map therefore imply

\eq{
 \mathcal D_n^\varepsilon(\omega_{\sigma,x/\sqrt n}^{\otimes m_n})
 -\omega_{\rho,x/\sqrt n}^{\otimes \ell_{\rho,n}}\otimes
    \mathcal P_{\rho,k_\rho}(\omega_{\rho,x/\sqrt n}^{\otimes k_\rho})
 \longrightarrow0
 \label{eq:tb_local_reference_channel}
}
in trace norm, uniformly on every fixed ball.
Under this local conversion, the main block becomes the desired local state, while the composition of compression and recovery remains on the small block.
Using this local conversion, the expansion of the normalization factors in \eqref{eq:tb_analytic_bounds}, and the bound $\|\bm d_n\|\le D_0$ in \eqref{eq:tb_local_mean_functional}, the entire function
\eq{
 M_\rho(z/\sqrt n)^{\otimes \ell_{\rho,n}}\otimes
   \mathcal P_{\rho,k_\rho}(M_\rho(z/\sqrt n)^{\otimes k_\rho})
 -e^{d_n(z)/\sqrt n}\mathcal D_n^\varepsilon(M_\sigma(z/\sqrt n)^{\otimes m_n})
 \label{eq:tb_global_analytic_difference}
}
converges to zero on real balls.
By the local boundedness in \eqref{eq:tb_analytic_bounds} and the holomorphic-function argument, the same is true at $z=ix$.
Since $\mathcal P_{\rho,k}$ fixes $\omega_\rho^{\otimes k}$, the error due to information lost from the small block satisfies
\eq{
 \|M_\rho(ix/\sqrt n)^{\otimes k}
       -\mathcal P_{\rho,k}(M_\rho(ix/\sqrt n)^{\otimes k})\|_1
 &\le2\|M_\rho(ix/\sqrt n)^{\otimes k}-\omega_\rho^{\otimes k}\|_1\nonumber\\
 &\le2\sqrt{k\,x^{\mathsf T}\VNS(\rho)x/n}.
 \label{eq:tb_small_block_error}
}
The last inequality uses, for $\Omega=\omega_\rho^{\otimes k}$ and $D=D_\rho(x)^{(k)}$,
$\|\sqrt\Omega(e^{iD/\sqrt n}-I)\sqrt\Omega\|_1
\le\sqrt{\Tr\Omega D^2/n}$.
Centering makes the cross terms between different copies vanish.

We have $M_\tau(ix/\sqrt n)=e^{i\mu_\tau(x)/\sqrt n}B_\tau(g(x/\sqrt n))$ and
$\gamma_n(g(x/\sqrt n))=e^{i\bm a(\gamma_n)\cdot x/\sqrt n}$.
Thus, setting $z=ix$ in \eqref{eq:tb_global_analytic_difference} cancels the phase difference in \eqref{eq:tb_local_mean_functional}. The Fourier error is then the sum of the difference remaining in the small block tensored with the main-block operator and a term tending to zero in trace norm.
Since the trace norm of a tensor product is the product of the trace norms of its factors,
\eq{
 \|F_n^\varepsilon(g(x/\sqrt n))\|_1
 &\le o(1)+\|B_\rho(g(x/\sqrt n))\|_1^{\ell_{\rho,n}}\nonumber\\
 &\quad\times\left\|M_\rho(ix/\sqrt n)^{\otimes k_\rho}
 -\mathcal P_{\rho,k_\rho}(M_\rho(ix/\sqrt n)^{\otimes k_\rho})\right\|_1.
 \label{eq:tb_local_error_factorization}
}
Substituting \eqref{eq:tb_small_block_error} proves \eqref{eq:tb_local_error_product}.
Here, $o(1)$ tends to zero uniformly for $x$ in a fixed ball and fixed $\varepsilon$.

This establishes all four properties.
The preceding application of Lemma~\ref{lem:app_strong_iid_fourier_bound} and the estimates \eqref{eq:tb_local_error_product} and \eqref{eq:tb_local_main_block_decay} therefore prove the local estimate \eqref{eq:tb_pointwise_outer}.
We have now proved both \eqref{eq:tb_integral_global_bound} and \eqref{eq:tb_pointwise_outer}.
The derivation in \eqref{eq:app_strong_iid_integrated_error} thus gives the integral estimate \eqref{eq:tb_reference_strategy_integral}.
Together with the Kraus-operator condition verified in \eqref{eq:app_strong_iid_reference_intertwiner}, this establishes the two conditions stated at the beginning of the proof.
Hence, the estimate \eqref{eq:tb_final_master} and the choice \eqref{eq:tb_reference_strategy_diagonal} yield the desired sequence of maps $\mathcal D_N$.

Finally, we prove the local comparison \eqref{eq:ref_local} used above.
We construct $\calL_n$ by mapping the local input states to a Gaussian model, converting between Gaussian models, and then mapping back to the original local output states.
We first state the error estimates for these three stages and derive the local comparison from them, and then prove each stage in turn.
We use QLAN as a theorem providing channels in both directions between many-copy local models and Gaussian models.
To cover reference states with degenerate eigenvalues, we introduce a small perturbation $t>0$ and reverse this perturbation on the output side.

First, for a positive definite state $\omega$ and a $C^2$ family of states $\omega_x$ satisfying $\omega_0=\omega$, set $T_i:=\partial_i\omega_x|_0$.
Choose a real basis $F_a$ of the Hermitian operators satisfying $\Tr\omega F=0$, and set
\eq{
 S_{ab}:=\operatorname{Re}\Tr\omega F_aF_b,\quad
 \Omega_{ab}:=\operatorname{Im}\Tr\omega F_aF_b,\quad
 \mathsf K_{ai}:=\Tr T_iF_a,\quad \mathsf G:=S+i\Omega>0.
 \label{eq:tb_gaussian_data}
}
The positive definiteness in the last expression follows from
$c^\dagger\mathsf Gc=\Tr\omega(\sum_ac_aF_a)^\dagger(\sum_bc_bF_b)>0$.
We identify the many-copy quantities corresponding to these data. Define the sums of observables centered by $\Tr\omega F_a=0$ as
\eq{
 F_{a,L}:=\frac1{\sqrt L}\sum_{\ell=1}^L
 I^{\otimes(\ell-1)}\otimes F_a\otimes I^{\otimes(L-\ell)}.
 \label{eq:tb_qlan_fluctuation_generator}
}
Their means in the local states are
\eq{
 \Tr[\omega_{x/\sqrt L}^{\otimes L}F_{a,L}]
 =\sqrt L\,\Tr(\omega_{x/\sqrt L}F_a)
 =(\mathsf Kx)_a+O_R(L^{-1/2}).
 \label{eq:tb_qlan_fluctuation_mean}
}
Moreover, since expectations factorize across distinct copies, at the reference state we have
\eq{
 \operatorname{Re}\Tr[\omega^{\otimes L}F_{a,L}F_{b,L}]&=S_{ab},\nonumber\\
 \Tr[\omega^{\otimes L}[F_{a,L},F_{b,L}]]&=2i\Omega_{ab}.
 \label{eq:tb_qlan_fluctuation_second_moments}
}
In the local states as well, the covariance after subtracting the mean and the expectation of the commutator converge to $S$ and $2i\Omega$, respectively.
Thus, $\mathsf Kx$ is the mean fluctuation induced by the local change, $S$ is its covariance, and $\Omega$ describes the commutation relations.
The corresponding Gaussian state $\Gamma[m]$ has mean $m$, covariance $S$, and commutator matrix $2i\Omega$.
Here, $x$ parametrizes the original state family, whereas $m$ is the mean vector of the full collection of fluctuation observables. They are related by $m=\mathsf Kx$ and need not have the same number of components.
A Gaussian model here is a family of states in which every real linear combination of the limiting fluctuation observables $\widehat F_a$ has a Gaussian characteristic function.
For a real vector $z$ with as many components as the basis $\{F_a\}$, the Weyl operator is the exponential of the corresponding linear combination:
\eq{
 [\widehat F_a,\widehat F_b]=2i\Omega_{ab}I,\qquad
 W(z):=\exp\!\left(i\sum_a z_a\widehat F_a\right).
 \label{eq:tcm_weyl_definition}
}
Here, $\Gamma[m](W(z))$ denotes the expectation of $W(z)$ in the state $\Gamma[m]$.
Thus, the second of the following identities is the counterpart of the characteristic function $\mathbb E[e^{iz\cdot X}]$ of an ordinary probability distribution:
\eq{
 W(z)W(w)=e^{-iz^{\mathsf T}\Omega w}W(z+w),\qquad
 \Gamma[m](W(z))=e^{iz^{\mathsf T}m-z^{\mathsf T}Sz/2}.
 \label{eq:tb_gaussian_character}
}
A real linear change of coordinates decomposes $\Omega$ into two-dimensional antisymmetric blocks and a zero block. The former correspond to position--momentum pairs, and the latter to classical Gaussian variables that commute with everything.
Writing the classical variable as $u$ and the operator-valued density of the quantum part as $\Gamma[m](u)$, the expectation and norm mean $\int\Tr[\Gamma[m](u)W(z;u)]du$ and $\int\|\Gamma[m](u)\|_1du$, respectively.
The norm used below to measure errors between Gaussian states is also this norm, incorporating both the classical and quantum parts.

We apply these definitions to the reference state families $\omega_{\tau,x}$, which are positive definite on their supports $E_\tau$, and use the subscripts $\tau=\rho,\sigma$.
On a side with one-dimensional support, there are no fluctuation variables, and the corresponding Gaussian system is interpreted as the trivial system with a single state.
For each $\tau$, choose a positive definite state $\kappa_\tau$ commuting with $\omega_\tau$ and having distinct eigenvalues within each degenerate eigenspace. Set
\eq{
 \mathcal E_{\tau,t}(A):=(1-t)A+t\Tr(A)\kappa_\tau,\qquad
 \omega_{\tau,t,x}:=\mathcal E_{\tau,t}(\omega_{\tau,x}),\quad
 \omega_{\tau,t}:=\omega_{\tau,t,0}.
 \label{eq:tb_qlan_perturbation}
}
Both reference states are then nondegenerate for sufficiently small $t>0$.
After perturbation, use the centered common basis $F_{\tau,a,t}:=F_{\tau,a}-\Tr(\omega_{\tau,t}F_{\tau,a})I$ and
denote the data in \eqref{eq:tb_gaussian_data} by $\mathsf G_{\tau,t},\mathsf K_{\tau,t}$.
Also denote the corresponding Gaussian state with mean $m$ by $\Gamma_{\tau,t}[m]$.

Recall that the copy numbers in the local comparison \eqref{eq:ref_local} satisfy $M_n/n\to\beta>0$ and $N_n/n\to\alpha>0$.
We construct CPTP maps independent of $x$ that satisfy the following three estimates.
The first estimate concerns perturbing the input and mapping it to the Gaussian model.
For fixed sufficiently small $t>0$, we can choose a map $\mathcal A_{\sigma,t,M_n}$ such that
\eq{
 \sup_{\|x\|\le R}\left\|
 \mathcal A_{\sigma,t,M_n}\!\left(
 \mathcal E_{\sigma,t}^{\otimes M_n}(\omega_{\sigma,x/\sqrt n}^{\otimes M_n})\right)
 -\Gamma_{\sigma,t}[\sqrt\beta\mathsf K_{\sigma,t}x]
 \right\|_1\longrightarrow0.
 \label{eq:tb_qlan_input_goal}
}

The second concerns conversion from the input Gaussian model to the output Gaussian model.
We can choose a map $\mathcal G_t$ such that
\eq{
 \sup_{\|x\|\le R}\left\|
 \mathcal G_t\!\left(\Gamma_{\sigma,t}[\sqrt\beta\mathsf K_{\sigma,t}x]\right)
 -\Gamma_{\rho,t}[\sqrt\alpha\mathsf K_{\rho,t}x]
 \right\|_1\le C_Rt.
 \label{eq:tb_qlan_gaussian_goal}
}

The third concerns mapping the output Gaussian model back to the local states and reversing the perturbation.
We can choose maps $\mathcal B_{\rho,t,N_n}$ and $\mathcal R_t$ such that
\eq{
 \sup_{\|x\|\le R}\left\|
 \mathcal R_t^{\otimes N_n}\!\left(
 \mathcal B_{\rho,t,N_n}(\Gamma_{\rho,t}[\sqrt\alpha\mathsf K_{\rho,t}x])\right)
 -\omega_{\rho,x/\sqrt n}^{\otimes N_n}
 \right\|_1\le C_Rt+o_n(1).
 \label{eq:tb_qlan_output_goal}
}

Here, $R<\infty$ is arbitrary, and $o_n(1)$ tends to zero on a fixed ball for fixed $t>0$.
The constant $C_R$, on the other hand, is independent of sufficiently small $t$.

We first derive the local comparison \eqref{eq:ref_local} from these three estimates.
Compose the three stages in order to define
\eq{
 \mathcal L_{n,t}:=
 \bigl(\mathcal R_t^{\otimes N_n}\circ\mathcal B_{\rho,t,N_n}\bigr)
 \circ\mathcal G_t\circ
 \bigl(\mathcal A_{\sigma,t,M_n}\circ\mathcal E_{\sigma,t}^{\otimes M_n}\bigr).
 \label{eq:tb_local_channel_constructed}
}

Combining the three estimates \eqref{eq:tb_qlan_input_goal}--\eqref{eq:tb_qlan_output_goal} using the triangle inequality and trace-norm contractivity of CPTP maps, and redefining $C_R$, gives
\eq{
 \sup_{\|x\|\le R}
 \|\mathcal L_{n,t}(\omega_{\sigma,x/\sqrt n}^{\otimes M_n})
       -\omega_{\rho,x/\sqrt n}^{\otimes N_n}\|_1
 \le C_Rt+o_n(1).
 \label{eq:tb_qlan_total_error}
}

Choose $t_j\downarrow0$ such that $C_jt_j\le1/j$ for each integer $j$, and choose an increasing sequence $n_j$ such that, for this $t_j$ and $n\ge n_j$, the remaining error is at most $1/j$ on the ball $\|x\|\le j$.
Setting $\mathcal L_n:=\mathcal L_{n,t_j}$ for $n_j\le n<n_{j+1}$ makes the error at most $2/j$ on the ball $\|x\|\le j$, yielding the convergence in \eqref{eq:ref_local} on every fixed ball.
Furthermore, replace this map by its group average:
\eq{
 \mathcal L_n\ \longmapsto\quad
 \int_H\operatorname{Ad}_{\overline{U_{h,E_{\rho}}}^{\dagger\otimes N_n}}
 \circ\mathcal L_n\circ
 \operatorname{Ad}_{\overline{U'_{h,E_{\sigma}}}^{\otimes M_n}}\,dh.
 \label{eq:tb_qlan_covariant_average}
}

The local states satisfy $\overline{U^\circ_{h,E_{\tau}}}\omega_{\tau,x}\overline{U^\circ_{h,E_{\tau}}}^\dagger=\omega_{\tau,O_hx}$, and $O_h$ preserves balls, so the error bound does not increase.
This average is exactly $H$-weak covariant.
Thus, the three estimates imply the local comparison \eqref{eq:ref_local}, including the required covariance.
We now prove these estimates in turn.

We first prove the first estimate \eqref{eq:tb_qlan_input_goal}.
The QLAN theorem of Kahn--Gu\c{t}\u{a}~\cite[Theorem~4.3 and Sec.~4.5]{KahnGuta2009LAN} states that, if $\omega$ has distinct eigenvalues, there exist parameter-independent CPTP maps $\mathcal A_L,\mathcal B_L$ such that, uniformly on fixed balls,
\eq{
 \mathcal A_L(\omega_{x/\sqrt L}^{\otimes L})
     &\longrightarrow\Gamma[\mathsf Kx],\qquad
 \mathcal B_L(\Gamma[\mathsf Kx])
     \longrightarrow\omega_{x/\sqrt L}^{\otimes L}.
 \label{eq:tb_nondegenerate_lan}
}
Both convergences are in trace norm.
The eigenvalue and eigenvector coordinates used in the original theorem are smooth at nondegenerate points, so the theorem can be restricted to the submodel used here.
Specifically, let $\omega^{\mathrm{full}}(\theta)$ be the coordinate representation of states in the original theorem, and write the submodel as $\omega_u=\omega^{\mathrm{full}}(\theta(u))$.
Writing the coordinates of the reference point as $\theta_0$ and the derivative of the coordinate map as $J$, we have
\eq{
 \theta(u)&=\theta_0+Ju+O(\|u\|^2),\qquad
 \widetilde\omega_u:=\omega^{\mathrm{full}}(\theta_0+Ju),\nonumber\\
 \|\omega_u-\widetilde\omega_u\|_2&=O(\|u\|^2).
 \label{eq:tb_qlan_local_coordinate_linearization}
}
Applying the original theorem to the linear local parameter $Jx$ gives a Gaussian approximation of $\widetilde\omega_{x/\sqrt L}^{\otimes L}$.
Moreover, this limit is determined by the first derivatives alone.
Indeed, for states $A,B\succeq cI$,
\eq{
 \|A^{\otimes L}-B^{\otimes L}\|_1
 \le2\sqrt L\,\|\sqrt A-\sqrt B\|_2
 \le\sqrt{L/c}\,\|A-B\|_2.
 \label{eq:tb_faithful_product_bound}
}
To verify the first inequality, the overlap of the standard purifications of $A,B$ is $\Tr\sqrt A\sqrt B\in[0,1]$. Comparing the purifications of $L$ copies and then taking the partial trace gives
\eq{
 \|A^{\otimes L}-B^{\otimes L}\|_1
 &\le2\sqrt{1-(\Tr\sqrt A\sqrt B)^{2L}}\nonumber\\
 &\le2\sqrt{2L(1-\Tr\sqrt A\sqrt B)}
 =2\sqrt L\,\|\sqrt A-\sqrt B\|_2.
 \label{eq:tb_qlan_product_purification_bound}
}
The second inequality in \eqref{eq:tb_faithful_product_bound} follows by solving
$\sqrt A(\sqrt A-\sqrt B)+(\sqrt A-\sqrt B)\sqrt B=A-B$ in the eigenbases of the two square roots.
For two models whose values and first derivatives coincide, the single-copy difference is $O_R(L^{-1})$, so their $L$-copy difference is $O_R(L^{-1/2})$.
In particular, this estimate makes the difference between the two models in \eqref{eq:tb_qlan_local_coordinate_linearization} tend to zero, so the same channels in both directions can be used for $\omega_{x/\sqrt L}^{\otimes L}$.
Thus, we can use the mean $\mathsf Kx$ defined in \eqref{eq:tb_gaussian_data}.

Apply this theorem to the perturbed model $\omega_{\tau,t,x}$ with fixed $t>0$, and denote the channels in the two directions by $\mathcal A_{\tau,t,L},\mathcal B_{\tau,t,L}$.
On a side with one-dimensional support, there is only one state, so preparation maps give exact conversions to and from the trivial Gaussian system.
It remains to account for the difference between $\sqrt L$ in \eqref{eq:tb_nondegenerate_lan} and $\sqrt n$ here.
In general, suppose $L_n/n\to a>0$.
Smoothness of the local state family and a common strictly positive lower bound on the support allow us to apply \eqref{eq:tb_faithful_product_bound}, giving
\eq{
 &\sup_{\|x\|\le R}
 \left\|\omega_{\tau,t,x/\sqrt n}^{\otimes L_n}
       -\omega_{\tau,t,\sqrt a\,x/\sqrt{L_n}}^{\otimes L_n}\right\|_1\nonumber\\
 &\qquad\le C_R\sqrt{L_n}
       \sup_{\|x\|\le R}\left\|\frac{x}{\sqrt n}-\frac{\sqrt a\,x}{\sqrt{L_n}}\right\|
 \le C'_R\left|\sqrt{L_n/n}-\sqrt a\right|\longrightarrow0.
 \label{eq:tb_qlan_copy_rescaling}
}
The common strictly positive lower bound can be chosen uniformly for small $t$ and, for $x$ in a fixed ball, sufficiently large $n$.
On the input side, take $(L_n,a)=(M_n,\beta)$.
The output of $\mathcal E_{\sigma,t}^{\otimes M_n}$ is exactly $\omega_{\sigma,t,x/\sqrt n}^{\otimes M_n}$, so
combining \eqref{eq:tb_qlan_copy_rescaling} with the first QLAN convergence in \eqref{eq:tb_nondegenerate_lan} proves the first estimate \eqref{eq:tb_qlan_input_goal}.

Next, we prove the second estimate \eqref{eq:tb_qlan_gaussian_goal}.
We calculate the information matrix needed to compare Gaussian models. Writing $L_\omega(B):=\omega B$ and $R_\omega(B):=B\omega$, a finite-dimensional calculation gives
\eq{
 \mathsf K^{\mathsf T}(S-i(2q-1)\Omega)^{-1}\mathsf K
 =\left[\Tr T_i(qL_\omega+(1-q)R_\omega)^{-1}(T_j)\right]_{ij}.
 \label{eq:tb_information_identity}
}
Indeed, $L_j=(qL_\omega+(1-q)R_\omega)^{-1}(T_j)$ satisfies $\Tr\omega L_j=\Tr T_j=0$, so it can be expanded as $L_j=\sum_b\ell_{bj}F_b$ with complex coefficients.
Multiplying the defining identity $T_j=q\omega L_j+(1-q)L_j\omega$ by $F_a$ and taking the trace gives
\eq{
 \mathsf K_{aj}
 &=q\Tr(\omega L_jF_a)+(1-q)\Tr(\omega F_aL_j)\nonumber\\
 &=\sum_b\bigl(q\mathsf G_{ba}+(1-q)\mathsf G_{ab}\bigr)\ell_{bj}\nonumber\\
 &=\sum_b\bigl(S_{ab}-i(2q-1)\Omega_{ab}\bigr)\ell_{bj}.
 \label{eq:tb_qlan_information_coefficients}
}
The coefficient matrix is invertible because $q\mathsf G^{\mathsf T}+(1-q)\mathsf G>0$.
Solving this system and using $\Tr T_iL_j=\sum_a\mathsf K_{ai}\ell_{aj}$ proves \eqref{eq:tb_information_identity}.
We now check the relation between $C_\tau^q$, defined in \eqref{eq:tb_reference_C_definition}, and the information matrix of the reference model.
Set $A_{\tau,i}:=X_i^\circ-\mu_{\tau,i}I$ and $T_{\tau,i}:=\sqrt\tau A_{\tau,i}\sqrt\tau$.
Since $\Tr T_{\tau,i}=0$, differentiating \eqref{eq:ref_tilt} gives $\partial_i\omega_{\tau,x}|_0=-T_{\tau,i}^{\mathsf T}$.
Let $\tau=\sum_a p_a|a\rangle\langle a|$, and write $L_\tau(B):=\tau B$ and $R_\tau(B):=B\tau$.
Applying the identity
$ab/(qa+(1-q)b)=qb+(1-q)a-q(1-q)(a-b)^2/(qa+(1-q)b)$ to \eqref{eq:tb_reference_C_definition} gives
\eq{
 (C_\tau^q)_{ij}
 &=\sum_{a,b:p_ap_b>0}\frac{p_ap_b}{qp_a+(1-q)p_b}
 (A_{\tau,i})_{ba}(A_{\tau,j})_{ab}\nonumber\\
 &=\Tr[T_{\tau,i}(qL_\tau+(1-q)R_\tau)^+(T_{\tau,j})].
 \label{eq:tb_reference_information_matrix}
}
Here, $+$ denotes the inverse on the support.
Terms with at least one zero eigenvalue vanish by the identity above, with terms where both eigenvalues vanish also interpreted as zero. Thus, this expression includes rank-deficient states.
Substituting $\omega=\omega_\tau=\tau^{\mathsf T}$ and $T_i=-T_{\tau,i}^{\mathsf T}$ into \eqref{eq:tb_information_identity} therefore gives $(C_\tau^q)^{\mathsf T}$ on the right-hand side.
That is, the Gaussian data of the reference model satisfy
\eq{
 \mathsf K_\tau^{\mathsf T}
 (S_\tau-i(2q-1)\Omega_\tau)^{-1}\mathsf K_\tau
 =(C_\tau^q)^{\mathsf T}.
 \label{eq:tb_qlan_reference_information_match}
}

If the output support is one-dimensional, the output Gaussian system is trivial, and discarding the input proves the second estimate with zero error.
If the input support is one-dimensional, then $C_\sigma^{1/2}=0$, and the assumption together with \eqref{eq:tb_reference_information_matrix} implies that all output tangent vectors vanish.
The perturbed tangent vectors are multiplied by $(1-t)$, so $\mathsf K_\rho=\mathsf K_{\rho,t}=0$. Preparing the fixed Gaussian state $\Gamma_{\rho,t}[0]$ again proves the second estimate with zero error.
We therefore consider the case where both bases are nonempty.
Since transposition preserves positive semidefiniteness, the assumption $\beta C_\sigma^q\succeq\alpha C_\rho^q$ implies, for every $0<q<1$,
\eq{
 &\beta\mathsf K_\sigma^{\mathsf T}
 (S_\sigma-i(2q-1)\Omega_\sigma)^{-1}\mathsf K_\sigma\nonumber\\
 &\qquad\succeq
 \alpha\mathsf K_\rho^{\mathsf T}
 (S_\rho-i(2q-1)\Omega_\rho)^{-1}\mathsf K_\rho.
 \label{eq:tb_qlan_information_comparison}
}
This is the hypothesis of Theorem~H.8 in Ref.~\cite{YT2026_1} for the mean matrices $\sqrt\beta\mathsf K_\sigma$ and $\sqrt\alpha\mathsf K_\rho$.
Hence, we can choose a real matrix $X$ mapping the input Gaussian means to the output Gaussian means such that
\eq{
 X\sqrt\beta\mathsf K_\sigma=\sqrt\alpha\mathsf K_\rho,
 \qquad \mathsf G_\rho-X\mathsf G_\sigma X^{\mathsf T}\succeq0.
 \label{eq:tb_gaussian_matrix_comparison}
}
This theorem (Theorem~H.8 in Ref.~\cite{YT2026_1}) assumes neither nondegeneracy of $\Omega$ nor full rank of $\mathsf K$.

The perturbation in \eqref{eq:tb_qlan_perturbation} and the centered basis defined above give
\eq{
 \mathsf G_{\tau,t}=\mathsf G_\tau+O(t),\qquad
 \mathsf K_{\tau,t}=(1-t)\mathsf K_\tau.
 \label{eq:tb_perturbed_gaussian_data}
}
The operators $F_{\tau,a,t}$ remain a centered basis, and the second identity uses the fact that the tangent vectors have trace zero.

We next construct a channel from the perturbed input Gaussian system to the perturbed output Gaussian system.
The difference matrix in \eqref{eq:tb_gaussian_matrix_comparison} may have zero eigenvalues. We therefore slightly shrink $X$ to create a positive margin that preserves nonnegativity after perturbation.
Fix $c>0$ such that $\mathsf G_\rho\succeq cI$ and $\mathsf G_{\rho,t}\succeq cI$ for small $t$, and set
\eq{
 e_t&:=\|\mathsf G_{\rho,t}-\mathsf G_\rho\|_\infty
       +\|X\|_\infty^2\|\mathsf G_{\sigma,t}-\mathsf G_\sigma\|_\infty,\nonumber\\
 a_t&:=\sqrt{1-(e_t+t)/c}.
}
Since $e_t=O(t)$, for small $t$ we have $0<a_t<1$ and $1-a_t=O(t)$.
This choice gives
\eq{
 (1-a_t^2)\mathsf G_\rho&\succeq(e_t+t)I,\nonumber\\
 (\mathsf G_{\rho,t}-\mathsf G_\rho)
 -a_t^2X(\mathsf G_{\sigma,t}-\mathsf G_\sigma)X^{\mathsf T}
 &\succeq-e_tI
 \label{eq:tb_qlan_noise_margin}
}
so that even after subtracting the loss due to perturbation, $(e_t+t)I-e_tI=tI$ remains.
By \eqref{eq:tb_gaussian_matrix_comparison},
\eq{
 \mathsf G_{\rho,t}-a_t^2X\mathsf G_{\sigma,t}X^{\mathsf T}
 &=a_t^2(\mathsf G_\rho-X\mathsf G_\sigma X^{\mathsf T})+(1-a_t^2)\mathsf G_\rho\nonumber\\
 &\quad+(\mathsf G_{\rho,t}-\mathsf G_\rho)-a_t^2X(\mathsf G_{\sigma,t}-\mathsf G_\sigma)X^{\mathsf T}
 \succeq tI.
 \label{eq:tb_strict_gaussian_noise}
}
Thus, with $N_t:=S_{\rho,t}-a_t^2XS_{\sigma,t}X^{\mathsf T}$, the map
\eq{
 \mathcal G_t^*(W_{\rho,t}(z))
 =e^{-z^{\mathsf T}N_tz/2}W_{\sigma,t}(a_tX^{\mathsf T}z)
 \label{eq:tb_perturbed_gaussian_channel}
}
is completely positive and unital~\cite[Proposition~24]{DammeierWerner2023Hybrid}.
Here, $\mathcal G_t^*$ is the dual map acting on observables. For an input state $\Theta$ and an output observable $A$, its action on states is defined by
\eq{
 (\mathcal G_t(\Theta))(A)=\Theta(\mathcal G_t^*(A)).
 \label{eq:tb_qlan_gaussian_duality}
}
Complete positivity requires the matrix combining the added-noise covariance $N_t$ with the difference between the output and input commutation relations to be nonnegative.
This matrix is precisely
\eq{
 N_t+i\bigl(\Omega_{\rho,t}-a_t^2X\Omega_{\sigma,t}X^{\mathsf T}\bigr)
 =\mathsf G_{\rho,t}-a_t^2X\mathsf G_{\sigma,t}X^{\mathsf T}
 \succeq tI
 \label{eq:tb_qlan_gaussian_cp_matrix}
}
so \eqref{eq:tb_strict_gaussian_noise} satisfies the cited complete-positivity condition.
Furthermore, $N_t\succeq tI$ bounds every output characteristic function in absolute value by $e^{-t\|z\|^2/2}$.
The output therefore has an ordinary classical--quantum density~\cite[Proposition~7]{DammeierWerner2023Hybrid}, and $\mathcal G_t$ can be composed as a CPTP map on the corresponding density spaces.
For a Gaussian input state with mean $m$, substituting \eqref{eq:tb_gaussian_character} into the dual representation gives
\eq{
 &(\mathcal G_t(\Gamma_{\sigma,t}[m]))(W_{\rho,t}(z))\nonumber\\
 &\quad=\exp\!\left(iz^{\mathsf T}a_tXm
 -\tfrac12z^{\mathsf T}(N_t+a_t^2XS_{\sigma,t}X^{\mathsf T})z\right)\nonumber\\
 &\quad=\exp\!\left(iz^{\mathsf T}a_tXm-\tfrac12z^{\mathsf T}S_{\rho,t}z\right).
 \label{eq:tb_qlan_gaussian_character_action}
}
Thus, this channel makes the covariance agree exactly with the target and maps the mean as $m\mapsto a_tXm$.
Using also $\mathsf K_{\tau,t}=(1-t)\mathsf K_\tau$ and \eqref{eq:tb_gaussian_matrix_comparison}, it maps the means of the local models as
\eq{
 \sqrt\beta\mathsf K_{\sigma,t}x
 \longmapsto a_t\sqrt\alpha\mathsf K_{\rho,t}x.
 \label{eq:tb_gaussian_mean_action}
}

The error due to the remaining discrepancy in the mean converges to zero uniformly on a fixed ball $\|x\|\le R$ as $t\to0$.
Fix $0<\eta<c$. Since $\mathsf G_{\rho,t}-\eta I\succeq0$, writing $\widetilde\Gamma_t[y]$ for the Gaussian state with covariance $S_{\rho,t}-\eta I$, we have the mixture representation
\eq{
 \Gamma_{\rho,t}[m]=\int p_m(y)\widetilde\Gamma_t[y]dy,
 \qquad p_m=\mathcal N(m,\eta I).
}
This follows by evaluating the Gaussian integral of the characteristic functions.
Consequently,
\eq{
 \|\Gamma_{\rho,t}[m]-\Gamma_{\rho,t}[m']\|_1
 \le\int|p_m-p_{m'}|\le\frac{\|m-m'\|}{\sqrt\eta}.
 \label{eq:tb_gaussian_shift_bound}
}
The last estimate follows by differentiating the density along $m_s=m'+s(m-m')$,
using Cauchy--Schwarz to obtain $\int|\partial_sp_{m_s}|\le\|m-m'\|/\sqrt\eta$, and integrating.
The error in \eqref{eq:tb_gaussian_mean_action} is therefore $O_R(t)$ on fixed balls.

This proves the conversion estimate \eqref{eq:tb_qlan_gaussian_goal} between Gaussian models.

Finally, we prove the third estimate \eqref{eq:tb_qlan_output_goal}.
Using the second QLAN convergence in \eqref{eq:tb_nondegenerate_lan} and \eqref{eq:tb_qlan_copy_rescaling} with $(L_n,a)=(N_n,\alpha)$, for fixed $t>0$ we obtain
\eq{
 \sup_{\|x\|\le R}\left\|
 \mathcal B_{\rho,t,N_n}(\Gamma_{\rho,t}[\sqrt\alpha\mathsf K_{\rho,t}x])
 -\omega_{\rho,t,x/\sqrt n}^{\otimes N_n}\right\|_1\longrightarrow0.
 \label{eq:tb_qlan_output_perturbed}
}
It remains to reverse the perturbation of this output state. For this purpose, we use a map that exactly recovers the reference state.
With $\omega=\omega_\rho$, $\kappa=\kappa_\rho$, and $\omega_t=(1-t)\omega+t\kappa$, set
\eq{
 \mathcal R_t(B)
 =(1-t)\sqrt\omega\omega_t^{-1/2}B\omega_t^{-1/2}\sqrt\omega
  +t\Tr(\kappa\omega_t^{-1/2}B\omega_t^{-1/2})\omega.
 \label{eq:tb_qlan_recovery_map}
}
Each term is completely positive, their traces sum to $\Tr B$, and $\mathcal R_t(\omega_t)=\omega$.

Moreover, $\mathcal R_t\mathcal E_{\rho,t}=\mathrm{id}+O(t)$ as $t\downarrow0$.
Here, $O(t)$ is measured in the map norm induced by the Hilbert--Schmidt norm,
$\|\Phi\|_{2\to2}:=\sup_{B\ne0}\|\Phi(B)\|_2/\|B\|_2$.
Indeed, in the common eigenbasis of $\omega$ and $\kappa$, we can perform a Taylor expansion for each eigenvalue. Since the eigenvalues of $\omega$ have a positive lower bound, we obtain
\eq{
 \omega_t^{-1/2}=\omega^{-1/2}+O(t),\qquad
 \|\mathcal R_t\mathcal E_{\rho,t}-\mathrm{id}\|_{2\to2}\le Ct
 \label{eq:tb_qlan_recovery_map_norm}
}
uniformly for small $t$. The latter estimate follows by substituting the former expansion into \eqref{eq:tb_qlan_recovery_map} and using $\mathcal E_{\rho,t}=\mathrm{id}+O(t)$.
Exact recovery of the reference state gives, for $x$ in a fixed ball,
\eq{
 &\mathcal R_t\mathcal E_{\rho,t}(\omega_{\rho,x/\sqrt n})
       -\omega_{\rho,x/\sqrt n}\nonumber\\
 &\qquad=(\mathcal R_t\mathcal E_{\rho,t}-\mathrm{id})
       (\omega_{\rho,x/\sqrt n}-\omega_\rho)
 \label{eq:tb_qlan_recovery_centering}
}
so its Hilbert--Schmidt norm is $O_R(t/\sqrt n)$.
For small $t$ and sufficiently large $n$, the two states being compared have a common strictly positive lower bound on the support $E_\rho$.
By \eqref{eq:tb_faithful_product_bound},
\eq{
 \sup_{\|x\|\le R}\left\|
 \mathcal R_t^{\otimes N_n}(\omega_{\rho,t,x/\sqrt n}^{\otimes N_n})
 -\omega_{\rho,x/\sqrt n}^{\otimes N_n}\right\|_1
 \le C_Rt\sqrt{N_n/n}\le C'_Rt.
 \label{eq:tb_qlan_recovery_product_error}
}

Applying norm contractivity of $\mathcal R_t^{\otimes N_n}$ to \eqref{eq:tb_qlan_output_perturbed} and adding this estimate proves the third estimate \eqref{eq:tb_qlan_output_goal}.
This establishes all three estimates.
The composition \eqref{eq:tb_local_channel_constructed}, the error bound \eqref{eq:tb_qlan_total_error}, and the subsequent diagonal selection and group averaging therefore prove the local comparison \eqref{eq:ref_local}, completing the proof of the lemma.
\endgroup
\end{proofof}

The final lemma obtains a channel on the physical system from an approximation on the reference system and holds for finite numbers of copies.

\begin{lemma}\label{lem:app_strong_iid_reconstruction}
Consider unitary representations $\widehat{U},\widehat{U}'$ and states $\alpha,\beta$ on finite-dimensional spaces $\calK,\calK'$, respectively, together with a character $\gamma\in\CharG$.
For a trace-nonincreasing completely positive map $\calD:\calB(E_\beta)\to\calB(E_\alpha)$, set
\eq{
\delta:=\bigl\|\Omega_{\alpha,\widehat{U}}
-(\mathrm{id}\otimes\calD)\Omega_{\beta,\gamma\widehat{U}'}\bigr\|_1
\label{eq:app_strong_iid_reconstruction_error}
}
Here, $\mathrm{id}$ is the identity map on each irreducible representation space.
Then there exist a finite-dimensional auxiliary system $A$, a representation $W$ on it, a pure state $\eta_A$ satisfying $W_g\eta_A=\gamma(g)\eta_A$, and a $(\widehat{U},\widehat{U}'\otimes W)$-strong covariant CPTP map $\Phi$ defined on the entire input space $\calB(\calK)$ such that
\eq{
\|\Phi(\alpha)-\beta\otimes\eta_A\|_1\le2\sqrt{2\delta}
\label{eq:app_strong_iid_reconstruction_bound}
}

\end{lemma}
\begin{proofof}{Lemma~\ref{lem:app_strong_iid_reconstruction}}
Below, we apply the same formula for the fidelity $F(A,B)=\|\sqrt A\sqrt B\|_1$, defined in the main text, also to positive operators with trace less than $1$.

We first extend the map $\calD$ on the reference system to a trace-preserving map.
Let $\calD^*$ be its dual with respect to the trace inner product, and set
\eq{
D_R:=I_{E_\beta}-\calD^*(I_{E_\alpha})\ge0,
\qquad \widehat E_\alpha:=E_\alpha\oplus\mathbb C\ket f
}

Here, $\ket f$ is a unit vector orthogonal to $E_\alpha$.
The map into $\widehat E_\alpha$ defined by
\eq{
\widehat{\calD}(B)
:=\calD(B)\oplus\Tr[D_RB]\ket f\bra f
\label{eq:app_strong_iid_reference_completion}
}
is completely positive.
Indeed, diagonalizing $D_R=\sum_jd_j\ket j\bra j$ gives the Kraus operators $\sqrt{d_j}\ket f\bra j$ for the second term, while the Kraus operators for the first term are those of $\calD$, up to the embedding of $E_\alpha$.
Moreover,
\eq{
\Tr\widehat{\calD}(B)
=\Tr[\calD^*(I)B]+\Tr[D_RB]=\Tr B
}
so $\widehat{\calD}$ is CPTP.

Extend $\Omega_{\alpha,\widehat{U}}$ by zero to $\widehat E_\alpha$, and set
\eq{
Q&:=(\mathrm{id}\otimes\calD)\Omega_{\beta,\gamma\widehat{U}'},\nonumber\\
\widehat Q&:=(\mathrm{id}\otimes\widehat{\calD})\Omega_{\beta,\gamma\widehat{U}'}
}

The operator $\widehat Q$ is a normalized state. Its block corresponding to $E_\alpha$ is $Q$, and the block corresponding to the flag is positive and has trace $1-\Tr Q$.
On the other hand, $\Tr\Omega_{\alpha,\widehat{U}}=1$ and the definition of the trace norm imply
\eq{
0\le1-\Tr Q
=\Tr(\Omega_{\alpha,\widehat{U}}-Q)\le\delta.
}
Since the two blocks are orthogonal,
\eq{
\|\Omega_{\alpha,\widehat{U}}-\widehat Q\|_1
=\|\Omega_{\alpha,\widehat{U}}-Q\|_1+1-\Tr Q
\le2\delta.
}
The Fuchs--van de Graaf inequalities for normalized states~\cite{Fuchs1999} give
\eq{
F(\Omega_{\alpha,\widehat{U}},\widehat Q)
\ge1-\frac12\|\Omega_{\alpha,\widehat{U}}-\widehat Q\|_1
\ge1-\delta.
\label{eq:app_strong_iid_reference_fidelity}
}

Next, we use this comparison on the reference system to construct a contraction on the physical system.
Take the irreducible decompositions of $\widehat{U}$ and $\gamma\widehat{U}'$ as
\eq{
\calK&=\bigoplus_\lambda\calV_\lambda\otimes\calM_{\alpha,\lambda},\nonumber\\
\calK'&=\bigoplus_\lambda\calV_\lambda\otimes\calM_{\beta,\lambda}
}

In these decompositions,
$\widehat{U}_g=\bigoplus_\lambda U_g^\lambda\otimes I$
and
$\gamma(g)\widehat{U}'_g=\bigoplus_\lambda U_g^\lambda\otimes I$,
where absent irreducible components are assigned zero-dimensional multiplicity spaces.

Choose a Kraus representation $\widehat{\calD}(B)=\sum_j\widehat D_jB\widehat D_j^\dagger$, and use an orthonormal basis $\{\ket j\}$ of an auxiliary space $\calE$ to define
\eq{
J&:=\sum_j\widehat D_j\otimes\ket j:
E_\beta\longrightarrow\widehat E_\alpha\otimes\calE,\nonumber\\
J^\dagger J&=\sum_j\widehat D_j^\dagger\widehat D_j=I,\qquad
\widehat{\calD}(B)=\Tr_\calE[JBJ^\dagger]
}
and set $\ket\Xi:=(I_{\calK'}\otimes J)\ket{\Psi_\beta}$.
Since $J^\dagger J=I$, the vector $\ket\Xi$ is normalized, and its marginal state on $\calK'$ is $\beta$.
We also naturally embed the reference system of $\ket{\Psi_\alpha}$ into $\widehat E_\alpha$.

Let $\ket{\Psi_{\alpha,\lambda}}$ and $\ket{\Xi_\lambda}$ be the vectors obtained by projecting their respective physical systems onto the irreducible component $\lambda$.
Reorder the tensor factors and regard $\calV_\lambda\otimes\widehat E_\alpha$ as the common system.
The two vectors are then, respectively, possibly unnormalized purifications of $\Omega_{\alpha,\widehat{U}}^\lambda$ and $\widehat Q^\lambda$.
Their purifying systems are $\calM_{\alpha,\lambda}$ and $\calM_{\beta,\lambda}\otimes\calE$, respectively.
By Uhlmann's theorem~\cite{Uhlmann1976}, we can choose a contraction between the multiplicity spaces,
\eq{
V_\lambda:\calM_{\alpha,\lambda}
\longrightarrow\calM_{\beta,\lambda}\otimes\calE,
\qquad V_\lambda^\dagger V_\lambda\le I
}
such that
\eq{
\bra{\Xi_\lambda}
(I_{\calV_\lambda}\otimes V_\lambda\otimes I_{\widehat E_\alpha})
\ket{\Psi_{\alpha,\lambda}}
=F(\Omega_{\alpha,\widehat{U}}^\lambda,\widehat Q^\lambda)\ge0
\label{eq:app_strong_iid_block_uhlmann}
}

Here, the systems are ordered so that $V_\lambda$ acts only on the multiplicity space.
This form of Uhlmann's theorem follows by isometrically embedding both multiplicity spaces into the same larger space and compressing an optimal unitary on that space to a map between the original two spaces.
The compression is a contraction, and since the original purifications are supported on the embedded subspaces, the value of the overlap is unchanged.
For unnormalized vectors, divide each vector by its norm, apply the theorem, and then multiply by the two norms.
If either block is zero, both sides vanish, and we can take $V_\lambda=0$.

Define the map on the physical system by
\eq{
V:=\bigoplus_\lambda I_{\calV_\lambda}\otimes V_\lambda:
\calK\longrightarrow\calK'\otimes\calE
}

Since each $V_\lambda$ is a contraction, $V^\dagger V\le I$. Moreover, it acts as the identity on each irreducible representation space, so
\eq{
V\widehat{U}_g=(\gamma(g)\widehat{U}'_g\otimes I_\calE)V
\label{eq:app_strong_iid_success_intertwiner}
}

Since fidelity is additive under direct sums, summing \eqref{eq:app_strong_iid_block_uhlmann} gives
\eq{
\bra\Xi(V\otimes I_{\widehat E_\alpha})\ket{\Psi_\alpha}
&=\sum_\lambda F(\Omega_{\alpha,\widehat{U}}^\lambda,\widehat Q^\lambda)\nonumber\\
&=F(\Omega_{\alpha,\widehat{U}},\widehat Q)\ge1-\delta.
\label{eq:app_strong_iid_purification_overlap}
}
Here, the use of a single direct-sum operator $V$, without measuring the irreducible components, also preserves coherence between different irreducible components.

Define $\calT(B):=\Tr_\calE[VBV^\dagger]$.
This is a trace-nonincreasing completely positive map. Using an orthonormal basis $\{\ket a\}$ of $\calE$, its Kraus operators are given by
\eq{
K_a:=(I_{\calK'}\otimes\bra a)V,
\qquad K_a\widehat{U}_g=\gamma(g)\widehat{U}'_gK_a
}

Since $(V\otimes I)\ket{\Psi_\alpha}$ and $\ket\Xi$ are purifications of $\calT(\alpha)$ and $\beta$, respectively, another application of Uhlmann's theorem gives
\eq{
F(\calT(\alpha),\beta)
\ge\left|\bra\Xi(V\otimes I)\ket{\Psi_\alpha}\right|
\ge1-\delta.
\label{eq:app_strong_iid_success_fidelity}
}
If the first purification is unnormalized but has positive norm, the same inequality follows by dividing by its norm and then multiplying by that norm again.
If its norm is zero, both the overlap and the fidelity vanish, so the first inequality still holds.

Finally, we complete $\calT$ to a strong covariant CPTP map on the entire input space. Set
\eq{
D:=I_\calK-\sum_aK_a^\dagger K_a=I_\calK-V^\dagger V\ge0
}

By \eqref{eq:app_strong_iid_success_intertwiner} and the unitarity of the representations,
\eq{
\widehat{U}_g^\dagger V^\dagger V\widehat{U}_g
=V^\dagger(\overline{\gamma(g)}(\widehat{U}'_g)^\dagger\otimes I)
(\gamma(g)\widehat{U}'_g\otimes I)V
=V^\dagger V
}

Thus, both $D$ and $\sqrt D$ commute with $\widehat{U}_g$.

Define the auxiliary system and its representation by
\eq{
\calH_A&:=\mathbb C\ket\gamma
\oplus(\overline{\calK'}\otimes\calK),\nonumber\\
W_g&:=\gamma(g)\oplus(\overline{\widehat{U}'_g}\otimes \widehat{U}_g)
\label{eq:app_strong_iid_auxiliary_completion}
}

Here, $\overline{\calK'}$ is the complex conjugate space of $\calK'$, and $\ket\gamma$ is a unit vector in the first, one-dimensional direct-sum component.
Let $\dim\calK'=d'$, and for an orthonormal basis $\{\ket j\}$ set
\eq{
\ket{\Omega_{\calK'}}
:=\frac1{\sqrt{d'}}\sum_{j=1}^{d'}\ket j\otimes\ket{\overline j}
}

Unitarity gives
$(\widehat{U}'_g\otimes\overline{\widehat{U}'_g})\ket{\Omega_{\calK'}}=\ket{\Omega_{\calK'}}$.

Therefore, the isometry whose auxiliary-system factor lies in the second direct-sum component,
\eq{
J_f\ket\psi:=\ket{\Omega_{\calK'}}\otimes\ket\psi
\quad\bigl(\ket\psi\in\calK\bigr)
}
satisfies
\eq{
(\widehat{U}'_g\otimes W_g)J_f\ket\psi
=\ket{\Omega_{\calK'}}\otimes \widehat{U}_g\ket\psi
=J_f\widehat{U}_g\ket\psi
}

Set $\eta_A:=\ket\gamma\bra\gamma$ and define
\eq{
\Phi(B)
:=\calT(B)\otimes\eta_A
+J_f\sqrt D\,B\sqrt D\,J_f^\dagger
\label{eq:app_strong_iid_full_channel_completion}
}

The Kraus operators of this map are $L_a\ket\psi:=(K_a\ket\psi)\otimes\ket\gamma$ and $L_f:=J_f\sqrt D$.
They satisfy
\eq{
\sum_aL_a^\dagger L_a+L_f^\dagger L_f
=\sum_aK_a^\dagger K_a+D=I_\calK
}
so $\Phi$ is CPTP on the entire input space.
Furthermore,
\eq{
L_a\widehat{U}_g&=(\widehat{U}'_g\otimes W_g)L_a,\nonumber\\
L_f\widehat{U}_g&=J_f\widehat{U}_g\sqrt D=(\widehat{U}'_g\otimes W_g)L_f
}

Thus, for every input operator $B$,
\eq{
\Phi(\widehat{U}_gB)
&=\sum_aL_a\widehat{U}_gB L_a^\dagger+L_f\widehat{U}_gB L_f^\dagger\nonumber\\
&=(\widehat{U}'_g\otimes W_g)\Phi(B),
}
that is, $\Phi$ is $(\widehat{U},\widehat{U}'\otimes W)$-strong covariant.
Moreover, $W_g\eta_A=\gamma(g)\eta_A$ shows that $\eta_A$ is the required strong symmetric pure state.

The two terms in \eqref{eq:app_strong_iid_full_channel_completion} have orthogonal supports on the auxiliary system, and the target state $\beta\otimes\eta_A$ is supported only on the first component.
Consequently,
\eq{
F(\Phi(\alpha),\beta\otimes\eta_A)
=F(\calT(\alpha),\beta)\ge1-\delta.
}
If $0\le\delta\le1$, another application of the Fuchs--van de Graaf inequalities gives
\eq{
\|\Phi(\alpha)-\beta\otimes\eta_A\|_1
&\le2\sqrt{1-F(\Phi(\alpha),\beta\otimes\eta_A)^2}\nonumber\\
&\le2\sqrt{1-(1-\delta)^2}
\le2\sqrt{2\delta}.
}
If $\delta>1$, the same bound holds because the trace-norm distance between two normalized states is at most $2$.
This proves \eqref{eq:app_strong_iid_reconstruction_bound}.
\end{proofof}

\subsection{Proof of Theorem~\ref{thm:compact_Lie_strong_iid}}
\label{app:strong_iid_proof}

\begin{proofof}{Theorem~\ref{thm:compact_Lie_strong_iid}}
We first prove necessity.
Assume a sequence of conversions satisfying \eqref{eq:app_strong_iid_error}, and choose an infinite set $\mathcal N$ satisfying \eqref{eq:app_strong_iid_subsequence_error}.
Throughout the necessity proof, $N$ denotes a copy number in $\mathcal N$, and all limits written as $N\to\infty$ are taken along this set.
In particular, $\epsilon_N\to0$ and $M_N/N\to r$.

We first establish (W). Set $\Lambda_N:=\Tr_{A_N}\circ\Phi_N$.
A strong covariant map is also weak covariant, and taking the partial trace gives
\eq{
\Lambda_N(U_g^{\otimes N}AU_g^{\dagger\otimes N})
=U_g'^{\otimes M_N}\Lambda_N(A)U_g'^{\dagger\otimes M_N}
}
Thus, $\Lambda_N$ is a weak covariant CPTP map, and trace-norm contractivity gives
\eq{
\|\Lambda_N(\rho^{\otimes N})-\sigma^{\otimes M_N}\|_1\le\epsilon_N
\label{eq:app_strong_iid_reduced_error}
}

Hence weak covariant i.i.d.\ conversion at rate $r$ is possible along the same subsequence $\mathcal N$.
The necessity argument in Ref.~\cite[Theorem~1 and Appendix~B]{YT2026_1} applies to this subsequence because it uses estimates at each copy number and the limit of vanishing error.
This gives the metric adjusted quantum geometric tensor inequalities and the inclusion of symmetry subgroups, namely (W).

We next prove (S).
Strong covariance and trace preservation give, for every $g\in G$,
\eq{
&\Tr\!\left[(U_g'^{\otimes M_N}\otimes W_N(g))
\Phi_N(\rho^{\otimes N})\right]\nonumber\\
&\quad=\Tr\!\left[\Phi_N(U_g^{\otimes N}\rho^{\otimes N})\right]\nonumber\\
&\quad=\Tr[U_g^{\otimes N}\rho^{\otimes N}]
=\chi_\rho(g)^N
}
The characteristic function of the target state is
\eq{
&\Tr\!\left[(U_g'^{\otimes M_N}\otimes W_N(g))
(\sigma^{\otimes M_N}\otimes\eta_{A_N})\right]\nonumber\\
&\quad=\chi_\sigma(g)^{M_N}\Tr[W_N(g)\eta_{A_N}]\nonumber\\
&\quad=\gamma_N(g)\chi_\sigma(g)^{M_N}
}

Since a unitary operator has operator norm $1$, the inequality $|\Tr[AB]|\le\|A\|_\infty\|B\|_1$ gives
\eq{
\sup_{g\in G}\left|\chi_\rho(g)^N
-\gamma_N(g)\chi_\sigma(g)^{M_N}\right|\le\epsilon_N
\label{eq:app_strong_iid_character_convergence}
}

We first evaluate this inequality near the identity.
Using the common notation $(\tau^\circ,U^\circ,\{X_i^\circ\})$ introduced in \eqref{eq:app_strong_iid_common_notation} and $X^\circ(\bm\lambda)$ from \eqref{eq:app_strong_iid_common_generator}, set
\eq{
\mu_{\tau^\circ}(\bm{\lambda})&:=\bm\mu_{\tau^\circ}\cdot\bm{\lambda},\nonumber\\
v_{\tau^\circ}(\bm{\lambda})&:=\Tr[\tau^\circ X^\circ(\bm{\lambda})^2]-\mu_{\tau^\circ}(\bm{\lambda})^2\nonumber\\
&=\bm{\lambda}^{\mathsf T}\operatorname{Re}\VNS(\tau^\circ)\bm{\lambda}
}

Finite dimensionality gives the expansion
\eq{
\chi_{\tau^\circ}(g(\bm{\lambda}/\sqrt N))
={}&1+\frac{i\mu_{\tau^\circ}(\bm{\lambda})}{\sqrt N}
-\frac{\Tr[\tau^\circ X^\circ(\bm{\lambda})^2]}{2N}
+O_R(N^{-3/2})
}
uniformly on $\|\bm{\lambda}\|\le R$ for every fixed $R<\infty$.
Using $\log(1+z)=z-z^2/2+O(|z|^3)$ and choosing the branch near the identity with $\log1=0$, we obtain
\eq{
\log\chi_{\tau^\circ}(g(\bm{\lambda}/\sqrt N))
={}&\frac{i\mu_{\tau^\circ}(\bm{\lambda})}{\sqrt N}-\frac{\Tr[\tau^\circ X^\circ(\bm{\lambda})^2]-\mu_{\tau^\circ}(\bm{\lambda})^2}{2N}
+O_R(N^{-3/2})\nonumber\\
={}&\frac{i\mu_{\tau^\circ}(\bm{\lambda})}{\sqrt N}
-\frac{v_{\tau^\circ}(\bm{\lambda})}{2N}+O_R(N^{-3/2})
\label{eq:app_strong_iid_log_expansion}
}

Since $M_N=O(N)$, it follows that
\eq{
\chi_\rho(g(\bm{\lambda}/\sqrt N))^N
&=e^{i\sqrt N\mu_\rho(\bm{\lambda})-v_\rho(\bm{\lambda})/2}
\bigl(1+O_R(N^{-1/2})\bigr),\nonumber\\
\chi_\sigma(g(\bm{\lambda}/\sqrt N))^{M_N}
&=e^{iM_N\mu_\sigma(\bm{\lambda})/\sqrt N-M_Nv_\sigma(\bm{\lambda})/(2N)}
\bigl(1+O_R(N^{-1/2})\bigr).
\label{eq:app_strong_iid_gaussian_expansion}
}
For the character, the identity
\eq{
\gamma_N(g(\bm{\lambda}/\sqrt N))
=e^{i\bm a(\gamma_N)\cdot\bm{\lambda}/\sqrt N}
}
holds exactly.

Comparing moduli in \eqref{eq:app_strong_iid_character_convergence} and taking $N\to\infty$, we obtain, for every real vector $\bm{\lambda}$,
\eq{
e^{-v_\rho(\bm{\lambda})/2}=e^{-rv_\sigma(\bm{\lambda})/2}
}
Thus, $v_\rho(\bm{\lambda})=rv_\sigma(\bm{\lambda})$.
For the standard basis vectors $\bm e_i$, setting $\bm{\lambda}=\bm e_i$ gives equality of the diagonal entries, while $\bm{\lambda}=\bm e_i+\bm e_j$ gives
\eq{
2\operatorname{Re}\VNS(\rho)_{ij}
&=v_\rho(\bm e_i+\bm e_j)-v_\rho(\bm e_i)-v_\rho(\bm e_j)\nonumber\\
&=r\bigl(v_\sigma(\bm e_i+\bm e_j)-v_\sigma(\bm e_i)-v_\sigma(\bm e_j)\bigr)\nonumber\\
&=2r\operatorname{Re}\VNS(\sigma)_{ij}
}
Therefore,
\eq{
\operatorname{Re}\VNS(\rho)=r\operatorname{Re}\VNS(\sigma)
\label{eq:app_strong_iid_real_covariance}
}

We next compare the imaginary parts of the covariance matrices.
Write the structure constants of the common Lie algebra as real numbers $f_{ij}^{\ k}$ such that
\eq{
[X_i,X_j]&=i\sum_k f_{ij}^{\ k}X_k,\nonumber\\
[X'_i,X'_j]&=i\sum_k f_{ij}^{\ k}X'_k
}
For fixed $i,j$, set
\eq{
\bm f_{ij}:=(f_{ij}^{\ 1},\ldots,f_{ij}^{\ d_G}),\qquad
b_{\tau^\circ,ij}:=\sum_k f_{ij}^{\ k}\mu_{\tau^\circ,k}
\quad(\tau^\circ=\rho,\sigma)
}

The generators of the one-dimensional representation $\gamma_N$ are scalars, so
\eq{
\sum_k f_{ij}^{\ k}a_k(\gamma_N)=0
}
Consequently, for every real $t$,
\eq{
\gamma_N(g(t\bm f_{ij}/N))
=\exp\!\left(\frac{it}{N}\sum_k f_{ij}^{\ k}a_k(\gamma_N)\right)=1
}
holds exactly.
For fixed $t$, the first-order expansion of the characteristic function gives
\eq{
\chi_{\tau^\circ}(g(t\bm f_{ij}/N))
=1+\frac{itb_{\tau^\circ,ij}}{N}+O_t(N^{-2})
}

Using $\log(1+z)=z+O(|z|^2)$ and $M_N=O(N)$, we obtain
\eq{
\chi_\rho(g(t\bm f_{ij}/N))^N
&=\exp\!\left(itb_{\rho,ij}+O_t(N^{-1})\right)
\longrightarrow e^{itb_{\rho,ij}},\nonumber\\
\chi_\sigma(g(t\bm f_{ij}/N))^{M_N}
&=\exp\!\left(it\frac{M_N}{N}b_{\sigma,ij}+O_t(N^{-1})\right)
\longrightarrow e^{itrb_{\sigma,ij}}
}

Applying \eqref{eq:app_strong_iid_character_convergence} at $g(t\bm f_{ij}/N)$ and using $\gamma_N(g(t\bm f_{ij}/N))=1$ and $\epsilon_N\to0$ gives
\eq{
e^{itb_{\rho,ij}}=e^{itrb_{\sigma,ij}}
\qquad(t\in\mathbb R)
}
Differentiating both sides at $t=0$ yields $b_{\rho,ij}=r b_{\sigma,ij}$.
By the definition of $\VNS$,
\eq{
2i\operatorname{Im}\VNS(\rho)_{ij}
&=\Tr[\rho[X_i,X_j]]=ib_{\rho,ij},\nonumber\\
2i\operatorname{Im}\VNS(\sigma)_{ij}
&=\Tr[\sigma[X'_i,X'_j]]=ib_{\sigma,ij}
}

Thus the imaginary parts of the covariance matrices are also proportional. Together with \eqref{eq:app_strong_iid_real_covariance}, this gives
\eq{
\VNS(\rho)=r\VNS(\sigma)
\label{eq:app_strong_iid_full_covariance}
}

We next establish the condition on the auxiliary characters.
We first compare phases near the identity to prove \eqref{eq:app_strong_iid_small_mean}.
Let $\bm d_N$ be defined by \eqref{eq:app_strong_iid_d}, and set
\eq{
a_\rho(\bm{\lambda}):=e^{-v_\rho(\bm{\lambda})/2},\qquad
a_{\sigma,N}(\bm{\lambda}):=e^{-M_Nv_\sigma(\bm{\lambda})/(2N)}
}

From \eqref{eq:app_strong_iid_character_convergence} and \eqref{eq:app_strong_iid_gaussian_expansion}, there is a constant $c_R$ independent of $N$ such that
\eq{
\sup_{\|\bm{\lambda}\|\le R}
\left|a_\rho(\bm{\lambda})e^{i\bm d_N\cdot\bm{\lambda}/\sqrt N}
-a_{\sigma,N}(\bm{\lambda})\right|
\le\epsilon_N+c_RN^{-1/2}
}

By \eqref{eq:app_strong_iid_real_covariance} and $M_N/N\to r$, the convergence $a_{\sigma,N}\to a_\rho$ is uniform on this ball.
Moreover, $\inf_{\|\bm{\lambda}\|\le R}a_\rho(\bm{\lambda})>0$, so the triangle inequality gives
\eq{
\sup_{\|\bm{\lambda}\|\le R}
\left|e^{i\bm d_N\cdot\bm{\lambda}/\sqrt N}-1\right|\longrightarrow0
\label{eq:app_strong_iid_uniform_phase}
}

We now deduce $\|\bm d_N\|/\sqrt N\to0$.
Otherwise, for some $0<\eta\le\pi$ and a subsequence in $\mathcal N$, the vectors $\bm t_N:=\bm d_N/\sqrt N$ satisfy $\|\bm t_N\|\ge\eta$.
Along this subsequence, choose
\eq{
\bm{\lambda}_N:=\frac{\eta\bm t_N}{\|\bm t_N\|^2}
}
Then $\|\bm{\lambda}_N\|\le1$ and $\bm t_N\cdot\bm{\lambda}_N=\eta$.
Consequently,
\eq{
\left|e^{i\bm d_N\cdot\bm{\lambda}_N/\sqrt N}-1\right|
=|e^{i\eta}-1|>0
}
contradicting \eqref{eq:app_strong_iid_uniform_phase} with $R=1$. This proves \eqref{eq:app_strong_iid_small_mean}.

Next, compare the characteristic functions on the whole group.
For any state $\tau$,
\eq{
\mathrm{Sym}_{G,\mathrm{strong}}(\tau)=\{g\in G:|\chi_\tau(g)|=1\}
\label{eq:app_strong_iid_unit_modulus}
}
Indeed, $U_g\tau=c\tau$ implies $\chi_\tau(g)=c$.
Conversely, if $c:=\chi_\tau(g)$ satisfies $|c|=1$, then the Hilbert--Schmidt norm satisfies
\eq{
\|(U_g-cI)\sqrt\tau\|_2^2
&=2-2\operatorname{Re}(\overline c\chi_\tau(g))=0
}
Hence $U_g\sqrt\tau=c\sqrt\tau$, and therefore $U_g\tau=c\tau$.
If $g\in \mathrm{Sym}_{G,\mathrm{strong}}(\rho)\setminus \mathrm{Sym}_{G,\mathrm{strong}}(\sigma)$, then
\eq{
\left|\chi_\rho(g)^N-\gamma_N(g)\chi_\sigma(g)^{M_N}\right|
\ge1-|\chi_\sigma(g)|^{M_N}\longrightarrow1
}
contradicts \eqref{eq:app_strong_iid_character_convergence} and $\epsilon_N\to0$.
For $g\in \mathrm{Sym}_{G,\mathrm{strong}}(\sigma)\setminus \mathrm{Sym}_{G,\mathrm{strong}}(\rho)$, the lower bound $1-|\chi_\rho(g)|^N$ gives the same contradiction.
Thus, $\mathrm{Sym}_{G,\mathrm{strong}}(\rho)=\mathrm{Sym}_{G,\mathrm{strong}}(\sigma)=:S^{\mathrm{strong}}_{G,\rho,\sigma}$.

On $S^{\mathrm{strong}}_{G,\rho,\sigma}$, we have $\chi_\rho=c_\rho$ and $\chi_\sigma=c_\sigma$, which are continuous one-dimensional characters.
For distinct characters $\alpha,\beta$, the character $\theta:=\alpha\overline\beta$ is nontrivial.
Choose $k_0\in S^{\mathrm{strong}}_{G,\rho,\sigma}$ with $\theta(k_0)\ne1$. Invariance of the normalized Haar measure gives
\eq{
\int_{S^{\mathrm{strong}}_{G,\rho,\sigma}}\theta(k)\,dk
&=\int_{S^{\mathrm{strong}}_{G,\rho,\sigma}}\theta(k_0k)\,dk\nonumber\\
&=\theta(k_0)\int_{S^{\mathrm{strong}}_{G,\rho,\sigma}}\theta(k)\,dk
}
and hence $\int_{S^{\mathrm{strong}}_{G,\rho,\sigma}}\theta(k)\,dk=0$.
Therefore,
\eq{
\int_{S^{\mathrm{strong}}_{G,\rho,\sigma}}|\alpha(k)-\beta(k)|^2\,dk
&=2-2\operatorname{Re}\int_{S^{\mathrm{strong}}_{G,\rho,\sigma}}\theta(k)\,dk\nonumber\\
&=2.
}
Apply this to $\alpha=c_\rho^{\,N}$ and $\beta=\gamma_N|_{S^{\mathrm{strong}}_{G,\rho,\sigma}}\,c_\sigma^{M_N}$.
By \eqref{eq:app_strong_iid_character_convergence}, this integral of the squared absolute difference is at most $\epsilon_N^2$, so the two characters agree for all sufficiently large $N\in\mathcal N$, giving
\eq{
\gamma_N|_{S^{\mathrm{strong}}_{G,\rho,\sigma}}=c_\rho^{\,N}c_\sigma^{-M_N}
}

Together with \eqref{eq:app_strong_iid_small_mean}, already proved above, this establishes (iii) of Lemma~\ref{lem:app_strong_iid_character_lift} for the same $\mathcal N$.
Part (i) of that lemma then gives the residual-mean-charge condition $m_G(\rho)=r m_G(\sigma)$.
Combining this with \eqref{eq:app_strong_iid_full_covariance} and equality of the strong-symmetry subgroups proves (S).
This completes the proof of necessity.

We now prove sufficiency.
Assume (W) and (S), and set $H:=\mathrm{Sym}_G(\rho)$ and $S^{\mathrm{strong}}_{G,\rho,\sigma}:=\mathrm{Sym}_{G,\mathrm{strong}}(\rho)=\mathrm{Sym}_{G,\mathrm{strong}}(\sigma)$.
Condition (W) gives $H\subseteq\mathrm{Sym}_G(\sigma)$.

First, choose the characters carried by the strong symmetric auxiliary states.
By (S) and Lemma~\ref{lem:app_strong_iid_character_lift}, we can choose an infinite set $\mathcal N$ and a sequence of characters $\gamma_N$ satisfying \eqref{eq:app_strong_iid_auxiliary_character} and \eqref{eq:app_strong_iid_bounded_mean} for all sufficiently large $N\in\mathcal N$.
Fix this set and this sequence of characters throughout what follows.

Next, construct maps from the output reference system to the input reference system.
The metric adjusted quantum geometric tensor inequalities in (W) and covariance proportionality in (S) ensure that the remaining assumptions of Lemma~\ref{lem:app_strong_iid_reference_implementation} also hold.
The lemma provides, for each sufficiently large $N\in\mathcal N$, a completely positive trace-nonincreasing map
\eq{
\calD_N:\calB(E_\sigma^{\otimes M_N})\longrightarrow\calB(E_\rho^{\otimes N})
}
for which \eqref{eq:app_strong_iid_reference_Omega_convergence} gives
\eq{
\delta_N&:=\bigl\|\Omega_{\rho,N}-(\mathrm{id}\otimes\calD_N)
\Omega_{\sigma,M_N}^{\gamma_N}\bigr\|_1,\nonumber\\
\lim_{\substack{N\to\infty\\N\in\mathcal N}}\delta_N&=0
}
Here, $\mathrm{id}$ is the identity map on each irreducible representation space.
Thus, acting only on the reference system of the output correlation state with the auxiliary character attached approximates the input correlation state.

Finally, turn this reference-system approximation into a conversion on the physical system.
Apply Lemma~\ref{lem:app_strong_iid_reconstruction} with $\alpha=\rho^{\otimes N}$, $\beta=\sigma^{\otimes M_N}$, $\widehat{U}=U^{\otimes N}$, $\widehat{U}'=U'^{\otimes M_N}$, $\gamma=\gamma_N$, and $\calD=\calD_N$.
It gives a finite-dimensional auxiliary system with a strong symmetric pure state $\eta_{A_N}$ and a strong covariant CPTP map $\Phi_N$ defined on the full input space such that
\eq{
\|\Phi_N(\rho^{\otimes N})-\sigma^{\otimes M_N}\otimes\eta_{A_N}\|_1
\le2\sqrt{2\delta_N}\longrightarrow0
\qquad(N\to\infty,\ N\in\mathcal N)
}

For all remaining copy numbers, namely $N\notin\mathcal N$ and the finitely many $N\in\mathcal N$ to which the above construction does not apply, use the same reconstruction lemma with $\gamma=1$ and $\calD=0$.
In this case, $\delta=\|\Omega_{\rho,N}\|_1=1$, and the required strong covariant CPTP map and strong symmetric auxiliary state exist. When $M_N=0$, interpret zero copies as a one-dimensional system.
This defines the operations and auxiliary states for every $N\in\mathbb N$.
The error need not vanish outside $\mathcal N$, but the construction gives $\epsilon_N\to0$ along $\mathcal N$. Hence
\eq{
0\le\liminf_{N\to\infty}\epsilon_N
\le\lim_{\substack{N\to\infty\\N\in\mathcal N}}\epsilon_N=0.
}
Achievability therefore holds under the same $\liminf$ definition as in the main text.
The output copy number is $M_N=\lfloor rN\rfloor$ throughout, so the rate $r$ itself is achieved.
This proves sufficiency and completes the proof of the theorem.
\end{proofof}

As the necessity proof shows, the phase matching in \eqref{eq:strong_iid_character} must hold exactly for all sufficiently large $N$ along any subsequence on which the error vanishes.
Under the $\liminf$ definition, (W) and (S) allow us to choose a subsequence satisfying this equality. A definition requiring $\lim_{N\to\infty}\epsilon_N=0$, however, requires the equality for every sufficiently large integer $N$.
Applying the equivalence of (i)--(iii) in Lemma~\ref{lem:app_strong_iid_character_lift} and the sufficiency construction above with $\mathcal N=\mathbb N$ shows that the necessary and sufficient conditions for the latter definition are (W) and (S) together with this phase condition. Thus, the two definitions generally give different conversion criteria.

\end{document}